\documentclass[a4paper,USenglish,cleveref,thm-restate,numberwithinsect]{no-lipics-v2022}

\usepackage[utf8]{inputenc}
\usepackage{color}
\usepackage{csquotes}
\usepackage{xspace}
\usepackage{amsmath}
\usepackage{amssymb}
\usepackage{mathtools}
\usepackage{bbm}
\usepackage{subcaption}
\usepackage{graphicx}
\usepackage{MnSymbol}
\nolinenumbers

\usepackage{regexpatch}
\makeatletter
\xpatchcmd\thmt@restatable{%
\csname #2\@xa\endcsname\ifx\@nx#1\@nx\else[{#1}]\fi
}{%
\ifthmt@thisistheone
\csname #2\@xa\endcsname\ifx\@nx#1\@nx\else[{#1}]\fi
\else
\csname #2\@xa\endcsname[{Restated}]
\fi}{}{}
\makeatother

\newtheorem{hypothesis}[theorem]{Hypothesis}

\newcommand{\Oh}{\ensuremath{\mathcal{O}}\xspace}
\newcommand{\Ohtilde}{\ensuremath{\smash{\rlap{\raisebox{-0.1ex}{$\widetilde{\phantom{\Oh}}$}}\Oh}}\xspace}
\newcommand{\RR}{\mathbb{R}\xspace}
\newcommand{\NN}{\mathbb{N}\xspace}
\newcommand{\ZZ}{\mathbb{Z}\xspace}
\renewcommand{\emptyset}{\varnothing}
\newcommand{\ie}{i.e.\xspace}
\DeclareMathOperator*{\argmax}{arg\,max}
\DeclareMathOperator{\subw}{subw}
\DeclareMathOperator{\sharpw}{\#subw}
\newcommand{\mwidth}{\ensuremath{m\textup{-width}}\xspace}
\DeclareMathOperator{\clemb}{clemb}
\DeclareMathOperator{\wed}{wed}
\DeclareMathOperator{\tw}{tw}
\DeclareMathOperator{\poly}{poly}
\DeclareMathOperator{\polylog}{polylog}
\newcommand{\maybe}{}
\newcommand{\lo}{\textup{lo}}
\newcommand{\hi}{\textup{hi}}
\newcommand{\tOh}{\Ohtilde}
\newcommand{\eps}{\varepsilon}
\newcommand{\bu}{\overline{u}}
\newcommand{\bv}{\overline{v}}
\newcommand{\partenc}{\mathcal{E}}
\newcommand{\cliqueconj}{MinClique Hypothesis\xspace}
\newcommand{\cliquelistconj}{Clique Listing Hypothesis\xspace}
\newcommand{\zerocliqueconj}{ZeroClique Hypothesis\xspace}
\newcommand{\minplusconvconj}{MinConv Hypothesis\xspace}
\newcommand{\threesumconj}{3SUM Hypothesis\xspace}
\newcommand{\miniso}[1]{Min-Weight-$#1$-Subgraph\xspace}
\newcommand{\listiso}[1]{$#1$-Listing\xspace}
\newcommand{\enumiso}[1]{$#1$-Enumeration\xspace}
\newcommand{\enciso}[1]{$#1$-Encoding\xspace}
\newcommand{\countiso}[1]{$#1$-Counting\xspace}
\newcommand{\minH}{\miniso{H}}
\newcommand{\listH}{\listiso{H}}
\newcommand{\enumH}{\enumiso{H}}
\newcommand{\Hminiso}{\miniso{H}}
\newcommand{\Hlistiso}{\listiso{H}}
\newcommand{\Henumiso}{\enumiso{H}}
\newcommand{\Henciso}{\enciso{H}}
\newcommand{\Hcountiso}{\countiso{H}}
\newcommand{\minplusverif}{MinConv Verification\xspace}
\newcommand{\minplusverifconj}{MinConv Verification Hypothesis\xspace}
\newcommand{\threesum}{3SUM\xspace}
\newcommand{\threesumlist}{3SUM listing\xspace}
\newcommand{\family}{\mathcal{P}}
\newcommand{\Tfamily}{\mathcal{P}'}
\newcommand{\funcc}{\mathcal{F}}
\newcommand{\fabc}{\funcc(\alpha, \beta, \gamma)}
\newcommand{\pabc}{P(\alpha, \beta, \gamma \times 2)}
\newcommand{\goggr}{H_{\textup{gog}}}
\newcommand{\eemb}{f}
\newcommand{\corr}{I}
\newcommand{\jj}{\delta}
\newcommand{\jmax}{\jj_{\max}}
\newcommand{\vmax}{d_{\max}}
\newcommand{\edg}{\leftrightline}%
\newcommand{\dd}{.\,.}
\newcommand{\subp}[1]{V_{#1}^{\jj_{#1}}}
\newcommand{\bends}{B_{\textup{ends}}}
\newcommand{\inq}{(\textasteriskcentered)\xspace}
\newcommand{\inqq}{(\textasteriskcentered\textasteriskcentered)\xspace}
\newcommand{\prj}[2]{\left. #1 \right |_{#2}}%
\newcommand{\cC}{\mathcal{C}}
\newcommand{\subf}{\varphi}
\newcommand{\fwidth}{\subf\text{-}\mathsf{width}}
\newcommand{\funwidth}[1]{#1\text{-}\mathsf{width}}
\newcommand{\submodset}{\mathsf{S}}
\newcommand{\subdeg}{t}
\newcommand{\submax}{\subdeg_{\max}}
\newcommand{\casea}{\textbf{\boldmath$\fabc = 2 \beta \gamma + \frac{\alpha \beta}{2} - \frac{\beta^2}{2} + \frac{\beta}{2} - \frac{\alpha}{2} - 2 \gamma + 2$, $\alpha + \beta$ is even, $\alpha > \beta$, and $\beta < \gamma + 2$.}\xspace}
\newcommand{\caseb}{\textbf{\boldmath$\fabc = 2 \beta \gamma + \frac{\alpha \beta}{2} - \frac{\beta^2}{2} + \frac{3\beta}{2} - \frac{\alpha}{2} - 3 \gamma$, $\alpha + \beta$ is even, $3 \beta < \alpha + 6 \gamma + 8$, and ($\alpha = \beta$ or $\beta \ge \gamma + 2$).}\xspace}
\newcommand{\casec}{\textbf{\boldmath$\fabc = 2 \beta \gamma + \frac{\alpha \beta}{2} - \frac{\beta^2}{2} + 3\beta - \alpha - 6 \gamma - 4$, $\alpha + \beta$ is even, $2 \beta \le \alpha + 4 \gamma + 6$, and $3 \beta \ge \alpha + 6 \gamma + 8$.}\xspace}
\newcommand{\cased}{\textbf{\boldmath$\fabc = 2 \beta \gamma + \frac{\alpha \beta}{2} - \frac{\beta^2}{2} + \beta - \frac{\alpha}{2} - 2 \gamma + \frac{3}{2}$, $\alpha + \beta$ is odd, and $\beta < 2 \gamma + 3$.}\xspace}
\newcommand{\casee}{\textbf{\boldmath$\fabc = 2 \beta \gamma + \frac{\alpha \beta}{2} - \frac{\beta^2}{2} + 2 \beta - \frac{\alpha}{2} - 4 \gamma - \frac{3}{2}$, $\alpha + \beta$ is odd, $2 \beta \le \alpha + 4 \gamma + 6$, and $\beta \ge 2 \gamma + 3$.}\xspace}
\newcommand{\casef}{\textbf{\boldmath$\fabc = 2 \gamma^2 + \alpha \gamma + \frac{\alpha^2}{8} + \frac{\alpha}{2}$, $\alpha = 0 \bmod 4$, and $2 \beta > \alpha + 4 \gamma + 6$.}\xspace}
\newcommand{\caseg}{\textbf{\boldmath$\fabc = 2 \gamma^2 + \alpha \gamma + \frac{\alpha^2}{8} + \frac{\alpha}{2} + \frac{3}{8}$, $\alpha$ is odd, and $2 \beta > \alpha + 4 \gamma + 6$.}\xspace}
\newcommand{\caseh}{\textbf{\boldmath$\fabc = 2 \gamma^2 + \alpha \gamma + \frac{\alpha^2}{8} + \frac{\alpha}{2} + \frac{1}{2}$, $\alpha = 2 \bmod 4$, and $2 \beta > \alpha + 4 \gamma + 6$.}\xspace}

\title{A Fine-Grained Classification of Subquadratic Patterns for Subgraph Listing and Friends}

\titlerunning{A Fine-Grained Classification of Subquadratic Patterns for Subgraph Listing and Friends} %

\author{Karl Bringmann}{Saarland University and Max Planck Institute for Informatics, Saarland Informatics Campus, Germany}{bringmann@cs.uni-saarland.de}{https://orcid.org/0000-0003-1356-5177}{}%

\author{Egor Gorbachev}{Saarland University and Max Planck Institute for Informatics, Saarland Informatics Campus, Germany}{egorbachev@cs.uni-saarland.de}{https://orcid.org/0009-0005-5977-7986}{}%

\authorrunning{K. Bringmann and E. Gorbachev} %

\Copyright{Karl Bringmann and Egor Gorbachev} %


\category{} %

\relatedversion{} %

\acknowledgements{K.B. learned about the database perspective on the problems studied in this paper at the Simons Institute program ``Logic and Algorithms in Database Theory and AI''. He is particularly thankful for vivid discussions with Nofar Carmeli, Mahmoud Abo Khamis, Paris Koutris, Stefan Mengel, Hung Ngo, and Dan Suciu.}

\funding{This work is part of the project TIPEA that has received funding from the European Research Council (ERC) under the European Unions Horizon 2020 research and innovation programme (grant agreement No.\ 850979).}

\begin{document}

\maketitle



\begin{abstract}
    In an $m$-edge host graph $G$, all triangles can be listed in time $\Oh(m^{1.5})$ [Itai, Rodeh '78], and all $k$-cycles can be listed in time $\Oh(m^{2-1/{\lceil k/2 \rceil}} + t)$ where $t$ is the output size [Alon, Yuster, Zwick '97].
These classic results also hold for the colored problem variant, where the nodes of the host graph $G$ are colored by nodes in the pattern graph $H$, and we are only interested in subgraphs of $G$ that are isomorphic to the pattern $H$ and respect the colors. 
We study the problem of listing all $H$-subgraphs in the colored setting, for fixed pattern graphs $H$.

As our main result, we determine all pattern graphs $H$ such that all $H$-subgraphs can be listed in subquadratic time $\Oh(m^{2-\eps} + t)$, where $t$ is the output size. Moreover, for each such subquadratic pattern $H$ we determine the smallest exponent $c(H)$ such that all $H$-subgraphs can be listed in time $\Oh(m^{c(H)} + t)$. This is a vast generalization of the classic results on triangles and cycles. In particular, it answers an open problem from the database community [Joglekar, R\'e '18].

We also show the same results for two related problems: finding an $H$-subgraph of minimum total edge-weight in time $\Oh(m^{c(H)})$, and enumerating all $H$-subgraphs in $\Oh(m^{c(H)})$ preprocessing time and constant delay. Again we determine all pattern graphs $H$ that have complexity $c(H) < 2$, and for each such subquadratic pattern we determine the optimal complexity $c(H)$.

To prove these results, we design new algorithms and prove conditional lower bounds based on standard hypotheses from fine-grained complexity theory. 
For our algorithms we introduce a new ingredient that we call hyper-degree splitting, where we split tuples of nodes into high degree and low degree depending on their number of common neighbors. This technique generalizes the classic high-degree-low-degree approach (which only splits along degrees of single nodes) and is crucial to obtain optimal algorithms.
Furthermore, we contribute the simple but fundamental insight that to determine the optimal complexity of some natural families of algorithms and lower bounds it suffices to study graphs without clique separators (i.e., cliques whose removal disconnects the pattern graph). This insight immediately implies some existing results in the area and is crucial to obtain a complete classification.

\end{abstract}

\maketitle

\newpage
\thispagestyle{empty}
\tableofcontents
\pagenumbering{gobble}
\newpage
\setcounter{page}{1}
\pagenumbering{arabic}

\section{Introduction} \label{sec:introduction}

The subgraph isomorphism problem is one of the most fundamental graph problems: Given two graphs $H$ and $G$, decide whether the host graph $G$ contains a subgraph that is isomorphic to the pattern graph $H$. This problem and its variants have a large number of applications in areas such as databases, network motifs, statistical physics, and probabilistic inference. 
In this paper we study the setting where the pattern graph $H$ is fixed and thus of constant size, while the host graph~$G$ is given as the input, and we analyze running time in terms of the number of edges $m$ of $G$ (\ie, the input size) and occasionally its number of nodes $n$.

In 1978 Itai and Rodeh showed that triangle detection is in time $\Oh(m^{1.5})$~\cite{ItaiR78}. Alon, Yuster, and Zwick found a generalization to $k$-cycle detection in time $\Oh(m^{2-1/\lceil k/2 \rceil})$~\cite{AlonYZ97}. This can be improved using fast matrix multiplication, in particular triangle detection is in time $\Oh(m^{2\omega/(\omega+1)})$ where $\omega$ is the exponent of matrix multiplication~\cite{AlonYZ97} (see~\cite{DalirrooyfardVW21} for $k$-cycles). %
These works inspired a long line of research on various variants of subgraph isomorphism, including counting, listing, enumeration, and weighted problem variants, and variants with colored or uncolored host graph.

Our work is driven by the vision to \emph{obtain near-optimal algorithms, with matching fine-grained lower bounds, for as many pattern graphs as possible}. Naturally, any such attempt must focus on problem variants for which triangles and $k$-cycles are well understood.

In particular, we do not consider subgraph \emph{detection}, because already for triangle detection tight lower bounds are elusive and a big open problem in fine-grained complexity.\footnote{An $\Omega(m^{\omega/2})$ lower bound follows from the conjecture that triangle detection on dense graphs requires time $\Omega(n^\omega)$. An $m^{6/5 - o(1)}$ lower bound holds under the Strong 3SUM Hypothesis \cite[Theorem 1.14]{JinX23}. Both do not match the upper bound $\Oh(m^{2\omega/(\omega+1)})$, either for the current value of $\omega$ or for $\omega = 2$.} Instead, we study \emph{minimum-weight}, \emph{listing}, and \emph{enumeration} variants of subgraph isomorphism.

Similarly, we do not consider the \emph{uncolored} problem variant, where each node in $G$ could be mapped to any node in $H$, because $k$-cycles are open in this setting.\footnote{See \cite{AbboudBF23,JinX23,AbboudKLS23} for recent tight bounds on 4-cycle listing, and \cite{JinVWZ23} for the ongoing quest on 6-cycle listing; 8-cycle listing is widely open.}
Instead, we study the \emph{colored} variant of subgraph isomorphism, where each node of $G$ is colored by a node in~$H$, and we ask for subgraphs respecting these colors.
This variant naturally arises in database applications, where subgraph isomorphism corresponds to join queries and colors represent different columns of the database table. Naturally, the colored variant is extensively studied in the database community and beyond (see, e.g.,~\cite{AtseriasGM13,Khamis0S17,JoglekarR18,KhamisCMNNOS20,
FanKZ23}).

Finally, we study patterns with \emph{subquadratic complexity}, i.e., the dependence on $m$ is $\Oh(m^{2-\eps})$ for $\eps > 0$, because classifying all subquadratic patterns was posed as an open problem in the database community~\cite{JoglekarR18} and because this is already a rich class of patterns.

In summary, we study min-weight subgraph, subgraph listing, and subgraph enumeration in the colored setting for patterns with subquadratic complexity, as we discuss next.

\subsection{Min-Weight Subgraph}
\label{sec:intro_minweight}

In the \minH problem the host graph $G$ comes equipped with edge weights, and the task is to find the minimum total edge weight of any $H$-subgraph of~$G$.
Standard dynamic programming solves \minH in time $\Oh(m)$ when $H$ is a tree.
It is well-known that a generalization of Itai and Rodeh's triangle detection algorithm~\cite{ItaiR78} solves min-weight triangle in time $\Oh(m^{1.5})$, and a generalization of Alon, Yuster, and Zwick's algorithm~\cite{AlonYZ97} solves min-weight $k$-cycle in time $\Oh(m^{2-1/\lceil k/2 \rceil})$.
Min-weight subgraph problems are of key importance in fine-grained complexity (see, e.g., \cite{WilliamsW13,WilliamsW18,AbboudWY18,WilliamsX20,AbboudGW23,
BringmannS21,ChanX23,ChanWX23}), 
mainly due to the equivalence of All-Pairs-Shortest-Paths (APSP) and min-weight triangle~\cite{WilliamsW18}.\footnote{Note that many fine-grained hypotheses we refer to were recently refuted~\cite{AlmanW26}. The initial publication of our paper predated these results and therefore we keep it mostly unchanged. See \cref{subpar:new-note} for more details.} %

For any pattern graph $H$ let us define its \emph{min-weight complexity} $c_W(H)$ as the infimum over all $c \ge 1$ such that \minH can be solved in time $\Oh(m^c)$. The min-weight complexity of the triangle is completely understood: it is $1.5$ assuming the APSP hypothesis~\cite{ItaiR78,WilliamsW18}. We ask how far this understanding can be generalized to further pattern graphs. As discussed before, we focus on patterns with subquadratic complexity:

\begin{center}
\emph{Subquadratic Min-Weight Question: 
Determine all pattern graphs $H$ with min-weight complexity $c_W(H) < 2$, and for each such pattern determine the number $c_W(H)$.
}
\end{center}
In other words, the question asks to determine for every pattern $H$ the value $\min\{c_W(H),2\}$.
This number is known for some natural families of pattern graphs: the min-weight complexity of any tree is 1, of the $k$-cycle is $2-1/{\lceil k/2 \rceil}$~\cite{AlonYZ97,LincolnWW18,FanKZ23}, and of the biclique $K_{k,2}$ is $2 - 1/k$~\cite{FanKZ23,KhamisCMNNOS20}; see \cite{FanKZ23} for some further examples.
However, beyond some natural families of pattern graphs, in its generality the question was open. 

\smallskip
In this paper, we answer the Subquadratic Min-Weight Question, conditional on standard hypotheses from fine-grained complexity. Specifically, for every pattern graph $H$, we either prove a fine-grained lower bound showing that \minH is not solvable in time $\Oh(m^{2-\eps})$ for any $\eps > 0$, or we determine a value $c(H) < 2$ for which we design an $\Oh(m^{c(H)})$-time algorithm for \minH and prove a fine-grained lower bound showing that \minH is not solvable in time $\Oh(m^{c(H)-\eps})$ for any $\eps > 0$.
Our results imply a characterization of the subquadratic patterns by the forbidden induced minors shown in \Cref{fig:forbidden-induced-minors}, see \Cref{cor:alternative-characterization}. 
The description of the exact exponent $\min\{c_W(H),2\}$ requires several additional definitions, see \Cref{sec:formal-results}.

\subsection{Subgraph Listing}
\label{sec:intro_listing}

For a pattern graph $H$, the \listH problem asks to list all $H$-subgraphs of a given $m$-edge host graph $G$. Denote the total number of $H$-subgraphs by $t$, i.e., $t$ is the output size of the listing problem. 
A standard dynamic programming algorithm solves \listH in time $\Oh(m + t)$ when $H$ is a tree, see, e.g.,~\cite{BaganDG07}.
It is well-known that a generalization of Itai and Rodeh's triangle detection algorithm~\cite{ItaiR78} solves triangle listing in time $\Oh(m^{1.5} + t) = \Oh(m^{1.5})$, and a generalization of Alon, Yuster, and Zwick's algorithm~\cite{AlonYZ97} solves $k$-cycle listing in time $\Oh(m^{2-1/\lceil k/2 \rceil} + t)$. %
For triangle listing, further investigations~\cite{WilliamsW18,BjorklundPWZ14} led to an algorithm with running time\footnote{By $\tOh$-notation we hide poly-logarithmic factors in the input size $m$.} $\tOh(m^{2\omega/(\omega+1)} + m^{3(\omega-1)/(\omega+1)} t^{(3-\omega)/(\omega+1)})$~\cite{BjorklundPWZ14}. 
In a landmark result in fine-grained complexity, P\v{a}tra\c{s}cu proved that listing $t=m$ triangles in an $m$-edge graph cannot be solved in time $\Oh(m^{4/3-\eps})$ for any $\eps > 0$ assuming the 3SUM hypothesis~\cite{Patrascu10}, which matches the known algorithm~\cite{BjorklundPWZ14} if $\omega = 2$ (see also~\cite{KopelowitzPP16,WilliamsX20}). However, without assuming $\omega = 2$ no tight conditional lower bounds are known for listing, say, $t=m$ triangles; in this sense triangle listing is not well understood.

The situation clears up once we focus on listing algorithms that are \emph{output-linear}, i.e., whose running time depends (near-)linearly on $t$ (such algorithms are the most commonly studied listing algorithms in graph algorithms and database theory).
Specifically, for any pattern graph $H$ we define its \emph{listing complexity} $c_L(H)$ as the infimum over all $c \ge 1$ such that \listH can be solved in time $\tOh(m^c + t)$.
It is known that the triangle has listing complexity $1.5$, as Itai and Rodeh's algorithm has exponent $1.5$ and no algorithm can have exponent $c < 1.5$ assuming the 3SUM hypothesis\footnote{This lower bound follows from slightly changing parameters in \cite{Patrascu10}, or from plugging in appropriate parameters in \cite[Theorem 3.4]{WilliamsX20}, or as a special case of the results in this paper.}, so triangle listing is well understood in this setting.
We ask how far this understanding can be generalized to further pattern graphs. Again we focus on patterns with subquadratic complexity:

\begin{center}
\emph{Subquadratic Listing Question: 
Determine all pattern graphs $H$ with listing complexity $c_L(H) < 2$, and for each such pattern determine the number $c_L(H)$.
}
\end{center}
Characterizing all patterns with listing complexity $c_L(H) < 2$ was posed as an open problem in the database community by Joglekar and R{\'{e}}~\cite[Problem 1]{JoglekarR18}. Our question is even more fine-grained, since it asks for the precise value of $\min\{c_L(H),2\}$.
So far the number $c_L(H)$ is known only for the triangle and for trees (assuming standard fine-grained hypotheses).
Joglekar and R{\'{e}}~\cite{JoglekarR18} started an attack on this question, where they discovered additional pattern graphs with subquadratic listing complexity which also play an important role in our work.

\smallskip
In this paper, we answer the Subquadratic Listing Question, and thus in particular we answer the open problem by Joglekar and R{\'{e}}, conditional on standard fine-grained hypotheses. 
Specifically, for every pattern graph $H$, we either prove a fine-grained lower bound showing that \listH is not solvable in time $\tOh(m^{2-\eps} + t)$, or we determine a value $c(H) < 2$ for which we design an $\Oh(m^{c(H)} + t)$-time algorithm for \listH and prove a fine-grained lower bound showing that \listH is not solvable in time $\tOh(m^{c(H)-\eps} + t)$, for any $\eps > 0$.

\subsection{Subgraph Enumeration}
\label{sec:intro_enumeration}

An \enumH algorithm with preprocessing time $P(m)$ and delay $D(m)$ is an algorithm that lists all $H$-subgraphs of a given $m$-edge host graph $G$ such that the time spent by the algorithm before writing the first solution is at most $P(m)$, and for any $i$ the time spent between writing the $i$-th solution and the $(i+1)$-st solution is at most $D(m)$. Note that such an enumeration algorithm also yields an \listH algorithm that lists all $t$ solutions in total time $\Oh(P(m) + t \cdot D(m))$. Enumeration algorithms come with the added benefit that they can be aborted early and still yield some solutions, while a listing algorithm might not have produced any solutions at the same point in time. Enumeration is widely studied in the database theory community, and recently there is a growing interest in the fine-grained complexity of enumeration algorithms, see, e.g.,~\cite{DurandG07,BerkholzKS18,
Durand20,DeepHK20,CarmeliK20,CarmeliK21,
CarmeliZBCKS22,BringmannCM22,BringmannC22,CarmeliS23}, for more background on enumeration see the survey~\cite{Strozecki19}. 
Any pattern $H$ that is a tree can be enumerated with preprocessing time $\Oh(m)$ and delay $\Oh(1)$, see, e.g.,~\cite{BaganDG07}.
Trivially, the listing version of Itai and Rodeh's triangle detection algorithm~\cite{ItaiR78} is an enumeration algorithm with preprocessing time $\Oh(m^{1.5})$ and delay $\Oh(1)$.

We study enumeration algorithms with (near-)constant delay; for the same reasons as for focusing on output-linear listing algorithms. Accordingly, for any pattern graph $H$ we define its \emph{enumeration complexity} $c_E(H)$ as the infimum over all $c \ge 1$ such that \enumH can be solved in preprocessing time $\Oh(m^c)$ and delay $\tOh(1)$. We ask: %

\begin{center}
\emph{Subquadratic Enumeration Question: 
Determine all pattern graphs $H$ with enumeration complexity $c_E(H) < 2$, and for each such pattern determine the number $c_E(H)$.
}
\end{center}
In other words, this question asks to determine for every pattern $H$ the value $\min\{c_E(H),2\}$.
So far the number $c_E(H)$ is only known for the triangle and for trees (assuming standard fine-grained hypotheses).
Super-linear lower bounds for all non-trees were shown in~\cite{DurandG07,brault2013pertinence}.

\medskip
In this paper, we answer the Subquadratic Enumeration Question, conditional on standard fine-grained hypotheses, similarly as for the other two questions. 
In fact, for all three questions described above we obtain the same time complexities, showing that $\min\{c_W(H),2\} = \min\{c_L(H),2\} = \min\{c_E(H),2\}$ for all $H$.

\subsection{Formal Statement of Results}
\label{sec:formal-results}

We start by formally defining the problems considered in this paper.

\begin{definition}[Problem Definitions]
  A \emph{pattern graph} is a connected, undirected graph $H$ with $|V(H)| \ge 2$; in what follows we fix $H$. A \emph{host graph} is an undirected graph $G = (V,E)$ together with a partitioning $V = \bigcup_{a \in V(H)} V_a$. An \emph{$H$-subgraph of $G$} is a tuple of nodes $\bv = (v_a)_{a \in V(H)} \in \prod_{a \in V(H)} V_a$ such that for each $\{a,b\} \in E(H)$ we have $\{v_a, v_b\} \in E(G)$. 
  
  In the \emph{\minH} problem, given a host graph $G$ and edge weights $w \colon E(G) \to \mathbb{R}$, the task is to compute the smallest total edge weight of any $H$-subgraph of $G$, i.e., the weight of $\bv = (v_a)_{a \in V(H)}$ is $\sum_{\{a,b\} \in E(H)} w(v_a,v_b)$. If $G$ has no $H$-subgraph, the result is $\infty$. We assume that arithmetic operations on edge weights can be performed in constant time. 
  
  In the \emph{\listH} problem, the task is to output all $H$-subgraphs of a given host graph $G$. 
  
  In the \emph{\enumH} problem, the task is the same as in the \listH problem. However, instead of the total running time, we consider the \emph{preprocessing time}, which is the time spent by the algorithm before writing the first solution, as well as the \emph{delay}, which is the maximum time spent by the algorithm between writing two consecutive solutions.
  
  We denote by $m := |E(G)|$ the input size, and we write $n := |V(G)|$. In \listH, we write $t$ for the total number of $H$-subgraphs (i.e., the output size).
    \lipicsEnd
\end{definition}

Note that the assumption that the pattern $H$ is connected and has at least two nodes is without loss of generality. Indeed, if $H$ is disconnected, then the above problems easily split over the connected components of $H$. If $|V(H)| = 1$, then the problems are trivial.
We may further assume that there are no isolated nodes in $G$, because such nodes are not part of any $H$-subgraph of $G$ and thus can be removed. In particular, we assume $m = \Omega(n)$.

\medskip
To state our main result, we need two ingredients. The first ingredient is clique separators and the corresponding decomposition, a concept known in the literature on graph algorithms for decades but not applied in the context of subgraph isomorphism before. %

\begin{definition}[Clique Separator, \cite{Gavril77}]
    Let $H$ be a graph and $C \subsetneq V(H)$ (possibly empty). We say that $C$ is a \emph{clique separator} in $H$ if $H[C]$ is a clique and $H[V(H) \setminus C]$ is disconnected. Here, $H[C]$ denotes the subgraph of $H$ induced by the node set $C$.
    \lipicsEnd
\end{definition}

For example, the empty set is a clique separator of size 0 in a disconnected graph, and a cut vertex is a clique separator of size 1.
See \Cref{clique-separator-decomposition-of-subquadratic-graphs} for a more complicated example; clique separators of $H$ are given in brown/dotted.

\begin{definition}[Clique Separator Decomposition] \label{def:clique-separator-decomp}
  Let $H$ be a graph. We define $D(H)$ as the set of all induced subgraphs $H[V']$ such that $H[V']$ has no clique separator and $V' \subseteq V(H)$ is an inclusionwise maximal subset with this property.
    \lipicsEnd
\end{definition}

The following lemma makes it easy to determine the decomposition $D(H)$ for a given graph $H$: It suffices to split along any minimal clique separator and recursively construct the decomposition for each part.

\renewcommand{\maybe}{\lipicsEnd}
\begin{restatable}{lemma}{dsetcomputation}
\label{d-set-computation}
    Let $H$ be a graph and $C$ be a clique separator in $H$ which is minimal, i.e., no proper subset of $C$ is a clique separator. Let $S_1, S_2, \ldots, S_k \subset V(H)$ be the connected components of $H[V(H) \setminus C]$. Then $D(H) = \bigcup_{i \in [k]} D(H[S_i \cup C])$.\maybe
\end{restatable}
\renewcommand{\maybe}{}

We remark that combining this lemma with an algorithm to find clique separators~\cite{Whitesides81} allows to compute the decomposition $D(H)$ in time $\textup{poly}(|V(H)|)$, see \cref{sec:appendix-decomposition-algorithm}. (This is not needed for our main results because we anyways assume that $H$ has constant size.)

\begin{figure}
    \begin{center}
        \includegraphics[scale=0.55]{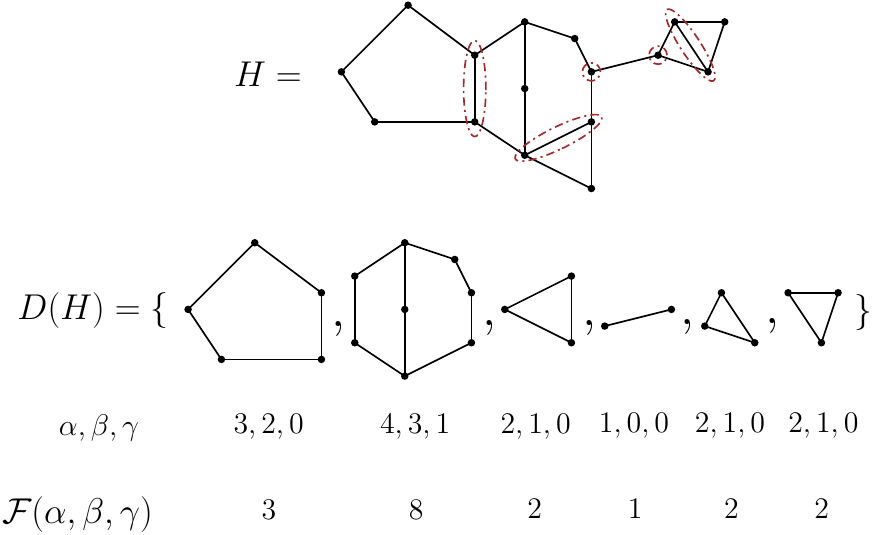}
    \end{center}

    \caption{Example application of our main result. All graphs in $D(H)$ are in $\family$, and the maximum value of $\funcc(\alpha,\beta,\gamma)$ over all parts is $8$, so we obtain complexity $c(H) = 2-1/8$. All minimal clique separators of $H$ are circled in brown.}
    \label{clique-separator-decomposition-of-subquadratic-graphs}
\end{figure}

\smallskip
The second ingredient is parallel path graphs, or $P$-graphs for short.

\begin{definition}[$P$-graphs]
    For positive integers $\ell_1, \ell_2, \ldots, \ell_k$, the \emph{$P$-graph} $P(\ell_1, \ell_2, \ldots, \ell_k)$ is constructed by starting from path graphs of lengths $\ell_1, \ell_2, \ldots, \ell_k$ and then identifying all their start nodes and identifying all their end nodes, see Figure \ref{p-example} for an example. Note that at most one of $\ell_1, \ell_2, \ldots, \ell_k$ can be~1 (as otherwise we would not obtain a simple graph).
    We often abbreviate $P(\ell_1, \ell_2, \ldots, \ell_k, 2, 2, \ldots, 2)$ with $\gamma$ many 2's as $P(\ell_1, \ell_2, \ldots, \ell_k, \gamma \times 2)$.
	To unify notation, we also define $P(\ell_1, \ell_2, 0 \times 2) := P(\ell_1, \ell_2)$ and $P(\ell_1, 0, 0 \times 2) := P(\ell_1)$.
    \lipicsEnd
\end{definition}

For example, the $P$-graph with one path $P(\ell_1)$ is the path with $\ell_1$ edges, the $P$-graph with two paths $P(\ell_1, \ell_2)$ is the cycle $C_{\ell_1 + \ell_2}$, and $P(k \times 2)$ is the bipartite clique $K_{k,2}$.

We are now ready to state our main result. We start by defining simple patterns $\family$ and then a function $\funcc$ that governs the time complexity of simple patterns.

\begin{figure}
\begin{center}
    \includegraphics[scale=0.41]{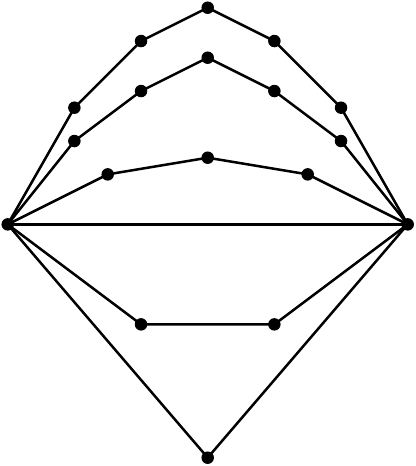}
\end{center}

\caption{A drawing of the $P$-graph $P(6, 6, 4, 1, 3, 2)$.}
\label{p-example}
\end{figure}

\begin{definition}[Simple Patterns] \label{def:family}
    Let $\family := \{ P(\alpha,\beta,\gamma \times 2) : (\alpha,\beta,\gamma) \in \Tfamily \}$ be the $P$-graphs corresponding to the set of triples $\Tfamily \subseteq \mathbb{Z}_{\ge 0}^3$ defined by 
	\begin{align*}
	\Tfamily := &\{ (\alpha, \beta, \gamma) : \alpha \ge \beta \ge 2, \gamma \ge 1 \}  \\ 
	&\cup \{ (k-2,2,0) : k \ge 4 \} & \text{($P(k-2,2,0 \times 2) = P(k-2,2)$ is the $k$-cycle $C_k$)} \\
	&\cup \{(1,0,0)\} & \text{($P(1,0,0 \times 2) = P(1)$ is the single edge $K_2$)} \\
	&\cup \{(2,1,0)\}. & \text{($P(2,1,0 \times 2) = P(2,1)$ is the triangle $K_3$)} 
	\end{align*}

\end{definition}

\begin{definition}[Savings Function] \label{def:funcc}
    For any $(\alpha, \beta, \gamma) \in \Tfamily$ we define $\funcc(\alpha, \beta, \gamma) = $

    \[             
            \begin{cases}
                2 \beta \gamma + \frac{\alpha \beta}{2} - \frac{\beta^2}{2} + \frac{\beta}{2} - \frac{\alpha}{2} - 2 \gamma + 2,\\
                \hfill\text{if $\alpha + \beta$ even, $\alpha > \beta$, and $\beta < \gamma + 2$;}\\
                2 \beta \gamma + \frac{\alpha \beta}{2} - \frac{\beta^2}{2} + \frac{3\beta}{2} - \frac{\alpha}{2} - 3 \gamma, \\
                \hfill\text{\;\;\;\;\;\;\;\;\;\;\;if $\alpha + \beta$ even, $3 \beta < \alpha + 6 \gamma + 8$, and ($\alpha = \beta$ or $\beta \ge \gamma + 2$);}\\
                2 \beta \gamma + \frac{\alpha \beta}{2} - \frac{\beta^2}{2} + 3\beta - \alpha - 6 \gamma - 4, \\
                \hfill\text{if $\alpha + \beta$ even, $2 \beta \le \alpha + 4 \gamma + 6$, and $3 \beta \ge \alpha + 6 \gamma + 8$;}\\
                2 \beta \gamma + \frac{\alpha \beta}{2} - \frac{\beta^2}{2} + \beta - \frac{\alpha}{2} - 2 \gamma + \frac{3}{2}, 
                \hfill\text{if $\alpha + \beta$ odd and $\beta < 2 \gamma + 3$;}\\
                2 \beta \gamma + \frac{\alpha \beta}{2} - \frac{\beta^2}{2} + 2 \beta - \frac{\alpha}{2} - 4 \gamma - \frac{3}{2}, \\
                \hfill\text{if $\alpha + \beta$ odd, $2 \beta \le \alpha + 4 \gamma + 6$, and $\beta \ge 2 \gamma + 3$;}\\
                2 \gamma^2 + \alpha \gamma + \frac{\alpha^2}{8} + \frac{\alpha}{2}, 
                \hfill\text{if $\alpha = 0 \bmod 4$ and $2 \beta > \alpha + 4 \gamma + 6$;}\\
                2 \gamma^2 + \alpha \gamma + \frac{\alpha^2}{8} + \frac{\alpha}{2} + \frac{3}{8}, 
                \hfill\text{if $\alpha$ odd and $2 \beta > \alpha + 4 \gamma + 6$;}\\
                2 \gamma^2 + \alpha \gamma + \frac{\alpha^2}{8} + \frac{\alpha}{2} + \frac{1}{2}, 
                \hfill\text{if $\alpha = 2 \bmod 4$ and $2 \beta > \alpha + 4 \gamma + 6$.}\\
            \end{cases}
    \]
\end{definition}

We verify that $\funcc$ is well-defined in \cref{lm:funcc-correctness}. The function values in the last three cases can be unified to $\lceil 2 \gamma^2 + \alpha \gamma + \frac{\alpha^2}{8} + \frac{\alpha}{2} \rceil$ (see Lemma~\ref{lm:funcc-last-cases}); we do not know comparable simplifications for the remaining cases. For an alternative way to define $\funcc$ see \cref{lm:alternative-definition-of-funcc}.

The following is our main result. For background on the hypotheses, see \cref{sec:prelims-hypotheses}.

\begin{theorem}[Main Result] \label{full-characterization}
    Let $H$ be a pattern graph. If there exists $H' \in D(H)$ such that $H'$ is not isomorphic to any graph in $\family$, then for any $\varepsilon > 0$
    \begin{itemize}
        \item \Hminiso cannot be solved in time $\Oh(m^{2 - \varepsilon})$ assuming the \cliqueconj and \minplusconvconj;
        \item \Hlistiso cannot be solved in time $\tOh(m^{2 - \varepsilon} + t)$ assuming the \zerocliqueconj and \threesumconj; and
        \item \Henumiso cannot be solved in $\Oh(m^{2 - \varepsilon})$ preprocessing time and $\tOh(1)$ delay assuming the \zerocliqueconj and \threesumconj.
    \end{itemize}
    Otherwise, every graph in $D(H)$ is isomorphic to a graph in $\family$, so up to isomorphism, $D(H) = \{P(\alpha_1, \beta_1, \gamma_1 \times 2), \ldots, P(\alpha_k, \beta_k, \gamma_k \times 2) \}$ for some triples $(\alpha_1, \beta_1, \gamma_1), \ldots, (\alpha_k, \beta_k, \gamma_k) \in \Tfamily$. Let $c(H) := 2 - 1 / \max_{i \in [k]} \funcc(\alpha_i, \beta_i, \gamma_i)$. Then for any $\varepsilon > 0$
    \begin{itemize}
        \item \Hminiso can be solved in time $\Oh(m^{c(H)})$, but not in time $\Oh(m^{c(H) - \varepsilon})$ assuming the \cliqueconj;
        \item \Hlistiso can be solved in time $\Oh(m^{c(H)} + t)$, but not in time $\tOh(m^{c(H) - \varepsilon} + t)$ assuming the \zerocliqueconj; and
        \item \Henumiso can be solved in preprocessing time $\Oh(m^{c(H)})$ and delay $\Oh(1)$, but not in preprocessing time $\Oh(m^{c(H) - \varepsilon})$ and delay $\tOh(1)$ assuming the \zerocliqueconj.\lipicsEnd
    \end{itemize}
\end{theorem}

Informally, \Cref{full-characterization} shows that all subquadratic patterns arise by stitching together graphs in~$\family$ along cliques, in a tree-like fashion. Actually, stitching along nodes and edges is sufficient, because no graph in $\family$ contains a triangle as a proper subgraph (see \cref{no-large-separators-needed}). The overall complexity is the maximum of $2 - 1/\funcc(\alpha,\beta,\gamma)$ over all stitched parts $P(\alpha, \beta, \gamma \times 2)$.

To help digest this theorem, we discuss some examples (with mostly known complexities):
\begin{itemize}
\item If $H$ is a tree, then $D(H)$ consists of all (subgraphs induced by) edges of $H$. So each graph in $D(H)$ is isomorphic to $K_2 = P(1) = P(1,0,0 \times 2) \in \family$. We have $\funcc(1,0,0) = 1$ by case 4 of \Cref{def:funcc}, so \Cref{full-characterization} yields complexity $c(H) = 2 - 1/1 = 1$.
\item The triangle $H = K_3$ has no clique separator, so $D(H) = \{H\} = \{P(2,1,0 \times 2)\}$. We have $\funcc(2,1,0) = 2$ from case 4 of \Cref{def:funcc}, so \Cref{full-characterization} yields complexity $c(K_3) = 1.5$.
\item The $k$-cycle $H = C_k$ for $k \ge 4$ also has no clique separator, so we have $D(H) = \{H\} = \{P(k-2,2,0 \times 2)\}$. For even $k$, case 2 of \Cref{def:funcc} applies and yields $\funcc(k-2,2,0) = k/2$. For odd $k$, case 4 applies and yields $\funcc(k-2,2,0) = (k+1)/2$. In either case we have $\funcc(k-2,2,0) = \lceil k/2 \rceil$, so we obtain complexity $c(C_k) = 2 - 1/\lceil k/2 \rceil$. 
\item The biclique $H = K_{k,2} = P(2,2,(k-2) \times 2)$ for $k \ge 3$ has no clique separator, so $D(H) = \{H\}$. We have $\funcc(2,2,k-2) = k$ by case 2 of \Cref{def:funcc}, so we obtain complexity $2 - 1/k$.
\item See \Cref{clique-separator-decomposition-of-subquadratic-graphs} for another non-trivial example.
\end{itemize}

In what follows we discuss several implications of our main result. The first is a characterization of all subquadratic patterns in terms of forbidden induced minors:

\begin{corollary}[Characterization by Forbidden Induced Minors] \label{cor:alternative-characterization}
    Let $H$ be a pattern graph. 
    If $H$ contains any of the three graphs from \Cref{fig:forbidden-induced-minors} as an induced minor then
    \begin{itemize}
        \item \Hminiso cannot be solved in time $\Oh(m^{2 - \varepsilon})$; %
        \item \Hlistiso cannot be solved in time $\tOh(m^{2 - \varepsilon} + t)$; and %
        \item \Henumiso cannot be solved in preprocessing time $\Oh(m^{2 - \varepsilon})$ and delay $\tOh(1)$, %
    \end{itemize}
    for any $\varepsilon > 0$ and assuming the \cliqueconj and \minplusconvconj.
    
    Otherwise, there exists an $\eps = \eps(H) > 0$ such that 
    \begin{itemize}
        \item \Hminiso can be solved in time $\Oh(m^{2 - \varepsilon})$;
        \item \Hlistiso can be solved in time $\Oh(m^{2 - \varepsilon} + t)$; and
        \item \Henumiso can be solved in preprocessing time $\Oh(m^{2 - \varepsilon})$ and delay $\Oh(1)$.\lipicsEnd
    \end{itemize}
\end{corollary}

    \begin{figure}
        \begin{center}
            \includegraphics[scale=0.8]{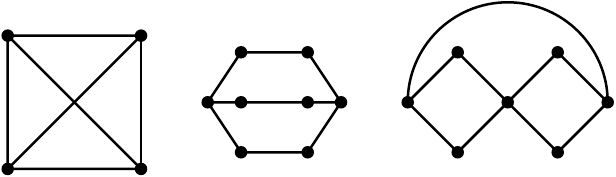}
        \end{center}

        \caption{Forbidden induced minors for subquadratic patterns as characterized by \Cref{cor:alternative-characterization}: $K_4$, $P(3, 3, 3)$, and the goggles graph $\goggr$ of \Cref{goggles-lb}.}
        \label{fig:forbidden-induced-minors}
    \end{figure}

Another surprising implication of our main result is that all subquadratic complexities are of the form $2-1/k$ for $k \in \mathbb{N}$. E.g., there is no pattern with complexity strictly between $1.5$ and $1.66$. (This follows from our main result and the fact that $\funcc(\alpha, \beta, \gamma)$ is a positive integer for every $(\alpha,\beta,\gamma) \in \Tfamily$, see Lemma~\ref{lm:funcc-is-positive-integer}.)

\smallskip
From our main result it is also easy to obtain an explicit description of all patterns with exponent $2-1/k$ for a given $k \in \mathbb{N}$. Indeed, we first determine all graphs from $\family$ attaining $\funcc(\alpha,\beta,\gamma) \le k$. Since every $(\alpha, \beta, \gamma) \in \Tfamily$ satisfies $\alpha \ge \beta$ and $\funcc(\alpha,\beta,\gamma) \ge \alpha/2 + \gamma$ (see \cref{F-gets-big}), the inequality $\funcc(\alpha,\beta,\gamma) \le k$ implies $\beta \le \alpha \le 2k$ and $\gamma \le k$. Hence, only a finite number of candidates in~$\Tfamily$ can attain value $k$, and we can calculate the value $\funcc(\alpha,\beta,\gamma)$ for each candidate. This is tedious to do by hand, but simple for a computer search. 
This search yields the following exhaustive list of patterns in $\family$ attaining $\funcc(\alpha,\beta,\gamma) = 1,2,\ldots,10$: 
\begin{enumerate}[1:]
  \item $K_2 = P(1, 0, 0 \times 2)$,
  \item $K_3 = P(2, 1, 0 \times 2)$, $C_4 = P(2, 2, 0 \times 2)$,
  \item $C_5 = P(3, 2, 0 \times 2)$, $C_6 = P(4, 2, 0 \times 2)$, $K_{3,2} = P(2, 2, 1 \times 2)$,
  \item $C_7 = P(5, 2, 0 \times 2)$, $C_8 = P(6, 2, 0 \times 2)$, $K_{4,2} = P(2, 2, 2 \times 2)$,
  \item $C_9 = P(7, 2, 0 \times 2)$, $C_{10} = P(8, 2, 0 \times 2)$, $K_{5,2} = P(2, 2, 3 \times 2)$, $P(3,2,1 \times 2)$, $P(4, 2, 1 \times 2)$,
  \item $C_{11} = P(9, 2, 0 \times 2)$, $C_{12} = P(10, 2, 0 \times 2)$, $K_{6,2} = P(2, 2, 4 \times 2)$, $P(5,2,1 \times 2)$, $P(6, 2, 1 \times 2)$, $P(3, 3, 1 \times 2)$,
  \item $C_{13} = P(11, 2, 0 \times 2)$, $C_{14} = P(12, 2, 0 \times 2)$, $K_{7,2} = P(2, 2, 5 \times 2)$, $P(7,2,1 \times 2)$, $P(8, 2, 1 \times 2)$, $P(3, 2, 2 \times 2)$, $P(4, 2, 2 \times 2)$,
  \item $C_{15} = P(13, 2, 0 \times 2)$, $C_{16} = P(14, 2, 0 \times 2)$, $K_{8,2} = P(2, 2, 6 \times 2)$, $P(9,2,1 \times 2)$, $P(10, 2, 1 \times 2)$, $P(5, 2, 2 \times 2)$, $P(6, 2, 2 \times 2)$, $P(4, 3, 1 \times 2)$, $P(5, 3, 1 \times 2)$,
  \item $C_{17} = P(15, 2, 0 \times 2)$, $C_{18} = P(16, 2, 0 \times 2)$, $K_{9,2} = P(2, 2, 7 \times 2)$, $P(11,2,1 \times 2)$, $P(12, 2, 1 \times 2)$, $P(7, 2, 2 \times 2)$, $P(8, 2, 2 \times 2)$, $P(3, 2, 3 \times 2)$, $P(4, 2, 3 \times 2)$, $P(3, 3, 2 \times 2)$, $P(4, 4, 1 \times 2)$,
  \item $C_{19} = P(17, 2, 0 \times 2)$, $C_{20} = P(18, 2, 0 \times 2)$, $K_{10,2} = P(2, 2, 8 \times 2)$, $P(13,2,1 \times 2)$, $P(14, 2, 1 \times 2)$, $P(9, 2, 2 \times 2)$, $P(10, 2, 2 \times 2)$, $P(5, 2, 3 \times 2)$, $P(6, 2, 3 \times 2)$, $P(6, 3, 1 \times 2)$, $P(7, 3, 1 \times 2)$.
\end{enumerate}
It is well known that all patterns with complexity $1 = 2-1/1$ are trees, i.e., they arise by stitching together any number of $K_2$ graphs in a tree-like fashion.
Our results show that there are no patterns with exponent strictly between $2-1/1$ and $2-1/2$, and all patterns with exponent $2-1/2$ arise by stitching together any number of $K_2, K_3$, and $C_4$ along nodes or edges in a tree-like fashion, using at least one $K_3$ or $C_4$. Similarly, there are no patterns with exponent strictly between $2-1/2$ and $2-1/3$, and all patterns with exponent $2-1/3$ arise by stitching together any number of $K_2, K_3, C_4, C_5, C_6$, and $K_{3,2}$ along nodes or edges in a tree-like fashion, using at least one $C_5, C_6$, or $K_{3,2}$. This continues; see the above list for exponents up to $2-1/10$. %
We find it surprising that we arrive at such an incredibly simple classification of all patterns with any fixed subquadratic complexity.
In \cref{sec:p-graph-subw} we also show a characterization of the complexity in terms of the submodular width $\subw(H)$~\cite{Marx13,Khamis0S17,KhamisNS24}, i.e., we prove that $\min\{\subw(H), 2\} = \min\{c(H), 2\}$ for all pattern graphs $H$.

\subsection{Discussion}

Our results are a vast generalization of prior results on min-weight, listing, and enumeration variants of subgraph isomorphism. 
It was well known that trees are exactly the patterns with linear complexity. 
We show the \emph{first results that extend this knowledge beyond the linear-time regime}, characterizing the class of patterns with complexity $c$ for any $1 < c < 2$.

Our framework is a unified approach to min-weight, listing, and enumeration, which allows us to easily generalize prior results in either of these settings to all three settings. However, even the union of prior results for these three settings was far from a complete understanding: While for the triangle, $k$-cycle, and single edge the complexity is long known, for the remaining graphs in $\family$ Joglekar and Ré~\cite{JoglekarR18} showed that their listing complexity is strictly below 2, but they did not achieve the optimal complexity (e.g., for $K_{k,2}$ their framework gives exponent $2 - 2^{1-k}$~\cite[Example 37]{JoglekarR18}, while the optimal exponent is $2 - 1/k$). We thus in particular had to (1) design improved algorithms for most patterns in $\family$, (2) prove matching fine-grained lower bounds, (3) show that stitching together patterns in $\family$ along nodes and edges keeps the complexity below 2, and (4) show that all subquadratic patterns can be constructed in this way.
Points (1) and (2) are particularly challenging, since the complexity of graphs in $\mathcal{P}$ is \emph{very complicated}, as can be seen from the 8 cases in \Cref{def:funcc}; naturally this also leads to many cases in our algorithms and lower bounds.
Since we prove matching upper and conditional lower bounds, we show that this complicated complexity is not just an artifact of our algorithms, but is actually the true complexity of these patterns (conditional on standard fine-grained hypotheses).

\subsection{Open Problems}
\label{sec:openproblems}

As discussed above, most choices for the setting of our main result are necessary, as other settings are already open for triangles or $k$-cycles (this is why we study colored instead of uncolored subgraph isomorphism, consider min-weight, listing, and enumeration instead of detection, and restrict to output-linear listing).

Let us further discuss the limitation to \emph{subquadratic complexity}. 
Subgraph isomorphism is intensely studied in the area of parameterized algorithms (see, e.g.,~\cite{Marx10,MarxP14,AminiFS12,FominLRSR12,Pratt19}).
In terms of the treewidth $\tw(H)$ a standard dynamic programming algorithm detects an $H$-subgraph in time $\tOh(n^{\tw(H)+1}) \le \tOh(m^{\tw(H)+1})$~\cite{AlonYZ95}.
However, the best conditional lower bound only rules out time $n^{o(\tw(H)/\log \tw(H))}$ (for any family of pattern graphs)~\cite{Marx10}, which leaves a super-constant gap of $\Omega(\log \tw(H))$ in the exponent. 
It is a big open problem in parameterized complexity to reduce this gap to a constant.
(A similar gap exists for minimum-weight, listing, and enumeration.)
Therefore, for pattern graphs $H$ with large treewidth, meaning high time complexity, it is impossible to determine the precise time complexity before answering a big open problem in parameterized complexity. 

It is therefore necessary to study patterns of bounded complexity. In addition, it is also \emph{natural} to study patterns of bounded complexity, because they are the most relevant for network motif search. Our choice to cap specifically at complexity 2 is mainly motivated by the open problem from the database community to characterize all subquadratic patterns~\cite{JoglekarR18}.

We leave as an open problem to extend our results beyond subquadratic complexity, i.e., to determine the numbers $\min\{c_W(H),\kappa\}$, $\min\{c_L(H),\kappa\}$, and $\min\{c_E(H),\kappa\}$ for all patterns $H$ for larger constants $\kappa > 2$. 
The natural next step in this direction is to classify all patterns with complexity equal to 2. This already seems difficult, because it includes patterns of treewidth 3 (while in this paper we can focus on patterns of treewidth 2).

\medskip
We also leave it open to find any significantly simpler proof of our main result, although this seems impossible due to the complicated nature of the optimal complexity, see \Cref{def:funcc}. %

\section{Proof Overview} \label{sec:proof-overview}

In this section we provide an overview of our tools and the structure of our proof. 
At the end of this section we will prove our main result, assuming the lemmas that we state throughout this section.
In later sections it then remains to prove the lemmas stated here.

\subsection{Lower Bounds}
\label{sec:proof-overview-lower-bounds}

For our lower bounds, we adopt the framework of clique embeddings from \cite{FanKZ23}.

\begin{definition}[Clique Embedding, \cite{FanKZ23}]
    Let $H$ be a graph. Sets $A,B \subseteq V(H)$ \emph{touch} in $H$ if $A \cap B \neq \emptyset$ or there exists an edge $\{a, b\} \in E(H)$ with $a \in A$ and $b \in B$.
    
    Let $H$ be a pattern graph and $k \in \mathbb{N}$. An \emph{embedding} from a clique $K_k$ to $H$ is a function~$\psi$ that maps every node $u \in V(K_k)$ to a non-empty subset $\psi(u) \subseteq V(H)$ such that $H[\psi(u)]$ is connected for every $u \in V(K_k)$, and $\psi(u)$ and $\psi(v)$ touch in $H$ for all $u,v \in V(K_k)$.

    For an embedding $\psi$, the \emph{weak depth} of an edge $e = \{a, b\} \in E(H)$ is $d_{\psi}(e) := |\{ u \in V(K_k) : \{a,b\} \cap \psi(u) \ne \emptyset \}|$. The \emph{weak edge depth} of the embedding $\psi$ is $\wed(\psi) := \max_{e \in E(H)} d_{\psi}(e)$. The \emph{clique embedding power} of $H$ is $\clemb(H) := \sup_{k \ge 3} \max_{\psi \colon K_k \to H} \frac{k}{\wed(\psi)}$, where the maximum is taken over all clique embeddings $\psi$ from $K_k$ to $H$.
    \lipicsEnd
\end{definition}

When describing clique embeddings, for $\psi(u) = B$ we say that node $u$ is embedded into the set~$B$. We will see that it is often advantageous to embed multiple nodes of $K_k$ into the same set $B$.
Since the nodes of a clique are indistinguishable, an embedding $\psi$ is given by a multiset $\{ \psi(u) : u \in V(K_k) \}$. In particular, it is not important \emph{which} nodes are mapped to a set $B$, but \emph{how many} nodes are mapped to $B$. Let us see an example of a clique embedding:

\begin{figure}
\begin{center}
    \centering
    \includegraphics[height=4cm]{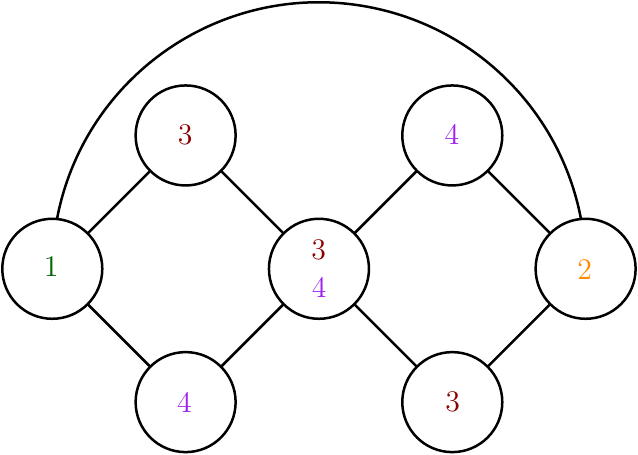}
\caption{Clique embedding for the goggles graph. Every number represents a subgraph into which one node is embedded.}
\label{goggles-lower-bound}
\end{center}
\end{figure}

\begin{lemma} \label{goggles-lb}
    We define the \emph{goggles graph} $\goggr$ by taking a $4$-cycle $(a_1, \ldots, a_4)$ and a $4$-cycle $(b_1, \ldots, b_4)$, identifying $a_1$ and $b_1$ and connecting $a_3$ and $b_3$ by an edge; see Figure~\ref{goggles-lower-bound}.
    
    We have $\clemb(\goggr) \ge 2$.
\end{lemma}

\begin{proof}
	Consider the clique embedding $\psi$ from $K_4$ to $\goggr$ that embeds one node into each of the sets $\{b_3\}$, $\{a_3\}$, $\{b_2, a_1, a_2\}$, and $\{b_4, a_1, a_4\}$, see Figure \ref{goggles-lower-bound}. It can be easily checked that $\psi$ is indeed a clique embedding (i.e., all chosen subgraphs are connected and all pairs of chosen subgraphs touch). It can also be checked that every edge of $\goggr$ sees at most two chosen subgraphs, so $\psi$ has weak edge depth $2$. This yields $\clemb(\goggr) \ge \frac{4}{2} = 2$.
\end{proof}

Fan, Koutris, and Zhao \cite{FanKZ23} used clique embeddings to prove conditional lower bounds for min-weight subgraph finding (and generalizations to hypergraphs). Here we show that clique embeddings also yield analogous lower bounds for listing and enumeration.

\begin{restatable}[Lower Bounds via $\clemb$]{lemma}{allHisocliquehardness} \label{all-H-iso-clique-hardness}
    For any pattern $H$ and constant $\varepsilon > 0$, \Hminiso cannot be solved in time $\Oh(m^{\clemb(H) - \varepsilon})$ assuming the \cliqueconj. \Hlistiso cannot be solved in time $\tOh(m^{\clemb(H) - \varepsilon} + t)$, and \Henumiso cannot be solved in preprocessing time $\Oh(m^{\clemb(H) - \varepsilon})$ and delay $\tOh(1)$, assuming the \zerocliqueconj.
\end{restatable}

\begin{proof}[Proof Sketch]
	For min-weight this was shown in \cite[Proposition 12]{FanKZ23}.
	For listing, there is a known reduction from detecting a triangle of total edge weight zero to triangle listing~\cite{WilliamsX20}. We adapt this reduction to start from detecting a $k$-clique of total edge weight zero and end at $H$-Listing.
	For enumeration, we use a simple reduction from listing.
\end{proof}

We next show that induced minors have smaller or equal time complexity for min-weight, listing, and enumeration. The same holds for the clique embedding power.

\renewcommand{\maybe}{\lipicsEnd}
\begin{restatable}[Induced Minors for Subgraph Problems]{lemma}{inducedminor} \label{induced-minor}
    Let $H$ and $H'$ be patterns such that $H'$ is an induced minor of~$H$. Then for any functions $T,P,D$,
    \begin{itemize}
        \item if \Hminiso can be solved in time $T(m)$, then \miniso{H'} can be solved in time $\Oh(T(m))$,
        \item if \Hlistiso can be solved in time $T(m, t)$, then \listiso{H'} can be solved in time $\Oh(T(m, t))$,
        \item if \Henumiso can be solved in preprocessing time $P(m)$ and delay $D(m)$, then \enumiso{H'} can be solved in preprocessing time $\Oh(P(m))$ and delay $\Oh(D(m))$.\maybe
    \end{itemize}
\end{restatable}
\renewcommand{\maybe}{}

\renewcommand{\maybe}{\lipicsEnd}
\begin{restatable}[Induced Minors for clemb]{lemma}{clembofinducedminors}  \label{clemb-of-induced-minors}
    Let $H$ and $H'$ be patterns such that $H'$ is an induced minor of~$H$. Then $\clemb(H) \ge \clemb(H')$.\maybe
\end{restatable}
\renewcommand{\maybe}{}

\smallskip
Thus, any lower bound for an induced minor $H'$ also gives a lower bound for~$H$. 
In~particular, it turns out that it suffices to prove lower bounds for patterns without clique separators, since all graphs in the decomposition $D(H)$ have this property and are induced minors of $H$.

We next show that any pattern without clique separator that is not a $P$-graph has clique embedding power at least 2, and thus complexity at least 2.

\begin{restatable}{lemma}{smallclembimpliesp} \label{small-clemb-implies-p}
	Let $H$ be a pattern without clique separator that is not a $P$-graph (i.e., $H$ is not isomorphic to $P(\ell_1,\ldots,\ell_k)$ for any integers $\ell_1,\ldots,\ell_k \ge 1$). Then $H$ has $\goggr$ or $K_4$ as an induced minor. Furthermore, we have $\clemb(H) \ge 2$.
\end{restatable}

\begin{proof}[Proof Sketch]
	For a graph that has no clique separator and thus no cut vertex, it is known that $H$ has no $K_4$ (induced) minor if and only if $H$ has treewidth at most 2 if and only if $H$ is a series-parallel graph. By an intricate graph theory argument we show that a series-parallel graph with no clique separator and no $\goggr$ induced minor is a $P$-graph, thus leading to a contradiction.
    If $H$ has an $\goggr$ induced minor then $\clemb(H) \ge 2$ (by Lemmas~\ref{goggles-lb} and~\ref{clemb-of-induced-minors}).
    Similarly, if $H$ has an (induced) $K_4$ minor then it is easy to see that $\clemb(H) \ge 2$.
\end{proof}

Not all $P$-graphs admit truly subquadratic algorithms, but surprisingly lower bounds based on clique embeddings are not sufficient to show this: We have $\clemb(P(3, 3, 3)) = 2 - 1/9$ according to~\cite[Section 7.1]{FanKZ23}, but \listiso{P(3, 3, 3)} cannot be solved in time $\tOh(m^{2-\eps} + t)$ for any $\eps > 0$ assuming the 3XOR hypothesis according to~\cite[Lemma 35]{JoglekarR18}. 
In this paper, we replace the 3XOR hypothesis by the more standard \threesumconj, and we prove analogous lower bounds for \miniso{P(3, 3, 3)} and \enumiso{P(3, 3, 3)}.

\renewcommand{\maybe}{\lipicsEnd}
\begin{restatable}[Lower bounds for P(3,3,3)]{lemma}{allpthreethreethreeisocliquehardness} \label{all-p333-iso-clique-hardness}
    Let $\varepsilon > 0$. \miniso{P(3, 3, 3)} cannot be solved in time $\Oh(m^{2 - \varepsilon})$ assuming the \minplusconvconj. \listiso{P(3, 3, 3)} cannot be solved in time $\tOh(m^{2 - \varepsilon} + t)$, and \enumiso{P(3, 3, 3)} cannot be solved in preprocessing time $\Oh(m^{2 - \varepsilon})$ and delay $\tOh(1)$, assuming the \threesumconj.
    \maybe
\end{restatable}
\renewcommand{\maybe}{}

Thus, all $P$-graphs that have at least three paths of length at least three do not admit truly subquadratic algorithms. 
We argue that every $P$-graph that does not have at least three paths of length at least three and has no clique separator is one of our simple patterns in $\family$ (recall \Cref{def:family}). All of these patterns indeed have complexity below 2.

It remains to prove lower bounds for graphs in $\family$. Here clique embeddings are sufficient.

\begin{restatable}{lemma}{pgraphlowerbound} \label{lm:p-graph-lower-bound}
    For any $(\alpha, \beta, \gamma) \in \Tfamily$, $\clemb(P(\alpha, \beta, \gamma \times 2)) \ge 2 - 1 / \funcc(\alpha, \beta, \gamma)$.
    \maybe
\end{restatable}
\begin{proof}[Proof Sketch]
	In \Cref{sec:p922-example} we present an example of this lower bound, namely a clique embedding for the concrete pattern $P(9,2,2)$.
	For the full proof of the lemma, we present several families of clique embeddings that together cover all graphs in $\family$. Each family is easy to verify, but it was difficult to find a complete set of families, matching our algorithm for each pattern in $\family$. The number of families is quite large, making our proof quite long with many cases --- this is likely to be necessary due to the complicated nature of the optimal complexity $\funcc$, see \Cref{def:funcc}.
\end{proof}

\noindent
In \Cref{sec:proof_main_result} we show that the above lemmas easily imply the lower bounds of our main~result.

\subsection{Algorithms}
\label{sec:proof-overview-algorithms}

To treat min-weight, listing, and enumeration in a unified way, we introduce the following natural intermediate problem that we call \Henciso. See \cref{sec:prelims-treedecomp} for background on tree decompositions. 

\begin{definition}[Encoding Problem]
	Let $H$ be a pattern and $G = (V,E)$ be a host graph with partitioning $V = \bigcup_{a \in V(H)} V_a$. A \emph{partial $H$-encoding} $\partenc$ of $G$ is a tree decomposition $\mathcal{T}$ of~$H$ and for every bag $B$ of $\mathcal{T}$ a set $S_B \subseteq \prod_{a \in B} V_a$ (which we call \emph{submaterialization}). We say that an $H$-subgraph $\bv$ of $G$ is \emph{encoded} by $\partenc$ if we have $\prj{\bv}{B} \in S_B$ for every bag $B$ of $\mathcal{T}$. The \emph{size} of $\partenc$ is the total size of all sets $S_B$, summed over all bags $B$.

    A (full) \emph{$H$-encoding} of $G$ is a %
    collection of partial $H$-encodings $\partenc_1, \ldots, \partenc_k$ such that every $H$-subgraph $\bv$ of $G$ is encoded by exactly one $\partenc_i$. We require that the number of partial encodings $k$ is bounded by a constant (i.e., it may depend on $H$ but not on $G$).
    The \emph{size} of a full $H$-encoding is the sum of all sizes of its partial encodings $\partenc_1, \ldots, \partenc_k$.

    The \emph{\Henciso problem} asks to compute a full $H$-encoding for given host graph~$G$.
    \lipicsEnd
\end{definition}

Similar problems have appeared in the literature, e.g., the ``Boolean tensor decomposition'' from~\cite{KhamisNOS19} (which differs in requiring that every $H$-subgraph is encoded by \emph{at least} one~$\partenc_i$).
Many existing algorithms for subgraph isomorphism problems use a variation of the encoding problem either explicitly or implicitly, see, e.g.,~\cite{AlonYZ97,Khamis0S17,JoglekarR18,KhamisCMNNOS20,KhamisNS24}.
We illustrate the encoding problem by the following example.

\begin{example}[4-Cycle Detection by Encoding]
    Consider the 4-cycle $H = C_4$.
    We show how the 4-cycle detection algorithm of~\cite{AlonYZ97} can be interpreted as an algorithm solving \Henciso in time $\Oh(m^{1.5})$.
    Let $a_1, a_2, a_3, a_4$ be the nodes of $H$ in cyclic order.
    Let $G$ be a given host graph. We call a node of $G$ a high-degree node if its degree is at least $\sqrt m$, and a low-degree node otherwise.
    Note that there are at most $\sqrt m$ high-degree nodes in each part of $G$.
    
    First, we create a partial encoding of all $H$-subgraphs of $G$ in which the chosen node in $V_{a_1}$ is a high-degree node.
    To this end, we consider the tree decomposition $\{a_1, a_2, a_3\} \edg \{a_1, a_3, a_4\}$ of $H$.
    We create a submaterialization of $\{a_1, a_2, a_3\}$ by listing all pairs of a high-degree node from $V_{a_1}$ and an edge between $V_{a_2}$ and $V_{a_3}$.
    This can be done in time $\Oh(m^{1.5})$ as there are at most $\sqrt m$ high-degree nodes in $V_{a_1}$ and at most $m$ edges between $V_{a_2}$ and $V_{a_3}$.
    We create a submaterialization of $\{a_1, a_3, a_4\}$ analogously. This finishes this partial encoding.
    We can now delete high-degree nodes from $V_{a_1}$.
    
    Next, we analogously create a partial encoding of all $H$-subgraphs of $G$ in which the chosen node in $V_{a_1}$ is a low-degree node and the chosen node in $V_{a_3}$ is a high-degree node.
    After that, we can delete all high-degree nodes from $V_{a_3}$.
    
    It remains to create a partial encoding of all $H$-subgraphs of $G$ in which the chosen nodes in $V_{a_1}$ and $V_{a_3}$ are low-degree nodes.
    To this end, we consider the tree decomposition $\{a_1, a_2, a_4\} \edg \{a_2, a_3, a_4\}$ of $H$.
    We create a submaterialization of $\{a_1, a_2, a_4\}$ by listing all edges between $V_{a_1}$ and $V_{a_2}$, and for each such edge generating all possible neighbors of the chosen low-degree node $v_{a_1} \in V_{a_1}$ in $V_{a_4}$.
    This can be done in time $\Oh(m^{1.5})$ as there are at most $m$ edges between $V_{a_1}$ and $V_{a_2}$, and for a low-degree node $v_{a_1}$ there are at most $\sqrt m$ neighbors in $V_{a_4}$.
    We create a submaterialization of $\{a_2, a_3, a_4\}$ analogously.
    
    One can easily verify that these three partial $H$-encodings constitute a full $H$-encoding of the initial host graph $G$.
    The described procedure mimics the 4-cycle detection algorithm of~\cite{AlonYZ97}. %
    \lipicsEnd
\end{example}

\smallskip
\Henciso is a unifying problem by means of the following reduction lemma.

\renewcommand{\maybe}{\lipicsEnd}
\begin{restatable}{lemma}{allreductions} \label{all-reductions}
    For any pattern $H$ if \Henciso is solvable in time $T(m)$ then 
    \begin{itemize}
    \item \Hminiso can be solved in time $\Oh(T(m))$,
    \item \Hlistiso can be solved in time $\Oh(T(m) + t)$, where $t$ is the output size, and
    \item \Henumiso can be solved in preprocessing time $\Oh(T(m))$ and delay $\Oh(1)$.\maybe
    \end{itemize}
\end{restatable}
\renewcommand{\maybe}{}

\smallskip
Hence, from now of we focus on designing algorithms for \Henciso.
We first show that the complexity of Encoding is smaller or equal for induced minors, as for our main problems.

\renewcommand{\maybe}{\lipicsEnd}
\begin{restatable}[Induced Minors for Encoding]{lemma}{inducedminorencoding} \label{induced-minor-encoding}
    Let $H$ and $H'$ be patterns such that $H'$ is an induced minor of $H$. If \Henciso can be solved in time $T(m)$, then \enciso{H'} can be solved in time $\Oh(T(m))$.\maybe
\end{restatable}
\renewcommand{\maybe}{}

\smallskip
A crucial property of Encoding is that it splits over a clique separator, i.e., we can combine algorithms for the parts induced by a clique separator to an algorithm for the full graph:

\begin{restatable}[Encoding Splits Over Clique Separator]{lemma}{cliqueseparatoralgocombination} \label{clique-separator-algo-combination}
    Let $H$ be a pattern and $C$ be a clique separator in $H$. Let $S_1, \ldots, S_k \subset V(H)$ be the connected components of $H - C$. If \enciso{H[S_i \cup C]} can be solved in time $T(m)$ for every $i \in [k]$, then \Henciso can be solved in time $\Oh(T(m))$.
\end{restatable}

\begin{proof}[Proof Sketch]
	We solve \enciso{H[S_i \cup C]} on the corresponding parts of $G$ and combine the resulting $H[S_i \cup C]$-encodings to an $H$-encoding of $G$. To this end, for each combination of partial $H[S_i \cup C]$-encodings we connect the corresponding tree decompositions of the subgraphs $H[S_i \cup C]$ to a tree decomposition of~$H$. This is always possible because $C$ is a clique in $H[S_i \cup C]$ and thus every tree decomposition of $H[S_i \cup C]$ has a bag containing~$C$; connecting these bags yields a tree decomposition of $H$.
\end{proof}

\renewcommand{\maybe}{\lipicsEnd}
\begin{restatable}{corollary}{tcfromD} \label{tc-from-D}
    For any pattern $H$, if \enciso{H'} can be solved in time $T(m)$ for every $H' \in D(H)$, then \Henciso can be solved in time $\Oh(T(m))$.\maybe
\end{restatable}
\renewcommand{\maybe}{}

Note that \Cref{induced-minor-encoding} trivially implies the converse (i.e., if \Henciso can be solved in time $T(m)$ then \enciso{H'} can be solved in time $\Oh(T(m))$ for every $H' \in D(H)$). Hence, \Cref{tc-from-D} is a characterization of the optimal running time for \Henciso as the maximum of the optimal running times for \enciso{H'} over $H' \in D(H)$. 
It follows that a near-optimal algorithm for the Encoding problem on all patterns without clique separators yields a near-optimal algorithm for Encoding on all patterns. It thus suffices to study patterns without clique separators, without loss of generality. 
This is a simple but powerful insight that is crucial in our work to obtain a full classification of all subquadratic patterns. It also significantly simplifies some prior work, e.g., \cite[Theorem 24]{FanKZ23} investigates the complexity of chordal patterns, but since the only chordal graphs without clique separator are cliques, their results follow immediately from~\Cref{tc-from-D} and algorithms for cliques. 

We believe that clique separators could be a crucial ingredient for understanding the optimal complexity of subgraph isomorphism in many other contexts. As a case in point, we demonstrate in \cref{sec:subw-general-case-to-p-graphs} that an analogous property holds for submodular width, which is another algorithmic framework for subgraph isomorphism.

\medskip
In our context, if $D(H)$ contains a graph not in $\family$, then we know that the complexity is at least $2$. It thus only remains to prove upper bounds on the complexity of graphs in $\family$.

\begin{restatable}{lemma}{pgraphupperbound} \label{lm:p-graph-upper-bound}
    For any $(\alpha, \beta, \gamma) \in \Tfamily$, \enciso{P(\alpha, \beta, \gamma \times 2)} can be solved in time $\Oh(m^{2 - 1 / \funcc(\alpha, \beta, \gamma)})$.\maybe
\end{restatable}

\begin{proof}[Proof Sketch]
	Algorithms for subgraph finding often follow the high-degree-low-degree idea of Alon, Yuster, and Zwick~\cite{AlonYZ97}: First, the nodes of the host graph are split into high degree or low degree according to some threshold (known as degree splitting). More generally, one can partition the nodes into $\log n$ parts depending on the interval $[2^i,2^{i+1})$ that their degree lies in (known as degree uniformization). It remains to solve subgraph finding with certain degree constraints on the host graph, which is solved by algorithmic ideas exploiting that small-degree nodes have few neighbors and that there are few high-degree nodes. 
	
	This high-degree-low-degree idea does not suffice to obtain optimal algorithms for all patterns in~$\family$. Our new ingredient to go beyond this standard approach is a second level of high-degree-low-degree that we call \emph{hyper-degree splitting}: After the usual degree splitting, in some cases we can afford to enumerate a set of possible tuples $T \subseteq V_1 \times \cdots \times V_\ell$, i.e., we can compute a superset $T$ of the projection of all $H$-subgraphs to $V_1 \times \cdots \times V_\ell$. We then consider the hyper-degree $\deg (v_1,\ldots,v_{\ell-1}) = |\{v_\ell : (v_1,\ldots,v_\ell) \in T\}|$, and we split the tuples $(v_1,\ldots,v_{\ell-1})$ according to this hyper-degree into high degree and low degree.\footnote{In \cite{BjorklundPWZ14} a similar concept was introduced. They call an edge $\Lambda$-light if it belongs to at most $\Lambda$ triangles.}
	
	In \Cref{sec:p922-example} we present our algorithm for the concrete example $P(9,2,2)$. This shows our hyper-degree splitting in action, on an example where it is necessary.
	
	For the full proof of the lemma, we again need to present several families of algorithms that together cover all graphs in $\family$. Each family is easy to verify, but it was difficult to find a complete set of families. As in our lower bound, the number of algorithm families is quite large, so the proof is long with many cases --- this is likely to be necessary due to the complicated nature of the optimal complexity~$\funcc$, see \Cref{def:funcc}.
\end{proof}

\subsection{Proof of the Main Result}
\label{sec:proof_main_result}

We are ready to prove our main result, assuming the lemmas stated in the last two sections.

\begin{proof}[Proof of \Cref{full-characterization}]
	We start by showing that every $H' \in D(H)$ is a pattern, i.e., $H'$ is connected and $|V(H')| \ge 2$. 
	The first property follows from the facts that any disconnected graph has a clique separator of size zero, and graphs in $D(H)$ have no clique separator.
	For the second property, suppose for the sake of contradiction that $|V(H')| = 1$, and let $u \in V(H)$ be such that $H' = H[\{u\}]$.
	Since $H$ is connected and has at least two nodes, there exists a neighbor $v$ of $u$ in $H$. Since $H[\{u,v\}]$ is a single edge and thus has no clique separator, $H' = H[\{u\}]$ is no maximal induced subgraph of $H$ without clique separator, which contradicts $H' \in D(H)$.
	Hence, each $H' \in D(H)$ is a pattern.

	\smallskip
	For the first part of the theorem statement, assume that there is some $H' \in D(H)$ that is not isomorphic to any graph in $\family$. If $\clemb(H') \ge 2$, then $\clemb(H) \ge 2$ by \Cref{clemb-of-induced-minors}, and thus we get the desired quadratic lower bounds by \Cref{all-H-iso-clique-hardness}.
	So we can assume $\clemb(H') < 2$. Note that $H'$ has no clique separator, as it is an element of $D(H)$. Thus, \Cref{small-clemb-implies-p} implies that $H'$ is a $P$-graph, i.e, $H'$ is isomorphic to $P(\ell_1,\ldots,\ell_k)$ for some integers $\ell_1,\ldots,\ell_k \ge 1$.
	If $k = 1$, then $H'$ is isomorphic to the path $P(\ell_1)$. If $\ell_1 > 1$ then this graph has a clique separator, contradicting $H' \in D(H)$, and if $\ell_1 = 1$ then $H'$ is isomorphic to $P(1) = P(1, 0, 0 \times 2) \in \family$, contradicting that $H'$ is not isomorphic to any graph in $\family$.
	If $k = 2$, then $H'$ is isomorphic to $P(\ell_1,\ell_2)$, and thus to the cycle of length $\ell := \ell_1 + \ell_2$. If $\ell = 3$ then $H'$ is isomorphic to $P(2,1) = P(2,1,0 \times 2) \in \family$, and if $\ell \ge 4$ then $H'$ is isomorphic to $P(\ell-2,2) = P(\ell-2,2,0 \times 2) \in \family$; both contradict that $H'$ is not isomorphic to any graph in~$\family$. 
	If $k \ge 3$, then none of $\ell_1,\ldots,\ell_k$ can be 1, as the path of length 1 would be a size-2 clique separator of $H'$. If less than three of $\ell_1,\ldots,\ell_k$ are at least three, then $H'$ is isomorphic to $P(\alpha,\beta,\gamma \times 2)$ for some $\alpha \ge \beta \ge 2$ and $\gamma \ge 1$, which is in $\family$, again contradicting the choice of $H'$. 
	It follows that $H'$ has at least three paths of length at least three, so $H'$ (and thus~$H$) has $P(3, 3, 3)$ as an induced minor. We now obtain the desired quadratic lower bounds by \Cref{all-p333-iso-clique-hardness,,induced-minor}.

	\smallskip
	For the second part of the theorem, assume that every $H' \in D(H)$ is isomorphic to some pattern graph in $\family$, i.e., up to isomorphism, $D(H) = \{P(\alpha_1, \beta_1, \gamma_1 \times 2), \ldots, P(\alpha_k, \beta_k, \gamma_k \times 2) \}$ for some triples $(\alpha_1, \beta_1, \gamma_1), \ldots, (\alpha_k, \beta_k, \gamma_k) \in \Tfamily$.
    Let $c(H) := 2 - 1 / \max_{i \in [k]} \funcc(\alpha_i, \beta_i, \gamma_i)$. By \Cref{lm:p-graph-upper-bound}, \enciso{P(\alpha_i, \beta_i, \gamma_i \times 2)} can be solved in time $\Oh(m^{2 - 1 / \funcc(\alpha_i, \beta_i, \gamma_i)}) \le \Oh(m^{c(H)})$ for each~$i$.
    Hence, by \Cref{tc-from-D}, \Henciso can be solved in $\Oh(m^{c(H)})$ time.
    \Cref{all-reductions} now yields the desired upper bounds.
	For the lower bounds, by \Cref{lm:p-graph-lower-bound} we have $\clemb(P(\alpha_i, \beta_i, \gamma_i \times 2)) \ge 2 - 1 / \funcc(\alpha_i, \beta_i, \gamma_i)$ for each $i$.
    Since each of the $P(\alpha_i, \beta_i, \gamma_i \times 2)$ is an induced subgraph of $H$, by \Cref{clemb-of-induced-minors} we obtain $\clemb(H) \ge \max_i \clemb(P(\alpha_i, \beta_i, \gamma_i \times 2)) \ge 2 - 1/\max_i \funcc(\alpha_i, \beta_i, \gamma_i) = c(H)$.
    The desired lower bounds now follow from \Cref{all-H-iso-clique-hardness}.
\end{proof}

We now show how \cref{cor:alternative-characterization} follows from \cref{full-characterization}.

\begin{proof}[Proof of \cref{cor:alternative-characterization}]
    For the first part of the corollary statement, assume that $H$ contains $K_4, P(3,3,3)$ or $\goggr$ as an induced minor.
	If $H$ has an induced $\goggr$ minor, then we have $\clemb(H) \ge \clemb(\goggr) \ge 2$ by \cref{clemb-of-induced-minors,goggles-lb}, so \cref{all-H-iso-clique-hardness} implies the desired lower bounds.
	If $H$ has an induced $K_4$ minor, then we have $\clemb(H) \ge \clemb(K_4)$ by \cref{clemb-of-induced-minors}. The trivial embedding $\psi(v) = \{v\}$ of $K_4$ into $K_4$ has weak edge depth $2$, which shows $\clemb(H) \ge \clemb(K_4) \ge 4 / 2 = 2$, so \cref{all-H-iso-clique-hardness} implies the desired lower bounds.
    Finally, if $H$ has an induced $P(3, 3, 3)$ minor, \cref{induced-minor,all-p333-iso-clique-hardness} imply the desired lower bounds.
    (Note that in the corollary statement we simplified the set of fine-grained hypotheses by using that the \cliqueconj implies the \zerocliqueconj and the \minplusconvconj implies the \threesumconj; see \cref{sec:prelims-hypotheses}.)

	\smallskip
    For the second part of the corollary statement, assume that $H$ contains neither of $K_4, P(3,3,3)$, and $\goggr$ as an induced minor.
    Consider any $H' \in D(H)$.
    Since $H$ and thus~$H'$ has no $K_4$ or $\goggr$ induced minor, and since $H'$ has no clique separator, 
    \cref{small-clemb-implies-p} implies that $H'$ is a $P$-graph, i.e., $H'$ is isomorphic to $P(\ell_1,\ldots,\ell_k)$ for some integers $\ell_1,\ldots,\ell_k \ge 1$.
    Similarly to the proof of \cref{full-characterization}, we now derive that $H'$ is isomorphic to some graph in $\family$: If $k = 1$ then $\ell_1 = 1$ (since $H'$ has no clique separator) and thus $H'$ is isomorphic to $P(1) = P(1,0,0 \times 2) \in \family$. If $k = 2$ then $H'$ is isomorphic to either $P(2,1) = P(2,1,0 \times 2) \in \family$ or $P(\ell-2,2) = P(\ell-2,2,0\times 2) \in \family$ for $\ell := \ell_1 + \ell_2$. If $k \ge 3$ then none of $\ell_1,\ldots,\ell_k$ can be 1 (since $H'$ has no clique separator). Since $H$ and thus $H'$ has no $P(3, 3, 3)$ induced minor, less than three of $\ell_1,\ldots,\ell_k$ are at least three. Hence, $H'$ is isomorphic to $P(\alpha,\beta,\gamma \times 2)$ for some $\alpha \ge \beta \ge 2$ and $\gamma \ge 1$, which is in $\family$.
    This shows that every $H' \in D(H)$ is isomorphic to some graph in $\family$. \cref{full-characterization} now yields algorithms with time complexity $\Oh(m^{c(H)})$ for some $c(H) < 2$, which is the desired subquadratic upper bound.
\end{proof}

\subparagraph{Organization}
At this point it only remains to prove the lemmas stated in the last two sections.
We start by presenting an example of our algorithms and lower bounds on a specific, relatively small pattern graph in \cref{sec:p922-example}.
In \cref{sec:preliminaries}, we give preliminaries and prove \cref{d-set-computation}. %
We then prove the lemmas from \cref{sec:proof-overview-lower-bounds} in Sections~\ref{sec:lower-bounds-until-p-graphs} and~\ref{sec:p-graph-lower-bounds}, and we prove the lemmas from \cref{sec:proof-overview-algorithms} in Sections~\ref{sec:upper-bounds-until-p-graphs} and~\ref{sec:p-graph-upper-bounds}. These sections rely on various inequalities on $\funcc(\alpha,\beta,\gamma)$ which are proven in \cref{sec:algebraic-lemmas}.
In \cref{app:further-related-work} we discuss further related work, in particular the connection to join queries in database theory and submodular width.
In \cref{sec:counting} we present some implications for counting $H$-subgraphs.
In \cref{sec:appendix-decomposition-algorithm} we present an efficient algorithm for computing the decomposition $D(H)$.
In \cref{sec:p-graph-subw} we present an analogue of \cref{full-characterization} for submodular width and show how to efficiently compute the submodular width of subquadratic patterns.

\section{Example: Algorithm and Lower Bound for \boldmath$P(9,2,2)$} \label{sec:p922-example}

In this section we show that the pattern graph $P(9,2,2)$ has complexity $2-1/8$. This is an example for our algorithms for graphs in $\family$ (see \Cref{lm:p-graph-upper-bound}) and for our lower bounds for graphs in $\family$ (see \Cref{lm:p-graph-lower-bound}). 

As mentioned in the proof overview, algorithms for subgraph finding typically use the high-degree-low-degree idea of Alon, Yuster, and Zwick~\cite{AlonYZ97} (at least in the algorithm theory community). Here we generalize the usual degree splitting to ``hyper-degree splitting'', where we have already computed a superset $T$ of the projection of all $H$-subgraphs to some node parts $V_1 \times \ldots \times V_\ell$, and we split the tuples $(v_1,\ldots,v_{\ell-1})$ into high degree and low degree based on their hyper-degree $\deg(v_1,\ldots,v_{\ell-1}) = |\{v_\ell : (v_1,\ldots,v_\ell) \in T\}|$.
For an example that involves $\ell = 3$ parts, see Main Cases 1 and 2 in the proof of \Cref{lm:p922-ub} below.
We remark that for the pattern $P(9,2,2)$, the use of hyper-degree splitting is necessary, since without hyper-degree splitting (i.e., in the classic high-degree-low-degree framework) the best exponent one can hope for is $2-1/9$, while with hyper-degree splitting we obtain the optimal exponent $2-1/8$.\footnote{The idea of high-degree-low-degree-based algorithms was never fully formalized in the algorithms literature, but one can observe that all its applications, such as \cite{AlonYZ97}, perform degree splittings and for each resulting subinstance compute a partial $H$-encoding. Arguably, the most general formalization that captures this approach is the \mwidth from the database community \cite{JoglekarR18}.
Via a computer-aided search one can show that on host graphs $G$ with degrees as depicted in \Cref{fig:P922-worst-degrees} we have $\mwidth(G) = \Omega(m^{2 - 1 / 9})$ (proving this fact goes beyond the scope of this paper and will be discussed in detail in our upcoming work~\cite{BG26}). Therefore, at least in the \mwidth formalization of the high-degree-low-degree approach, no high-degree-low-degree-based algorithm solves $P(9,2,2)$-Encoding in time $\Oh(m^{2-1/9-\eps})$ for any $\eps > 0$. On the other hand, in \Cref{lm:p922-ub} we use hyper-degree splitting to solve $P(9,2,2)$-Encoding in time $\Oh(m^{2-1/8})$. To the best of our knowledge, no example of a pattern for which high-degree-low-degree-based algorithms do not give the optimal exponent of the time complexity was known before.}
To the best of our knowledge, no example of a pattern for which high-degree-low-degree-based algorithms do not give the optimal exponent of the time complexity was known before.

\begin{figure}
\begin{center}
    \includegraphics[scale=0.4]{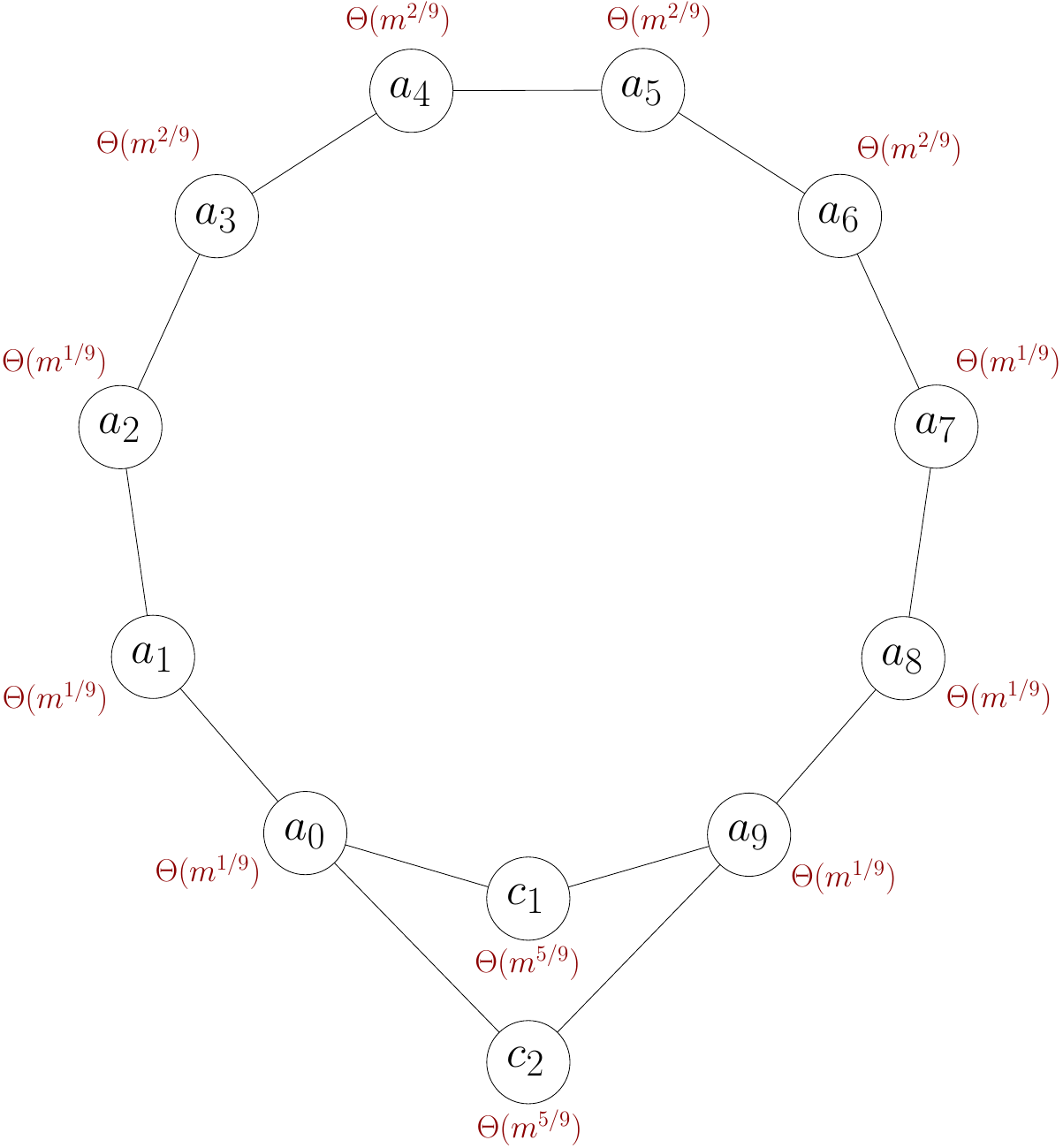}
\end{center}

\caption{Illustration of the node numbering that we use for $P(9,2,2)$. In red, we depict a bad degree sequence for the host graph, resulting in complexity at least $2-1/9$ for classic high-degree-low-degree.}
\label{fig:P922-worst-degrees}
\end{figure}

\begin{lemma} \label{lm:p922-ub}
    \enciso{P(9, 2, 2)} can be solved in time $\Oh(m^{2 - \frac{1}{8}})$.
\end{lemma}

\begin{proof}
	We write $H = P(9,2,2)$.
    Let $a_0, \ldots, a_9$ be the nodes on the path of length $9$ and $c_1$ and $c_2$ be the middle nodes on the two paths of length $2$ of $H$.
    Consider an arbitrary $H$-subgraph $\bv$ in $G$. We create several partial $H$-encodings of $G$ and show that exactly one of them encodes $\bv$.

    See \Cref{fig:P922-worst-degrees} for the node numbering of $P(9, 2, 2)$.
    We split the nodes in $V_{a_0}, V_{a_1}, \ldots, V_{a_9}$ into the ones that have degree less than $m^{\frac{1}{8}}$ (low-degree nodes, sets $V_{a_i}^{\lo}$) and the ones that have degree at least $m^{\frac{1}{8}}$ (high-degree nodes, sets $V_{a_i}^{\hi}$). For each of the $2^{10}$ choices of high or low degrees, we construct several partial $H$-encodings that encode exactly the $H$-subgraphs that satisfy these degree constraints. As $\bv$ satisfies exactly one of these $2^{10}$ cases, it will be encoded exactly once. Once we fix for each of the parts $V_{a_0}, V_{a_1}, \ldots, V_{a_9}$ whether we consider a low-degree or a high-degree node in that part, we filter out all other nodes. We first consider two special cases.

\begin{itemize}
    \item \emph{Special Case 1: $v_{a_0}$ has high degree.} The sum of degrees of nodes from $V_{a_0}^{\hi}$ in the original host graph (i.e., before filtering) is at least $m^{\frac{1}{8}} \cdot |V_{a_0}^{\hi}|$. As the sum of all degrees in a graph is equal to $2m$, we obtain $|V_{a_0}^{\hi}| = \Oh(m^{1-{\frac{1}{8}}})$. Since $H-\{a_0\}$ is a tree, it has a tree decomposition where every bag consists of exactly two adjacent nodes. Adding $a_0$ to all these bags, we obtain a valid tree decomposition of $H$ (see \Cref{P922-upper-bound-1}). Each one of these bags can be materialized in time $\Oh(m^{2-\frac{1}{8}})$ because there are $\Oh(m)$ choices for the two adjacent nodes and $\Oh(|V_{a_0}^{\hi}|) = \Oh(m^{1 - \frac{1}{8}})$ choices for $v_{a_0}$.

    \begin{figure}
    \begin{center}
        \includegraphics[scale=0.65]{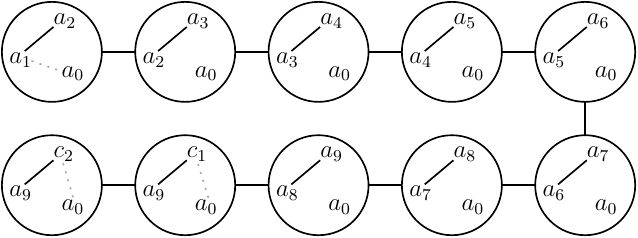}
    \end{center}

    \caption{Tree decomposition of $P(9, 2, 2)$ from Special Case 1 of \Cref{lm:p922-ub}. Here and later solid black lines inside a bag represent the edges used to create a submaterialization of the bag, and dotted gray edges represent all the other edges that the bag covers.}
    \label{P922-upper-bound-1}
    \end{figure}

\item \emph{Special Case 2: $v_{a_0}$ has low degree, and $v_{a_9}$ has high degree.} This case is symmetric to the previous one.

\end{itemize}

From now on we may assume that both $v_{a_0}$ and $v_{a_9}$ have low degrees.
As the degrees of nodes in $V_{a_0}^{\lo}$ are less than $m^{\frac{1}{8}}$, in time $\Oh(m^{1 + \frac{1}{8}})$ we can materialize a bag $B_{a_0 \to c} \coloneqq \{a_0, c_1, c_2\}$: There are $\Oh(m)$ choices to pick an adjacent pair of nodes from $V_{a_0}^{\lo}$ and $V_{c_1}$, and there are $\Oh(m^{\frac{1}{8}})$ choices to pick a neighbor in $V_{c_2}$ of the chosen node in $V_{a_0}^{\lo}$. By materializing all such triples, we ensure that we list the triple $(v_{a_0},v_{c_1},v_{c_2})$ of the nodes in the fixed solution. By using a binary search tree as a dictionary, we can store for every pair of nodes from parts $V_{c_1}$ and $V_{c_2}$ all their common neighbors in $V_{a_0}^{\lo}$ in time $\tOh(m^{1 + \frac{1}{8}}) \le \Oh(m^{2- \frac{1}{8}})$.
    Now we split the listed pairs of nodes from $V_{c_1} \times V_{c_2}$ into the pairs that have less than $m^{\frac{1}{4}}$ common neighbors in $V_{a_0}$ (low-degree pairs) and the ones that have at least $m^{\frac{1}{4}}$ common neighbors (high-degree pairs); here we apply hyper-degree splitting. Similarly to nodes from parts $V_{a_i}$, we deal with low- and high-degree pairs separately. We now distinguish four cases.

    \begin{itemize}
        \item \emph{Main Case 1: $(v_{c_1}, v_{c_2})$ is a high-degree pair.} The total number of triples in $V_{a_0}^{\lo} \times V_{c_1} \times V_{c_2}$ that we generated is $\Oh(m^{1 + \frac{1}{8}})$. Every high-degree pair in $V_{c_1} \times V_{c_2}$ has at least $m^{\frac{1}{4}}$ common neighbors in $V_{a_0}$ which means that there are $\Oh(m^{1 + \frac{1}{8}} / m^{\frac{1}{4}}) = \Oh(m^{1-\frac{1}{8}})$ such pairs. Consider the graph $H-\{c_1, c_2\}$. It is a tree, so it has a tree decomposition where each bag consists of two adjacent nodes. Adding $c_1$ and $c_2$ to every bag, we obtain a valid tree decomposition of $H$ (see Figure \ref{P922-upper-bound-2}). Submaterializations of all of these bags can be computed in time $\Oh(m^{2-\frac{1}{8}})$ because there are $\Oh(m)$ choices for the two adjacent nodes and $\Oh(m^{1 - \frac{1}{8}})$ choices for a high-degree pair from $V_{c_1} \times V_{c_2}$.
            Note that this tree decomposition encodes exactly such $H$-subgraphs $\bv$ for which $(v_{c_1}, v_{c_2})$ is a high-degree pair, because for all bags only high-degree pairs $(v_{c_1}, v_{c_2})$ are listed.
        \begin{figure}
        \begin{center}
            \includegraphics[scale=0.6]{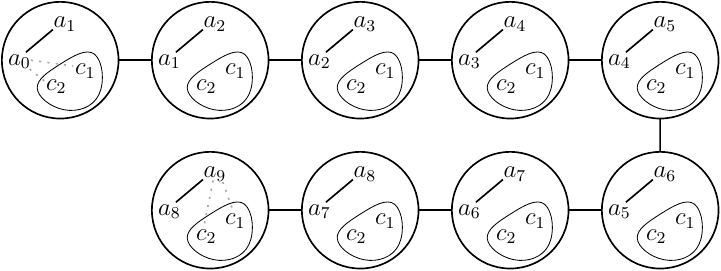}
        \end{center}

    \caption{Tree decomposition of $P(9, 2, 2)$ from Main Case 1 of \Cref{lm:p922-ub}.}
        \label{P922-upper-bound-2}
        \end{figure}
    \item \emph{Main Case 2: $(v_{c_1}, v_{c_2})$ is a low-degree pair, and at least one of $v_{a_1}, v_{a_2}, v_{a_3}$, and $v_{a_4}$ has high degree.} Let $r$ be the smallest number among $\{1, 2, 3, 4\}$ such that $v_{a_r}$ has high degree.
        We create the following tree decomposition of $H$ (see \Cref{P922-upper-bound-3}). We define the bags $B_0 \coloneqq \{a_0, a_1, \ldots, a_r, a_9, c_1, c_2\}$ and $B_{r + 1}, B_{r+2}, \ldots, B_8$, where for each $i \in [r + 1, 8]$ we let $B_i \coloneqq \{a_r, a_i, a_{i + 1}\}$. We connect these bags in a path graph: $B_0 \edg B_{8} \edg B_{7} \edg \dots \edg B_{r+2} \edg B_{r+1}$.
        It is easy to see that it is indeed a valid tree decomposition of $H$. As we chose high-degree nodes in $V_{a_r}$, we have $|V_{a_r}^{\hi}| = \Oh(m^{1-\frac{1}{8}})$.
        All bags $B_{r+1}, \ldots, B_8$ can be materialized in time $\Oh(m^{2-\frac{1}{8}})$ because there are $\Oh(m)$ choices for the two adjacent nodes and $\Oh(|V_{a^r}^{\hi}|) = \Oh(m^{1 - \frac{1}{8}})$ choices for a node from $V_{a^r}^{\hi}$.
        It remains to materialize $B_0$. In time $\Oh(m^{1 + \frac{1}{8}})$ we can materialize $\{a_9, c_1, c_2\}$: there are $\Oh(m)$ choices for an adjacent pair of nodes $(u_{a_9},u_{c_1})$ from $V_{a_9}^{\lo}$ and $V_{c_1}$, and $\Oh(m^{\frac{1}{8}})$ choices for a neighbor $u_{c_2}$ of $u_{a_9}$ in $V_{c_2}$. 
        We filter the generated triples by removing every triple $(u_{a_9},u_{c_1},u_{c_2})$ where $u_{c_1}$ and $u_{c_2}$ do not form a low-degree pair. For each remaining triple, we use our binary search tree to list all possible choices of a node from $V_{a_0}^{\lo}$ such that it is connected to both $u_{c_1}$ and $u_{c_2}$.
        This can be done in time $\tOh(m^{\frac{1}{4}})$ for each low-degree pair, giving a total time complexity $\tOh(m^{1 + \frac{3}{8}})$ to list all tuples of nodes $(u_{a_0}, u_{a_9}, u_{c_1}, u_{c_2})$ such that $u_{c_1}$ and $u_{c_2}$ form a low-degree pair and $u_{a_0}$ and $u_{a_9}$ are low-degree nodes. There are $\Oh(m^{1 + \frac{3}{8}})$ such tuples. For each one of them, we materialize nodes in $V_{a_1}^{\lo}, V_{a_2}^{\lo}, \ldots, V_{a_{r-1}}^{\lo}$, and $V_{a_r}^{\hi}$ in time $\Oh(m^{\frac{r}{8}}) \le \Oh(m^{\frac{4}{8}})$ by choosing the node in $V_{a_i}$ as one of the $\Oh(m^{\frac{1}{8}})$ neighbors of the chosen node in $V_{a_{i-1}}^{\lo}$ for $i=1,2,\ldots, r$, where we are using that $r$ is the smallest index for which $v_{a_r}$ has high degree.
    In total, we create a submaterialization of $B_0$ in time $\Oh(m^{1 + \frac{3}{8}} \cdot (\log m + m^{\frac{4}{8}})) = \Oh(m^{2 - \frac{1}{8}})$.
    Note that this tree decomposition encodes exactly such $H$-subgraphs $\bv$ for which $(v_{c_1}, v_{c_2})$ is a low-degree pair, because for $B_0$ we materialize only low-degree pairs $(v_{c_1}, v_{c_2})$.
        \begin{figure}
        \begin{center}
            \includegraphics[scale=1.1]{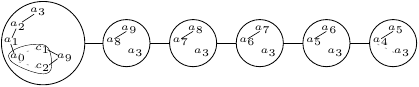}
        \end{center}
        \vspace{-.3cm}

        \caption{Example of a tree decomposition of $P(9, 2, 2)$ from Main Case 2 of \Cref{lm:p922-ub} for $r = 3$.}
        \label{P922-upper-bound-3}
        \end{figure}
    \item \emph{Main Case 2': $(v_{c_1}, v_{c_2})$ is a low-degree pair, $v_{a_1}, v_{a_2}, v_{a_3}$, and $v_{a_4}$ all have low degree, and at least one of $v_{a_5}, v_{a_6}, v_{a_7}$, and $v_{a_8}$ has high degree.} This case is symmetric to the previous one.
    \item \emph{Main Case 3: $(v_{c_1}, v_{c_2})$ is a low-degree pair, and $v_{a_1},v_{a_2},\ldots,v_{a_8}$ all have low degree.}
        We create a tree decomposition with two adjacent bags: $B_1 \coloneqq \{a_0, a_1, a_2, a_3, a_4, a_5, c_1, c_2\}$ and $B_2 \coloneqq \{a_5, a_6, a_7, a_8, a_9, c_1, c_2\}$ (see Figure \ref{P922-upper-bound-4}).
        It is easy to see that this is indeed a valid tree decomposition of $H$. To materialize $B_1$, note that there are $\Oh(m)$ choices for an adjacent pair of nodes $(u_{a_0},u_{c_1})$ from $V_{a_0}^{\lo}$ and $V_{c_1}$, and $\Oh(m^{\frac{1}{8}})$ choices for a neighbor $u_{c_2} \in V_{c_2}$ of $u_{a_0}$. We filter out triples $(u_{a_0},u_{c_1},u_{c_2})$ where $u_{c_1},u_{c_2}$ do not form a low-degree pair. For each remaining triple, we materialize nodes in $V_{a_1}^{\lo}, V_{a_2}^{\lo}, V_{a_3}^{\lo}, V_{a_4}^{\lo}$, and $V_{a_5}^{\lo}$ in time $\Oh(m^{\frac{5}{8}})$ by choosing the node in $V_{a_i}$ as one of the $\Oh(m^{\frac{1}{8}})$ neighbors of the chosen node in $V_{a_{i-1}}^{\lo}$ for $i=1,2,3,4,5$. The materialization of $B_1$ takes time $\Oh(m^{1 + \frac{1}{8}} \cdot (\log m + m^{\frac{5}{8}})) = \Oh(m^{1 + \frac{6}{8}})$. The second bag $B_2$ can be materialized similarly in time $\Oh(m^{1 + \frac{5}{8}})$.
        Note that this tree decomposition encodes exactly such $H$-subgraphs $\bv$ for which $(v_{c_1}, v_{c_2})$ is a low-degree pair, because while we materialize $B_1$ we discard the options for which $(v_{c_1}, v_{c_2})$ is a high-degree pair.
        \begin{figure}
        \begin{center}
            \includegraphics[scale=1]{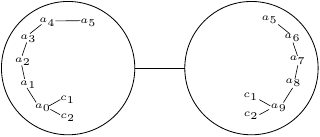}
        \end{center}
        \vspace{-.3cm}

        \caption{Tree decomposition of $P(9, 2, 2)$ from Main Case 3 of \Cref{lm:p922-ub}.}
        \label{P922-upper-bound-4}
        \vspace{-.3cm}
        \end{figure}
        \begin{figure}%
    
        \begin{center}
            \includegraphics[scale=0.27]{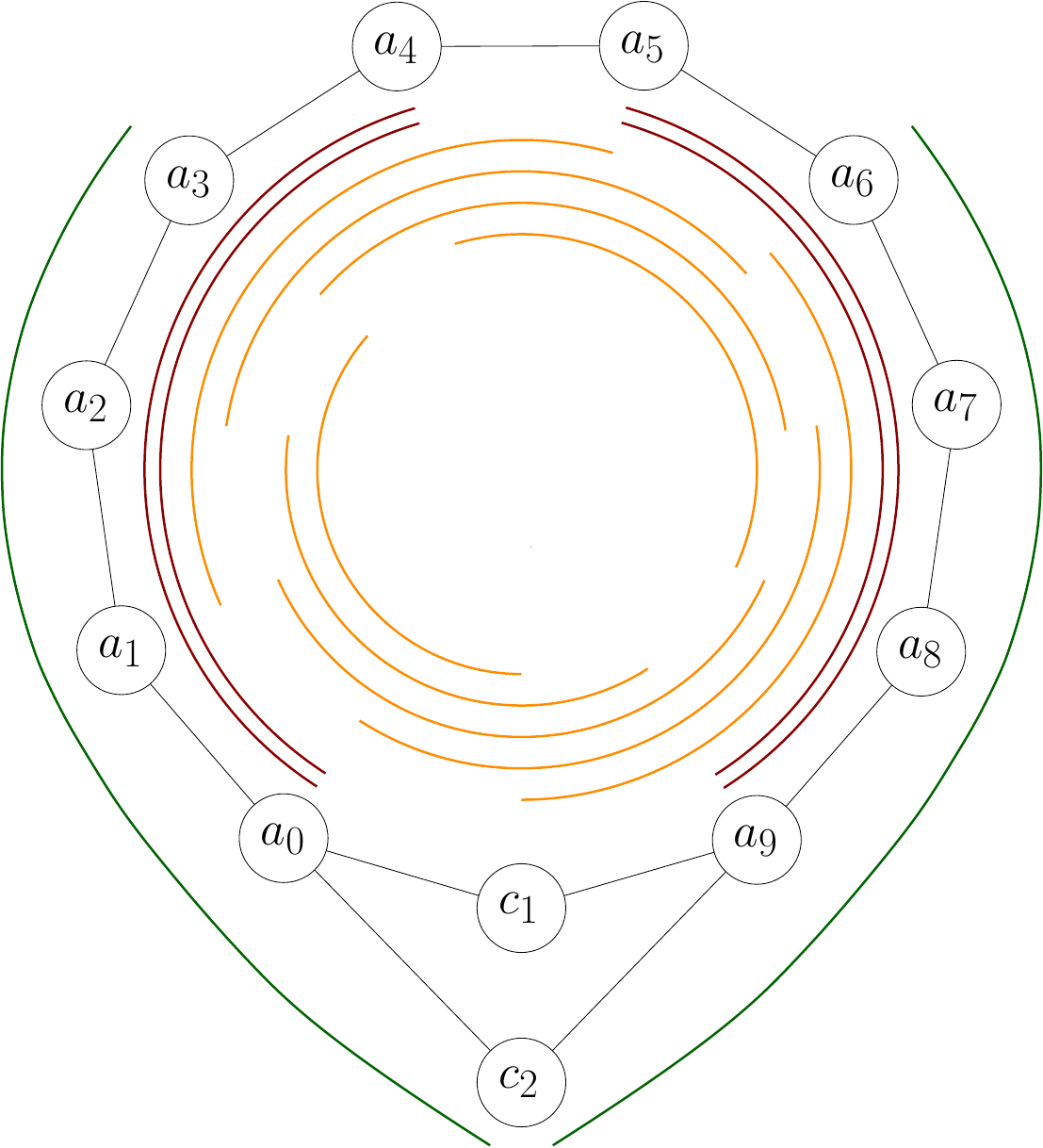}
        \end{center}
        \vspace{-.3cm}

        \caption{Clique embedding for $P(9, 2, 2)$. Each arc represents a node embedded into the subgraph covered by the arc.}
        \label{P922-alt-lower-bound}
        \end{figure}
    \end{itemize}

The cases are mutually exclusive and cover all $H$-subgraphs. Thus, we indeed create a full $H$-encoding of $G$. In each case, we computed a partial $H$-encoding of $G$ in time $\Oh(m^{2-\frac{1}{8}})$. Hence, we solved \Henciso in time $\Oh(m^{2 - \frac{1}{8}})$.
\end{proof}

We now show the matching conditional lower bound.

\begin{lemma} \label{lem:exampleP922-lower-bound}
    $\clemb(P(9, 2, 2)) \ge 2 - \frac{1}{8}$.
\end{lemma}

\begin{proof}
    We create a clique embedding $\psi$ from $K_{15}$ to $P(9, 2, 2)$ with $\wed(\psi) = 8$ which yields $\clemb(P(9, 2, 2)) \ge \frac{15}{8} = 2 - \frac{1}{8}$. See Figure \ref{P922-alt-lower-bound}.
    We first treat $a_0 \edg a_1 \edg a_2 \edg \dots \edg a_8 \edg a_9 \edg c_1$ as a cycle of length $11$ and embed one node into each of the paths of length four on this cycle.
    For the paths $a_0 \edg a_1 \edg a_2 \edg a_3 \edg a_4$ and $a_5 \edg a_6 \edg a_7 \edg a_8 \edg a_9$ we embed one additional node, making their total embedded number of nodes two. Furthermore, we embed one node into the path $c_2 \edg a_0 \edg a_1 \edg a_2 \edg a_3$ and one node into the path $a_6 \edg a_7 \edg a_8 \edg a_9 \edg c_2$. 
    In total, we embed one node into every one of $11$ paths of length $4$ in the cycle, additional one node for two of these paths, and one node for each of the two paths going through $c_2$. Thus, we used all $11 + 2 + 2 = 15$ nodes of $K_{15}$.

We omit the details, but it can be checked by inspecting \Cref{P922-alt-lower-bound} that each embedded subgraphs is connected and any two embedded subgraphs touch, so $\psi$ is a valid clique embedding. Furthermore, every edge has weak edge depth equal to $8$, thus proving the claim.
\end{proof}

\section{Preliminaries}
\label{sec:preliminaries}

All graphs in this paper are simple and undirected. In addition to standard notation from graph theory, we use the following notation. We write $[n] = \{1,\ldots,n\}$ and use $\NN$ to denote the set of positive integers.
For a vector $\bv$ we denote by $\prj{\bv}{B}$ the projection to the entries with indices in $B$.
Furthermore, for a set $S$ of vectors we denote by $\prj{S}{B}$ the set $\{\prj{\bv}{B} \mid \bv \in S\}$.

Let $H$ be a pattern and $G = (V,E)$ be a host graph with partitioning $V = \bigcup_{a \in V(H)} V_a$. Let $H'$ be an induced subgraph of $H$. We denote by $G[H']$ the subgraph of $G$ induced by all nodes in $G$ that are colored by nodes in $H'$. More precisely, $G[H'] := G[ \bigcup_{a \in V(H')} V_a ]$.

For $B \subseteq V(H)$, we call a set $S \subseteq \prod_{a \in B} V_a$ a \emph{materialization} of $B$ if for every $H$-subgraph $\bv$ of $G$ we have $\prj{\bv}{B} \in S$.
We say that $B$ is materialized if we found some materialization of $B$.

\subsection{Fine-Grained Hypotheses}
\label{sec:prelims-hypotheses}

Let us introduce the hypotheses that our conditional lower bounds are based upon.

\subparagraph{\threesumconj}
In the 3SUM problem, given sets $A,B,C$ of $n$ integers the task is to decide whether there exist $a \in A, b \in B, c \in C$ with $a+b=c$. This has a classic $\Oh(n^2)$-time algorithm~\cite{GajentaanO95}.
The \threesumconj postulates that for any $\eps > 0$ there exists $c \ge 0$ such that the 3SUM problem on integers in the range $\{-n^c,\ldots,n^c\}$ cannot be solved in time $\Oh(n^{2-\eps})$. This is one of the oldest fine-grained hypotheses, dating back to the 90s~\cite{GajentaanO95}, so lower bounds based on the \threesumconj have a long history, see, e.g.,~\cite{GajentaanO95,Patrascu10,KopelowitzPP16,
AbboudBKZ22,AbboudBF23,JinX23,ChanX23}.

\subparagraph{\minplusconvconj}
In the $(\min,+)$-Convolution problem, or MinConv problem for short, given sequences $A[1..n]$ and $B[1..n]$ the task is to compute the sequence $C[1..n]$ with $C[k] = \min\{A[i] + B[k-i] : i \in [k] \}$, where out-of-bounds entries are interpreted as $\infty$. Naively this can be solved in time $\Oh(n^2)$. 
The \minplusconvconj postulates that for any $\eps > 0$ there exists $c \ge 0$ such that MinConv with integer entries in the range $\{-n^c,\ldots,n^c\}$ cannot be solved in time $\Oh(n^{2-\eps})$. This hypothesis has been frequently used to prove fine-grained lower bounds, see, e.g.,~\cite{KunnemannPS17,CyganMWW19,BringmannN21,JansenR23,Klein22,
ChanWX23}. 

\smallskip
It is known that the \minplusconvconj implies the \threesumconj~\cite{KunnemannPS17,CyganMWW19}. That is, all our lower bounds based on the \threesumconj also hold assuming the \minplusconvconj (but phrasing a lower bound under the \threesumconj gives a stronger result).

\subparagraph{\cliqueconj}
The MinWeight-$k$-Clique problem is identical to the \miniso{K_k} problem. Naively, the problem can be solved in time $\Oh(n^k)$, where $n$ is the number of nodes in the host graph. 
For the purpose of fine-grained lower bounds, we restrict the edge weights to polynomially bounded integers, i.e., to integers in the range $\{-n^c,\ldots,n^c\}$ for some constant $c \ge 0$. 
The \cliqueconj postulates that there exist no $k \ge 3$ and $\eps > 0$ such that for all $c \ge 0$, MinWeight-$k$-Clique with edge weights in $\{-n^c,\ldots,n^c\}$ can be solved in time $\Oh(n^{k-\eps})$. 

\subparagraph{\zerocliqueconj}
The Zero-$k$-Clique problem is similar to the MinWeight-$k$-Clique problem, but instead of asking for the minimum total edge weight of a $k$-clique we ask whether there exists a $k$-clique with total edge weight 0. Again the naive running time is $\Oh(n^k)$. 
The \zerocliqueconj postulates that there exist no $k \ge 3$ and $\eps > 0$ such that for all $c \ge 0$, Zero-$k$-Clique with edge weights in $\{-n^c,\ldots,n^c\}$ can be solved in time $\Oh(n^{k-\eps})$. 

\smallskip
It is known that the \cliqueconj implies the \zerocliqueconj, see, e.g.,~\cite{AbboudBDN18}. That is, all our lower bounds based on the \zerocliqueconj also hold assuming the \cliqueconj (but phrasing a lower bound assuming the \zerocliqueconj gives a stronger result).
It is also known that the special case of the \cliqueconj for $k=3$ is implied by the All-Pairs-Shortest-Paths Hypothesis~\cite{WilliamsW18}. Similarly, 
the special case of the \zerocliqueconj for $k=3$, also called ZeroTriangle Hypothesis, is implied
by the 3SUM Hypothesis~\cite{WilliamsW13} and the All-Pairs-Shortest-Paths Hypothesis~\cite{WilliamsW18}.
For $k > 3$, these hypotheses are natural generalizations, but are not known to be implied by any other fine-grained hypothesis. 
Both the \cliqueconj and the \zerocliqueconj have been frequently used to prove fine-grained lower bounds, see, e.g.,~\cite{AbboudWW14,BackursDT16,BackursT17,
AbboudBDN18,LincolnWW18,BringmannGMW20,BringmannCM22}.

\medskip
In this paper, whenever we talk about algorithms we always implicitly include randomized algorithms. In particular, all hypotheses are assumed to hold also for randomized algorithms.

\medskip
\subparagraph{Note from November 2026.}\label{subpar:new-note}
Very recently, Josh Alman and Virginia Vassilevska Williams~\cite{AlmanW26} refuted all the above hypotheses.
For 3SUM, the hypothesis was subsequently broken even without the use of fast matrix multiplication~\cite{ChanIMY26}.
The initial publication of our paper predated these results and therefore we keep it mostly unchanged.
All the reductions we propose still stand.
The interpretation of the implied conditional hardness results, however, changes.

That said, the lower bounds we show based on the \zerocliqueconj also work under the \cliquelistconj (see \cref{hyp:cliquelistconj} and the proof of \cref{all-H-iso-clique-hardness}), which is not refuted by the results of Alman and Vassilevska Williams.
Therefore, the lower bounds for \Hlistiso and \Henumiso based on \zerocliqueconj still indicate a barrier for faster algorithms: before trying to improve the time complexities of these problems, we should aim to break the \cliquelistconj first.

\subsection{Tree Decompositions}
\label{sec:prelims-treedecomp}

Let us define tree decompositions and collect some facts about them. 
A \emph{tree decomposition} of a graph $H$ is a tree $T = ([k], E_T)$ and sets $B(1),\ldots,B(k) \subseteq V(H)$ (these sets are called \emph{bags}) with the following properties. (1) Every node $v \in V(H)$ is contained in some bag. (2) For every edge $\{u,v\} \in E(H)$ some bag contains both $u$ and $v$. And (3) for any $i,j \in [k]$ and $v \in V(H)$, if both $B(i)$ and $B(j)$ contain $v$, then $v$ is contained in $B(\ell)$ for every node $\ell$ on the unique path from $i$ to $j$ in $T$.

Note that if $H$ has no isolated nodes then (1) is subsumed by (2). Property (3) implies that for any node $v \in V(H)$ the bags containing $v$ form a connected subgraph of $T$. We will use the following fact.

\begin{lemma}[Clique Containment Lemma, {\cite[Lemma 3.1]{BodlaenderM93}}] \label{induced-clique-td}
  Let $\mathcal{T}$ be a tree decomposition of a graph $H$ and let $C \subseteq V(H)$ be a clique in $H$.
  Then for some bag $B$ of $\mathcal{T}$ we have $C \subseteq B$.
  \lipicsEnd
\end{lemma}

\subsection{Series-Parallel Graphs}
\label{sec:prelims-seriesparallel}

Any series-parallel graph is a graph $H$ with a distinguished source node $s$ and a distinguished sink node $t$. The family of all series-parallel graphs is defined as follows.
\begin{itemize}
    \item \emph{Single Edge.} The graph $H = (\{s, t\}, \{\{s, t\}\})$ with source $s$ and sink $t$ is series-parallel.
    \item \emph{Series Composition.} 
    Let $H_1$ and $H_2$ be series-parallel graphs with source nodes $s_1$ and~$s_2$ and sink nodes $t_1$ and $t_2$, respectively. Construct graph $H$ by taking the disjoint union of $H_1$ and $H_2$ and identifying $t_1$ with $s_2$. Then $H$ with source $s_1$ and sink $t_2$ is series-parallel. 
    \item \emph{Parallel Composition.} Let $H_1$ and $H_2$ be series-parallel graphs with source nodes $s_1$ and~$s_2$ and sink nodes $t_1$ and $t_2$, respectively. Construct graph $H$ by taking the disjoint union of $H_1$ and $H_2$ and identifying $s_1$ with $s_2$ and identifying $t_1$ with $t_2$. Then $H$ with source $s_1 = s_2$ and sink $t_1 = t_2$ is series-parallel. 
\end{itemize}

\subsection{Miscellaneous} \label{sec:preliminaries:misc}

In this section we prove statements that did not naturally fit in any other section, including \cref{d-set-computation} from \cref{sec:formal-results}.

\begin{observation} \label{high-degree-small-cnt-nodes}
    Let $G$ be a $k$-partite graph for $k \ge 2$ with parts $V_1, V_2, \ldots, V_k$. Let $m \ge |E(G)|$. Then for any $i \in [k]$ and any $\Delta \ge 1$ the number of nodes in $V_i$ with degree at least $\Delta$ is at most $m / \Delta$.
    \lipicsEnd
\end{observation}

\dsetcomputation*

\begin{proof}
    First consider the case that $C$ is empty.
    Hence, $H$ is disconnected.
    Note that any element of $D(H)$ does not have a clique separator, and thus is connected.
    Therefore, any $H' \in D(H)$ lies completely inside some connected component $S_i$ of $H$.
    It is easy to see that $H'$ is a maximal subgraph without a clique separator in $H[S_i]$ if and only if it is a maximal subgraph without a clique separator in $H$.
    Hence, we get $D(H) = \bigcup_{i \in [k]} D(H[S_i]) = \bigcup_{i \in [k]} D(H[S_i \cup C])$.
    It remains to consider the case $|C| \ge 1$.

    \medskip

    We first show the inclusion ``$\subseteq$''.
    Consider any $U \subseteq V(H)$ such that $ H[U] \in D(H)$.
    If $U \cap S_i \neq \emptyset$ and $U \cap S_j \neq \emptyset$ for some $i \neq j \in [k]$, then $U \cap C$ is a clique separator in $H[U]$ which contradicts $H[U] \in D(H)$.
    Hence, $U \subseteq S_i \cup C$ for some $i \in [k]$.
    As $H[U]$ is a maximal subgraph of $H$ without a clique separator, it is also a maximal subgraph of $H[S_i \cup C]$ without a clique separator, and thus $H[U] \in D(H[S_i \cup C])$.
    Therefore, $D(H) \subseteq \bigcup_{i \in [k]} D(H[S_i \cup C])$.

    \medskip

    We now show the inclusion ``$\supseteq$''. Consider any $U \subseteq V(H[S_i \cup C])$ such that $H[U] \in D(H[S_i \cup C])$ for some $i \in [k]$.
    Assume for the sake of contradiction that $H[U]$ is not a maximal subgraph without a clique separator in $H$.
    It is still a subgraph without a clique separator, so there is some $U' \subseteq V(H)$ such that $U \subsetneq U'$ and $H[U'] \in D(H)$ does not contain a clique separator.
    As shown before, we have $U' \subseteq S_j \cup C$ for some $j \in [k]$.
    If $U \not \subseteq C$, then $j = i$, and $U' = U$ because $H[U]$ is a maximal subgraph in $H[S_i \cup C]$ without a clique separator.
    It follows that $H[U] = H[U'] \in D(H)$.
    
    Otherwise, we have $U \subseteq C$.
    Since $C$ is a clique, it does not have a clique separator in $H[S_i \cup C]$.
    By maximality of $U$ we have $U = C$.
    From now on we can assume $U = C$ and $H[U] \in D(H[S_i \cup C])$, and we want to show $H[U] \in D(H)$.
    Note that every node $a \in C$ is adjacent to at least one node of $S_i$ as otherwise removal of $C \setminus \{a\}$ from $H$ would disconnect $a$ from $S_i$ which contradicts the fact that $C$ is a minimal clique separator in $H$.

    If $|C| = 1$, let $a \in V(H)$ be such that $C = \{a\}$. We have that $a$ is adjacent to some $b \in S_i$. Since $\{a, b\}$ does not have a clique separator, we get a contradiction as $C=U$ is not a maximal subgraph of $H[S_i \cup C]$ without a clique separator.
    Thus, from now on we can assume $|C| \ge 2$.

    Let $H' \coloneqq (S_i \cup C, E(H[S_i \cup C]) \setminus E(H[C]))$.
    That is, $H'$ is the graph $H[S_i \cup C]$ without the edges in $H[C]$.
    Define $T$ as a tree in $H'$ of minimum cardinality such that $C \subseteq V(T)$ and each node of $C$ has degree one in $T$.
    In other words, $T$ is a Steiner tree with $C$ as a set of terminals such that all nodes of $C$ are leaves in $T$.
    Note that such a tree exists as $H[S_i]$ is connected and each node of $C$ is adjacent to some node in $S_i$.
    We choose $T$ as a minimal such tree.
    Note that as $|C| \ge 2$ and $H'$ does not have any edges between nodes in $C$, we have $C \subsetneq V(T)$.
    We claim that $H[V(T)]$ does not have a clique separator.
    As $C$ is a strict subset of $V(T)$, this leads to a contradiction with the fact that $C$ is a maximal subgraph of $H[S_i \cup C]$ without a clique separator.
    
    It thus remains to prove the claim that $H[V(T)]$ does not have a clique separator.
    Assume for the sake of contradiction that $H[V(T)]$ has a clique separator $C'$.
    In case $C' \subseteq C$, the subtree $T[V(T) \setminus C']$ is connected as we removed some leaves from $T$.
    Hence, $H[V(T) \setminus C']$ is connected which contradicts the fact that $C'$ is a clique separator in $H[V(T)]$.
    It remains to consider the case $C' \not \subseteq C$.
    Let $a$ be some node in $C' \setminus C$.
    Note that all nodes from $C \setminus C'$ (if any) are adjacent in $H[V(T) \setminus C']$.
    Hence, there is some connected component of $H[V(T) \setminus C']$ that contains all nodes from $C \setminus C'$.
    As $C'$ is a clique separator in $H[V(T)]$, there is some other connected component $S'_{\ell} \subseteq S_i$ of $H[V(T) \setminus C']$.
    We build a new tree $T'$ in $H'$ such that $C \subseteq V(T')$ and each node of $C$ has degree one in $T'$ such that $T'$ only contain nodes from $V(T) \setminus S'_{\ell}$, thus contradicting the fact that $T$ is such a tree of minimum cardinality.
    To create $T'$ we first connect $a$ with all other nodes from $C'$ by edges. 
    Then for each terminal $b \in C \setminus C'$ we consider the path from $b$ to $a$ in $T$, go along it and add edges of it to $T'$ until the current node is already in $T'$.

    At every point of this process, $T'$ is a connected tree with $V(T') \subseteq V(T)$, and at the end it connects all terminals $C$. Furthermore, all nodes from $C$ have degrees one in $T'$.
    On the other hand, we claim that $T'$ does not contain any nodes from $S'_{\ell}$.
    In the initial step we only added nodes from $C'$ to $T'$ where $S'_{\ell} \cap C' = \emptyset$.
    For each terminal $b \in C$ we do not visit any nodes from $S'_{\ell}$ on the path, because as $b$ lies in a different connected component of $H[V(T) \setminus C']$, any path from $b$ to some node in $S'_{\ell}$ goes through $C'$ (as it is a clique separator), and as soon as we enter $C'$, we stop adding edges (or we stop even earlier).
    Hence, we indeed have $V(T') \subseteq V(T) \setminus S'_{\ell}$. This proves the claim and thus finishes the proof.
\end{proof}

\begin{lemma} \label{no-large-separators-needed}
    Let $H$ be a pattern that has a clique separator and $\clemb(H) < 2$. Then every minimal clique separator of $H$ has cardinality one or two.
\end{lemma}

\begin{proof}
    As $H$ is a pattern, it is connected, and thus $H$ does not have a clique separator of size zero.

    If $H$ has a clique separator of size at least four, $H$ has an (induced) $K_4$ minor.
    The trivial embedding $\psi(v) = \{v\}$ of $K_4$ into $K_4$ has weak edge depth $2$.
    Together with \cref{clemb-of-induced-minors}, we thus arrive at the contradiction $\clemb(H) \ge \clemb(K_4) \ge 4 / 2 = 2$.

    It remains to show that no minimal clique separator of $H$ has size three.
    For the sake of contradiction assume that nodes $a, b$, and $c$ form a minimal clique separator of size three in $H$.
    Let $S_1$ be one of the connected components of $H-\{a, b, c\}$.
    If $a$ is not adjacent to any node of $S_1$, then $\{b, c\}$ is a clique separator in $H$ because its removal disconnects nodes of $S_1$ from $a$.
    This would violate the fact that $\{a, b, c\}$ is a minimal clique separator in $H$.
    Thus, $a$ is adjacent to some node in $S_1$.
    Analogously, $b$ and $c$ are adjacent to some nodes in $S_1$.
    Removing all nodes from $V(H) \setminus S_1 \setminus \{a, b, c\}$ and contracting $S_1$ into a single node yields that $H$ has a $K_4$ minor, and thus $\clemb(H) \ge \clemb(K_4) \ge 2$, which is a contradiction.
\end{proof}

\section{Lower Bounds: from General Case to Family \boldmath$\family$} \label{sec:lower-bounds-until-p-graphs}

In this section we prove \cref{all-H-iso-clique-hardness,induced-minor,clemb-of-induced-minors,small-clemb-implies-p,all-p333-iso-clique-hardness}. This covers most of the lemmas from \cref{sec:proof-overview-lower-bounds}, except for \cref{lm:p-graph-lower-bound} whose proof is postponed to \cref{sec:p-graph-lower-bounds}.

\smallskip
We start by proving that the \zerocliqueconj implies a lower bound for the problem of listing cliques. 

\begin{hypothesis}[\cliquelistconj]\label{hyp:cliquelistconj}
    For any $k \ge 3$ and $\eps > 0$, there is no algorithm that lists all $k$-cliques of a given $n$-node host graph $G$ in time $\tOh(n^{k - \varepsilon} + t)$, where $t$ is the number of $k$-cliques in $G$.
    \lipicsEnd
\end{hypothesis}

\begin{lemma} \label{lm:zero-clique-to-clique-list}
    The \zerocliqueconj implies the \cliquelistconj.
\end{lemma}

In the case of triangles, i.e, $k=3$, this lemma follows from the work of Vassilevska Williams and Xu~\cite[Theorem 3.4]{WilliamsX20}. Here we essentially generalize the proof of~\cite{WilliamsX20} to the case $k > 3$. However, their result also shows lower bounds for a related problem (AllEdgesTriangle) and bounds the number of edges in the constructed graph, both of which is not necessary for the result that we need to prove here, which allows to somewhat simplify their proof.

\begin{proof}
    For the sake of contradiction assume that the \cliquelistconj does not hold, i.e., there exists an algorithm listing all $k$-cliques in time $\tOh(n^{k - \varepsilon} + t)$ for some $k \ge 3$ and $0 < \varepsilon < 1$.
    We use it to design an algorithm that, for a given host graph $G$ with edge weights absolutely bounded by $n^c$, decides whether $G$ has a $k$-clique of total weight zero in time $\Oh(n^{k - \varepsilon'})$ for some $\varepsilon' > 0$. This contradicts the \zerocliqueconj.

    For $100 n^{k - \varepsilon}$ times we pick random nodes from each part of $G$ and check whether they form a zero-$k$-clique. If so, we return such a $k$-clique. Otherwise, we proceed to the main algorithm.

    We first pick a prime number in $[100 k^2 n^c, \Gamma k^2 n^c]$ for a large enough constant $\Gamma$.
    By the prime number theorem, a random integer in this range is prime with probability $\Omega(1 / \log n)$, so we can find such a prime number $p$ in time $\Oh(\polylog n)$ by checking random numbers from this range for primality.
    After picking $p$, we view all edge weights in $G$ as numbers in the prime field $\mathbb{F}_p$ by reducing them modulo $p$.
    Since all edges in $G$ have weights in $[-n^c, n^c]$, any $k$-clique in $G$ has weight in $[-k^2 n^c, k^2 n^c]$.
    Since $p \ge 100 k^2 n^c$, the set of zero-$k$-cliques with respect to the new weights stays the same.

    Pick a uniformly random number $x$ in $\mathbb{F}_p$. Let $G'$ be a copy of~$G$, but all weights $w'$ of edges are multiplied by $x$ over $\mathbb{F}_p$, i.e., $w'(e) = x \cdot w(e) \in \mathbb{F}_p$.
    Note that for any $k$-clique $C$ we have $w'(C) = x \cdot w(C)$ over $\mathbb{F}_p$.

    Let $\rho \coloneqq \frac{\varepsilon}{3k^2}$. We split the numbers $\{0,1,\ldots,p-1\}$ into $\ell \coloneqq \left\lfloor n^{\rho} \right\rfloor$ contiguous ranges $R_1, R_2, \ldots, R_{\ell}$, so that each range has size $\left\lfloor p / \ell \right\rfloor$ or $\left\lceil p / \ell \right\rceil$.
    Let $\mathcal{R}$ be the set of all tuples $(R_{i_e})_{e \in E(K_k)}$ such that $0 \in \sum_{e \in E(K_k)} R_{i_e}$, where $i_e \in \{1,2,\ldots,\ell\}$ for each $e$.
    We compute $\mathcal{R}$ by enumerating all possible tuples $(R_{i_e})_{e \in E(K_k)}$ and checking whether the resulting tuple satisfies $0 \in \sum_{e \in E(K_k)} R_{i_e}$. There are $\Oh(\ell^{|E(K_k)|}) \le \Oh(n^{k^2 \rho}) = \Oh(n^{\varepsilon / 3})$ such tuples, and they can be enumerated in this time bound as well.

    For each tuple $r = (R_{i_e})_{e \in E(K_k)} \in \mathcal{R}$ we construct the set of edges \[E_r \coloneqq \bigcup_{\{a, b\} \in E(K_k)} \{\{v_a, v_b\} \in E(G) \mid v_a \in V_a, v_b \in V_b, w'(\{v_a, v_b\}) \in R_{i_{\{a, b\}}}\}.\]
    Let $G_r = (V(G), E_r, w')$ be the subgraph of $G'$ with edge set $E_r$.
    Observe that $G'$ has a zero-$k$-clique if and only if at least one of the graphs $G_r$ has a zero-$k$-clique.

    We list all $k$-cliques in each graph $G_r$ for all $r \in \mathcal{R}$. For each listed $k$-clique, we check whether it is a zero-$k$-clique in $G$, and if so, we return it.

    Correctness of the presented algorithm is immediate: Since we check the output to be a zero-$k$-clique in $G$, if we return a clique, then it is guaranteed to be a zero-$k$-clique, so we never return a false positive.
    Moreover, for any zero-$k$-clique $C$ in $G$, we have $w'(C) = x \cdot w(C) = 0$ over $\mathbb{F}_p$, so $C$ is also a zero-$k$-clique in $G'$, and thus a zero-$k$-clique in some graph $G_r$ for some $r \in \mathcal{R}$.
    Since we list all $k$-cliques in $G_r$, in particular we list $C$.
    Thus, if $G$ has any zero-$k$-cliques, the algorithm will return one.

    It remains to analyze the running time of the algorithm.
    If the number $t$ of zero-$k$-cliques in $G$ is at least $n^{2 \varepsilon}$, a collection of random nodes picked in the first stage of the algorithm is a zero-$k$-clique with probability at least $n^{2\varepsilon - k}$.
    Hence, by repeating the procedure $100 n^{k - \varepsilon}$ time we find a zero-$k$-clique with probability at least $1 - (1 - n^{2 \varepsilon - k})^{100 n^{k - \varepsilon}} = 1 - \Oh(e^{-n^{\varepsilon}})$.
    Thus, in this case we continue to the main stage of the algorithm with probability $\Oh(e^{-n^{\varepsilon}})$.

    Note that since $|R_e| = \Oh(p / n^{\rho})$, the constraint $0 \in \sum_{e \in E(K_k)} R_{i_e}$ for a tuple $(R_{i_e})_{e \in E(K_k)}$ to be in $\mathcal{R}$ implies that $\sum_{e \in E(K_k)} R_{i_e} \subseteq [-\Oh(p / n^{\rho}), \Oh(p / n^{\rho})]$ over $\mathbb{F}_p$.
    Thus, if a $k$-clique $C$ appears as a $k$-clique of one of graphs $G_s$, then $w'(C) \in [-\Oh(p / n^{\rho}), \Oh(p / n^{\rho})]$ over $\mathbb{F}_p$.
    If $C$ is not a zero-$k$-clique, then $w'(C) = x \cdot w(C)$ is uniformly distributed in $\mathbb{F}_p$, since $p$ is prime and $x$ a uniformly random number in $\mathbb{F}_p$.
    Hence, $w'(C) \in [-\Oh(p / n^{\rho}), \Oh(p / n^{\rho})]$ happens with probability $\Oh(n^{- \rho})$.
    It follows that if $t \le n^{2 \varepsilon}$, then the expected total number of listed cliques in all graphs $G_r$ over all $r \in \mathcal{R}$ is $\Oh(n^{2 \varepsilon} + n^{k - \rho}) = \Oh(n^{k - \rho})$, where the first summand counts the number of zero-$k$-cliques in $G$ and the second counts the number of false positives.
    The remaining case $t \ge n^{2 \varepsilon}$ is handled with high probability by the first step of the algorithm which samples random $k$-tuples of nodes and checks whether they are zero-$k$-cliques. Specifically, with high probability in this case we do not list any cliques (and with probability $\Oh(e^{-n^{\varepsilon}})$ we list $O(n^k)$ cliques).
    In both cases, the expected total number of listed cliques is $\Oh(n^{k - \rho})$.

    Since $|\mathcal{R}| = \Oh(n^{\varepsilon / 3})$, we get that the expected time the algorithm takes is $\tOh(n^{k - \varepsilon} + n^{\varepsilon / 3} \cdot n^{k - \varepsilon} + n^{k - \rho}) \le \tOh(n^{k - \rho})$, where we use that the total output size of all calls to $k$-clique listing is $\Oh(n^{k - \rho})$ in expectation.

    With standard boosting methods we can turn this algorithm with expected running time into an algorithm that is correct with high probability.
    Using Markov's bound, by aborting this algorithm after time $\tOh(n^{k - \rho})$, we obtain a randomized algorithm running in time $\tOh(n^{k - \rho})$ that is correct with probability at least $0.9$.
    Repeating this algorithm $\Oh(\log n)$ times and returning a zero-$k$-clique if at least one repetition returned a zero-$k$-clique, we obtain a randomized algorithm running in time $\tOh(n^{k - \rho}) \le \Oh(n^{k - \rho / 2})$ that is correct with high probability.
    This contradicts the \zerocliqueconj.
\end{proof}

\subsection{Lower Bounds Based on Clique Embeddings}\label{sec:lower-bounds-until-p-graphs:clemb}

The following is a helper lemma to prove lower bounds based on clique embeddings.

\begin{lemma} \label{lm:clemb-reduction}
    Let $H$ be a pattern, $k \ge 3$ be an integer, and $\psi : K_k \to H$ be a clique embedding.
    There is an algorithm that given a (weighted) host graph $G$ for $K_k$-subgraph isomorphism with $|V(G)| = n$ computes a (weighted) host graph $G'$ for $H$-subgraph isomorphism with $|E(G')| = \Oh(n^{\wed(\psi)})$ in time $\Oh(n^{\wed(\psi)})$, such that there is a (weight-preserving) constant-time-computable bijection between $K_k$-subgraphs in $G$ and $H$-subgraphs in $G'$.
\end{lemma}

\begin{proof}
	It suffices to prove the weighted case of the lemma, as the unweighted case follows from setting all weights to zero. The proof is a rather standard embedding.
	
	Choose a function $\eemb : E(K_k) \to E(H)$ that maps every edge $\{x, y\}$ of $K_k$ to an arbitrary edge $\{a, b\} \in E(H)$ such that $x, y \in \psi^{-1}(a) \cup \psi^{-1}(b)$. 
	Such a function $\eemb$ exists, because for any $\{x, y\} \in E(K_k)$ the subgraphs induced by $\psi(x)$ and $\psi(y)$ touch. That is, there is an edge $\{a,b\} \in E(H)$ with $a \in \psi(x)$ and $b \in \psi(y)$, or there is a node $a \in V(H)$ with $a \in \psi(x) \cap \psi(y)$. In the former case we can choose $f(\{x,y\}) = \{a,b\}$, in the latter case we pick an arbitrary neighbor $b$ of $a$ in $H$ and set $f(\{x,y\}) = \{a,b\}$, where we used that $H$ is connected and $|V(H)| \ge 2$ since $H$ is a pattern.
	
	Denote the partitioning of $G$ by $V_x$ for $x \in V(K_k)$. We assume that the edges of $G$ are stored in an adjacency matrix, so given two nodes we can query whether there is an edge between them and, if so, query its weight in constant time. Reading the adjacency matrix of $G$ takes time $\Oh(n^2) \le \Oh(n^{\wed(\psi)})$. Here, $\wed(\psi) \ge 2$ follows from the fact that any two embedded subgraphs $\psi(x), \psi(y)$ touch, so there is some edge $e \in E(H)$ with $x, y \in \psi^{-1}(e)$, and thus the depth of $e$ is at least two.

    We construct the graph $G'$ as follows. The nodes of $G'$ consist of parts $V'_a$ for $a \in V(H)$. 
    For each $a \in V(H)$, we put a node in $V'_a$ for each tuple of nodes in $\prod_{x \in \psi^{-1}(a)} V_x$. That is, each node in $V'_a$ corresponds to a tuple of nodes from the parts of $G$ corresponding to nodes of $K_k$ that are embedded into $a$ by $\psi$. For a node $v_a \in V'_a$, denote by $v_{a,x} := v_a |_x$ the node from $V_x$ that $v_a$ encodes for $x \in \psi^{-1}(a)$.
    To construct the edges of $G'$, for each edge $e = \{a, b\} \in E(H)$ and all nodes $v_a \in V'_a$ and $v_b \in V'_b$, we add the edge $\{v_a, v_b\}$ to $E(G')$ if and only if $C \coloneqq \bigcup_{x \in \psi^{-1}(a)} \{v_{a,x}\} \cup \bigcup_{x \in \psi^{-1}(b)} \{v_{b,x}\}$ forms a clique in $G$. In particular, since there are no edges in $G$ within a part, this condition checks that $v_{a,x} = v_{b,x}$ for every $x \in \psi^{-1}(a) \cap \psi^{-1}(b)$.
    In case we put an edge, we set its weight to $\sum_{\{x,y\} \in \eemb^{-1}(e)} w(\{v_{a,x}, v_{b,y}\})$, where $w$ is the weight function of edges in $G$.

    We now prove that there is a constant-time-computable weight-preserving one-to-one correspondence between $K_k$-subgraphs in $G$ and $H$-subgraphs in $G'$. Fix a $k$-clique $C$ in $G$. For each $x \in V(K_k)$, denote by $C_x$ the node of $C$ from $V_x$. Consider nodes $v_a = C |_{\psi^{-1}(a)} \in V'_a$ for all $a \in V(H)$. We claim that $\corr(C) \coloneqq \bv \coloneqq (v_a)_{a \in V(H)}$ forms an $H$-subgraph in $G'$, and its weight is the same as the weight of $C$. Indeed, for every edge $\{a, b\} \in E(H)$ we have $\{v_a,v_b\} \in E(G')$ since the subset $\bigcup_{x \in \psi^{-1}(a)} \{v_{a,x}\} \cup \bigcup_{x \in \psi^{-1}(b)} \{v_{b,x}\} = \bigcup_{x \in \psi^{-1}(a) \cup \psi^{-1}(b)} \{C_x\} \subseteq C$ forms a clique in $G$. Furthermore, the weight of the subgraph formed by $\bv$ is $\sum_{e \in E(H)} \sum_{\{x,y\} \in \eemb^{-1}(e)} w(\{C_x, C_y\}) = \sum_{\{x,y\} \in E(K_k)} w(\{C_x, C_y\})$, which is the weight of $C$, because the function $\eemb$ embeds every edge into exactly one summand. Thus, every $K_k$-subgraph in $G$ corresponds to an $H$-subgraph in $G'$ of the same weight.
    Furthermore, such a correspondence is an injection.
    Indeed, let $C'$ be some other $k$-clique in $G$.
    As $C \neq C'$, we have $C_x \neq C'_x$ for some $x \in V(K_k)$.
    Pick some $a \in \psi(x)$.
    It is easy to see that $\corr(C)$ and $\corr(C')$ pick different nodes in $V'_a$.
    Hence, $\corr$ is indeed an injection.

    We now show the inverse correspondence. Fix an $H$-subgraph $\bv$ in $G'$. For each $a \in V(H)$, denote $v_a$ the node of $\bv$ from $V'_a$. Consider nodes $C_x \coloneqq v_a^x$ for each $x \in V(K_k)$ for some $a \in \psi(x)$. We claim that $C \coloneqq (C_x)_{x \in V(K_k)}$ forms a $k$-clique in $G$ and $\corr(C) = \bv$.

    First, we show that $v_b^x = C_x$ for all $b \in \psi(x)$. Let $a$ be the node in $\psi(x)$, for which $C_x$ was defined as $v_a^x$. Thus, $v_a^x = C_x$. By the definition of a clique embedding, $H[\psi(x)]$ is connected, thus there exists a path $a = d_1 \edg d_2 \edg \dots \edg d_{i} = b$ between $a$ and $b$ in $H$, such that all nodes of this path are from $\psi(x)$. If we now consider $v_{d_1}, v_{d_2}, \ldots, v_{d_i}$, such nodes form a path in $G'$ because $\bv$ is an $H$-subgraph in $G'$. For each $j \in [i-1]$, an edge between $v_{d_j}$ and $v_{d_{j+1}}$ was added to $G'$ only if $v_{d_j}$ and $v_{d_{j+1}}$ represent the same nodes in parts $V_y$ for all $y \in \psi^{-1}(d_j) \cap \psi^{-1}(d_{j+1})$. In particular, they represent the same node in $V_x$. Thus, by transitivity $v_{d_1}$ and $v_{d_i}$ represent the same node in $V_x$, so $v_b^x = v_{d_i}^x = v_{d_1}^x = v_a^x = C_x$, proving the claim.
    Thus, indeed $\corr(C) = \bv$.
    Consider some $x \neq y \in V(K_k)$.
    By the definition of a clique embedding, there is an edge $\{a, b\} \in E(H)$ such that $x, y \in \psi^{-1}(\{a, b\})$.
    As $\bv$ forms an $H$-subgraph in $G$, we have that $v_a$ and $v_b$ are adjacent, and thus $C_x$ and $C_y$ are adjacent due to the definition of graph $G'$.
    Hence, $C$ forms a $k$-clique in $G$.

    Thus, $\corr$ is an injection and a surjection.
    Furthermore, computation of $\corr$ and $\corr^{-1}$ can be straightforwardly done in constant time.

    It remains to analyze the time complexity of building graph $G'$. We first build its nodes. For each $a \in H$, we create $\prod_{x \in \psi^{-1}(a)} |V_x| \le n^{|\psi^{-1}(a)|}$ nodes in $V'_a$. For any edge $\{a, b\} \in H$, weak edge depth of $\{a, b\}$ is at least $|\psi^{-1}(a)|$, thus $|\psi^{-1}(a)| \le \wed(\psi)$.
    Thus, nodes of $G'$ can be computed in time $\Oh(n^{\wed(\psi)})$. Then for each $\{a, b\} \in E(H)$, we create edges between $V'_a$ and $V'_b$. We fix a node $v_a \in V'_a$, and to generate all $v_b \in V'_b$ with $v_b^x = v_a^x$ for all $x \in \psi^{-1}(a) \cap \psi^{-1}(b)$, we generate all tuples of nodes from $\prod_{x \in \psi^{-1}(b) \setminus \psi^{-1}(a)} V_x$ and generate a node $v_b$ that encodes them and $v_a^x$ for all $x \in \psi^{-1}(a) \cap \psi^{-1}(b)$. In total, this process works in time $\Oh(\prod_{x \in \psi^{-1}(\{a, b\})} |V_x|) \le \Oh(n^{\wed(\psi)})$. For a fixed pair of nodes $v_a \in V'_a$ and $v_b \in V'_b$, it takes constant time to decide whether an edge should be created and what should be its weight. Thus, in total it takes time $\Oh(n^{\wed(\psi)})$ to build $G'$.
\end{proof}

Despite the supremum in the definition of clique embeddings, it is known that the clique embedding power is always attained by an embedding from a finite clique $K_k$:

\begin{lemma}[{\cite[Proposition 16]{FanKZ23}}] \label{fct:clemb-for-fixed-k}
    For any pattern $H$, there is some $k$ and a clique embedding $\psi : K_k \to H$, such that $\clemb(\psi) = \clemb(H)$.
    \lipicsEnd
\end{lemma}

We can now turn clique embeddings into conditional lower bounds for subgraph problems.

\allHisocliquehardness*

\begin{proof}[Proof of \cref{all-H-iso-clique-hardness}]
    We start with the claim on \Hminiso. For the sake of contradiction assume that there is an algorithm that solves \Hminiso in time $\Oh(m^{\clemb(H) - \varepsilon})$ for some $\varepsilon > 0$. According to \cref{fct:clemb-for-fixed-k}, there is some $k$ and an embedding $\psi : K_k \to H$, such that $\clemb(\psi) = \clemb(H)$. We assume that $\clemb(H) > 1$ as otherwise the claim is obvious, since \Hminiso is not solvable in sublinear time. We now prove that \miniso{K_k} is solvable in time $\Oh(n^{k - \varepsilon'})$ for some $\varepsilon' > 0$, contradicting the \cliqueconj.
    Given an instance $G$ of \miniso{K_k} with $|V(G)| = n$, we use \cref{lm:clemb-reduction} to build an instance $G'$ of \Hminiso with $|E(G')| = \Oh(n^{\wed(\psi)})$ in time $\Oh(n^{\wed(\psi)})$, such that there is a weight-preserving constant-time-computable bijection between $K_k$-subgraphs in $G$ and $H$-subgraphs in $G'$. 
    Thus, solving \Hminiso on $G'$ in time $\Oh(|E(G')|^{\clemb(H) - \varepsilon}) = \Oh(n^{\wed(\psi) \cdot (\clemb(\psi) - \varepsilon)}) = \Oh(n^{k - \wed(\psi) \cdot \varepsilon})$, we solve \miniso{K_k} on $G$ in time $\Oh(n^{\wed(\psi)}) + \Oh(n^{k - \wed(\psi) \cdot \varepsilon}) = \Oh(n^{k - \varepsilon'})$ for $\varepsilon' \coloneqq \min\{k - \wed(\psi), \wed(\psi) \cdot \varepsilon\}$.
    Since $\clemb(\psi) > 1$ we have $\wed(\psi) < k$ and thus $\eps' > 0$, contradicting the \cliqueconj.

    \smallskip

    For proving the claim on \Hlistiso, by \cref{lm:zero-clique-to-clique-list} we may assume the \cliquelistconj.
    For the sake of contradiction assume that there is an algorithm that solves \Hlistiso in time $\tOh(m^{\clemb(H) - \varepsilon} + t)$ for some $\varepsilon > 0$, where $t$ is the number of $H$-subgraphs in the host graph. According to \cref{fct:clemb-for-fixed-k}, there exists some $k$ and an embedding $\psi : K_k \to H$, such that $\clemb(\psi) = \clemb(H)$. We assume that $\clemb(H) > 1$ as otherwise the lemma claim is obvious, since \Hlistiso is not solvable in sublinear time $\tOh(m^{1-\eps} + t)$. We now prove that \listiso{K_k} is solvable in time $\tOh(n^{k - \varepsilon'} + t')$ for some $\varepsilon' > 0$, where $t'$ is the number of $k$-cliques in the host graph, which contradicts the \cliquelistconj.
    Given an instance $G$ of \listiso{K_k} with $|V(G)| = n$, we use \cref{lm:clemb-reduction} (in the unweighted setting) to create an instance $G'$ of \Hlistiso with $|E(G')| = \Oh(n^{\wed(\psi)})$ in time $\Oh(n^{\wed(\psi)})$, such that there is a constant-time-computable bijection between $K_k$-subgraphs in $G$ and $H$-subgraphs in $G'$.
    Thus, solving \Hlistiso on $G'$ in time $\tOh(|E(G')|^{\clemb(H) - \varepsilon} + t) = \tOh(n^{\wed(\psi) \cdot (\clemb(\psi) - \varepsilon)} + t) = \tOh(n^{k - \wed(\psi) \cdot \varepsilon} + t)$, we solve \listiso{K_k} on $G$ in time $\Oh(n^{\wed(\psi)}) + \tOh(n^{k - \wed(\psi) \cdot \varepsilon} + t) = \tOh(n^{k - \varepsilon'} + t)$ for $\varepsilon' \coloneqq \min\{k - \wed(\psi), \wed(\psi) \cdot \varepsilon\}$, where again $\eps' > 0$. This contradicts the \cliquelistconj, and thus the \zerocliqueconj.
    
    \smallskip

	Now consider the claim on \Henumiso. For the sake of contradiction assume that there is an algorithm that solves \Henumiso in $\Oh(m^{\clemb(H) - \varepsilon})$ preprocessing time and $\tOh(1)$ delay for some $\varepsilon > 0$.
    We can then use this algorithm to solve \Hlistiso in time $\tOh(m^{\clemb(H) - \varepsilon} + t)$, where $t$ is the number of $H$-subgraphs in the host graph, by enumerating all $H$-subgraphs in the host graph. By the second claim of the lemma this contradicts the \zerocliqueconj.
\end{proof}

\subsection{Lower Bounds Based on Induced Minors}\label{sec:lower-bounds-until-p-graphs:induced-minor}

For induced minors, we prove a similar helper lemma as for clique embeddings (cf.~\cref{lm:clemb-reduction}).

\begin{lemma} \label{induced-minor-reduction}
    Let $H$ and $H'$ be patterns, where $H'$ is an induced minor of $H$. Let $G'$ be a (weighted) $m$-edge host graph for pattern $H'$.
    There is an algorithm that in time $\Oh(m)$ builds a (weighted) host graph $G$ for pattern $H$ such that $|V(G)| = \Oh(m)$, $|E(G)| = \Oh(m)$, and there is a (weight-preserving) constant-time-computable bijection $\corr$ between $H$-subgraphs in~$G$ and $H'$-subgraphs in $G'$.
\end{lemma}

\begin{proof}
	Again it suffices to prove the lemma in the weighted setting, as in the unweighted setting we can set all weights to zero. The proof is similar to \cref{lm:clemb-reduction}.
	
    As $H'$ is an induced minor of $H$, there exists a function $\varphi : V(H') \to 2^{V(H)}$, such that nodes of $H'$ are mapped to disjoint non-empty connected subsets of nodes of $H$, where $\{x, y\} \in E(H')$ for $x \neq y$ if and only if $\{a, b\} \in E(H)$ for some $a \in \varphi(x)$ and $b \in \varphi(y)$. Pick an arbitrary such function $\varphi$ and a function $\eemb : E(H') \to E(H)$ that maps each $\{x, y\} \in E(H')$ to some $\{a, b\} \in E(H)$ with $a \in \varphi(x)$ and $b \in \varphi(y)$.
Denote by $\varphi^{-1} : V(H) \to V(H') \cup \{\bot\}$ a function that maps each node $a$ of $H$ to the only node $x$ of $H'$, such that $a \in \varphi(x)$, or $\bot$ if no such node $x$ exists. As sets $\varphi(x)$ are disjoint, $\varphi^{-1}$ is uniquely defined.

Denote the parts of $G'$ by $V'_x$ for $x \in V(H')$.

We construct the graph $G$ as follows. For each $a \in V(H)$ we construct a part $V_a$ in $G$. 
If $\varphi^{-1}(a) = x$ for some $x \in V(H')$, then $V_a$ consists of $|V'_x|$ nodes that are in a one-to-one correspondence with nodes of $V'_x$. We denote the nodes in $V'_x$ by $v_{x,i}$ for $i \in [|V'_x|]$, and the nodes in $V_a$ by $v_{a,i}$ for $i \in [|V'_x|]$. 
Otherwise, we have $\varphi^{-1}(a) = \bot$, in this case $V_a$ consists of a single node $v_{a,1}$. 
To construct the edges of $G$, for each $\{a, b\} \in E(H)$, we add edges between $V_a$ and $V_b$ as follows.
If $\varphi^{-1}(a) = \bot$ or $\varphi^{-1}(b) = \bot$, we add edges of weight zero between all pairs of nodes in $V_a$ and $V_b$. If $\varphi^{-1}(a) = x = \varphi^{-1}(b)$ for some $x \in V(H')$, we add edges of weight zero between $v_{a,i}$ and $v_{b,i}$ for all $i \in [|V'_x|]$.
If $\varphi^{-1}(a) = x \neq y = \varphi^{-1}(b)$ for some $x, y \in V(H')$, for each edge $\{v'_{x,i}, v'_{y,j}\} \in E(G')$ between nodes in $V'_x$ and $V'_y$, we add an edge $\{v_{a,i}, v_{b,j}\}$ to $G$. Furthermore, if $\eemb(\{x, y\}) = \{a, b\}$, we set the weight of the edge $\{v_{a,i}, v_{b,j}\}$ to the weight of the edge $\{v'_{x,i}, v'_{y,j}\}$, otherwise we set it to zero.

We now prove that $|V(G)|, |E(G)| = \Oh(m)$. In each part $V_a$, we have either a single node or $|V'_x| = \Oh(m)$ nodes for $x = \varphi^{-1}(a)$. Thus, there are $\Oh(m)$ nodes in $G$ in total. Regarding edges, if $\varphi^{-1}(a) = \bot$ or $\varphi^{-1}(b) = \bot$, either $V_a$ or $V_b$ consists of a single node, and we add $\Oh(|V_a| + |V_b|) = \Oh(m)$ edges between these parts. If $\varphi^{-1}(a) = x = \varphi^{-1}(b)$, we add $|V'_x| = \Oh(m)$ edges between these parts. If $\varphi^{-1}(a) = x \neq y = \varphi^{-1}(b)$, the number of edges we add is equal to the number of edges between $V'_x$ and $V'_y$ in $G'$, and thus is $\Oh(m)$. Hence, there are $\Oh(m)$ edges in $G$ in total. It is easy to see that $G$ can also be constructed in time $\Oh(m)$.

We now show that there is a weight-preserving constant-time-computable bijection between $H$-subgraphs in $G$ and $H'$-subgraphs in $G'$. Fix an $H'$-subgraph $\bu'$ in $G'$. For each $x \in V(H')$, denote by $u'_x$ the node of $\bu'$ from $V'_x$. For any $a \in V(H)$, if $x = \phi^{-1}(a)$ then we have $u'_x = v_{x,i}$ for some $i$, and we set $u_a := v_{a,i}$. Otherwise, we have $\varphi^{-1}(a) = \bot$ and we set $u_a := v_{a,1}$. We claim that $\corr(\bu') \coloneqq \bu \coloneqq (u_a)_{a \in V(H)}$ forms an $H$-subgraph in $G$ and its weight is the same as the weight of $\bu'$.
Indeed, consider some edge $\{a, b\} \in E(H)$.
If $\varphi^{-1}(a) = \bot$ or $\varphi^{-1}(b) = \bot$, all pairs of nodes in $V_a$ and $V_b$ are connected, in particular, $u_a$ and $u_b$. If $\varphi^{-1}(a) = x = \varphi^{-1}(b)$, $u_a = v_{a,i}$ and $u_b = v_{b,i}$, where $u'_x = v_{x,i}$, and thus they are connected by an edge. If $\varphi^{-1}(a) = x \neq y = \varphi^{-1}(b)$, $u_a$ and $u_b$ are connected by an edge as they are nodes in these parts corresponding to $u'_x$ and $u'_y$ respectively, and they are connected by an edge in $G'$ due to the definition of $\varphi$ and the fact that $\bu'$ forms an $H'$-subgraph in $G'$.
Furthermore, for each $\{x, y\} \in E(H)$, the weight of $\{u'_x, u'_y\}$ is embedded into exactly one edge on $\bu$, specifically it is embedded in the weight of $\{u_a, u_b\}$ where $\{a, b\} = \eemb(\{x, y\})$. Thus, every $H'$-subgraph in $G'$ corresponds to an $H$-subgraph in $G$ of the same weight.
Furthermore, the mapping $\corr$ is an injection.
Indeed, let $\bu''$ be some other $H'$-subgraph in $G'$.
As $\bu' \neq \bu''$, we have $u'_x \neq u''_x$ for some $x \in V(H')$.
Pick some $a \in \varphi(x)$.
As $u'_x \neq u''_x$, $\corr(\bu')$ and $\corr(\bu'')$ pick different nodes in $V_a$.
Hence, $\corr$ is indeed an injection.

We now show the inverse correspondence. Fix some $H$-subgraph $\bu$ in $G$. For each $a \in V(H)$, denote $u_a$ the node of $\bu$ from $V_a$. 
For each $x \in V(H')$ pick an $a \in \varphi(x)$, write $u_a = v_{a,i}$ for some $i$, and set $u'_x \coloneqq v'_{x,i}$. We claim that $\bu' \coloneqq (u'_x)_{x \in V(H')}$ forms an $H'$-subgraph in $G'$ and $\corr(\bu') = \bu$.

First, we prove that $u_b = v_{b,i}$ for all $b \in \varphi(x)$, where $u'_x = v'_{x,i}$. Let $a$ be the node that we picked in $\varphi(x)$ to set $u'_x$. Thus, $u_a = v_{a,i}$. By the definition of $\varphi$, $H[\varphi(x)]$ is connected, thus there exists a path $a = d_1 \edg d_2 \edg \cdots \edg d_i = b$ between $a$ and $b$ in $H$, such that all nodes of this path are from $\varphi(x)$. If we now consider $u_{d_1}, u_{d_2}, \ldots, u_{d_i}$, such nodes form a path in $G$ because $\bu$ is an $H$-subgraph in $G$. For each $j \in [i-1]$, an edge between $u_{d_j}$ and $u_{d_{j+1}}$ exists if and only if $u_{d_j} = v_{d_j,k}$ and $u_{d_{j+1}} = v_{d_{j+1},k}$ for some $k$ because $\varphi^{-1}(d_j) = x = \varphi^{-1}(d_{j+1})$. Thus, by transitivity $u_{d_1} = u_a = v_{a,k}$ and $u_{d_i} = u_b = v_{b,k}$ for some $k$. But as $u_a = v_{a,i}$, we get that $k=i$, thus proving the claim. Furthermore, for each $a \in V(H)$, such that $\varphi^{-1}(a) = \bot$, $\bu$ chooses the only available node of $V_a$.
Thus, indeed $\corr(\bu') = \bu$.
Consider some $\{x, y\} \in E(H')$.
By the definition of $\varphi$, there is an edge $\{a, b\} \in E(H)$ such that $a \in \varphi(x)$ and $b \in \varphi(y)$.
As $\bu$ forms an $H$-subgraph in $G$, we have that $u_a$ and $u_b$ are adjacent, and thus $u'_x$ and $u'_y$ are adjacent due to the definition of graph $G$.
Hence, $\bu'$ forms a $H'$-subgraph in $G'$.

Thus, $\corr$ is an injection and a surjection.
Furthermore, computation of $\corr$ and $\corr^{-1}$ can be straightforwardly done in constant time.
\end{proof}

\inducedminor*

\begin{proof}[Proof of \cref{induced-minor}]
    We first prove the claim for \Hminiso. Given a host graph $G'$ for \miniso{H'}, we apply \cref{induced-minor-reduction} to build a host graph $G$ for \Hminiso in time $\Oh(m)$. Given $G$, we solve \Hminiso in time $\Oh(T(m))$. From the solution for \Hminiso on $G$, in constant time we get the solution for \miniso{H'} on $G'$. The overall time complexity is $\Oh(T(m))$.

    \smallskip

    For \Hlistiso, we apply \cref{induced-minor-reduction} in the unweighted setting. Given a list of all $H$-subgraphs in $G$, we can convert each one of them into a corresponding $H'$-subgraph in $G'$, which solves the problem in time $\Oh(T(m, t))$.

    \smallskip

    For \Henumiso, we proceed similarly as for \Hlistiso. Whenever the enumeration algorithm generates an $H$-subgraph in $G$, we convert it into an $H'$-subgraph in $G'$ in constant time. The delay thus remains bounded by $\Oh(D(m))$.
\end{proof}

\clembofinducedminors*

\begin{proof}[Proof of \cref{clemb-of-induced-minors}]
	Let $\varphi, \varphi^{-1}$, and $f$ be the same functions as in the proof of \cref{induced-minor-reduction}.

Fix some clique embedding $\psi' : K_k \to 2^{V(H')}$.
We create a clique embedding $\psi : K_k \to 2^{V(H)}$ with $\wed(\psi) \le \wed(\phi')$.
Hence, we obtain $\clemb(H') \le \clemb(H)$ by taking supremum over all $k$ and all clique embeddings $\psi'$.

We define $\psi$ in the following way:
\[
    \psi(v) = \bigcup_{a \in \psi'(v)} \varphi(a)
\]
for every $v \in V(K_k)$.

Denote $\psi'^{-1}(a) \coloneqq \{v \in V(K_k) \mid a \in \psi'(v)\}$ for every $a \in V(H')$.
Analogously, $\psi^{-1}(a) \coloneqq \{v \in V(K_k) \mid a \in \psi(v)\}$ for every $a \in V(H)$.

We first show that it is a valid clique embedding.
Note that $\psi(v)$ is an image in $V(H)$ of $\varphi$ of some connected subset $\psi'(v)$ of nodes of $V'(H)$.
An image of a connected subset of nodes in a minor is a connected subset of nodes in the graph.
Furthermore, as $\psi'$ is a clique embedding, for every $u, v \in V(K_k)$ with $u \ne v$, there is some edge $\{x, y\} \in E(H')$, such that $\psi'(v) \cap \{x, y\} \neq \emptyset$ and $\psi'(u) \cap \{x, y\} \neq \emptyset$.
Let $\{a, b\} \coloneqq \eemb(\{x, y\})$.
By the definition of $\psi$, we have $\psi(v) \cap \{a, b\} \neq \emptyset$ and $\psi(u) \cap \{a, b\} \neq \emptyset$.
Thus, $\psi$ is indeed a clique embedding.

It remains to show that $\wed(\psi) \le \wed(\psi')$.
Consider any edge $\{a, b\} \in E(H)$.
We claim that there is some edge $\{x, y\} \in E(H')$ such that $\psi^{-1}(a) \cup \psi^{-1}(b) \subseteq \psi'^{-1}(x) \cup \psi'^{-1}(y)$.
If $\bot \neq \varphi^{-1}(a) \neq \varphi^{-1}(b) \neq \bot$, there is an edge $\{x \coloneqq \varphi^{-1}(a), y \coloneqq \varphi^{-1}(b)\} \in E(H')$ as $H'$ is an induced minor of $H$. We have $\psi^{-1}(a) \subseteq \psi'^{-1}(x)$ and $\psi^{-1}(b) \subseteq \psi^{-1}(y)$.
If $\varphi^{-1}(a) = \varphi^{-1}(b)$, we have $\psi^{-1}(a), \psi^{-1}(b) \subseteq \psi'^{-1}(x)$ for $x \coloneqq \psi^{-1}(a)$.
As $H'$ is a pattern, there is some $y \in V(H')$ with $\{x, y\} \in E(H')$.
If $\varphi^{-1}(a) = \bot$, we have $\psi^{-1}(a) = \emptyset$ by the definition of $\psi$.
Such a case is similar to the case $\varphi^{-1}(a) = \varphi^{-1}(b)$.
The case of $\varphi^{-1}(b) = \bot$ is symmetric.

As $\psi^{-1}(a) \cup \psi^{-1}(b) \subseteq \psi'^{-1}(x) \cup \psi'^{-1}(y)$, we get that the depth of $\{a, b\}$ is at most the depth of $\{x, y\}$.
Hence, by taking maximum over all edges, we obtain $\wed(\psi) \le \wed(\psi')$.
\end{proof}

\subsection{Lower Bounds: from General Case to $P$-graphs}\label{sec:lower-bounds-until-p-graphs:lst}

In this section we prove \cref{small-clemb-implies-p}, which reduces general patterns to $P$-graphs.

Without defining treewidth, we will use the following two facts. 
The first is a characterization of graphs of treewidth at most two in terms of forbidden minors, see also~\cite{Ramachandramurthi97}. The next is a characterization of graphs of treewidth at most 2 in terms of series-parallel graphs. 

\begin{lemma}[\cite{ArnborgPC90,SatyanarayanaT90}] \label{lm:no-k4-is-tw-2}
  Let $H$ be a graph. The treewidth of $H$ is at most $2$ if and only if $H$ has no $K_4$ minor.
  \lipicsEnd
\end{lemma}

\begin{lemma}[\cite{BodlaenderF01}] \label{lm:tw-2-with-no-cut-vertices-is-series-parallel}
  Let $H$ be a graph. A \emph{block} in $H$ is a maximal connected subgraph of~$H$ that has no cut vertex. The treewidth of $H$ is at most $2$ if and only if every block of $H$ is a series-parallel graph.
  \lipicsEnd
\end{lemma}

\smallclembimpliesp*

\begin{proof}[Proof of \cref{small-clemb-implies-p}]
    For a proof by contradiction, assume that $H$ has neither $\goggr$ nor $K_4$ as an induced minor.
	Since $H$ has no $K_4$ induced minor, it also has no $K_4$ minor, so by \cref{lm:no-k4-is-tw-2,lm:tw-2-with-no-cut-vertices-is-series-parallel} every block of $H$ is a series-parallel graph.
	Note that a cut vertex is a clique separator of size one.
    Hence, as $H$ has no clique separator, $H$ is a series-parallel graph itself.
    The following claim now yields a contradiction.

    \begin{claim} \label{clm:series-parallel-without-goggles-is-p-graph}
        Let $H$ be a pattern. If $H$ is a series-parallel graph with no clique separator and no induced $\goggr$ minor then it is a $P$-graph.
    \end{claim}

    \begin{claimproof}

    If $H$ is a single edge, it is a $P$-graph $P(1)$. Otherwise, if the topmost composition in some series-parallel decomposition of $H$ is series, the middle node of this series composition is a cut vertex.
    Since $H$ has no clique separator and thus no cut vertex, we may assume that in every series-parallel decomposition of $H$ the topmost composition is parallel.
    Out of all series-parallel decompositions of $H$, choose the one that has the largest number of parallel compositions on the topmost level.
    Out of such decompositions choose the one where the distance between the source $s$ and the sink $t$ is maximized.

    Let $H_1, H_2, \ldots, H_k$ for $k \ge 2$ be the subgraphs of $H$ that constitute the topmost parallel composition, that is, (1) $H$ arises by taking the disjoint union of $H_1,\ldots,H_k$ and identifying all their source nodes and identifying all their sink nodes, and (2) no $H_i$ can be written as the parallel composition of series-parallel graphs.
    See \cref{series-parallel-general}.
    We consider two cases. %
\begin{figure}
\begin{center}
    \includegraphics[scale=0.75]{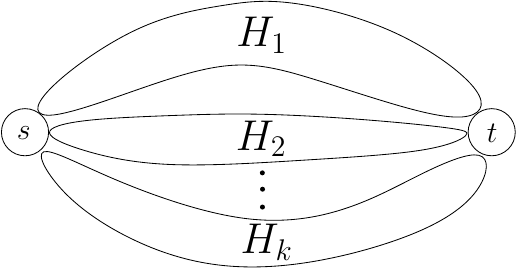}
\end{center}

\caption{Series-parallel graph $H$ from \cref{small-clemb-implies-p} decomposed into series parts.}
\label{series-parallel-general}
\end{figure}
    \begin{itemize}
        \item \emph{Case 1: $s$ and $t$ are adjacent in $H$.}
            The edge between $s$ and $t$ induces one of the subgraphs~$H_i$; without loss of generality let it be $H_1$.
            If $k \ge 3$, then the edge $\{s, t\}$ is a clique separator as its removal disconnects $H_2 \setminus \{s, t\}$ and $H_3 \setminus \{s, t\}$, which are nonempty as there cannot be two edges between $s$ and $t$.
            Hence, we may assume $k=2$.
            Let $H_2$ be composed as series of two subgraphs $H_{2, 1}$ and $H_{2, 2}$, and denote their common middle node by~$q$ (see \cref{series-parallel-ends-connected}).
            If $H_{2, 1}$ and $H_{2, 2}$ are single edges, the whole graph is a triangle, which is a $P$-graph $P(2, 1)$.
            Otherwise, assume that $H_{2, 1}$ is not a single edge (this is without loss of generality by symmetry).
            In this case $s$ and $q$ are not adjacent, as otherwise the edge $\{s,q\}$ would be a clique separator, as its removal disconnects $t$ and $V(H_{2, 1}) \setminus \{s, q\}$ (which is not empty since $H_{2, 1}$ is not a single edge).
            We then build a different series-parallel decomposition of $H$, where $s$ is the source, $q$ is the sink, and $H_{2, 1}$ and $(V(H_1) \cup V(H_{2, 2}), E(H_1) \cup E(H_{2, 2}))$ are parallel parts.
            See \cref{series-parallel-ends-connected}.
            We arrive at a contradiction as we assumed that the series-parallel decomposition we picked maximizes the number of parallel compositions on the topmost level, and out of such maximizes the distance between the source and the sink.
            The decomposition we picked has $k-1=1$ parallel compositions on the topmost level, and the distance between the source and the sink is one as $s$ and $t$ are adjacent.
            However, the new series-parallel decomposition has at least one parallel composition on the topmost level, and the distance between the source and the sink is at least two as $s$ and $q$ are not adjacent.
            Therefore, in this case the only valid graph is $P(2, 1)$.
    \begin{figure}
    \begin{center}
        \includegraphics[scale=0.65]{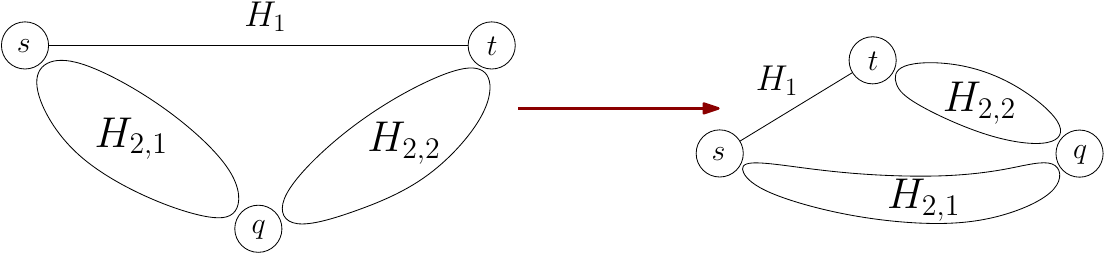}
    \end{center}

    \caption{Possible transformation of the series-parallel graph $H$ from Case 1 of \cref{small-clemb-implies-p}.}
    \label{series-parallel-ends-connected}
    \end{figure}
\item \emph{Case 2: $s$ and $t$ are not adjacent.}
    If $H$ is not a $P$-graph with endpoints $s$ and $t$, the parts $H_1, \ldots, H_k$ cannot all be paths, and thus for some part $H_i$ its constructions requires a parallel composition; 
    without loss of generality let this part be $H_1$.
    Then there are nodes $p,q$ in $H_1$ such that $H$ is the series composition of a graph $H_{1,1}$ with source $s$ and sink $p$, a graph $H_{1,2}$ with source $p$ and sink $q$, and a graph $H_{1,3}$ with source $q$ and sink $t$, and the graph $H_{1,2}$ is the parallel composition of two graphs $H_{1,2,1}$ and $H_{1,2,2}$ each with source $p$ and sink $q$. See \cref{series-parallel-ends-disconnected}. Note that here it may happen that $s$ coincides with $p$, in which case the graph $H_{1,1}$ consists of a single node $s=p$, or that $q$ coincides with $t$, in which case the graph $H_{1,3}$ consists of a single node $q=t$. However, both cannot happen at the same time, as otherwise $H_1$ itself arises from a parallel composition, contradicting our choice of $H_1,\ldots,H_k$ as the subgraphs that constitute the topmost parallel composition.
    Without loss of generality assume that $q \neq t$.
    Note that $V(H_{1,2}) \setminus \{p, q\}$ is not empty as $H_{1,2}$ arises from a parallel composition.
    It follows that $p$ and $q$ are not adjacent, as otherwise the edge $\{p,q\}$ would be a clique separator, as its removal disconnects $t$ and $V(H_{1,2}) \setminus \{p, q\}$.
    We now build a new series-parallel decomposition of $H$, where $p$ is the source, $q$ is the sink, and $H$ arises as the parallel composition of $H_{1,2}$ and everything else (i.e., the second part is formed by $H_{1,1}, H_{1,3}, H_2,\ldots, H_k$).
    See \cref{series-parallel-ends-disconnected}.
    Note that the new decomposition consists of at least 3 parallel parts, as $H_{1,2}$ arises by a parallel composition and thus consists of at least two parallel parts, and the remainder of the graph yields at least one additional parallel part. 
    Since we picked the original series-parallel decomposition to maximize the number of parallel parts, we obtain $k \ge 3$.

        \begin{figure}
        \begin{center}
            \includegraphics[scale=0.65]{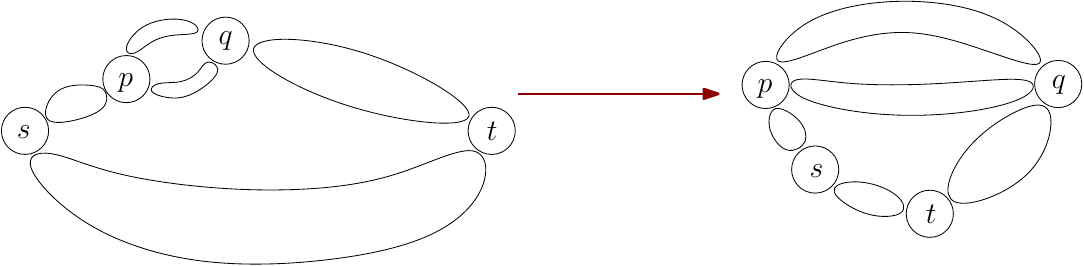}
        \end{center}

    \caption{Possible transformation of the series-parallel graph $H$ from Case 2 of \cref{small-clemb-implies-p}.}
        \label{series-parallel-ends-disconnected}
        \end{figure}
    
    We now turn our attention back to the original series-parallel decomposition, now with the additional knowledge that $k \ge 3$.
        We claim that $H$ has an induced $\goggr$ minor, which yields a contradiction.
        See \cref{series-parallel-goggles-minor}.
        We contract the graph $H_{1,1}$ with source $s$ and sink~$p$ into a single node, i.e., we contract the nodes $V(H_{1,1})$. 
        We contract the graph $H_{1,3}$ with source $q$ and sink $t$ to a single edge, i.e., we contract the nodes $V(H_{1,3}) \setminus \{q\}$.
        For the parallel parts $H_{1,2,1}$ and $H_{1,2,2}$ that constitute the graph $H_{1,2}$ with source $p$ and sink $q$, we contract both of them to paths of length two, i.e., we contract the nodes $V(H_{1,2,1}) \setminus \{p,q\}$ and the nodes $V(H_{1,2,2}) \setminus \{p,q\}$. This is possible because $p$ and $q$ are not adjacent.
        Analogously, we contract $H_2$ and $H_3$ into two paths of length two between $s$ and $t$; this is possible because $s$ and $t$ are not adjacent.
        We delete all nodes in $V(H_i) \setminus \{s, t\}$ for $i \ge 4$. This yields an induced $\goggr$ minor. We thus arrive at a contradiction.
        \begin{figure}
        \begin{center}
            \includegraphics[scale=0.65]{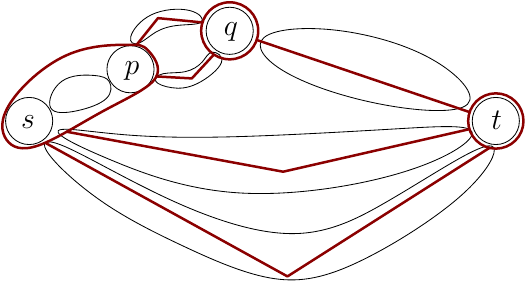}
        \end{center}

        \caption{Induced $H_{\textup{gog}}$ minor in $H$ in Case 2 of \cref{small-clemb-implies-p}. The graph $H$ is shown in black. An induced $H_{\textup{gog}}$ minor is shown in red.}
        \label{series-parallel-goggles-minor}
        \end{figure}
    \end{itemize}
    In each case we showed that $H$ is a $P$-graph, or we obtained a contradiction, which finishes the proof.
    \end{claimproof}

    This concludes the proof of the statement that $H$ has $\goggr$ or $K_4$ as an induced minor.
    We now show $\clemb(H) \ge 2$.
    Recall that in \cref{goggles-lb} we defined the goggles graph $\goggr$ and showed that $\clemb(\goggr) \ge 2$. 
	If $H$ has an induced $\goggr$ minor, then we have $\clemb(H) \ge \clemb(\goggr) \ge 2$ by \cref{clemb-of-induced-minors}.
	If $H$ has an induced $K_4$ minor, then $\clemb(H) \ge \clemb(K_4)$ by \cref{clemb-of-induced-minors}. The trivial embedding $\psi(v) = \{v\}$ of $K_4$ into $K_4$ has weak edge depth $2$, which shows $\clemb(H) \ge \clemb(K_4) \ge 4 / 2 = 2$.
\end{proof}

\subsection{Lower Bounds for $P(3, 3, 3)$} \label{sec:lower-bounds-until-p-graphs:p333}

We now turn to the specific pattern graph $P(3,3,3)$. As discussed in the proof overview, clique embeddings do not suffice to show that this pattern has complexity at least 2. We therefore prove conditional lower bounds assuming the \minplusconvconj or the \threesumconj, see \cref{sec:prelims-hypotheses} for their definitions.

\smallskip

For MinConv we make use of the standard equivalence with MinConv Verification.

\begin{hypothesis}[\minplusverifconj]
    For any $\varepsilon > 0$ there exists $c \ge 0$ such that there is no algorithm that, given sequences $A[1..n], B[1..n]$, and $C[1..n]$ with integer entries in the range $\{-n^c,\ldots,n^c\}$, in time $\Oh(n^{2 - \varepsilon})$ decides whether $C[k] \le A[i] + B[j]$ holds for all $i + j = k$.
    \lipicsEnd
\end{hypothesis}

\begin{lemma}[{\cite[Theorem 2]{CyganMWW19}}] \label{lm:minplus-to-minplus-verif}
    The \minplusconvconj is equivalent to the \minplusverifconj.
    \lipicsEnd
\end{lemma}

For our 3SUM-based lower bounds we make use of the following result that shows hardness of 3SUM Listing over small universe.\footnote{This result was known earlier, the specific reference \cite{FischerKP23} proves this result by a deterministic reduction.}

\begin{lemma}[{\cite[Lemma 4.6]{FischerKP23}}] \label{lm:three-sum-listing-hardness}
    For any $1 < \mu < 2$ and $\varepsilon > 0$, there is no algorithm that given a size-$n$ \threesum instance $A \subseteq [n^{\mu}]$ with $\Oh(n^{3 - \mu})$ solutions lists all solutions in time $\Oh(n^{2 - \varepsilon})$, unless the \threesumconj fails.
    \lipicsEnd
\end{lemma}

\allpthreethreethreeisocliquehardness*

\begin{proof}[Proof of \cref{all-p333-iso-clique-hardness}]
    For the first claim, by \cref{lm:minplus-to-minplus-verif}, we may use the \minplusverifconj instead of the \minplusconvconj.
    We reduce \minplusverif on sequences of length $n$ to \miniso{P(3, 3, 3)} on a host graph $G$ with $\Oh(n)$ nodes and edges in $\Oh(n)$ time.

    \minplusverif asks us to check, given integer sequences $A, B$, and $C$ of length~$n$, whether there are indices $i, j$, and $k$ such that $i + j = k$ and $A[i] + B[j] < C[k]$. Define $s \coloneqq \left\lceil \sqrt n \right\rceil $. Every index $i$ of $A$ is uniquely represented as $i_H \cdot s + i_L$ where $i_L \in [0, s - 1]$. Similarly, every index $j$ of $B$ is uniquely represented as $j_H \cdot s + j_L$ where $j_L \in [0, s - 1]$.

    Denote the two end-nodes of the paths of $P(3, 3, 3)$ by $u_H$ and $u_L$, the nodes of the first path by $u_H, a_H, a_L, u_L$, the nodes of the second path by $u_H, b_H, b_L, u_L$, and the nodes of the third path by $u_H, c_H, c_L, u_L$. See \cref{fig:p333-lower-bound-minplus}.

    We create a host graph $G$ as follows. The parts of $G$ are $V_{a_L} = V_{a_H} = V_{b_L} = V_{b_H} = V_{c_L} = V_{c_H} = [0, 2s]$ and $V_{u_L} = V_{u_H} = \{(i, j) \mid i, j \in [0, s] \}$. The edges are constructed in the following way.
Nodes $(i, j) \in V_{u_L}$ and $i' \in V_{a_L}$ are adjacent if and only if $i = i'$.
Nodes $(i, j) \in V_{u_L}$ and $j' \in V_{b_L}$ are adjacent if and only if $j = j'$.
Nodes $(i, j) \in V_{u_L}$ and $k' \in V_{c_L}$ are adjacent if and only if $i + j = k'$.
We add analogous edges between $V_{u_H}$ and $V_{a_H}$, between $V_{u_H}$ and $V_{b_H}$, and between $V_{u_H}$ and $V_{c_H}$.
All these edges are assigned weight zero.
Furthermore, $i_H \in V_{a_H}$ and $i_L \in V_{a_L}$ are adjacent if and only if $i_L \in [0, s - 1]$ and $i_H \cdot s + i_L \in [n]$, and the weight of this edge is $A[i_H \cdot s + i_L]$.
Similarly, $j_H \in V_{b_H}$ and $j_L \in V_{b_L}$ are adjacent if and only if $j_L \in [0, s - 1]$ and $j_H \cdot s + j_L \in [n]$, and the weight of this edge is $B[j_H \cdot s + j_L]$.
Finally, $k_H \in V_{c_H}$ and $k_L \in V_{c_L}$ are adjacent if and only if $k_H \cdot s + k_L \in [n]$, and the weight of this edge is $-C[k_H \cdot s + k_L]$.

\begin{figure}
    \begin{center}
        \includegraphics[scale=0.8]{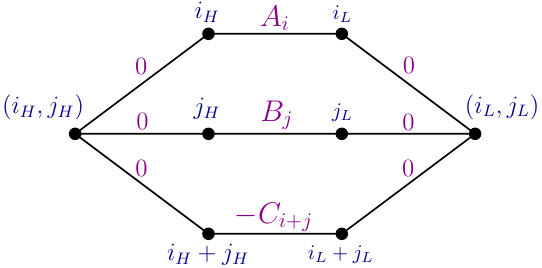}
    \end{center}

    \caption{Illustration of the \miniso{P(3, 3, 3)} instance constructed in the proof of \cref{all-p333-iso-clique-hardness}. The figure depicts the $P(3, 3, 3)$-subgraph corresponding to some triple $(A[i], B[j], C[k])$ with $i + j = k$.}
    \label{fig:p333-lower-bound-minplus}
\end{figure}

    We claim that there is a constant-time-computable bijection between $P(3, 3, 3)$-subgraphs of $G$ and triples $(i, j, k) \in [n]^3$ with $i + j = k$. Furthermore, the weight of such a $P(3, 3, 3)$-subgraph is equal to $A[i] + B[j] - C[k]$. Thus, there is a triple $(i, j, k)$ with $A[i] + B[j] < C[k]$ if and only if there is a $P(3, 3, 3)$-subgraph in $G$ of negative weight. 
    Hence, by finding the minimum weight of a $P(3, 3, 3)$-subgraph in $G$ we can decide the given MinConv Verification instance. 
    It remains to show a bijection.

    Let $\bv$ be a $P(3, 3, 3)$-subgraph of $G$. Denote $i_H \coloneqq v_{a_H}, i_L \coloneqq v_{a_L}, j_H \coloneqq v_{b_H}, j_L \coloneqq v_{b_L}, i \coloneqq i_H \cdot s + i_L, j \coloneqq j_H \cdot s + j_L$, and $k \coloneqq i + j$. We say that $\bv$ corresponds to a triple $\corr(\bv) \coloneqq (i, j, k)$. As $i_H$ and $i_L$ are adjacent, $j_H$ and $j_L$ are adjacent and $k_H$ and $k_L$ are adjacent, we get that $(i, j, k) \in [n]^3$ be the definition of $G$. Furthermore, $i + j = k$ by the definition of $k$. The weight of the edge between $i_H$ and $i_L$ is $A[i]$, the weight of the edge between $j_H$ and $j_L$ is $B[j]$, and the weight of the edge between $k_H$ and $k_L$ is $-C[k]$. All other edges have weight zero. Thus, the weight of $\bv$ is indeed equal to $A[i] + B[j] - C[k]$.
    
    We show that $\corr$ is injective. For a given $i$ and $j$, there is exactly one pair of nodes $(i_H, i_L) \in V_{a_H} \times V_{a_L}$ with $i_H \cdot s + i_L = i$ that is connected by an edge, and there is exactly one pair of nodes $(j_H, j_L) \in V_{b_H} \times V_{b_L}$ with $j_H \cdot s + j_L = j$ that is connected by an edge. For fixed $i_H$ and $j_H$, there is exactly one node $(i_H, j_H)$ in $V_{u_H}$ that is connected to both of them, and $(i_H, j_H)$ is connected to exactly one node $i_H + j_H$ in $V_{c_H}$. Analogously, the nodes in $V_{u_L}$ and $V_{c_L}$ are defined uniquely. Thus, for each $(i, j, k)$, at most one $P(3, 3, 3)$-subgraph $\bv$ of $G$ has $\corr(\bv) = (i, j, k)$.
    
    Finally, $\corr$ is surjective. Given $(i, j, k) \in [n]^3$ with $i + j = k$, define $\bv$, where $v_{a_H} \coloneqq i_H, v_{a_L} \coloneqq i_L, v_{b_H} \coloneqq j_H, v_{b_L} \coloneqq j_L, v_{c_H} \coloneqq i_H + j_H, v_{c_L} \coloneqq i_L + j_L, v_{u_H} \coloneqq (i_H, j_H)$, and $v_{u_L} \coloneqq (i_L, j_L)$. It is easy to see that $\bv$ forms a $P(3, 3, 3)$-subgraph in $G$ and $\corr(\bv) = (i, j, k)$.

    We now analyze the time complexity. There are $\Oh(s) = \Oh(\sqrt n)$ nodes in parts $V_{a_H}$, $V_{a_L}$, $V_{b_H}$, $V_{b_L}$, $V_{c_H}$, and $V_{c_L}$. Furthermore, there are $\Oh(s \cdot s) = \Oh(n)$ nodes in parts $V_{u_H}$ and $V_{u_L}$. Each node in $V_{u_H}$ and $V_{u_L}$ is connected to a constant number of nodes in $V_{a_H}, V_{b_H}$, and $V_{c_H}$ and $V_{a_L}, V_{b_L}$, and $V_{c_L}$, respectively. Between $V_{a_H}$ and $V_{a_L}$ there are at most $\Oh(|V_{a_H}| \cdot |V_{a_L}|) = \Oh(s^2) = \Oh(n)$ edges. Analogously, there are $\Oh(n)$ edges between $V_{b_H}$ and $V_{b_L}$ and between $V_{c_H}$ and $V_{c_L}$. Thus, $G$ has $\Oh(n)$ nodes and $\Oh(n)$ edges, and it can be built in time $\Oh(n)$.
    Hence, if \miniso{P(3, 3, 3)} can be solved in time $\Oh(m^{2 - \varepsilon})$ for some $\varepsilon > 0$, where $m$ is the number of edges in the host graph, using this reduction we can solve \minplusverif in time $\Oh(n^{2 - \varepsilon})$, which contradicts the \minplusverifconj and thus the \minplusconvconj.
    This concludes the proof of the first claim of the lemma statement.

    \medskip
    
    For the second claim, assume for the sake of contradiction that \listiso{P(3, 3, 3)} can be solved in time $\tOh(m^{2 - \varepsilon} + t)$ for some $0 < \varepsilon < 1$.
    We prove that the \threesumconj fails.
    By \cref{lm:three-sum-listing-hardness}, it suffices to list all solutions of a size-$n$ \threesum instance $A \subseteq [n^{\mu}]$ with $\Oh(n^{3 - \mu})$ solutions in time $\Oh(n^{2 - \varepsilon'})$ for some $1 < \mu < 2$ and $\varepsilon' > 0$.
    We pick $\mu = 1 + \frac{\varepsilon}{2}$ and $\varepsilon' = \frac{\varepsilon^2}{3}$.
    
    We claim that, given such a \threesumlist instance $A$, in time $\Oh(n^{\mu})$ we can construct an $\Oh(n^{\mu})$-edge host graph $G$ for $P(3, 3, 3)$ such that there is a constant-time-computable bijection between $P(3, 3, 3)$-subgraphs in $G$ and \threesum solutions on $A$.
    Using the assumed \listiso{P(3, 3, 3)} algorithm, we then list all $P(3, 3, 3)$-subgraphs in $G$ in time $\tOh((n^{\mu})^{2 - \varepsilon} + n^{3 - \mu}) = \tOh(n^{(1 + \frac{\varepsilon}{2}) \cdot (2 - \varepsilon)} + n^{2 - \frac{\varepsilon}{2}}) = \tOh(n^{2 - \frac{\varepsilon^2}{2}} + n^{2 - \frac{\varepsilon}{2}})\le \Oh(n^{2 - \varepsilon'})$.
    In time $\Oh(n^{3 - \mu}) \le \Oh(n^{2 - \varepsilon'})$ we can then convert the listed $P(3, 3, 3)$-subgraphs in $G$ into the solutions of \threesum on $A$.
    It remains to show the construction of $G$.

    Define $s \coloneqq \left\lceil \sqrt{n^{\mu}} \right\rceil $. We build a similar graph as in the first claim of the lemma, with the following differences: The edges do not have weights. We add an edge between $i_H \in V_{a_H}$ and $i_L \in V_{a_L}$ if and only if $i_L \in [0, s - 1]$ and $i_H \cdot s + i_L \in A$. Similarly, we add an edge between $j_H \in V_{b_H}$ and $j_L \in V_{b_L}$ if and only if $j_L \in [0, s - 1]$ and $j_H \cdot s + j_L \in A$. Finally, we add an edge between $k_H \in V_{c_H}$ and $k_L \in V_{c_L}$ if and only if $k_H \cdot s + k_L \in A$. 
    We get the same bijection as in the first claim of the lemma between $P(3, 3, 3)$-subgraphs in $G$ and triples $(i, j, k)$ such that $i + j = k$ and $i, j, k \in A$.
    Finally, $G$ has $\Oh(s^2) = \Oh(n^{\mu})$ edges, and can be built in time $\Oh(s^2) = \Oh(n^{\mu})$.

    \medskip
    
    For the third claim of the lemma, note that if there is such an algorithm for \enumiso{P(3, 3, 3)}, we can use it to solve \listiso{P(3, 3, 3)} in time $\tOh(m^{2 - \varepsilon} + t)$, where $t$ is the number of $P(3, 3, 3)$-subgraphs in $G$. This contradicts the \threesumconj by the second claim of the lemma.
\end{proof}

\section{Lower Bounds for Patterns in \boldmath$\family$} \label{sec:p-graph-lower-bounds}

In this section we prove the following lemma that was stated in \cref{sec:proof-overview}.

\renewcommand{\maybe}{\lipicsEnd}
\pgraphlowerbound*
\renewcommand{\maybe}{}

Throughout this section, $\alpha, \beta$, and $\gamma$ denote integers with $\alpha \ge \beta \ge 3$ and $\gamma \ge 1$.
We start by considering three simple graph families for which the clique embedding number is known, see \cref{edge-lb,cycle-lb,biclique-lb}. Then in \cref{sec:lower-bound-P-graphs-alpha-gamma} we cover all graphs of the form $P(\alpha, \gamma \times 2)$, and in \cref{sec:lower-bound-P-graphs-alpha-beta-gamma} we cover all graphs of the form $P(\alpha, \beta, \gamma \times 2)$.
Finally, in \cref{sec:lower-bound-P-graphs-combination} we put the pieces together to prove \cref{lm:p-graph-lower-bound}.

\begin{lemma} \label{edge-lb}
    We have $\clemb(K_2) \ge 1$.
\end{lemma}

\begin{proof}
    Pick any node $a$ in $K_2$. We create a clique embedding $\psi : K_3 \to K_2$ by setting $\psi(v) = a$ for each node $v \in V(K_3)$. The depth of the only edge in $K_2$ is $3$, and thus $\wed(\psi) = 3$, which yields $\clemb(K_2) \ge 3 / 3 = 1$.
\end{proof}

\begin{lemma}[{\cite[Lemma 17]{FanKZ23}}] \label{cycle-lb}
    For any integer $k \ge 3$, we have $\clemb(C_k) \ge 2 - 1/\lceil \frac{k}{2} \rceil$.\lipicsEnd
\end{lemma}

\begin{lemma}[{\cite[Proposition 19]{FanKZ23}}] \label{biclique-lb}
    For any integer $k \ge 2$, we have $\clemb(K_{2, k}) \ge 2 - \frac{1}{k}$.\lipicsEnd
\end{lemma}

\subsection{Lower bound for $P(\alpha, \gamma \times 2)$}
\label{sec:lower-bound-P-graphs-alpha-gamma}

\begin{lemma} \label{paths-touch}
    For integers $k \ge 3$ and $\ell_1, \ell_2 \in [0, k - 1]$ such that $\ell_1 + \ell_2 \ge k - 3$, any two paths of lengths $\ell_1$ and $\ell_2$ respectively in a cycle of length $k$ touch.
\end{lemma}

\begin{proof}
    The first path of length $\ell_1$ goes through $\ell_1 + 1$ nodes.
    Therefore, if the second path contains any of these $\ell_1 + 1$ nodes or the ones adjacent to the ends of the first path, the two paths touch each other.
    There are $\min\{\ell_1 + 3, k\}$ such nodes.

    The second path goes through $\ell_2 + 1$ nodes.
    We have $\min\{\ell_1 + 3, k\} + \ell_2 + 1 \ge k + 1$ as $\ell_1 + \ell_2 \ge k - 3$.
    Therefore, by the pigeonhole principle, the second path goes through either one of the nodes of the first path or one of the nodes that are adjacent to its ends.
    Hence, the two paths touch each other.
\end{proof}

For a graph $P(\alpha, \gamma \times 2)$ for $\alpha \ge 3$ and $\gamma \ge 1$ we fix its node numbering as follows (see \cref{Pac-node-numbering}). We denote the nodes on the path of length $\alpha$ by $a_0$, $a_1$, $\ldots$, $a_{\alpha - 1}$, $a_{\alpha}$, where $a_0$ and $a_{\alpha}$ are the ends of the path. We denote the nodes on the paths of length $2$ that are adjacent to both $a_0$ and $a_{\alpha}$ by $c_1$, $c_2$, $\ldots$, $c_{\gamma}$.

\begin{figure}
\begin{center}
    \includegraphics[scale=0.4]{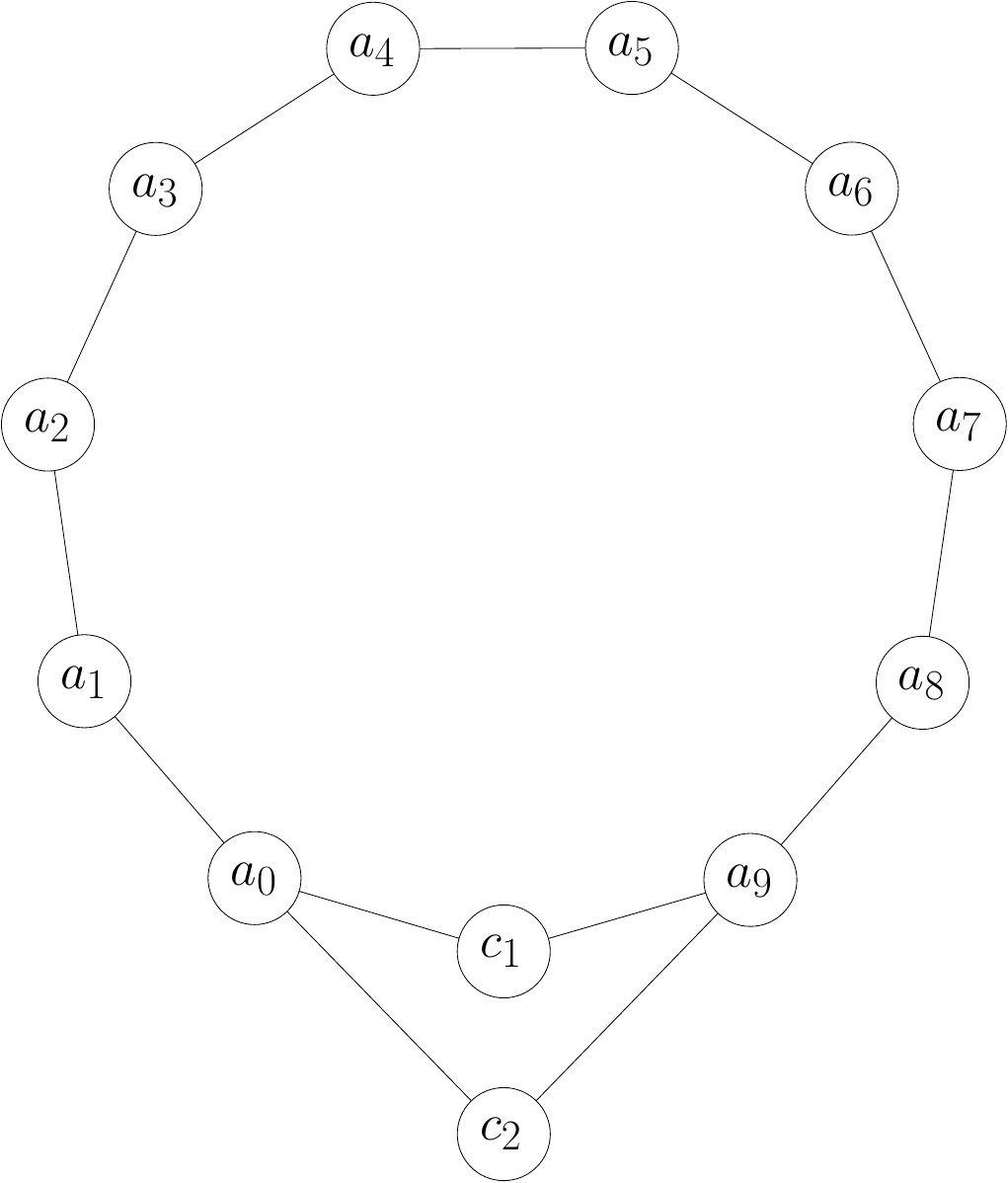}
\end{center}

\caption{An example of a node numbering for $P(\alpha, \gamma \times 2)$ for $\alpha = 9$ and $\gamma = 2$.}
\label{Pac-node-numbering}
\end{figure}

\begin{lemma} \label{lm:pa2c-clemb}
    For any $\alpha \ge 3$ and $\gamma \ge 2$ we have $\clemb(P(\alpha, \gamma \times 2)) \ge 2 - 1 / (2 \gamma + \left\lfloor \frac{\alpha - 1}{2} \right\rfloor)$.
\end{lemma}

\begin{proof}

    We prove the lemma for odd values of $\alpha$.
    For an even $\alpha$, the graph $P(\alpha, \gamma \times 2)$ has $P(\alpha - 1, \gamma \times 2)$ as an induced minor, and thus due to \cref{clemb-of-induced-minors} we get $\clemb(P(\alpha, \gamma \times 2)) \ge \clemb(P(\alpha - 1, \gamma \times 2)) \ge 2 - 1 / (2 \gamma + \left\lfloor \frac{\alpha - 2}{2} \right\rfloor) = 2 - 1 / (2 \gamma + \left\lfloor \frac{\alpha - 1}{2} \right\rfloor)$.
    
    Hence, assume $\alpha$ to be odd and let $\ell \coloneqq \frac{\alpha - 1}{2}$ (so $\alpha = 2\ell + 1$).
    We create a clique embedding $\psi$ from $K_{4 \gamma + 2\ell - 1}$ to $P(\alpha, \gamma \times 2)$ with $\wed(\psi) = 2 \gamma + \ell$, which yields $\clemb(P(\alpha, \gamma \times 2)) \ge \frac{4 \gamma + 2\ell - 1}{2 \gamma + \ell} = 2 - \frac{1}{2 \gamma + \ell}$.

    We treat the path of length $\alpha$ and node $c_1$ as a cycle $\cC$ of length $2 \ell + 3$.
    We use three different types of subgraphs that we embed into.
    See \cref{P922-alt-lower-bound} for an example.

    \begin{enumerate}
        \item[A.] We embed one node into each path of length $\ell$ in $\cC$.
        \item[B.] For the paths $a_0 \edg a_1 \edg \cdots \edg a_{\ell}$ and $a_{\ell + 1} \edg a_{\ell + 2} \edg \cdots \edg a_{\alpha}$ we embed additional $\gamma-1$ nodes (making their total embedded values $\gamma$).
        \item[C.] For each $i \in [2, \gamma]$, we embed one node into paths $c_i \edg a_0 \edg a_1 \edg \cdots \edg a_{\ell-1}$ and $c_i \edg a_{\alpha} \edg a_{\alpha - 1} \edg \cdots \edg a_{\ell + 2}$ each.
    \end{enumerate}

    In total, we embed one node into each one of $2\ell + 3$ subgraphs of type A, additional $\gamma-1$ node for two of subgraphs of type B, and one node for each of the $2 \gamma - 2$ subgraphs of type C.
    Hence, we embed $(2 \ell + 3) + 2 \cdot (\gamma - 1) + (2 \gamma - 2) = 4 \gamma + 2 \ell - 1$ nodes, which is exactly the number of nodes in $K_{4 \gamma + 2\ell - 1}$.

    It is clear that all subgraphs chosen in our embedding are connected. Let us argue that all pairs of embedded subgraphs touch.
    All pairs of subgraphs of type A and B touch due to \cref{paths-touch}.
    All pairs of subgraphs of type C touch either on the edge $\{a_0, c_i\}$ or on the edge $\{a_{\alpha}, c_i\}$ for some $i \in [\gamma]$.
    Furthermore, we claim that every path $p_1$ of type C touches every path $p_2$ of type A or B.
    Pick $i$ such that $p_1$ goes through $c_i$.
    If $a_0 \in p_2$, the two paths touch on the edge $\{a_0, c_i\}$.
    If $a_{\alpha} \in p_2$, they touch on the edge $\{a_{\alpha}, c_i\}$.
    Otherwise, $p_2$ lies inside the path $a_1 \edg a_2 \edg \cdots \edg a_{\alpha - 1}$.
    Hence, $p_1$ and $p_2$ lie inside the cycle consisting of nodes $a_0, a_1, \ldots, a_{\alpha}, c_i$.
    Such a cycle has length $2 \ell + 3$, and thus $p_1$ and $p_2$ touch due to \cref{paths-touch}.
    Hence, $\psi$ is a valid clique-embedding function.
    
    It remains to show $\wed(\psi) = 2 \gamma + \ell$.
    In isolation, the subgraphs of type A give weak edge depth of $\ell + 2$ to each edge in $\cC$ because every edge intersects exactly $\ell+2$ paths of length $\ell$.

    We claim that the subgraphs of types B and C increase weak edge depths of all edges in $\cC$ by $2 \gamma - 2$.
    Edges on the path $c_1 \edg a_0 \edg a_1 \edg \cdots \edg a_{\ell}$ intersect $\gamma - 1$ subgraphs of type C and the subgraph $a_0 \edg a_1 \edg \cdots \edg a_{\ell}$ of type B with weight $\gamma - 1$.
    Edges on the path $c_1 \edg a_{\alpha} \edg a_{\alpha-1} \edg \cdots \edg a_{\ell + 1}$ touch the symmetric subgraphs.
    The remaining edge $\{a_{\ell}, a_{\ell+1}\}$ touches both subgraphs of type B with weight $\gamma - 1$ each and no subgraphs of type C.
    Hence, every edge in $\cC$ has weak edge depth $\ell + 2 + 2 \gamma - 2 = 2 \gamma + \ell$.

    It remains to argue that edges outside of $\cC$ have weak edge depth $2 \gamma + \ell$.
    Such edges have form $\{a_0, c_i\}$ or $\{a_{\alpha}, c_i\}$ for some $i \in [2, \gamma]$.
    We already know that weak edge depth of $\{a_0, c_1\}$ is $2 \gamma + \ell$.
    If we switch from this edge to $\{a_0, c_i\}$, we lose a single path of type A that ends in $c_1$ and gain a single path of form $c_i \edg a_{\alpha} \edg a_{\alpha - 1} \edg \cdots \edg a_{\ell + 2}$ of type C.
    Therefore, weak edge depth of $\{a_0, c_i\}$ is the same as for $\{a_0, c_1\}$.
    Edges of the form $\{a_{\alpha}, c_i\}$ can be checked symmetrically.
\end{proof}

\subsection{Lower bounds for $P(\alpha, \beta, \gamma \times 2)$}
\label{sec:lower-bound-P-graphs-alpha-beta-gamma}

For a graph $P(\alpha, \beta, \gamma \times 2)$ for $\alpha \ge \beta \ge 3$ and $\gamma \ge 1$ we fix its node numbering as follows (see \cref{Pabc-node-numbering}). We denote the nodes on the path of length $\alpha$ by $a_0$, $a_1$, $\ldots$, $a_{\alpha - 1}$, $a_{\alpha}$, where $a_0$ and $a_{\alpha}$ are the ends of the path. We denote the nodes on the path of length $\beta$ by $b_0$, $b_1$, $\ldots$, $b_{\beta - 1}$, $b_{\beta}$. Note that $a_0 = b_0$ and $a_{\alpha} = b_{\beta}$. We denote the nodes in the paths of length $2$ that are adjacent to both $a_0$ and $a_{\alpha}$ by $c_1$, $c_2$, $\ldots$, $c_{\gamma}$.

\begin{figure}
\begin{center}
    \includegraphics[scale=0.35]{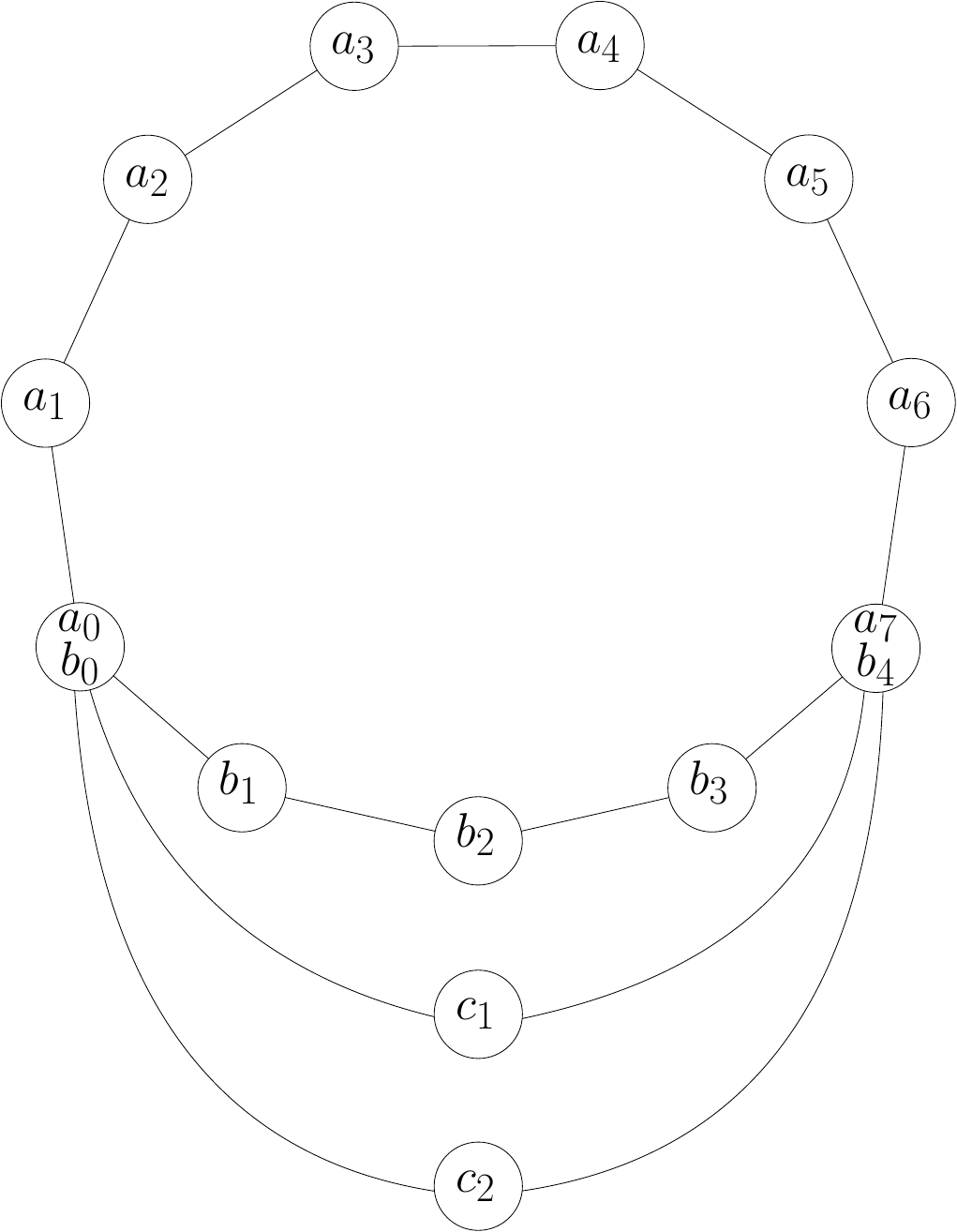}
\end{center}

\caption{An example of a node numbering for $P(\alpha, \beta, \gamma \times 2)$ for $\alpha = 7$, $\beta=4$, and $\gamma = 2$.}
\label{Pabc-node-numbering}
\end{figure}

\begin{lemma} \label{lb-parameterized-bound-a-b-c-1}
    For any $\alpha \ge \beta \ge 5$, $\alpha + \beta$ odd, and $\gamma \ge 1$, we have $\clemb(P(\alpha, \beta, \gamma \times 2)) \ge 2 - 1 / (2 \beta \gamma + \frac{\alpha \beta}{2} - \frac{\beta^2}{2} + 2 \beta - \frac{\alpha}{2} - 4 \gamma - \frac{3}{2})$.
\end{lemma}

\begin{proof}
    As $\alpha + \beta$ is odd and $\alpha \ge \beta$, we have $\alpha \ge \beta + 1$.

    Define $X \coloneqq 2 \beta \gamma + \frac{\alpha \beta}{2} - \frac{\beta^2}{2} + 2 \beta - \frac{\alpha}{2} - 4 \gamma - \frac{3}{2} = (\beta - 1) \cdot (2 \gamma + \frac{\alpha - \beta + 3}{2}) - 2 \gamma$.
    We create a clique embedding $\psi$ from $K_{2X - 1}$ to $P(\alpha, \beta, \gamma \times 2)$ with $\wed(\psi) = X$, which yields $\clemb(P(\alpha, \beta, \gamma \times 2)) \ge \frac{2X - 1}{X} = 2 - \frac{1}{X}$, as desired.
    We treat the paths $a_0 \edg a_1 \edg \cdots \edg a_{\alpha}$ and $b_0 \edg b_1 \edg \cdots \edg b_{\beta}$ as a cycle $\cC$ of length $\alpha + \beta$.
    We use $6$ different types of subgraphs that we embed into. See \cref{lb-parameterized-bound-a-b-c-1-picture-alt}.

    \begin{enumerate}
        \item[A.] We embed one node into each path of the form $a_{\ell} \edg a_{\ell + 1} \edg \cdots \edg a_{\ell + \frac{\alpha + \beta - 3}{2} - 1}$ for $\ell \in [\frac{\alpha - \beta + 3}{2}]$.
        \item[B.] We embed $\gamma + \frac{\alpha - \beta + 3}{2}$ nodes into the path $b_{\beta - 2} \edg b_{\beta - 3} \edg \cdots \edg b_0 \edg a_1 \edg \cdots \edg a_{\frac{\alpha - \beta + 3}{2}-1}$.
            Furthermore, we embed $2\gamma + \frac{\alpha - \beta + 3}{2}$ nodes into each path of the form $b_{\beta - 3 - \ell} \edg b_{\beta - 4 - \ell} \edg \cdots \edg b_0 \edg a_1 \edg \cdots \edg a_{\frac{\alpha - \beta + 3}{2} + \ell}$ for $\ell \in [0, \beta - 4]$ (note that there is at least one such $\ell$ as $\beta \ge 5$).
        \item[C.] Symmetrically, we embed $\gamma + \frac{\alpha - \beta + 3}{2}$ nodes into the path $a_{\frac{\alpha + \beta - 3}{2} + 1} \edg a_{\frac{\alpha + \beta - 3}{2} + 2} \edg \cdots \edg a_{\alpha} \edg b_{\beta - 1} \edg \cdots \edg b_{2}$.
            Furthermore, we embed $2\gamma + \frac{\alpha - \beta + 3}{2}$ nodes into each path of the form $a_{\frac{\alpha + \beta - 3}{2} - \ell} \edg a_{\frac{\alpha + \beta - 3}{2} - \ell + 1} \edg \cdots \edg a_{\alpha} \edg b_{\beta - 1} \edg \cdots \edg b_{\ell + 3}$ for $\ell \in [0, \beta - 4]$ (note that there is at least one such $\ell$ as $\beta \ge 4$).
        \item[D.] We embed one node into each path of the form $a_{\frac{\alpha + \beta - 3}{2} + \ell} \edg a_{\frac{\alpha + \beta - 3}{2} + \ell + 1} \edg \cdots \edg a_{\alpha} \edg c_1 \edg a_0 \edg a_1 \edg \cdots \edg a_{\ell - 2}$ for $\ell \in [2, \frac{\alpha - \beta + 3}{2}]$ (note that there is at least one such $\ell$ as $\alpha - \beta \ge 1$).
        \item[E.] We embed one node into each path of the form $c_{\ell} \edg a_0 \edg a_1 \edg \cdots \edg a_{\frac{\alpha - \beta + 3}{2} - 1}$ for $\ell \in [\gamma]$.
        \item[F.] We embed one node into each path of the form $a_{\frac{\alpha + \beta - 3}{2} + 1} \edg a_{\frac{\alpha + \beta - 3}{2} + 2} \edg \cdots \edg a_{\alpha} \edg c_{\ell}$ for $\ell \in [\gamma]$.
    \end{enumerate}

    \begin{figure}
    \begin{center}
        \includegraphics[scale=0.8]{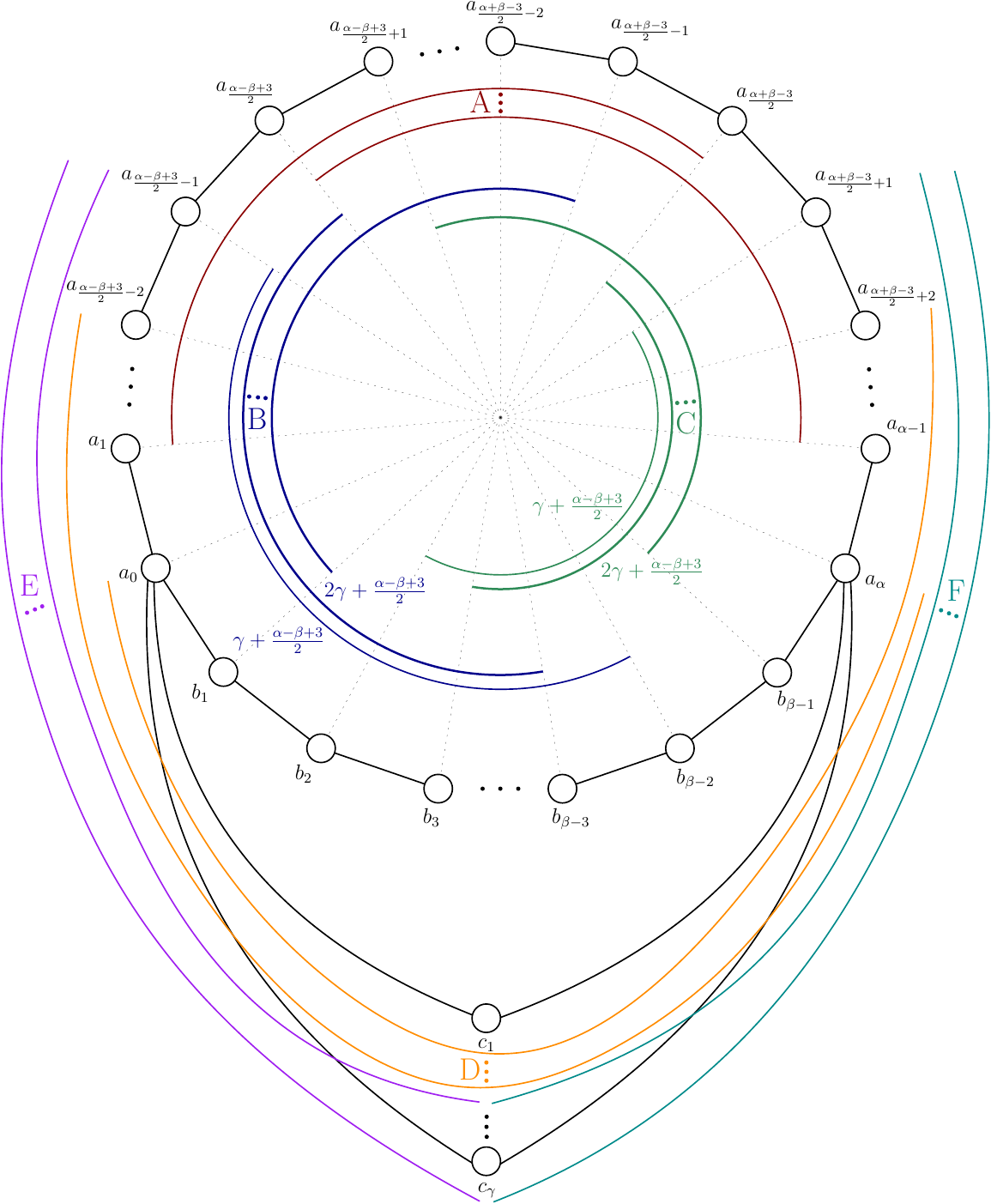}
    \end{center}

    \caption{Illustration of the proof of \cref{lb-parameterized-bound-a-b-c-1}. In addition to the abstract situation of \cref{lb-parameterized-bound-a-b-c-1}, the figure also shows the clique embedding for $P(12, 7, 2, 2)$.}
    \label{lb-parameterized-bound-a-b-c-1-picture-alt}
    \end{figure}

    We now analyze this construction.
    In total, we embed
    \begin{align*}
        &\frac{\alpha - \beta + 3}{2} \cdot 1 & \text{(A)}\\
        &+ (\gamma + \frac{\alpha - \beta + 3}{2}) + (\beta - 3) \cdot (2\gamma + \frac{\alpha - \beta + 3}{2}) & \text{(B)}\\
        &+ (\gamma + \frac{\alpha - \beta + 3}{2}) + (\beta - 3) \cdot (2\gamma + \frac{\alpha - \beta + 3}{2}) & \text{(C)}\\
        &+ (\frac{\alpha - \beta + 3}{2} - 2 + 1) \cdot 1 & \text{(D)}\\
        &+ \gamma \cdot 1 & \text{(E)}\\
        &+ \gamma \cdot 1 & \text{(F)}\\
        &= 2 \cdot ((\beta - 1) \cdot (2 \gamma + \frac{\alpha - \beta + 3}{2}) - 2 \gamma) - 1
	\end{align*}
    nodes, which is exactly the number of nodes in $K_{2X - 1}$.
    It is clear that all subgraphs into which we embed are connected.
    Next, we prove that all pairs of these subgraphs touch.
    For every type from A to F we show that the paths of this type touch each other and all later types.
    \begin{enumerate}
    \item[A.] All paths of type A start not later than node $a_{\frac{\alpha - \beta + 3}{2}}$ and end not earlier than node $a_{\frac{\alpha + \beta - 3}{2}}$.
        Note that $\frac{\alpha - \beta + 3}{2} < \frac{\alpha + \beta - 3}{2}$ as $\beta \ge 5$.
        Thus, all these paths contain all nodes $a_{\frac{\alpha - \beta + 3}{2}}, \ldots, a_{\frac{\alpha + \beta - 3}{2}}$, and all pairs of them touch.
        All paths of type B end not earlier than node $a_{\frac{\alpha - \beta + 3}{2} - 1}$, so they touch all paths of type A on the edge $\{a_{\frac{\alpha - \beta + 3}{2} - 1}, a_{\frac{\alpha - \beta + 3}{2}}\}$.
        Symmetrically, all paths of type C start not later than $a_{\frac{\alpha + \beta - 3}{2} + 1}$, so they all touch all paths of type A on the edge $\{a_{\frac{\alpha + \beta - 3}{2}}, a_{\frac{\alpha + \beta - 3}{2} + 1}\}$.
        For any fixed $\ell$, paths of types D (for $\ell=1$), E (for $\ell \in [\gamma]$), and F (for $\ell \in [\gamma]$) lie on the cycle $a_0 \edg a_1 \edg \cdots \edg a_{\alpha} \edg c_{\ell}$ of length $\alpha + 2$, which also contains all paths of type A.
        The paths of types D, E, and F have length $\frac{\alpha - \beta + 3}{2}$, while the paths of type A have length $\frac{\alpha + \beta - 3}{2} - 1$.
        Thus, \cref{paths-touch} implies that all pairs of them touch.

    \item[B.] All paths of type B go through node $a_0$, so all pairs of them touch.
        All paths of type B lie on the cycle $\cC$ of length $\alpha + \beta$ together with all paths of type C.
        All of these paths have length $\frac{\alpha + \beta - 3}{2}$.
        Thus, \cref{paths-touch} implies that each path of type B and each path of type C touch.
        For any fixed $\ell$, paths of types D (for $\ell=1$), E (for $\ell \in [\gamma]$), and F (for $\ell \in [\gamma]$) go through node $c_{\ell}$, so they touch paths of type B on the edge $\{a_0, c_{\ell}\}$.

    \item[C.] All paths of type C go through node $a_{\alpha}$, so all pairs of them touch.
        For any fixed $\ell$, paths of types D (for $\ell=1$), E (for $\ell \in [\gamma]$), and F (for $\ell \in [\gamma]$) go through node $c_{\ell}$, so they touch paths of type C on the edge $\{a_{\alpha}, c_{\ell}\}$.

    \item [D-F.] Any path of type D, E, or F goes either through node $a_0$ or through node $a_{\alpha}$.
        Any other such path goes through node $c_{\ell}$ for some $\ell \in [\gamma]$.
        Hence, they touch either on the edge $\{a_0, c_{\ell}\}$ or on the edge $\{a_{\alpha}, c_{\ell}\}$.
\end{enumerate}
    
We showed that all pairs of subgraphs we embed touch, so $\psi$ is a valid clique embedding. It remains to show that all edges have weak edge depth at most $X$.
There are
\begin{align*}
    &(\gamma + \frac{\alpha - \beta + 3}{2}) + (\beta - 3) \cdot (2\gamma + \frac{\alpha - \beta + 3}{2}) & \text{(B)}\\
    &+ (\frac{\alpha - \beta + 3}{2} - 2 + 1) \cdot 1 & \text{(D)}\\
    &+ \gamma \cdot 1 & \text{(E)}\\
    &= (\beta - 1) \cdot (2 \gamma + \frac{\alpha - \beta + 3}{2}) - 2 \gamma - 1\\
    &= X - 1
\end{align*}

embedded subgraphs that go through node $a_0$.
Node $b_1$ is not contained in any paths that $a_0$ is not contained in.
Hence, weak edge depth of $\{a_0, b_1\}$ is equal to $X - 1$.
For all edges of the form $\{a_0, c_{\ell}\}$ for $\ell \in [\gamma]$, the only new path that $c_{\ell}$ adds to the subgraphs that are already embedded into $a_0$ are paths of type F, contributing one to weak edge depth of these edges.
So each edge $\{a_0, c_{\ell}\}$ has weak edge depth $(X - 1) + 1 = X$.
The same holds for the edge $\{a_0, a_1\}$ as the only new path that $a_1$ adds to the subgraphs that are already embedded into $a_0$ is the path of type A that starts in $a_1$, which contributes one to weak edge depth of this edge.
Whenever we now move from an edge $\{a_{\ell}, a_{\ell + 1}\}$ to an edge $\{a_{\ell + 1}, a_{\ell + 2}\}$ for $\ell \in [0, \frac{\alpha - \beta + 3}{2} - 2]$ (note that there is at least one such $\ell$ as $\alpha \ge \beta + 1$), we lose a single path of type D that ends in $a_{\ell}$, and we gain a single path of type A that starts in $a_{\ell + 2}$.
Both have weight one, so weak edge depth remains at $X$.
When we move from the edge $\{a_{\frac{\alpha - \beta + 3}{2} - 1}, a_{\frac{\alpha - \beta + 3}{2}}\}$ to the edge $\{a_{\frac{\alpha - \beta + 3}{2}}, a_{\frac{\alpha - \beta + 3}{2} + 1}\}$, we lose a single path of type B with weight $\gamma + \frac{\alpha - \beta + 3}{2}$ that ends in $a_{\frac{\alpha - \beta + 3}{2} - 1}$ and all paths of type E with total weight $\gamma$.
At the same time we gain a single path of type C with weight $2\gamma + \frac{\alpha - \beta + 3}{2}$ that starts in node $a_{\frac{\alpha - \beta + 3}{2} + 1}$.
In total, weak edge depth of $\{a_{\frac{\alpha - \beta + 3}{2}}, a_{\frac{\alpha - \beta + 3}{2} + 1}\}$ is again $X$.
Next, whenever we move from an edge $\{a_{\ell}, a_{\ell + 1}\}$ to an edge $\{a_{\ell + 1}, a_{\ell + 2}\}$ for $\ell \in [\frac{\alpha - \beta + 3}{2}, \frac{\alpha + \beta - 3}{2} - 2]$ (there is at least one such $\ell$ because $\beta \ge 5$), we lose a single path of type B with weight $2\gamma + \frac{\alpha - \beta + 3}{2}$ that ends in node $a_{\ell}$, and we gain a single path of type C with weight $2\gamma + \frac{\alpha - \beta + 3}{2}$ that starts in node $a_{\ell + 2}$, so weak edge depth remains at $X$.
When we move from the edge $\{a_{\frac{\alpha + \beta - 3}{2} - 1}, a_{\frac{\alpha + \beta - 3}{2}}\}$ to the edge $\{a_{\frac{\alpha + \beta - 3}{2}}, a_{\frac{\alpha + \beta - 3}{2} + 1}\}$, we lose a single path of type B with weight $2\gamma + \frac{\alpha - \beta + 3}{2}$ that ends in node $a_{\frac{\alpha + \beta - 3}{2} - 1}$.
We gain a path of type C with weight $\gamma + \frac{\alpha - \beta + 3}{2}$ that starts in node $a_{\frac{\alpha + \beta - 3}{2} + 1}$ and all paths of type F with total weight $\gamma$.
Hence, weak edge depth of $\{a_{\frac{\alpha + \beta - 3}{2}}, a_{\frac{\alpha + \beta - 3}{2} + 1}\}$ is again $X$.
Whenever we move from an edge $\{a_{\ell}, a_{\ell + 1}\}$ to an edge $\{a_{\ell + 1}, a_{\ell + 2}\}$ for $\ell \in [\frac{\alpha + \beta - 3}{2}, \alpha - 2]$ (there is at least one such $\ell$ as $\alpha \ge \beta + 1$), we lose a single path of type A that ends in node $a_{\ell}$ and gain a single path of type D that starts in node $a_{\ell + 2}$.
Both of them have weight one, so weak edge depth of the current edge remains at $X$.
If we move from the edge $\{a_{\alpha - 1}, a_{\alpha}\}$ to an edge $\{a_{\alpha}, c_{\ell}\}$ for $\ell \in [\gamma]$, we lose a single path of type A that ends in $a_{\alpha - 1}$ and gain a single path of type E that starts in $c_{\ell}$, so each such edge has weak edge depth $X$.
When we move from the edge $\{a_{\alpha - 1}, a_{\alpha}\}$ to the edge $\{b_{\beta - 1}, b_{\beta}\}$, we do not gain any new subgraphs, but we lose a single path of type A that ends in  $a_{\alpha - 1}$, so weak edge depth of $\{b_{\beta - 1}, b_{\beta}\}$ is $X - 1$.
When we move from the edge $\{b_{\beta - 1}, b_{\beta}\}$ to the edge $\{b_{\beta - 2}, b_{\beta - 1}\}$, we lose all paths of types D and F of total weight $\frac{\alpha - \beta + 3}{2} - 1 + \gamma$, and we gain a single path of type B with weight $\gamma + \frac{\alpha - \beta + 3}{2}$ that starts in $b_{\beta - 2}$.
In total, weak edge depth of $\{b_{\beta - 2}, b_{\beta - 1}\}$ is $(X - 1) - (\frac{\alpha - \beta + 3}{2} - 1 + \gamma) + (\gamma + \frac{\alpha - \beta + 3}{2}) = X$.
Whenever we now move from an edge $\{b_{\ell + 1}, b_{\ell + 2}\}$ to an edge $\{b_{\ell}, b_{\ell + 1}\}$ for $\ell \in [\beta - 3]$ (note that there is at least one such $\ell$ as $\beta \ge 5$), we lose a single path of type C with weight $2\gamma + \frac{\alpha - \beta + 3}{2}$ that ends in node $b_{\ell + 2}$, and we gain a single path of type B with the same weight that starts in node $b_{\ell}$, so weak edge depth of the current edge remains at $X$.
This way we checked that weak edge depth of every edge is at most $X$, which concludes the proof of the fact that $\wed(\psi) = X$.
\end{proof}

The proofs of the following two lemmas are very similar to the previous one.

\begin{lemma} \label{lb-parameterized-bound-a-b-c-3}
    For any $\alpha \ge \beta \ge 3$, $\alpha + \beta$ odd, and $\gamma \ge 1$, we have $\clemb(P(\alpha, \beta, \gamma \times 2)) \ge 2 - 1 / (2\beta\gamma + \frac{\alpha\beta}{2} - \frac{\beta^2}{2} + \beta - \frac{\alpha}{2} - 2\gamma + \frac{3}{2})$.
\end{lemma}

\begin{proof}
    As $\alpha + \beta$ is odd and $\alpha \ge \beta$, we have $\alpha \ge \beta + 1$.

    Define $X \coloneqq 2\beta\gamma + \frac{\alpha\beta}{2} - \frac{\beta^2}{2} + \beta - \frac{\alpha}{2} - 2\gamma + \frac{3}{2} = (\beta - 1) \cdot (2\gamma + \frac{\alpha - \beta - 1}{2} + 1) + 2$.
    We create a clique embedding $\psi$ from $K_{2X - 1}$ to $P(\alpha, \beta, \gamma \times 2)$ with $\wed(\psi) = X$, which yields $\clemb(P(\alpha, \beta, \gamma \times 2)) \ge \frac{2X - 1}{X} = 2 - \frac{1}{X}$, as desired.
    We treat the paths $a_0 \edg a_1 \edg \cdots \edg a_{\alpha}$ and $b_0 \edg b_1 \edg \cdots \edg b_{\beta}$ as a cycle $\cC$ of length $\alpha + \beta$.
    We use $6$ different types of subgraphs that we embed into. See \cref{lb-parameterized-bound-a-b-c-3-picture}.

    \begin{enumerate}
        \item[A.] We embed one node into each path of the form $a_{\ell} \edg a_{\ell + 1} \edg \cdots \edg a_{\ell + \frac{\alpha + \beta + 1}{2} - 2}$ for $\ell \in [\frac{\alpha - \beta - 1}{2} + 1]$.
        \item[B.] We embed one node into the path $b_{\beta - 1} \edg b_{\beta - 2} \edg \cdots \edg b_1 \edg a_0 \edg \cdots \edg a_{\frac{\alpha - \beta - 1}{2}}$.
    Furthermore, we embed $2\gamma + \frac{\alpha - \beta - 1}{2} + 1$ nodes into each path of the form $b_{\beta - \ell - 1} \edg b_{\beta - \ell - 2} \edg \cdots \edg b_{1} \edg a_0 \edg a_1 \edg \cdots \edg a_{\frac{\alpha - \beta - 1}{2} + \ell}$ for $\ell \in [\beta - 2]$.
    Moreover, we embed $\gamma + 1$ nodes into the path $a_0 \edg a_1 \edg \cdots \edg a_{\frac{\alpha + \beta + 1}{2} - 2}$.
        \item[C.] Symmetrically, we embed one node into the path $a_{\frac{\alpha + \beta + 1}{2}} \edg a_{\frac{\alpha + \beta + 1}{2} + 1} \edg \cdots \edg a_{\alpha - 1} \edg b_{\beta} \edg b_{\beta - 1} \edg \cdots \edg b_{1}$.
    Furthermore, we embed $2\gamma + \frac{\alpha - \beta - 1}{2} + 1$ nodes  into each path of the form $a_{\frac{\alpha - \beta - 1}{2} + \ell} \edg a_{\frac{\alpha - \beta - 1}{2} + \ell + 1} \edg \cdots \edg a_{\alpha - 1} \edg b_{\beta} \edg b_{\beta - 1} \edg \cdots \edg b_{\beta - \ell + 2}$ for $\ell \in [3, \beta]$.
    Moreover, we embed $\gamma + 1$ nodes into the path $a_{\frac{\alpha - \beta - 1}{2} + 2} \edg a_{\frac{\alpha - \beta - 1}{2} + 3} \edg \cdots \edg a_{\alpha - 1} \edg a_{\alpha}$.
        \item[D.] If $\alpha \ge \beta + 3$, we embed one node into each path of the form $a_{\frac{\alpha + \beta + 1}{2} + \ell} \edg a_{\frac{\alpha + \beta + 1}{2} + \ell + 1} \edg \cdots \edg a_{\alpha} \edg c_{1} \edg a_0 \edg a_1 \edg \cdots \edg a_{\ell - 1}$ for $\ell \in [\frac{\alpha - \beta - 1}{2}]$.
        \item[E.] We embed one node into each path of the form $c_{\ell} \edg a_0 \edg a_1 \edg \cdots \edg a_{\frac{\alpha - \beta - 1}{2}}$ for $\ell \in [\gamma]$.
        \item[F.] We embed one node into each path of the form $a_{\frac{\alpha + \beta + 1}{2}} \edg a_{\frac{\alpha + \beta + 1}{2} + 1} \edg \cdots \edg a_{\alpha} \edg c_{\ell}$ for $\ell \in [\gamma]$.
    \end{enumerate}

    \begin{figure}
    \begin{center}

        \includegraphics[scale=0.7]{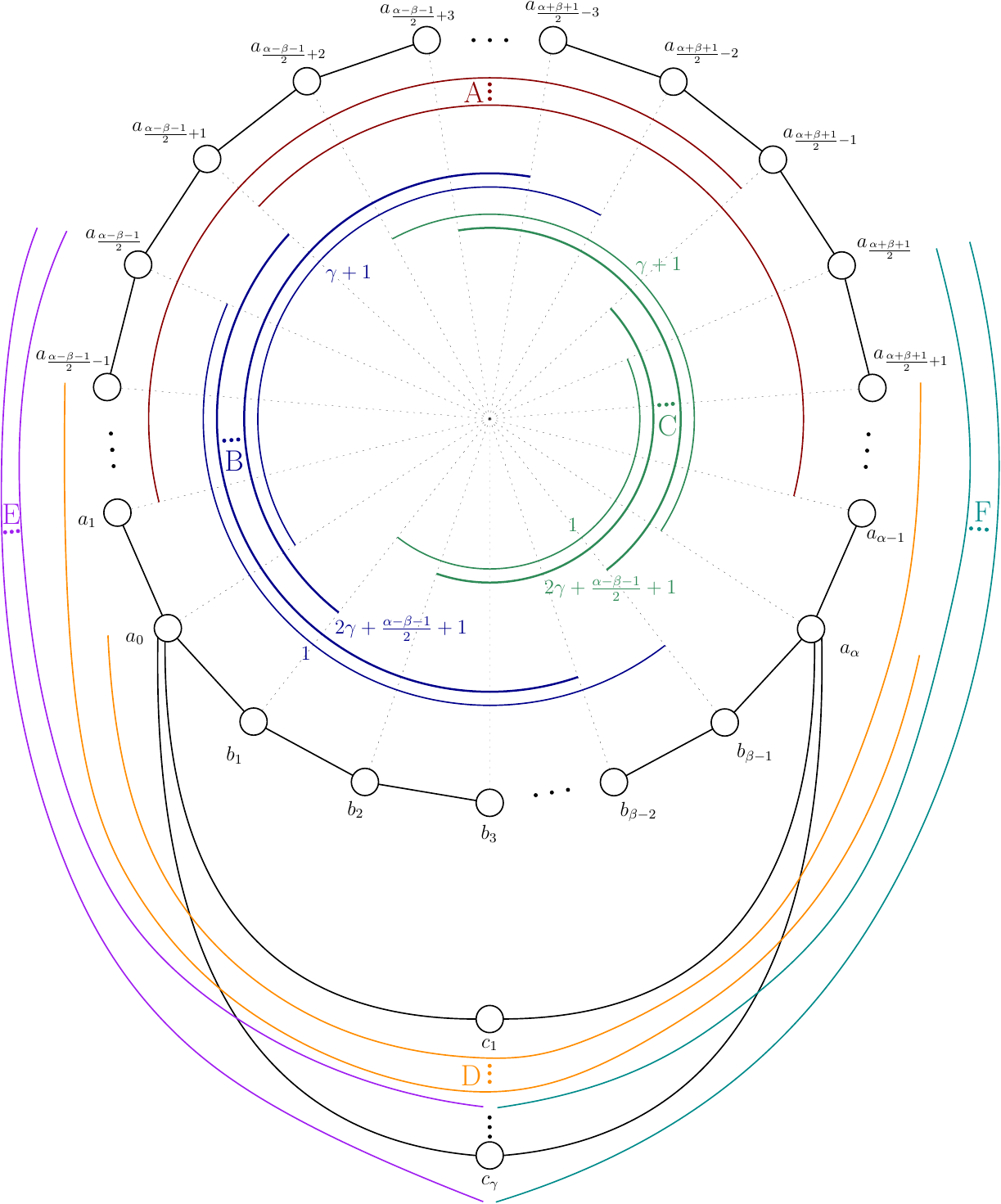}
    \end{center}

    \caption{Illustration of the proof of \cref{lb-parameterized-bound-a-b-c-3}. In addition to the abstract situation of \cref{lb-parameterized-bound-a-b-c-3}, the figure also shows the clique embedding for $P(13, 6, 2, 2)$.}
    \label{lb-parameterized-bound-a-b-c-3-picture}
    \end{figure}

    We now analyze this construction.
    In total, we embed
    \begin{align*}
        &(\frac{\alpha - \beta - 1}{2} + 1) \cdot 1 & \text{(A)}\\
        &+ 1 + (\beta - 2) \cdot (2\gamma + \frac{\alpha - \beta - 1}{2} + 1) + (\gamma + 1) & \text{(B)}\\
        &+ 1 + (\beta - 3 + 1) \cdot (2\gamma + \frac{\alpha - \beta - 1}{2} + 1) + (\gamma + 1) & \text{(C)}\\
        &+ \frac{\alpha - \beta - 1}{2} \cdot 1 & \text{(D)}\\
        &+ \gamma \cdot 1 & \text{(E)}\\
        &+ \gamma \cdot 1 & \text{(F)}\\
        &= 2 \cdot ((\beta - 1) \cdot (2\gamma + \frac{\alpha - \beta - 1}{2} + 1) + 2) - 1
	\end{align*}
    nodes, which is exactly the number of nodes in $K_{2X - 1}$.
    It is clear that all subgraphs into which we embed are connected.
    Next, we prove that all pairs of these subgraphs touch.
    For every type from A to F we show that the paths of this type touch each other and all later types.
    \begin{enumerate}
    \item[A.] All paths of type A start not later than node $a_{\frac{\alpha - \beta - 1}{2} + 1}$ and end not earlier than node $a_{\frac{\alpha + \beta + 1}{2} - 1}$.
        Note that $\frac{\alpha - \beta - 1}{2} + 1 < \frac{\alpha + \beta + 1}{2} - 1$ because $\beta \ge 3$.
        Thus, all these paths contain all nodes $a_{\frac{\alpha - \beta - 1}{2} + 1}, \ldots, a_{\frac{\alpha + \beta + 1}{2} - 1}$, and all pairs of them touch.
        All paths of type B end not earlier than node $a_{\frac{\alpha - \beta - 1}{2}}$, so they touch all paths of type A on the edge $\{a_{\frac{\alpha - \beta - 1}{2}}, a_{\frac{\alpha - \beta - 1}{2} + 1}\}$.
        Symmetrically, all paths of type C start not later than $a_{\frac{\alpha + \beta + 1}{2}}$, so they touch all paths of type A on the edge $\{a_{\frac{\alpha + \beta + 1}{2} - 1}, a_{\frac{\alpha + \beta + 1}{2}}\}$.
        For any fixed $\ell$, paths of types D (for $\ell=1$), E (for $\ell \in [\gamma]$), and F (for $\ell \in [\gamma]$) lie on the cycle $a_0 \edg a_1 \edg \cdots \edg a_{\alpha} \edg c_{\ell}$ of length $\alpha + 2$, which also contains all paths of type A.
        The paths of types D, E, and F have length $\frac{\alpha - \beta - 1}{2} + 1$, while the paths of type A have length $\frac{\alpha + \beta + 1}{2} - 2$.
        Thus, \cref{paths-touch} implies that all pairs of them touch.

    \item[B.] All paths of type B go through node $a_0$, so all pairs of them touch.
        All paths of type B lie on the cycle $\cC$ of length $\alpha + \beta$ together with all paths of type C.
        All of these paths have length $\frac{\alpha + \beta + 1}{2} - 2$.
        Thus, \cref{paths-touch} implies that each path of type B and each path of type C touch.
        For any fixed $\ell$, paths of types D (for $\ell=1$), E (for $\ell \in [\gamma]$), and F (for $\ell \in [\gamma]$) go through node $c_{\ell}$, so they touch all paths of type B on the edge $\{a_0, c_{\ell}\}$.

    \item[C.] All paths of type C go through node $a_{\alpha}$, so all pairs of them touch.
        For any fixed $\ell$, paths of types D (for $\ell=1$), E (for $\ell \in [\gamma]$), and F (for $\ell \in [\gamma]$) go through node $c_{\ell}$, so they touch paths of type C on the edge $\{a_{\alpha}, c_{\ell}\}$.

    \item [D-F.] Any path of type D, E, or F goes either through node $a_0$ or through node $a_{\alpha}$.
        Any other such path goes through node $c_{\ell}$ for some $\ell \in [\gamma]$.
        Hence, they touch either on the edge $\{a_0, c_{\ell}\}$ or on the edge $\{a_{\alpha}, c_{\ell}\}$.
\end{enumerate}

We showed that all pairs of subgraphs we embed touch, so $\psi$ is a valid clique embedding. It remains to show that all edges have weak edge depth at most $X$.
There are
\begin{align*}
    & 1 + (\beta - 2) \cdot (2\gamma + \frac{\alpha - \beta - 1}{2} + 1) + (\gamma + 1) & \text{(B)}\\
    &+ \frac{\alpha - \beta - 1}{2} \cdot 1 & \text{(D)}\\
    &+ \gamma \cdot 1 & \text{(E)}\\
    &=X - 1
\end{align*}

embedded subgraphs that go through node $a_0$.
The only subgraph that node $b_1$ is contained in that $a_0$ is not contained in is the path of type C with weight one that ends in $b_1$.
Hence, weak edge depth of $\{a_0, b_1\}$ is equal to $X$.
For all edges of the form $\{a_0, c_{\ell}\}$ for $\ell \in [\gamma]$, the only new path that $c_{\ell}$ adds to the subgraphs that are already embedded into $a_0$ are paths of type F, contributing one to weak edge depth of these edges.
So each edge $\{a_0, c_{\ell}\}$ has weak edge depth $(X - 1) + 1 = X$.
The same holds for the edge $\{a_0, a_1\}$ as the only new path that $a_1$ adds to the subgraphs that are already embedded into $a_0$ is the path of type A that starts in $a_1$, which contributes one to weak edge depth of this edge.
If $\alpha \ge \beta + 3$, whenever we now move from an edge $\{a_{\ell}, a_{\ell + 1}\}$ to an edge $\{a_{\ell + 1}, a_{\ell + 2}\}$ for $\ell \in [0, \frac{\alpha - \beta - 1}{2} - 1]$, we lose a single path of type D that ends in $a_{\ell}$, and we gain a single path of type A that starts in $a_{\ell + 2}$.
Both have weight one, so weak edge depth remains at $X$.
When we move from the edge $\{a_{\frac{\alpha - \beta - 1}{2}}, a_{\frac{\alpha - \beta - 1}{2} + 1}\}$ to the edge $\{a_{\frac{\alpha - \beta - 1}{2} + 1}, a_{\frac{\alpha - \beta - 1}{2} + 2}\}$, we lose a single path of type B with weight one that ends in $a_{\frac{\alpha - \beta - 1}{2}}$ and all paths of type E with total weight $\gamma$.
At the same time we gain a single path of type C with weight $\gamma + 1$ that starts in node $a_{\frac{\alpha - \beta - 1}{2} + 2}$.
In total, weak edge depth of $\{a_{\frac{\alpha - \beta - 1}{2} + 1}, a_{\frac{\alpha - \beta - 1}{2} + 2}\}$ is again $X$.
Next, whenever we move from an edge $\{a_{\ell}, a_{\ell + 1}\}$ to an edge $\{a_{\ell + 1}, a_{\ell + 2}\}$ for $\ell \in [\frac{\alpha - \beta - 1}{2} + 1, \frac{\alpha + \beta + 1}{2} - 3]$ (there is at least one such $\ell$ as $\beta \ge 3$), we lose a single path of type B with weight $2\gamma + \frac{\alpha - \beta - 1}{2} + 1$ that ends in node $a_{\ell}$, and we gain a single path of type C with weight $2\gamma + \frac{\alpha - \beta - 1}{2} + 1$ that starts in node $a_{\ell + 2}$ so weak edge depth remains at $X$.
When we move from the edge $\{a_{\frac{\alpha + \beta + 1}{2} - 2}, a_{\frac{\alpha + \beta + 1}{2} - 1}\}$ to the edge $\{a_{\frac{\alpha + \beta + 1}{2} - 1}, a_{\frac{\alpha + \beta + 1}{2}}\}$, we lose a single path of type B with weight $\gamma + 1$ that ends in node $a_{\frac{\alpha + \beta + 1}{2} - 2}$.
We gain a path of type C with weight one that starts in node $a_{\frac{\alpha + \beta + 1}{2}}$ and all paths of type F with total weight $\gamma$. 
Hence, weak edge depth of $\{a_{\frac{\alpha + \beta + 1}{2} - 1}, a_{\frac{\alpha + \beta + 1}{2}}\}$ is again $X$.
Now if $\alpha \ge \beta + 3$, whenever we move from an edge $\{a_{\ell}, a_{\ell + 1}\}$ to an edge $\{a_{\ell + 1}, a_{\ell + 2}\}$ for $\ell \in [\frac{\alpha + \beta + 1}{2} - 1, \alpha - 2]$, we lose a single path of type A that ends in node $a_{\ell}$, and we gain a single path of type D that starts in node $a_{\ell + 2}$.
Both of them have weight one, so weak edge depth of the current edge remains at $X$.
If we move from the edge $\{a_{\alpha - 1}, a_{\alpha}\}$ to an edge $\{a_{\alpha}, c_{\ell}\}$ for $\ell \in [\gamma]$, we lose a single path of type A that ends in $a_{\alpha - 1}$, and we gain a single path of type E that starts in $c_{\ell}$, so each such edge has weak edge depth $X$.
When we move from the edge $\{a_{\alpha - 1}, a_{\alpha}\}$ to the edge $\{b_{\beta}, b_{\beta - 1}\}$, we similarly lose a single path of type A with weight one that ends in $a_{\alpha - 1}$, and we gain a single path of type B with weight one that starts in $b_{\beta - 1}$, so weak edge depth of $\{b_{\beta}, b_{\beta - 1}\}$ is $X$.
Now when we move from the edge $\{b_{\beta}, b_{\beta-1}\}$ to the edge $\{b_{\beta-1}, b_{\beta - 2}\}$, we lose a single path of type C with weight $\gamma + 1$ that ends in node $b_{\beta}$ and all paths of types D and F of total weight $\frac{\alpha - \beta - 1}{2} + \gamma$.
At the same time we gain a single path of type B with weight $2\gamma + \frac{\alpha - \beta - 1}{2} + 1$ that starts in node $b_{\beta - 2}$.
In total, weak depth is again $X$.
Now if $\beta \ge 4$, whenever we move from an edge $\{b_{\ell + 2}, b_{\ell + 1}\}$ to an edge $\{b_{\ell + 1}, b_{\ell}\}$ for $\ell \in [\beta - 3]$, we lose a single path of type C with weight $2\gamma + \frac{\alpha - \beta - 1}{2} + 1$ that ends in $b_{\ell + 2}$, and we gain a single path of type B with the same weight that starts in node $b_{\ell}$, so weak edge depth of the current edge remains at $X$.
This way we checked that weak edge depth of every edge is at most $X$, which concludes the proof of the fact that $\wed(\psi) = X$.
\end{proof}

\begin{lemma} \label{lb-parameterized-bound-a-b-c-2}
    For any $\alpha \ge \beta \ge 3$, $\alpha + \beta$ even, and $\gamma \ge 1$, we have $\clemb(P(\alpha, \beta, \gamma \times 2)) \ge 2 - 1 / (2\beta\gamma + \frac{\alpha\beta}{2} - \frac{\beta^2}{2} + \frac{3\beta}{2} - \frac{\alpha}{2} - 3\gamma)$.
\end{lemma}

\begin{proof}
    Define $X \coloneqq 2\beta\gamma + \frac{\alpha\beta}{2} - \frac{\beta^2}{2} + \frac{3\beta}{2} - \frac{\alpha}{2} - 3\gamma = (\beta - 2) \cdot (2\gamma + \frac{\alpha - \beta}{2} + 1) + \frac{\alpha - \beta}{2} + \gamma + 2$.
    We create a clique embedding $\psi$ from $K_{2X - 1}$ to $P(\alpha, \beta, \gamma \times 2)$ with $\wed(\psi) = X$, which yields $\clemb(P(\alpha, \beta, \gamma \times 2)) \ge \frac{2X - 1}{X} = 2 - \frac{1}{X}$, as desired.
    We treat the paths $a_0 \edg a_1 \edg \cdots \edg a_{\alpha}$ and $b_0 \edg b_1 \edg \cdots \edg b_{\beta}$ as a cycle $\cC$ of length $\alpha + \beta$.
    We use $6$ different types of subgraphs that we embed into. See \cref{lb-parameterized-bound-a-b-c-2-picture}.

    \begin{enumerate}
        \item[A.] We embed one node into each path of the form $a_{\ell} \edg a_{\ell + 1} \edg \cdots \edg a_{\ell + \frac{\alpha + \beta}{2} - 2}$ for $\ell \in [\frac{\alpha - \beta}{2} + 1]$.
        \item[B.] We embed one node into the path $b_{\beta - 1} \edg b_{\beta - 2} \edg \cdots \edg b_1 \edg a_0 \edg \cdots \edg a_{\frac{\alpha - \beta}{2}}$.
    Furthermore, we embed $2\gamma + \frac{\alpha - \beta}{2} + 1$ nodes into each path of the form $b_{\beta - \ell - 1} \edg b_{\beta - \ell - 2} \edg \cdots \edg b_1 \edg a_0 \edg \cdots \edg a_{\frac{\alpha - \beta}{2} + \ell}$ for $\ell \in [\beta - 2]$.
        \item[C.] We embed $\gamma + 1$ nodes into the path $a_{\frac{\alpha - \beta}{2} + 2} \edg a_{\frac{\alpha - \beta}{2} + 3} \edg \cdots \edg a_{\alpha}$.
        Furthermore, if $\beta \ge 4$, we embed $2\gamma + \frac{\alpha - \beta}{2} + 1$ nodes into each path of the form $a_{\frac{\alpha - \beta}{2} + \ell} \edg a_{\frac{\alpha - \beta}{2} + \ell + 1} \edg \cdots \edg a_{\alpha} \edg b_{\beta - 1} \edg \cdots \edg b_{\beta - \ell + 2}$ for $\ell \in [3, \beta - 1]$.
        Moreover, we embed $\gamma + \frac{\alpha - \beta}{2} + 1$ nodes into the path $a_{\frac{\alpha + \beta}{2}} \edg a_{\frac{\alpha + \beta}{2} + 1} \edg \cdots \edg a_{\alpha} \edg b_{\beta - 1} \edg \cdots \edg b_{2}$.
        \item[D.] If $\alpha > \beta$, we embed one node into each path of the form $a_{\frac{\alpha + \beta}{2} + \ell} \edg a_{\frac{\alpha + \beta}{2} + \ell + 1} \edg \cdots \edg a_{\alpha} \edg c_1 \edg a_0 \edg a_1 \edg \cdots \edg a_{\ell - 1}$ for $\ell \in [\frac{\alpha - \beta}{2}]$.
        \item[E.] We embed one node into each path of the form $c_{\ell} \edg a_0 \edg a_1 \edg \cdots \edg a_{\frac{\alpha - \beta}{2}}$ for $\ell \in [\gamma]$.
        \item[F.] We embed one node into each path of the form $a_{\frac{\alpha + \beta}{2}} \edg a_{\frac{\alpha + \beta}{2} + 1} \edg \cdots \edg a_{\alpha} \edg c_{\ell}$ for $\ell \in [\gamma]$.
    \end{enumerate}

    \begin{figure}
    \begin{center}

        \includegraphics[scale=0.7]{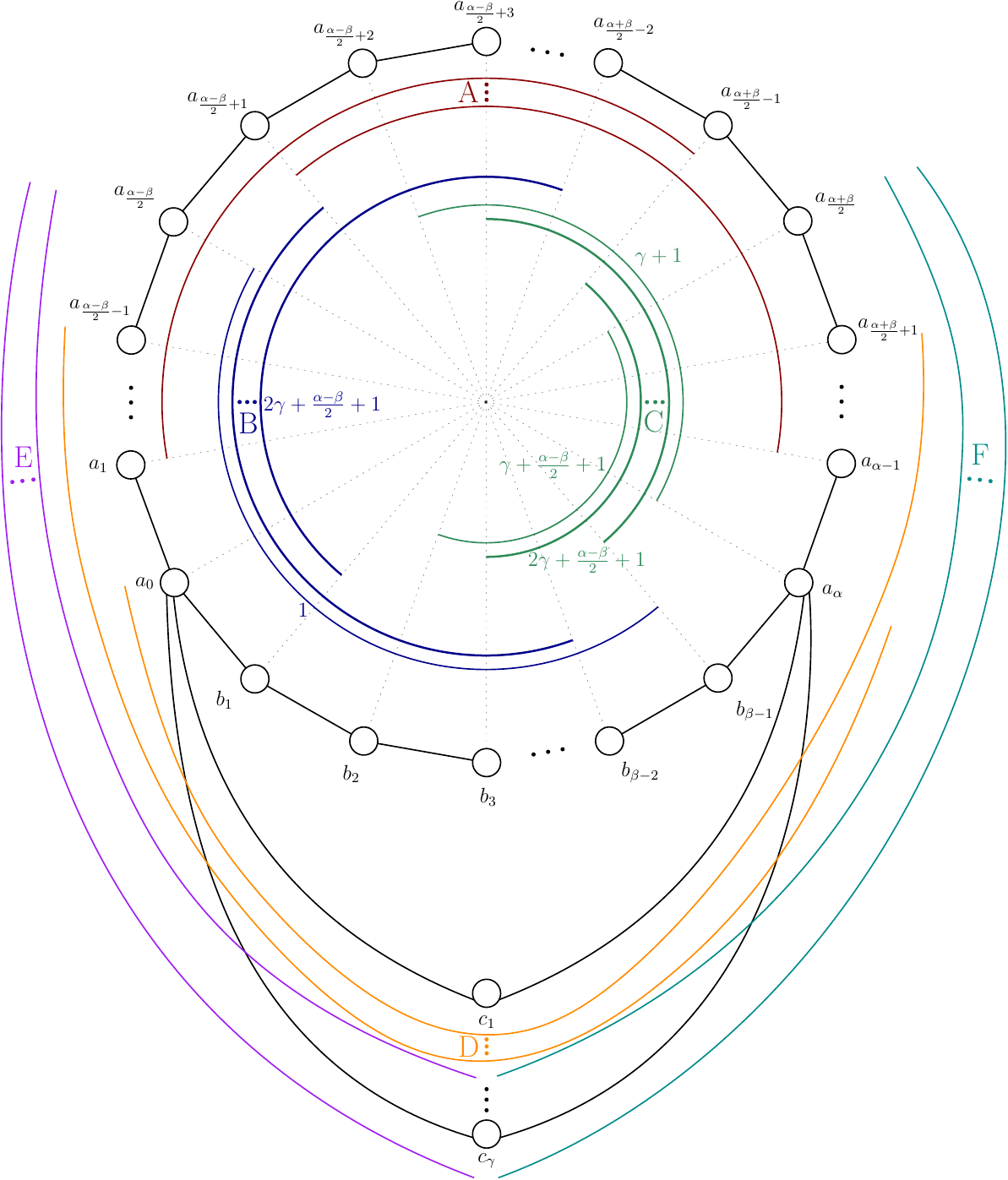}
    \end{center}

    \caption{Illustration of the proof of \cref{lb-parameterized-bound-a-b-c-2}. In addition to the abstract situation of \cref{lb-parameterized-bound-a-b-c-2}, the figure also shows the clique embedding for $P(12, 6, 2, 2)$.}
    \label{lb-parameterized-bound-a-b-c-2-picture}
    \end{figure}

    We now analyze this construction.
    In total, we embed
    \begin{align*}
        &(\frac{\alpha - \beta}{2} + 1) \cdot 1 & \text{(A)}\\
        &+ 1 + (\beta - 2) \cdot (2 \gamma + \frac{\alpha - \beta}{2} + 1) & \text{(B)}\\
        &+ (\gamma + 1) + (\beta - 3) \cdot (2 \gamma + \frac{\alpha - \beta}{2} + 1) + (\gamma + \frac{\alpha - \beta}{2} + 1) & \text{(C)}\\
        &+ \frac{\alpha - \beta}{2} \cdot 1 & \text{(D)}\\
        &+ \gamma \cdot 1 & \text{(E)}\\
        &+ \gamma \cdot 1 & \text{(F)}\\
        &= 2 \cdot ((\beta - 2) \cdot (2\gamma + \frac{\alpha - \beta}{2} + 1) + \frac{\alpha - \beta}{2} + \gamma + 2) - 1
    \end{align*}
    nodes, which is exactly the number of nodes in $K_{2X - 1}$.
    It is clear that all subgraphs into which we embed are connected.
    Next, we prove that all pairs of these subgraphs touch.
    For every type from A to F we show that the paths of this type touch each other and all later types.
    \begin{enumerate}
    \item[A.] All paths of type A start not later than node $a_{\frac{\alpha - \beta}{2} + 1}$ and end not earlier than node $a_{\frac{\alpha + \beta}{2} - 1}$.
        Note that $\frac{\alpha - \beta}{2} + 1 < \frac{\alpha + \beta}{2} - 1$ as $\beta \ge 3$.
        Thus, all these paths contain all the nodes $a_{\frac{\alpha - \beta}{2} + 1}, \ldots, a_{\frac{\alpha + \beta}{2} - 1}$, and all pairs of them touch.
        All paths of type B end not earlier than node $a_{\frac{\alpha - \beta}{2}}$, so they touch all paths of type A on the edge $\{a_{\frac{\alpha - \beta}{2}}, a_{\frac{\alpha - \beta}{2} + 1}\}$.
        Symmetrically, all paths of type C start not later than $a_{\frac{\alpha + \beta}{2}}$, so they all touch all paths of type A on the edge $\{a_{\frac{\alpha + \beta}{2} - 1}, a_{\frac{\alpha + \beta}{2}}\}$.
        For any fixed $\ell$, paths of types D (for $\ell=1$), E (for $\ell \in [\gamma]$), and F (for $\ell \in [\gamma]$) lie on the cycle $a_0 \edg a_1 \edg \cdots \edg a_{\alpha} \edg c_{\ell}$ of length $\alpha + 2$, which also contains all paths of type $A$.
        The paths of types D, E, and F have length $\frac{\alpha - \beta}{2} + 1$, while the paths of type A have length $\frac{\alpha + \beta}{2} - 2$.
        Thus, \cref{paths-touch} implies that all pairs of them touch.

    \item[B.] All paths of type B go through node $a_0$, so all pairs of them touch.
        All paths of type B lie on the cycle $\cC$ of length $\alpha + \beta$ together with all paths of type C.
        Paths of type B have length $\frac{\alpha + \beta}{2} - 1$, and paths of type C have length $\frac{\alpha + \beta}{2} - 2$.
        Thus, \cref{paths-touch} implies that each path of type B and each path of type C touch.
        For any fixed $\ell$, paths of types D (for $\ell=1$), E (for $\ell \in [\gamma]$), and F (for $\ell \in [\gamma]$) go through node $c_{\ell}$, so they touch paths of type B on the edge $\{a_0, c_{\ell}\}$.

    \item[C.] All paths of type C go through node $a_{\alpha}$, so all pairs of them touch.
        For any fixed $\ell$, paths of types D (for $\ell=1$), E (for $\ell \in [\gamma]$), and F (for $\ell \in [\gamma]$) go through node $c_{\ell}$, so they touch paths of type C on the edge $\{a_{\alpha}, c_{\ell}\}$.

    \item [D-F.] Any path of type D, E, or F goes either through node $a_0$ or through node $a_{\alpha}$.
        Any other such path goes through node $c_{\ell}$ for some $\ell \in [\gamma]$.
        Hence, they touch either on the edge $\{a_0, c_{\ell}\}$ or on the edge $\{a_{\alpha}, c_{\ell}\}$.
\end{enumerate}

We showed that all pairs of subgraphs we embed touch, so $\psi$ is a valid clique embedding. It remains to show that all edges have weak edge depth at most $X$.
There are
\begin{align*}
    &1 + (\beta - 2) \cdot (2\gamma + \frac{\alpha - \beta}{2} + 1) & \text{(B)}\\
    &+ \frac{\alpha - \beta}{2} \cdot 1 & \text{(D)}\\
    &+ \gamma \cdot 1 & \text{(E)}\\
    &= X - 1
\end{align*}
embedded subgraphs that go through node $a_0$.
Node $b_1$ is not contained in any paths that $a_0$ is not contained in.
Hence, weak edge depth of $\{a_0, b_1\}$ is equal to $X - 1$.
For all edges of the form $\{a_0, c_{\ell}\}$ for $\ell \in [\gamma]$, the only new path that $c_{\ell}$ adds to the subgraphs that are already embedded into $a_0$ are paths of type F, contributing one to weak edge depth of these edges.
So each edge $\{a_0, c_{\ell}\}$ has weak edge depth $(X - 1) + 1 = X$.
The same holds for the edge $\{a_0, a_1\}$ as the only new path that $a_1$ adds to the subgraphs that are already embedded into $a_0$ is the path of type A that starts in $a_1$, which contributes one to weak edge depth of this edge.
If $\alpha > \beta$, whenever we now move from an edge $\{a_{\ell}, a_{\ell + 1}\}$ to an edge $\{a_{\ell + 1}, a_{\ell + 2}\}$ for $\ell \in [0, \frac{\alpha - \beta}{2} - 1]$, we lose a single path of type D that ends in $a_{\ell}$, and we gain a single path of type A that starts in $a_{\ell + 2}$.
Both have weight one, so  weak edge depth remains at $X$.
When we move from the edge $\{a_{\frac{\alpha - \beta}{2}}, a_{\frac{\alpha - \beta}{2} + 1}\}$ to the edge $\{a_{\frac{\alpha - \beta}{2} + 1}, a_{\frac{\alpha - \beta}{2} + 2}\}$, we lose a single path of type B with weight one that ends in $a_{\frac{\alpha - \beta}{2}}$ and all paths of type E with total weight $\gamma$.
At the same time we gain a single path of type C with weight $\gamma + 1$ that starts in node $a_{\frac{\alpha - \beta}{2} + 2}$.
In total, weak edge depth of $\{a_{\frac{\alpha - \beta}{2} + 1}, a_{\frac{\alpha - \beta}{2} + 2}\}$ is again $X$.
Next, if $\beta \ge 4$, whenever we move from an edge $\{a_{\ell}, a_{\ell + 1}\}$ to an edge $\{a_{\ell + 1}, a_{\ell + 2}\}$ for $\ell \in [\frac{\alpha - \beta}{2} + 1, \frac{\alpha + \beta}{2} - 3]$, we lose a single path of type B with weight $2\gamma + \frac{\alpha - \beta}{2} + 1$ that ends in node $a_{\ell}$ and gain a single path of type C with weight $2\gamma + \frac{\alpha - \beta}{2} + 1$ that starts in node $a_{\ell + 2}$, so weak edge depth remains at $X$.
When we move from the edge $\{a_{\frac{\alpha + \beta}{2} - 2}, a_{\frac{\alpha + \beta}{2} - 1}\}$ to the edge $\{a_{\frac{\alpha + \beta}{2} - 1}, a_{\frac{\alpha + \beta}{2}}\}$, we lose a single path of type B with weight $2\gamma + \frac{\alpha - \beta}{2} + 1$ that ends in node $a_{\frac{\alpha + \beta}{2} - 2}$.
We gain a path of type C with weight $\gamma + \frac{\alpha - \beta}{2} + 1$ that starts in node $a_{\frac{\alpha + \beta}{2}}$ and all paths of type F with total weight $\gamma$.
Hence, weak edge depth of $\{a_{\frac{\alpha + \beta}{2} - 1}, a_{\frac{\alpha + \beta}{2}}\}$ is again $X$.
Now, if $\alpha > \beta$, whenever we move from an edge $\{a_{\ell}, a_{\ell + 1}\}$ to an edge $\{a_{\ell + 1}, a_{\ell + 2}\}$ for $\ell \in [\frac{\alpha + \beta}{2} - 1, \alpha - 2]$, we lose a single path of type A that ends in node $a_{\ell}$ and gain a single path of type D that starts in node $a_{\ell + 2}$.
Both of them have weight one, so weak edge depth of the current edge remains at $X$.
If we move from the edge $\{a_{\alpha - 1}, a_{\alpha}\}$ to an edge $\{a_{\alpha}, c_{\ell}\}$ for $\ell \in [\gamma]$, we lose a single path of type A that ends in $a_{\alpha - 1}$ and gain a single path of type E that starts in $c_{\ell}$, so each such edge has weak edge depth $X$.
When we move from the edge $\{a_{\alpha - 1}, a_{\alpha}\}$ to the edge $\{b_{\beta}, b_{\beta - 1}\}$, we lose a single path of type A with weight one that ends in $a_{\alpha - 1}$ and gain a single path of type B with weight one that starts in $b_{\beta - 1}$, so weak edge depth of $\{b_{\beta}, b_{\beta - 1}\}$ is $X$.
Now when we move from the edge $\{b_{\beta}, b_{\beta - 1}\}$ to the edge $\{b_{\beta - 1}, b_{\beta - 2}\}$, we lose a single path of type C with weight $\gamma + 1$ that ends in node $b_{\beta}$ and all paths of types D and F of total weight $\frac{\alpha - \beta}{2} + \gamma$.
At the same time we gain a single path of type B with weight $2\gamma + \frac{\alpha - \beta}{2} + 1$ that starts in node $b_{\beta - 2}$.
In total weak edge depth remains at $X$.
Now if $\beta \ge 4$, whenever we move from an edge $\{b_{\ell + 2}, b_{\ell + 1}\}$ to an edge $\{b_{\ell + 1}, b_{\ell}\}$ for $\ell \in [\beta - 3]$, we lose a single path of type C with weight $2\gamma + \frac{\alpha - \beta}{2} + 1$ that ends in $b_{\ell + 2}$, and we gain a single path of type B with the same weight that starts in node $b_{\ell}$, so weak edge depth of the current edge remains at $X$.
This way we checked that weak edge depth of every edge is at most $X$, which concludes the proof of the fact that $\wed(\psi) = X$.
\end{proof}

\begin{lemma} \label{lm:pabc-clemb}
    For $\alpha \ge \beta \ge 3$ and $\gamma \ge 1$, we have $\clemb(P(\alpha, \beta, \gamma \times 2)) \ge 2 - 1 / \funcc(\alpha, \beta, \gamma)$.
\end{lemma}

\begin{proof}
    We consider all the cases from the definition of $\funcc$.
    \begin{enumerate}
        \item \casea As $\alpha > \beta$ in this case, we have $\alpha - 1 \ge \beta$.
            Furthermore, as $\alpha + \beta$ is even, $(\alpha - 1) + \beta$ is odd.
            Finally, $P(\alpha - 1, \beta, \gamma \times 2)$ is an induced minor of $P(\alpha, \beta, \gamma \times 2)$.
            Hence, due to \cref{clemb-of-induced-minors} and \cref{lb-parameterized-bound-a-b-c-3}, we obtain
            \begin{align*}
                \clemb(P(\alpha, \beta, \gamma \times 2)) &\ge \clemb(P(\alpha - 1, \beta, \gamma \times 2))\\
                                &\ge 2 - 1 / (2\beta\gamma + \frac{(\alpha - 1)\beta}{2} - \frac{\beta^2}{2} + \beta - \frac{\alpha - 1}{2} - 2\gamma + \frac{3}{2})\\
                                &= 2 - 1 / (2\beta\gamma + \frac{\alpha\beta}{2} - \frac{\beta^2}{2} + \frac{\beta}{2} - \frac{\alpha}{2} - 2\gamma + 2)\\
                                &= 2 - 1 / \funcc(\alpha, \beta, \gamma).
            \end{align*}
        \item \caseb Follows from \cref{lb-parameterized-bound-a-b-c-2}.
        \item \casec In this case $3 \beta \ge \alpha + 6 \gamma + 8 \ge \beta + 6 \gamma + 8$, so $\alpha > \beta - 1 \ge (3 \gamma + 4) - 1 \ge 5$.
            Furthermore, as $\alpha + \beta$ is even, $\alpha + (\beta - 1)$ is odd.
            Finally, $P(\alpha, \beta - 1, \gamma \times 2)$ is an induced minor of $P(\alpha, \beta, \gamma \times 2)$.
            Hence, due to \cref{clemb-of-induced-minors} and \cref{lb-parameterized-bound-a-b-c-1}, we obtain
            \begin{align*}
                \clemb(P(\alpha, \beta, \gamma \times 2)) &\ge \clemb(P(\alpha, \beta - 1, \gamma \times 2))\\
                                    &\ge 2 - 1 / (2 (\beta - 1) \gamma + \frac{\alpha (\beta - 1)}{2} - \frac{(\beta - 1)^2}{2} + 2 (\beta - 1) - \frac{\alpha}{2} - 4 \gamma - \frac{3}{2})\\
                                    &= 2 - 1 / (2 \beta \gamma + \frac{\alpha \beta}{2} - \frac{\beta^2}{2} + 3 \beta - \alpha - 6 \gamma - 4)\\
                                    &= 2 - 1 / \funcc(\alpha, \beta, \gamma).
            \end{align*}
        \item \cased Follows from \cref{lb-parameterized-bound-a-b-c-3}.
        \item \casee This case follows from \cref{lb-parameterized-bound-a-b-c-1} as $\beta \ge 2 \gamma + 3 \ge 5$.
        \item \casef Here we consider all cases $6$ through $8$ together.
            Define $x \coloneqq (((\alpha - 1) \bmod 4) + 3) / 2$.
            Denote $\beta' \coloneqq \frac{\alpha}{2} + 2\gamma + x$.
            It is easy to see that $\beta'$ is an integer.
            Furthermore, $2 \alpha + 2 \beta' = 2 \alpha + \alpha + 4 \gamma + (\alpha - 1) + 3 = 2 \bmod 4$, hence $\alpha + \beta'$ is odd.
            Moreover, we have $\beta' = \frac{\alpha}{2} + 2 \gamma + x \ge \frac{3}{2} + 2 + \frac{3}{2} = 5$.
            Furthermore, $\beta' < \beta \le \alpha$ because $x \le 3$ and $2 \beta > \alpha + 4 \gamma + 6$ by case assumption, and thus $P(\alpha, \beta', \gamma \times 2)$ is an induced minor of $P(\alpha, \beta, \gamma \times 2)$.
            Hence, due to \cref{clemb-of-induced-minors} and \cref{lb-parameterized-bound-a-b-c-1}, we obtain
            \begin{align*}
                \clemb(P(\alpha, \beta, \gamma \times 2)) &\ge \clemb(P(\alpha, \beta', \gamma \times 2))\\
                                        &\ge 2 - 1 / (2 \beta' \gamma + \frac{\alpha \beta'}{2} - \frac{\beta'^2}{2} + 2 \beta' - \frac{\alpha}{2} - 4 \gamma - \frac{3}{2})\\
                                        &= 2 - 1 / (2 (\frac{\alpha}{2} + 2 \gamma + x) \gamma + \frac{\alpha (\frac{\alpha}{2} + 2 \gamma + x)}{2} - \frac{(\frac{\alpha}{2} + 2 \gamma + x)^2}{2} \\&+ 2 (\frac{\alpha}{2} + 2 \gamma + x) - \frac{\alpha}{2} - 4 \gamma - \frac{3}{2})\\
                                        &= 2 - 1 / (2 \gamma^2 + \alpha \gamma + \frac{\alpha^2}{8} + \frac{\alpha}{2} + (2x - \frac{x^2}{2} - \frac{3}{2}))\\
                                        &=2 - 1 / \fabc,
            \end{align*}
            where the last equality can be checked for each value of $\alpha \bmod 4$.
            \begin{enumerate}
                \item $\alpha = 0 \bmod 4$. In this case $x = (3 + 3) / 2 = 3$. Hence, $2x - \frac{x^2}{2} - \frac{3}{2} = 0$.
                \item $\alpha = 1 \bmod 4$. In this case $x = (0 + 3) / 2 = 3 / 2$. Hence, $2x - \frac{x^2}{2} - \frac{3}{2} = 3 / 8$.
                \item $\alpha = 2 \bmod 4$. In this case $x = (1 + 3) / 2 = 2$. Hence, $2x - \frac{x^2}{2} - \frac{3}{2} = 1 / 2$.
                \item $\alpha = 3 \bmod 4$. In this case $x = (2 + 3) / 2 = 5 / 2$. Hence, $2x - \frac{x^2}{2} - \frac{3}{2} = 3 / 8$.
            \end{enumerate}
        \item \caseg See case $6$.
        \item \caseh See case $6$.
    \end{enumerate}
\end{proof}

\subsection{Proof of \cref{lm:p-graph-lower-bound}}
\label{sec:lower-bound-P-graphs-combination}

Now we can put the pieces together to prove \cref{lm:p-graph-lower-bound}.

\pgraphlowerbound*

\begin{proof}[Proof of \cref{lm:p-graph-lower-bound}]
    Let $(\alpha, \beta, \gamma) \in \Tfamily$.

    For $(\alpha, \beta, \gamma) = (1, 0, 0)$, we have $\clemb(\pabc) \ge 1 = 2 - 1 / 1$ by \cref{edge-lb} and $\fabc = 1$ by \cref{lm:p1-time-complexity}.
    For $(\alpha, \beta, \gamma) = (2, 1, 0)$, we have $\clemb(\pabc) \ge 2 - 1 / 2$ by \cref{cycle-lb} and $\fabc = 2$ by \cref{lm:c3-time-complexity}.
    For $(\alpha, \beta, \gamma) = (k - 2, 2, 0)$ with $k \ge 4$, we have $\clemb(\pabc) \ge 2 - 1 / \left\lceil k / 2 \right\rceil$ by \cref{cycle-lb} and $\fabc = \left\lceil k / 2 \right\rceil$ by \cref{lm:ck-time-complexity}.

    Otherwise, we have $\alpha \ge \beta \ge 2$ and $\gamma \ge 1$.
    If $\alpha = \beta = 2$, then we have $\clemb(\pabc) \ge 2 - 1 / (\gamma + 2)$ by \cref{biclique-lb} and $\fabc = \gamma + 2$ by \cref{lm:kk2-time-complexity}.
    If $\alpha > \beta = 2$, then we have $\funcc(\alpha, \beta, \gamma) = 2 \gamma + 2 + \left\lfloor \frac{\alpha - 1}{2} \right\rfloor$ by \cref{a-c-time-complexity}, and \cref{lm:pa2c-clemb} proves the desired inequality.
    Finally, if $\alpha \ge \beta \ge 3$, then \cref{lm:pabc-clemb} implies the desired inequality. This covers all $(\alpha, \beta, \gamma) \in \Tfamily$.
\end{proof}

\section{Algorithms: from General Case to Family \boldmath$\family$} \label{sec:upper-bounds-until-p-graphs}

In this section we prove \cref{all-reductions,induced-minor-encoding,clique-separator-algo-combination,tc-from-D}. This covers most of the claims from \cref{sec:proof-overview-algorithms}, except for \cref{lm:p-graph-upper-bound} whose proof is postponed to \cref{sec:p-graph-upper-bounds}.

\smallskip
We start by presenting a reduction to tree patterns.

\begin{lemma} \label{lm:reduce-encoding-to-tree}
    Let $H$ be a pattern.
    There is an algorithm that, given a (weighted) $m$-edge host graph $G$ and a partial $H$-encoding $\partenc$ of size $s$, in time $\Oh(m + s)$ computes a tree $T'$ and a (weighted) host graph $G'$ for pattern $T'$ such that there is a (weight-preserving) bijection between $H$-subgraphs of $G$ encoded by $\partenc$ and $T'$-subgraphs of $G'$.
    Furthermore, given any $T'$-subgraph of $G'$, the corresponding $H$-subgraph of $G$ can be computed in constant time.
\end{lemma}

\begin{proof}
	It suffices to prove the weighted setting, as the unweighted setting follows by setting all weights to zero.
	Recall that a partial encoding $\partenc$ consists of a tree decomposition $\mathcal{T}$, which consists of a tree $T$ together with bags $B(v)$ for each node $v$ of $T$, as well as a submaterialization $S_{B(v)}$ for every bag $B(v)$ of $\mathcal{T}$. 
    
    We start by filtering submaterializations in time $\Oh(s)$.
    For each bag $B$ of $\mathcal{T}$, we remove all $\bv \in S_B$ such that $\bv$ does not form an $H[B]$-subgraph in $G[H[B]]$.
    Note that the filtered partial $H$-encoding encodes exactly the same $H$-subgraphs as before the filtering. Indeed, if $\bv \in S_B$ does not form an $H[B]$-subgraph in $G[H[B]]$, it cannot be a projection of any $H$-subgraph in $G$ onto $B$.

	We construct the new pattern $T'$ by subdividing every edge of $T$, i.e., we let $T' \coloneqq (V(T) \cup E(T), \{\{a, e\} \mid e \in E(T), a \in e\})$. Clearly, $T'$ is a tree and a pattern.
	
    For any edge $e = \{a, b\} \in E(T)$, let $B(e) \coloneqq B(a) \cap B(b)$ and $S_{B(e)} \coloneqq \prj{S_{B(a)}}{B(e)} \cup \prj{S_{B(b)}}{B(e)}$.
    With this notation, we construct the nodes of the new host graph $G'$ as follows. 
    For every $a \in V(T)$, we let $V'_a \coloneqq S_{B(a)}$ be the part of $G'$ corresponding to $a$.
    For every $e \in E(T)$, we let $V'_e \coloneqq S_{B(e)}$ be the part of $G'$ corresponding to $e$. %

	The edges of $G'$ simply encode projections:
    For every $\{a, e\} \in E(T')$, we let the edges between $V'_a$ and $V'_e$ be $\{\{v_a, \prj{v_a}{B(e)}\} \mid v_a \in S_{B(a)}\}$.
    
    We next define the edge weights. Note that, by the definition of a tree decomposition, for every edge $e' = \{x, y\} \in E(H)$ there is some node $a \in T$ with $x, y \in B(a)$.
    We pick any such node $a$ and an arbitrary edge $e \in E(T)$ incident to $a$ and denote $\varphi(e') \coloneqq \{a, e\}$.
    Define the edge weights in $G'$ by
    \[ w'(\{v_a, \prj{v_a}{B(e)}\}) \coloneqq \sum_{\{x, y\} \in E(H) :\, \varphi(\{x, y\}) = \{a, e\}} w(\{\prj{v_a}{x}, \prj{v_a}{y}\}). \]

    We now show how to build $G'$ in time $\Oh(m + s)$.
    We can build $V'_a$ for each $a \in V(T)$ trivially in time $\Oh(|S_{B(a)}|) = \Oh(s)$.
    To build $V'_e$ for each $e = \{a, b\} \in E(T)$, we first project $S_{B(a)}$ and $S_{B(b)}$ onto $B(e)$ in time $\Oh(s)$ and unite the two sequences.
    The resulting sequence is $S_{B(e)}$ but with repetitions.
    To delete repetitions, we sort the sequence lexicographically in time $\Oh(m + s)$ using radix sort.
    By remembering throughout this process for each element of $S_{B(e)}$ from which elements of $S_{B(a)}$ and $S_{B(b)}$ it was created through projection, we can build edges between $V'_a$ and $V'_e$ and between $V'_b$ and $V'_e$ in time $\Oh(s)$.
    To compute the edge weights $w'$ of edges between $V'_a$ and $V'_e$ for each $\{a, e\} \in E(T')$, we iterate over $\{x, y\} \in E(H)$, and if $\varphi(\{x, y\}) = \{a, e\}$, we add the corresponding weights from $G$ to $G'$.
    For a fixed $\{x, y\} \in E(H)$ with $\varphi(\{x, y\}) = \{a, e\}$, we first sort edges between $V_x$ and $V_y$ in $G$ in time $\Oh(m)$ using radix sort.
    Note that $x, y \in B(a)$ holds by the definition of $\varphi$.
    We then use radix sort to lexicographically sort $V'_a$ in time $\Oh(m + s)$ where $x$ and $y$ come first in the ordering of $B(a)$ according to which we sort lexicographically.
    Now use two pointers to go through all $v_a \in V'_a$ and corresponding edges $\{\prj{v_a}{x}, \prj{v_a}{y}\}$.
    We then add $w(\{\prj{v_a}{x}, \prj{v_a}{y}\})$ to the weight $w'$ of the edge $\{v_a, \prj{v_a}{B(e)}\}$, which is the single outgoing edge from $v_a$ to $V'_e$.
    Thus, in time $\Oh(m + s)$ we build $G'$.

    We now prove that there is a weight-preserving bijection between $H$-subgraphs of $G$ encoded by $\partenc$ and $T'$-subgraphs of $G'$.
    Fix some $H$-subgraph $\bv$ of $G$ encoded by $\partenc$.
    Let $I(\bv) \coloneqq \bv' \coloneqq (\prj{\bv}{B(a)})_{a \in V(T')}$.
    We claim that $\bv'$ is a $T'$-subgraph in $G'$ with the same weight as the $H$-subgraph $\bv$ in $G$.
    The fact that $\bv'$ is a $T'$-subgraph in $G'$ follows from the definition of $G'$.
    The weight is preserved as for every edge $\{x, y\} \in E(H)$, there is exactly one edge $\{a, e\} \in E(T')$, such that the weights of edges between $V_x$ and $V_y$ are encoded on the edges between $V'_a$ and $V'_e$.

    Furthermore, if $\bu$ is some other $H$-subgraph of $G$ encoded by $\partenc$, there is some $x \in V(H)$ such that $v_x \neq u_x$.
    There is some node $a \in V(T)$ such that $x \in B(a)$, and thus $I(\bv)_x \neq I(\bu)_x$.
    Hence, $I$ is an injection.

    It remains to show that $I$ is a surjection and how to compute $I^{-1}$ in constant time.
    Fix some $T'$-subgraph $\bv'$ of $G'$.
    For each $x \in V(H)$, fix some node $\nu(x) \in V(T)$ such that $x \in B(\nu(x))$.
    Let $\bv \coloneqq (\prj{v'_{\nu(x)}}{x})_{x \in V(H)}$.
    It is clear that $\bv$ can be computed from $\bv'$ in constant time.
    We claim that $\bv$ is an $H$-subgraph of $G$ encoded by $\partenc$ and $I(\bv) = \bv'$.
    First, we claim that $v_x = \prj{v'_a}{x}$ for any $a \in V(T')$ with $x \in B(a)$.
    That is, the definition of $\bv$ does not depend on the specific choice of nodes $\nu(x)$.
    Note that as $T$ corresponds to some tree decomposition of $H$, nodes $a \in V(T)$ with $x \in B(a)$ form a connected subgraph of $T$.
    After subdivision this fact still holds.
    That is, nodes $a \in V(T')$ with $x \in B(a)$ form a connected subgraph of $T'$.
    For any two such adjacent nodes $a, b \in V(T')$, by the definition of $G'$, we have $\prj{v'_a}{x} = \prj{v'_b}{x}$ as $v'_a$ and $v'_b$ are adjacent in $G'$ and edges in $G'$ are projections.
    Hence, by transitivity we get $\prj{v'_a}{x} = \prj{v'_{\nu(x)}}{x} = v_x$.
    From this, it follows that $I(\bv) = \bv'$.
    Furthermore, we claim that $\bv$ is an $H$-subgraph in $G$.
    For each $\{x, y\} \in E(H)$, by the definition of a tree decomposition there is some $a \in V(T)$ with $x, y \in B(a)$.
    As $\prj{\bv}{B(a)} = v'_a \in S_{B(a)}$, we have that $\prj{\bv}{B(a)}$ forms an $H[B(a)]$-subgraph in $G[H[B(a)]]$ as we initially have filtered out all elements of $S_{B(a)}$ that do not form $H[B(a)]$-subgraphs in $G[H[B(a)]]$.
    In particular, we have $\{v_x, v_y\} \in E(G)$ as $x, y \in B(a)$.
    As this holds for each $\{x, y\} \in E(H)$, we get that $\bv$ is indeed an $H$-subgraph in $G$.
    Finally, for each $a \in V(T)$, we have $\prj{\bv}{B(a)} = v'_a \in S_{B(a)}$, and thus $\bv$ is encoded by $\partenc$.
    This concludes the proof of the fact that $\bv$ is an $H$-subgraph of $G$ encoded by $\partenc$ and $I(\bv) = \bv'$, thus proving that $I$ is a surjection.
\end{proof}

\allreductions*

\begin{proof}[Proof of \cref{all-reductions}]
    For the first claim, we first solve \Henciso on $G$ in time $T(m)$, which returns an $H$-encoding consisting of partial $H$-encodings $\partenc_1, \partenc_2, \ldots, \partenc_k$. Recall that $k$ must be constant, i.e., it may depend on $H$ but not on $G$.
    Since the running time is $T(m)$, each partial $H$-encoding has size $\Oh(T(m))$. 
    For each partial encoding $\partenc_i$, we apply \cref{lm:reduce-encoding-to-tree} to build a tree $H'_i$ and a host graph $G'_i$ for pattern $H'_i$.
    We solve \miniso{H'_i} on $G'_i$ in linear time as $H'_i$ is a tree.
    This yields a weight for each $i$, and we return the minimum of all these weights.
    
    Since $\partenc_1, \partenc_2, \ldots, \partenc_k$ is a full $H$-encoding, the minimum-weight $H$-subgraph of $G$ is encoded by some partial encoding $\partenc_i$, and by the weight-preserving property of \cref{lm:reduce-encoding-to-tree} we obtain the same minimum weight by solving \miniso{H'_i} on $G'_i$. This proves correctness. 
    The running time of this algorithm is $\Oh(T(m))$ as $k$ is constant.

    \medskip

    The second claim of the lemma follows similarly.
    We apply \cref{lm:reduce-encoding-to-tree} (in the unweighted setting) to each $\partenc_i$.
    Since $H'_i$ is a tree, we can then solve \listiso{H'_i} on $G'_i$ in time $\Oh(T(m) + t_i)$ where $t_i$ is the number of $H$-subgraph of $G$ encoded by $\partenc_i$ (see, e.g.,~\cite{BaganDG07}).
    By the definition of a full $H$-encoding, we have $\sum_i t_i = t$, and thus in total we work in time $\Oh(T(m) + t)$.
    Each $H'_i$-subgraph of $G'_i$ can be converted into an $H$-subgraph of $G$ in constant time.

    \medskip

    It remains to show the third claim of the lemma.
    In the preprocessing phase we solve \Henciso on $G$ in time $\Oh(T(m))$ and apply \cref{lm:reduce-encoding-to-tree} to each $\partenc_i$ to build corresponding instances $G'_i$ of \enumiso{H'_i}.
    Since $H'_i$ is a tree, \enumiso{H'_i} can be solved in preprocessing time $\Oh(|E(G'_i)|) = \Oh(T(m))$ and constant delay for each $i$ (see, e.g.,~\cite{BaganDG07}).
    Thus, the preprocessing takes time $\Oh(T(m))$ in total.

    To enumerate the next $H$-subgraph of $G$, we pick any $i$ such that we did not yet enumerate all $H'_i$-subgraphs in $G'_i$ and ask \enumiso{H'_i} to generate a new  $H'_i$-subgraph of $G'_i$. This takes constant delay.
    Then, in constant time we can convert the computed $H'_i$-subgraph of $G'_i$ into an $H$-subgraph of $G$.
    As $\partenc_1, \partenc_2, \ldots, \partenc_k$ are a full $H$-encoding, every $H$-subgraph of $G$ is enumerated exactly once.
    Hence, we obtain constant-time delay.
\end{proof}

\inducedminorencoding*

\begin{proof}[Proof of \cref{induced-minor-encoding}]
    Given a host graph $G'$ for \enciso{H'}, we apply \cref{induced-minor-reduction} (in the unweighted setting) to build a host graph $G$ for \Henciso in time $\Oh(m)$. Given $G$, we solve \Henciso in time $\Oh(T(m))$. As a result, we get a full $H$-encoding of $G$ consisting of partial $H$-encodings $\partenc_1, \partenc_2, \ldots, \partenc_k$ of $G$. There is a bijection $\corr$ between $H$-subgraphs in $G$ and $H'$-subgraphs in $G'$.
    We will create partial $H'$-encodings $\partenc'_1, \ldots, \partenc'_k$ of $G'$, such that $\partenc'_i$ encodes exactly the $H'$-subgraphs of  $G'$ corresponding to the $H$-subgraphs of $G$ that are encoded by $\partenc_i$, which solves the problem.
    
    Fix some partial encoding $\partenc_i$.
    Given its underlying tree decomposition $\mathcal{T}$ of $H$, we can create a tree decomposition $\mathcal{T}'$ of $H'$ by preserving the tree structure of $\mathcal{T}$, and for each bag $B$ of~$\mathcal{T}$ creating a bag $\varphi^{-1}(B) \coloneqq \{x \in V(H') \mid \varphi(x) \cap B \neq \emptyset\}$ in $\mathcal{T}'$, where $\varphi$ is taken from the proof of \cref{induced-minor-reduction}.
    Consider the bijection $\corr$ (and its inverse $\corr^{-1}$) between $H$-subgraphs of $G$ and $H'$-subgraphs of $G'$ from \cref{induced-minor-reduction}, and denote the submaterializations of $\partenc_i$ by $S_B$, for each bag $B$ of $\mathcal{T}$.
    Fix some $\bu \in S_B$.
    For each $x \in \varphi^{-1}(B)$, pick some $a \in \varphi(x) \cap B$.
    Let $j$ be such that $u_a = v_{a, j}$.
    Denote $u'_x = v'_{x, j}$.
    Similarly to \cref{induced-minor-reduction}, one can see that the definition of $u'_x$ does not depend on the choice of $a$.
    Define $\corr^{-1}(\bu) \coloneqq (u'_x)_{x \in \varphi^{-1}(B)}$.
    We define the submaterialization of bag $\varphi^{-1}(B)$ by setting $S'_{\varphi^{-1}(B)} \coloneqq \{\corr^{-1}(\bu) \mid \bu \in S_B\}$.
    As the algorithm for \Henciso works in time $\Oh(T(m))$, all submaterializations of $\partenc_i$ have size at most $\Oh(T(m))$.
    Thus, creating $S'_{\varphi^{-1}(B)}$ takes time $\Oh(|S_B|) = \Oh(T(m))$.

    We now prove that the constructed $\partenc'_i$ is a partial $H'$-encoding.
    First, we prove that $\mathcal{T}'$ is a tree decomposition of $H'$.
    Consider any edge $\{x, y\} \in E(H')$.
    As $H'$ is an induced minor of $H$, there is an edge $\{a, b\} \in E(H)$ such that $a \in \varphi(x)$ and $b \in \varphi(y)$.
    There is a bag $B$ of $\mathcal{T}$ that contains $a$ and $b$.
    Hence, $\varphi^{-1}(B)$ contains $x$ and $y$.

    Now consider some $x \in V(H')$.
    Note that $x$ is contained in exactly the bags $\varphi^{-1}(B)$ of $\mathcal{T}'$, for which $\varphi(x) \cap B \neq \emptyset$.
    Note that $\varphi(x)$ is a connected subset of nodes of $H$.
    Nodes from a connected subset of nodes of $H$ lie in the bags of a connected subtree of the tree decomposition $\mathcal{T}$ of $H$.
    Therefore, $x$ lies in the bags of a connected subtree of $\mathcal{T}'$.
    Hence, $\mathcal{T}'$ is a tree decomposition of $H'$.
    Furthermore, for each bag $B$ of $\mathcal{T}$, it is clear that $S'_{\varphi^{-1}(B)}$ is a submaterialization of $\varphi^{-1}(B)$.
    Therefore, $\partenc'_i$ is a partial $H'$-encoding.

    It remains to show that $\partenc'_i$ encodes exactly the $H'$-subgraphs of $G'$ that correspond to the $H$-subgraphs of $G$ encoded by $\partenc_i$.
    Consider some $H$-subgraph $\bu$ in $G$ encoded by $\partenc_i$.
    We have $\prj{\bu}{B} \in S_B$ for all bags $B$ of $\mathcal{T}$.
    By the definition of $S'_{\varphi^{-1}(B)}$, we have that $\corr^{-1}(\prj{\bu}{B}) \in S'_{\varphi{-1}(B)}$.
    As $\corr^{-1}(\prj{\bu}{B}) = \prj{\corr^{-1}(\bu)}{\varphi^{-1}(B)}$, we have that $\partenc'_i$ encodes $\corr^{-1}(\bu)$.
    Furthermore, if $\partenc'_i$ encodes some $H'$-subgraph $\bu'$ of $G'$, we have $\corr(\prj{\bu'}{\varphi^{-1}(B)}) \in S_B$ for all bags $B$ of $\mathcal{T}$, and as $\corr(\prj{\bu'}{\varphi^{-1}(B)}) = \prj{\corr(\bu')}{B}$, we have that $\corr(\bu')$ is encoded by $\partenc_i$.
\end{proof}

\cliqueseparatoralgocombination*

\begin{proof}
    For the given host graph $G$, for each $i \in [k]$ solve \enciso{H[S_i \cup C]} on $G[H[S_i \cup C]]$ to obtain an encoding $\partenc_{i,1},\ldots,\partenc_{i,k_i}$. 
    We now consider all combinations of partial $H[S_i \cup C]$-encodings over all $i \in [k]$, i.e., we consider all tuples $(\partenc_{1,i_1},\ldots,\partenc_{k,i_k})$ with $i_j \in \{1,\ldots,k_j\}$ for all $j \in \{1,\ldots,k\}$.
    Note that there is a constant number of such combinations.
    By \cref{induced-clique-td} and since $C$ is a clique, each tree decomposition of each partial $H[S_i \cup C]$-encodings has a bag that contains $C$.
    We connect such a bag from the tree decomposition of $\partenc_{1,i_1}$ with the corresponding bags from the tree decompositions of $\partenc_{2,i_2},\ldots,\partenc_{k,i_k}$.
    Observe that this yields a tree decomposition of $H$.
    We preserve the already created submaterializations of all bags, thus creating a partial $H$-encoding of $G$.
    Note that any $H$-subgraph $\bv$ of $G$ is encoded exactly once, as for each of $H[S_i \cup C]$ we have that $\prj{\bv}{S_i \cup C}$ is encoded exactly once in the solution of \enciso{H[S_i \cup C]}, thus only the combination corresponding to all such partial $H[S_i \cup C]$-encodings encodes $\bv$.

    We now analyze the time complexity. Running the algorithm of \enciso{H[S_i \cup C]} for all $i \in [k]$ takes time $\Oh(T(m) \cdot k) = \Oh(T(m))$.
    As these algorithms work in time $\Oh(T(m))$, all partial encodings have size $\Oh(T(m))$.
    Hence, combining the partial encodings takes time $\Oh(T(m))$.
\end{proof}

\tcfromD*

\begin{proof}[Proof of \cref{tc-from-D}]
    We prove the lemma by induction on $|V(H)|$.
    If $H$ does not have a clique separator, then $D(H) = \{H\}$, and the claim is obvious.
    Otherwise, there is some minimal clique separator $C$ in $H$.
    Let $S_1, S_2, \ldots, S_k$ be the connected components of $H - C$.
    Note that by \cref{d-set-computation}, we have $D(H[S_i \cup C]) \subseteq D(H)$ for each $i$.
    Hence, for each $H' \in D(H[S_i \cup C])$, \enciso{H'} can be solved in time $\Oh(T(m))$.
    Consequently, by the induction hypothesis, \enciso{H[S_i \cup C]} can be solved in time $\Oh(T(m))$ for each $i$.
    Applying \cref{clique-separator-algo-combination}, we get that \Henciso can be solved in time $\Oh(T(m))$. Since $H$ has constant size, the number of induction steps is constant, so the overhead remains constant.
\end{proof}

\section{Algorithms for Patterns in \boldmath$\family$} \label{sec:p-graph-upper-bounds}

In this section we prove the following lemma that was stated in \cref{sec:proof-overview}.

\renewcommand{\maybe}{\lipicsEnd}
\pgraphupperbound*
\renewcommand{\maybe}{}

Throughout this section, $\alpha, \beta$, and $\gamma$ denote integers with $\alpha \ge \beta \ge 3$ and $\gamma \ge 1$.
We start by considering three simple graph families, see \cref{edge-ub,cycle-ub,biclique-ub}. Then in \cref{sec:upper-bound-P-graphs-alpha-gamma} we cover all graphs of the form $P(\alpha, \gamma \times 2)$, and in \cref{sec:upper-bound-P-graphs-alpha-beta-gamma} we cover all graphs of the form $P(\alpha, \beta, \gamma \times 2)$.
Finally, in \cref{sec:upper-bound-P-graphs-combination} we put the pieces together to prove \cref{lm:p-graph-upper-bound}.

\begin{lemma} \label{edge-ub}
    For $H = K_2$, \Henciso can be solved in time $\Oh(m)$.
\end{lemma}

\begin{proof}
    We create a single tree decomposition consisting of a single bag $B \coloneqq V(K_2)$. It takes $\Oh(m)$ time to materialize $B$ by picking adjacent pairs of nodes.
\end{proof}

\begin{lemma} \label{cycle-ub}
    For $H = C_k$, \Henciso can be solved in time $\Oh(m^{2 - 1 / \left\lceil \frac{k}{2} \right\rceil })$ for any $k \ge 3$.
\end{lemma}

\begin{proof}
    We prove the lemma for even values of $k$ because for odd values of $k$, the desired time complexity for $C_k$ is the same as for $C_{k+1}$, and thus by \cref{induced-minor-encoding}, the claim follows from the even case.

    Denote the nodes of $C_k$ by $a_1$, $a_2$, $\ldots$, $a_k$ in counterclockwise order. Define $\ell \coloneqq \frac{k}{2}$ (so $k=2\ell$). We want to solve the problem in time $\Oh(m^{2-1 / \ell})$.

    Consider some $H$-subgraph $\bv$ of $G$. We create several partial $H$-encodings of $G$ and claim that exactly one of them encodes $\bv$.

    We split the nodes in $V_{a_1}, V_{a_2}, \ldots, V_{a_k}$ into the ones that have degrees smaller than $m^{1 / \ell}$ (low-degree nodes, sets $V_{a_i}^{\textup{lo}}$) and the ones that have degrees at least $m^{1 / \ell}$ (high-degree nodes, sets $V_{a_i}^{\textup{hi}}$).
    For each of the $2^k$ choices of high or low degrees in all parts, we construct a partial $H$-encoding that encodes exactly the $H$-subgraphs that satisfy these degree constraints.
    As $\bv$ satisfies exactly one of these $2^k$ cases, it will be encoded exactly once. Once we fix for each part of $G$ whether we consider a low-degree or a high-degree node in that part, we filter out all other nodes.
    We distinguish two cases.

    \begin{itemize}
        \item \emph{Case 1: $v_{a_i}$ has high degree for some  $i$.}
            In this case \cref{high-degree-small-cnt-nodes} implies $|V_{a_i}^{\textup{hi}}| \le m^{1 - 1 / \ell}$.
            Since $H - \{a_i\}$ is a tree, it has a tree decomposition where each bag consists of two adjacent nodes.
            Adding $a_i$ to every bag, we obtain a tree decomposition of $H$ (see \cref{C12-upper-bound-1}).
            All of these bags can be materialized in time $\Oh(m^{2 - 1 / \ell})$ because there are $\Oh(m)$ choices for the two adjacent nodes and $\Oh(m^{1 - 1 / \ell})$ choices for a node from $V_{a_i}^{\textup{hi}}$.
    \begin{figure}
    \begin{center}
        \includegraphics[scale=0.7]{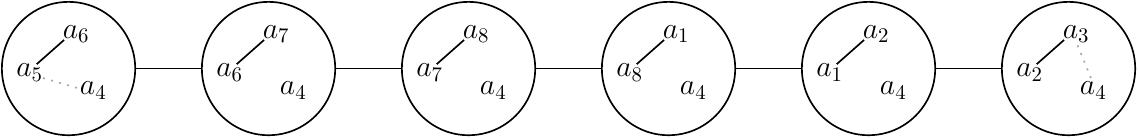}
    \end{center}

    \caption{An example of a tree decomposition of $C_k$ for Case 1 of \cref{cycle-ub} for $k = 8$  and $i = 4$.}
    \label{C12-upper-bound-1}
    \end{figure}
    
    \item \emph{Case 2: $v_{a_i}$ has low degree for all $i$.}
        In this case we create a tree decomposition of $H$ consisting of two bags: $B_1 \coloneqq \{a_1, a_2, \ldots, a_{\ell+1}\}$ and $B_2 \coloneqq \{a_{\ell + 1}, a_{\ell + 2}, \ldots, a_k, a_1\}$. 
        It is easy to see that it is indeed a valid tree decomposition of $H$ (see \cref{C12-upper-bound-2}).
        To materialize $B_1$, note that there are $\Oh(m)$  choices for a pair of adjacent nodes from $V_{a_1}^{\textup{lo}}$ and $V_{a_2}^{\textup{lo}}$, and then for each $i \in [3,\ell+1]$ there are $\Oh(m^{1 / \ell})$ choices for a neighbor in $V_{a_i}^{\textup{lo}}$ of the chosen node in $V_{a_{i - 1}}^{\textup{lo}}$, because all these nodes have degree at most $m^{1 / \ell}$. Thus, we materialize $B_1$ in time  $\Oh(m^{1 + (\ell - 1) \cdot 1 / \ell}) = \Oh(m^{2 - 1 / \ell})$ in total. Bag $B_2$ is materialized analogously.
    \begin{figure}
    \begin{center}
        \includegraphics[scale=0.85]{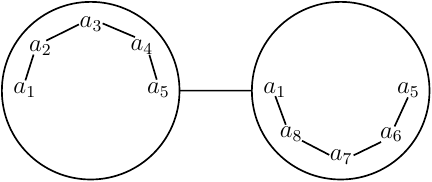}
    \end{center}

    \caption{An example of a tree decomposition of $C_k$ for Case 2 of \cref{cycle-ub} for $k = 8$.}
    \label{C12-upper-bound-2}
    \end{figure}
    \end{itemize}

    The cases are mutually exclusive and cover all $H$-subgraphs of $G$.
\end{proof}

\begin{lemma} \label{biclique-ub}
    Let $H = K_{2, k}$ for $k \ge 2$. \Henciso can be solved in time $\Oh(m^{2 - \frac{1}{k}})$.
\end{lemma}

\begin{proof}
    Denote the nodes in the first part of the bipartite graph as $a_1$, $a_2$, and the nodes in the second part as $b_1$, $b_2$, $\ldots$, $b_k$.

    Consider some $H$-subgraph $\bv$ of $G$. We create several partial $H$-encodings of $G$ and claim that exactly one of them encodes $\bv$.

    We split the nodes in $V_{a_1}$ and $V_{a_2}$ into the ones that have degrees smaller than $m^{1 / k}$ (low-degree nodes, sets $V_{a_1}^{\textup{lo}}$ and $V_{a_2}^{\textup{lo}}$) and the ones that have degrees at least $m^{1 / k}$ (high-degree nodes, sets $V_{a_1}^{\textup{hi}}$ and $V_{a_2}^{\textup{hi}}$).
    For each of the four choices of high or low degrees in parts $V_{a_1}$ and $V_{a_2}$, we construct a partial $H$-encoding that encodes exactly the $H$-subgraphs that satisfy these degree constraints.
    As $\bv$ satisfies exactly one of these four cases, it will be encoded exactly once.
    Once we fix for $V_{a_1}$ and $V_{a_2}$ whether we consider a low-degree or a high-degree node in that part, we filter out all other nodes. We distinguish three cases.

    \begin{itemize}
        \item \emph{Case 1: $v_{a_1}$ has high degree.}
            In this case Observation \ref{high-degree-small-cnt-nodes} implies $|V_{a_1}^{\textup{hi}}| \le m^{1 - 1/k}$.
            Since $H - \{a_1\}$ is a tree, it has a tree decomposition where each bag consists of two adjacent nodes.
            Adding $a_1$ to every bag, we obtain a tree decomposition of $H$.
            All of these bags can be materialized in time $\Oh(m^{2 - 1 / k})$ because there are $\Oh(m)$ choices for the two adjacent nodes and $\Oh(m^{1 - 1 / k})$ choices for a node from $V_{a_1}^{\hi}$.
        \item \emph{Case 1': $v_{a_1}$ has low degree and $v_{a_2}$ has high degree.} This case is symmetric to the previous one.
        \item \emph{Case 2: $v_{a_1}$ and $v_{a_2}$ have low degrees.}
        In this case we create a tree decomposition of $H$ consisting of two adjacent bags: $B_1 \coloneqq \{a_1, b_1, b_2, \ldots, b_k\}$ and $B_2 \coloneqq \{a_2, b_1, b_2, \ldots, b_k\}$.
        It is easy to see that it is indeed a valid tree decomposition of $H$.
        To materialize $B_1$, note that there are $\Oh(m)$ choices for a pair of adjacent nodes from $V_{a_1}^{\lo}$ and $V_{b_1}$, and then for each $i \in [2, k]$ there are $\Oh(m^{1 / k})$ choices for a neighbor in $V_{b_i}$ of the chosen node in $V_{a_1}^{\lo}$, because all nodes in $V_{a_1}^{\lo}$ have degree at most $m^{1 / k}$. Thus, we materialize $B_1$ in time $\Oh(m^{1 + (k - 1) \cdot 1 / k}) = \Oh(m^{2 - 1 / k})$. Bag $B_2$ is materialized analogously.
    \end{itemize}

    The cases are mutually exclusive and cover all $H$-subgraphs of $G$.
\end{proof}

\subsection{Algorithm for $P(\alpha, \gamma \times 2)$}
\label{sec:upper-bound-P-graphs-alpha-gamma}

\begin{lemma}[Biased Cycle] \label{lm:biased-cycle-algorithm}
    Let $H$ be a $(k+1)$-cycle on nodes $d_0, d_1, \ldots, d_{k}$. Let $G$ be a host graph, such that between all adjacent parts except $V_{d_0}$ and $V_{d_{k}}$ there are at most $m$ blue edges, and between $V_{d_0}$ and $V_{d_{k}}$ there are $\Oh(m^{2 - \ell f})$ red edges for some $\ell \in \NN$ and $f \in (0, 1)$. Say $\jj_i \in \NN$ for $i \in [0, k]$ are such that all nodes in part $V_{d_i}$ have at most $m^{\jj_i f}$ blue adjacent edges. Furthermore, if $\jj_i > 1$, then $|V_{d_i}| \le m^{1 - f}$. We call such a graph $G$ \emph{an instance of Biased Cycle}. Then there is an algorithm that solves \Henciso on $G$ in the following time with a single partial $H$-encoding.

    \begin{itemize}
        \item $\Oh(m^{2 - f} + m^{1 + (\jj_{\ell \dd k - \ell} + \ell - 1) f})$ if $\ell \le \frac{k + 1}{2}$, where $\jj_{x \dd y} \coloneqq \sum_{x \le i \le y} \jj_i$ for any $x, y \in [0, k]$.\footnote{Note that if $\ell = \frac{k + 1}{2}$, we have $\jj_{\ell \dd k - \ell} = \jj_{\ell \dd \ell - 1} = 0$.}
        \item $\Oh(m^{2 - f})$ if $\ell > \frac{k + 1}{2}$.
    \end{itemize}
    Furthermore, if $\ell \le \frac{k + 1}{2}$ and $\jj_r > 1$ for some $r \in \{0, \ldots, \ell - 1, k - \ell + 1, \ldots, k\}$, the algorithm works in time $\Oh(m^{2 - f})$.
    We call such an instance of Biased Cycle \emph{$\ell$-reachable}.

\end{lemma}

\begin{proof}
    We consider four cases depending on given values $\jj$, and for each case present a single partial $H$-encoding of $G$ that encodes all $H$-subgraphs in $G$.

    \begin{itemize}
        \item \emph{Case 1: $\jj_r > 1$ for some $r \in [0, \min\{\ell - 1, k - 1\}]$.}
            Choose minimal such $r$. We create the following tree decomposition of $H$.
            Bag $B_0$ consists of nodes $d_k, d_0, d_1, \ldots, d_r$.
            Furthermore, we create bags $B_i \coloneqq \{d_r, d_i, d_{i + 1}\}$ for $i \in [r + 1, k - 1]$.
            We connect them in a line: $B_0 \edg B_{k - 1} \edg B_{k - 2} \edg \cdots \edg B_{r + 1}$.
            It is clear that this is a tree decomposition of $H$.
            Note that as $\jj_r > 1$, the statement of the lemma implies $|V_{d_r}| \le m^{1 - f}$.
            Hence, each bag $B_i$ for $i \in [r + 1, k - 1]$ can be materialized in time $\Oh(m^{2 - f})$ because there are $\Oh(m)$ choices for the two adjacent nodes from $V_{d_i}$ and $V_{d_{i + 1}}$ and $\Oh(m^{1 - f})$ choices for a node from $V_{d_r}$.
        To materialize $B_0$, note that there are $\Oh(m^{2 - \ell f})$ choices for an adjacent pair of nodes from $V_{d_0}$ and $V_{d_k}$, and for each $i \in [r]$ we pick a node from $V_{d_i}$ as one of the $\Oh(m^f)$ neighbors of the chosen node from $V_{d_{i - 1}}$.
        Thus, we materialize $B_0$ in time $\Oh(m^{2 - \ell f} \cdot m^{r f}) \le \Oh(m^{2 - f})$ as $r \le \ell - 1$.
        \item \emph{Case 1': $\jj_r = 1$ for all $r \in [0, \min\{\ell - 1, k - 1\}]$ and $\jj_r > 1$ for some $r \in [\max\{k - \ell + 1, 1\}, k]$.} This case is symmetric to the previous one.
        \item \emph{Case 2: $\jj_r = 1$ for all $r \in [0, \min\{\ell - 1, k - 1\}] \cup [\max\{k - \ell + 1, 1\}, k]$ and $\ell > \frac{k + 1}{2}$.}
            In this case $k \le 2 \ell - 2$. Hence, $\jj_r = 1$ for all $r \in [0, k]$.
            We create a tree decomposition of $H$ consisting of two bags: $B_1 \coloneqq \{d_k, d_0, d_1, \ldots, d_{\min\{\ell - 1, k - 1\}}\}$ and $B_2 \coloneqq \{d_0, d_k, d_{k - 1}, \ldots, d_{\max\{k - \ell + 1, 1\}}\}$.
            It is easy to see that it is a tree decomposition of $H$.
            To materialize $B_1$, note that there are $\Oh(m^{2 - \ell f})$ choices for an adjacent pair of nodes from $V_{d_0}$ and $V_{d_k}$, and then for each $i \in [\min\{\ell - 1, k - 1\}]$ there are $\Oh(m^f)$ choices to pick a node from $V_{d_i}$ as a neighbor of the chosen node from $V_{d_{i - 1}}$.
            Hence, we materialize $B_1$ in time $\Oh(m^{2 - \ell f} \cdot m^{\min\{\ell - 1, k - 1\} \cdot f}) \le \Oh(m^{2 - f})$.
            Bag $B_2$ can be materialized in time $\Oh(m^{2 - f})$ analogously.
        \item \emph{Case 3: $\jj_r = 1$ for all $r \in [0, \min\{\ell - 1, k - 1\}] \cup [\max\{k - \ell + 1, 1\}, k]$ and $\ell \le \frac{k + 1}{2}$.}
            In this case $\min\{\ell - 1, k - 1\} = \ell - 1$ and $\max\{k - \ell + 1, 1\} = k - \ell + 1$.
            Furthermore, $\ell - 1 < k - \ell + 1$.
            We create a tree decomposition of $H$ consisting of two bags: $B_1 \coloneqq \{d_k, d_0, d_1, \ldots, d_{\ell - 1}\}$ and $B_2 \coloneqq \{d_{\ell - 1}, d_{\ell}, \ldots, d_k\}$.
            It is easy to see that it is a tree decomposition of $H$.
            We materialize $B_1$ in time $\Oh(m^{2 - f})$ the same way it was materialized in the previous case.
            To materialize $B_2$, note that there are $\Oh(m)$ choices for an adjacent pair of nodes from $V_{d_k}$ and $V_{d_{k - 1}}$, and then for each $i \in \{k - 2, k - 3, \ldots, \ell - 1\}$ there are at most $m^{\jj_{i + 1} f}$ choices to pick a node from $V_{d_i}$ as a neighbor of the chosen node from $V_{d_{i + 1}}$.
            Hence, we materialize $B_2$ in time $\Oh(m \cdot m^{\jj_{\ell \dd k - 1} f})$.
            As we know that $\jj_{k - \ell + 1} = \jj_{k - \ell + 2} = \cdots = \jj_{k - 1} = 1$, we get that $\Oh(m \cdot m^{\jj_{\ell \dd k - 1} f}) = \Oh(m^{1 + (\jj_{\ell \dd k - \ell} + \ell - 1) f})$.
            Hence, we materialize both $B_1$ and $B_2$ in time $\Oh(m^{2 - f} + m^{1 + (\jj_{\ell \dd k - \ell} + \ell - 1) f})$.
    \end{itemize}

    Note that if $\ell > \frac{k + 1}{2}$, only Cases 1, 1', and 2 appear, and if $G$ is an $\ell$-reachable instance of Biased Cycle, only Cases 1 and 1' appear. In both situations the time complexity is $\Oh(m^{2 - f})$.
\end{proof}

\begin{lemma} \label{lm:pa2c-ub}
    For any $\alpha \ge 3$, $\gamma \ge 2$, and $H = P(\alpha, \gamma \times 2)$, $H$-subgraph isomorphism encoding can be solved in time $\Oh(m^{2 - 1 / (2 \gamma + \left\lfloor \frac{\alpha - 1}{2} \right\rfloor)})$.
\end{lemma}

\begin{proof}
    We prove the lemma for even values of $\alpha$, because for odd values of $\alpha$ the desired time complexity for $P(\alpha, \gamma \times 2)$ is the same as for $P(\alpha + 1, \gamma \times 2)$, and thus by \cref{induced-minor-encoding}, the claim follows from the even case.

    Define $\ell \coloneqq \frac{\alpha - 2}{2}$ (so $\alpha = 2\ell + 2$) and $f \coloneqq \frac{1}{2 \gamma + \ell}$.
    We need to solve the problem in time $\Oh(m^{2-f}) = \Oh(m^{1 + (2\gamma + \ell - 1) f})$.

    Consider some $H$-subgraph $\bv$ of $G$.
    We create several partial $H$-encodings of $G$ and claim that exactly one of them encodes $\bv$.
    
    We split the nodes in $V_{a_0}, V_{a_1}, V_{a_2}, \ldots, V_{a_{\alpha}}$ of $G$ into the ones that have degrees smaller than $m^f$ (low-degree nodes, sets $V_{a_i}^{\textup{lo}}$) and the ones that have degrees at least $m^f$ (high-degree nodes, sets $V_{a_i}^{\textup{hi}}$).
    For each of the $2^{\alpha + 1} = \Oh(1)$ choices of high and low degrees, we construct several partial $H$-encodings that encode exactly the $H$-subgraphs that satisfy these degree constraints.
    As $\bv$ satisfies exactly one of these $2^{\alpha + 1}$ cases, it will be encoded exactly once.
    Once we fix for each of the parts $V_{a_0} \ldots, V_{a_{\alpha}}$ whether we consider a low-degree or a high-degree node in that part, we filter out all other nodes.
    By a slight abuse of notation we still call the new filtered graph $G$.

    \medskip

    We consider four cases.
    \begin{itemize}
        \item \emph{Case 1: $v_{a_0}$ has high degree.} In this case $|V_{a_0}^{\hi}| \le m^{1 - f}$ due to \cref{high-degree-small-cnt-nodes}.
            Since $H - \{a_0\}$ is a tree, it has a tree decomposition where every bag consists of exactly two adjacent nodes.
            Adding $a_0$ to all bags, we obtain a tree decomposition of $H$.
            Each one of these bags can be materialized in time $\Oh(m^{2-f})$ because there are $\Oh(m)$ choices for the two adjacent nodes and $|V_{a_0}^{\hi}| = \Oh(m^{1 - f})$ choices for $v_{a_0}$.
        \item \emph{Case 1': $v_{a_0}$ has low degree, and $v_{a_{\alpha}}$ has high degree.}
            This case is symmetric to the previous one.
        \item \emph{Case 2: $v_{a_0}, v_{a_1}, \ldots, v_{a_{\ell}}, v_{a_{\alpha - \ell}}, \ldots, v_{a_{\alpha - 1}}$, and $v_{a_{\alpha}}$ have low degrees.}
            As $\alpha - \ell = \ell + 2$, all nodes $v_{a_i}$ except for $v_{a_{\ell + 1}}$ have low degrees.
            We create a tree decomposition with two adjacent bags: $B_1 \coloneqq \{a_0, a_1, \ldots, a_{\ell + 1}, c_1, c_2, \ldots, c_{\gamma}\}$ and $B_2 \coloneqq \{a_{\ell + 1}, a_{\ell + 2}, \ldots, a_{\alpha}, c_1, c_2,\allowbreak \ldots, c_{\gamma}\}$.
        It is easy to see that it is a tree decomposition of $H$ (see \cref{Pac-upper-bound-2}).
        To materialize $B_1$, note that there are $\Oh(m)$ choices for an adjacent pair of nodes from $V_{a_0}^{\lo}$ and $V_{a_1}^{\lo}$. Then, for each $c_i$, for $i \in [\gamma]$, there are $\Oh(m^f)$ choices to pick a node in $V_{c_i}$ as a neighbor of the chosen node in $V_{a_0}^{\lo}$.
        Finally, for each $i \in [2, \ell]$, there are $\Oh(m^f)$ choices to pick a node in $V_{a_i}^{\lo}$ as a neighbor of the chosen node in $V_{a_{i - 1}}^{\lo}$, and there are $\Oh(m^f)$ choices to pick a node in $V_{a_{\ell+1}}$ as a neighbor of the chosen node in $V_{a_\ell}^{\lo}$.
        In total, to materialize $B_1$, we use $\Oh(m^{1 + \gamma f + \ell f}) \le \Oh(m^{1 + (2\gamma + \ell - 1) f}) = \Oh(m^{2-f})$ time.
        We materialize $B_2$ analogously.
    \begin{figure}
    \begin{center}
        \includegraphics[scale=1.0]{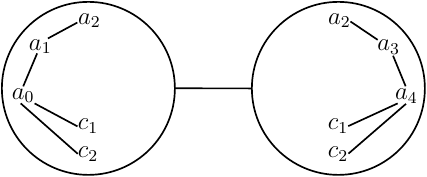}
    \end{center}

    \caption{An example of a tree decomposition of $P(\alpha, \gamma \times 2)$ for Case 2 of \cref{lm:pa2c-ub} for $\alpha = 4$ and $\gamma = 2$.}
    \label{Pac-upper-bound-2}
    \end{figure}
        \item \emph{Case 3: $v_{a_0}$ and $v_{a_{\alpha}}$ have low degrees, and at least one of $v_{a_1}, \ldots, v_{a_{\ell}}, v_{a_{\alpha - \ell}}, \ldots, v_{a_{\alpha - 1}}$ has a high degree.}
            As the degrees of nodes in $V_{a_0}^{\lo}$ are smaller than $m^f$, in time $\Oh(m^{1 + (\gamma - 1) f})$ we can materialize a bag $B_{a_0 \to c} \coloneqq \{a_0, c_1, c_2, \ldots, c_{\gamma}\}$: there are $\Oh(m)$ choices for an adjacent pair of nodes from $V_{a_0}^{\lo}$ and $V_{c_1}$, and there are $\Oh(m^f)$ choices to pick a node in $V_{c_i}$ as a neighbor of the chosen node in $V_{a_0}^{\lo}$, for each $i \in [2, \gamma]$.
            By using a binary search tree as a dictionary we can store for every tuple of nodes from parts $V_{c_1},V_{c_2}, \ldots, V_{c_{\gamma}}$ all their common neighbors in $V_{a_0}^{\lo}$ in time $\tOh(m^{1 + (\gamma - 1) f}) \le \Oh(m^{1 + (2\gamma + \ell - 1) f}) = \Oh(m^{2 - f})$.
            Note that $(v_{c_1}, \ldots, v_{c_{\gamma}})$ is one of the listed tuples as they have a common neighbor $v_{a_0} \in V_{a_0}^{\lo}$.
            Now we split the listed tuples of nodes into the ones that have less than $m^{\gamma f}$ common neighbors in $V_{a_0}^{\lo}$ (low-degree tuples) and the ones that have at least $m^{\gamma f}$ common neighbors (high-degree tuples).
            Similarly to nodes from parts $V_{a_i}$, we deal with low- and high-degree tuples separately.
            We distinguish two cases.
    \begin{itemize}
        \item \emph{Case 3.1: $(v_{c_1}, \ldots, v_{c_{\gamma}})$ is a high-degree tuple.}
            The total number of copies of $B_{a_0 \to c}$ that we generated is at most $m^{1 + (\gamma-1)f}$.
            Hence, there are at most $m^{1-f}$ high-degree tuples of nodes from parts $V_{c_1}, V_{c_2}, \ldots, V_{c_{\gamma}}$.
            Consider the graph $H - \{c_1, \ldots, c_{\gamma}\}$.
            Since it is a tree, it has a tree decomposition where each bag consists of two adjacent nodes.
            Adding $c_1, c_2, \ldots, c_{\gamma}$ to every bag, we obtain a tree decomposition of $H$ (see \cref{Pac-upper-bound-1}).
            We can materialize each bag in time $\Oh(m^{2-f})$, because there are $\Oh(m)$ choices for the two adjacent nodes and at most $m^{1 - f}$ choices for a high-degree tuple from $V_{c_1} \times V_{c_2} \times \cdots \times V_{c_{\gamma}}$.
            Note that this partial $H$-encoding encodes exactly such $H$-subgraphs $\bv$, for which $(v_{c_1}, \ldots, v_{c_{\gamma}})$ is a high-degree tuple because for all bags, only high-degree tuples $(v_{c_1}, \ldots, v_{c_{\gamma}})$ are listed.
    \begin{figure}
    \begin{center}
        \includegraphics[scale=1.0]{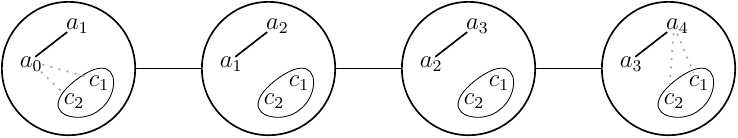}
    \end{center}

    \caption{An example of a tree decomposition of $P(\alpha, \gamma \times 2)$ for Case 3.1 of \cref{lm:pa2c-ub} for $\alpha = 4$ and $\gamma = 2$.}
    \label{Pac-upper-bound-1}
    \end{figure}

        \item \emph{Case 3.2: $(v_{c_1}, \ldots, v_{c_{\gamma}})$ is a low-degree tuple.}
            In time $\Oh(m^{1 + (\gamma - 1) f})$ we materialize a bag $B_{a_{\alpha} \to c} \coloneqq \{a_{\alpha}, c_1, c_2, \ldots, c_{\gamma} \}$ analogously to the materialization of $B_{a_0 \to c}$.
            If the chosen nodes from parts $V_{c_1}, V_{c_2}, \ldots, V_{c_{\gamma}}$ form a low-degree tuple, we spend time $\Oh(m^{\gamma f})$ to materialize all their common neighbors in $V_{a_0}$ by accessing the dictionary.
            This way we create a submaterialization $S_{\textup{down}}$ of $B_{\textup{down}} \coloneqq B_{a_{\alpha} \to c} \cup \{a_0\}$ in time $\Oh(m^{1 + (\gamma - 1) f} \cdot (\log m + m^{\gamma f})) = \Oh(m^{1 + (2\gamma-1) f}) = \Oh(m^{2 - (\ell + 1) f})$.
            Note that for an $H$-subgraph $\bv$, we have that $(v_{c_1}, \ldots, v_{c_{\gamma}})$ is a low-degree tuple if and only if $\prj{\bv}{B_{\textup{down}}} \in S_{\textup{down}}$.
        Thus, it remains to attach a decomposition of the path $a$ to $B_{\textup{down}}$.
        We attach a bag $B_{\textup{ends}} \coloneqq \{a_0, a_{\alpha}\}$ to $B_{\textup{down}}$ and create its submaterialization $S_{\textup{ends}}$ as a set of projections of the elements of $S_{\textup{down}}$ onto $B_{\textup{ends}}$.
        This takes time $\tOh(|S_{\textup{down}}|) = \tOh(m^{2 - (\ell + 1) f}) \le \tOh(m^{2 - 2f}) \le \Oh(m^{2 - f})$, and we have $|S_{\textup{ends}}| \le |S_{\textup{down}}| = \Oh(m^{2 - (\ell + 1) f})$.
        We consider $V_{a_0}, \ldots, V_{a_{\alpha}}$ as an instance of Biased Cycle, where we think of edges of $G$ between $V_{a_{i - 1}}$ and $V_{a_i}$ for $i \in [\alpha]$ as the blue edges, and we think of $S_{\textup{ends}}$ as the red edges between $V_{a_0}$ and $V_{a_{\alpha}}$.
        The values $\jj_i$ are defined as 1 if we pick low-degree nodes in $V_{a_i}$, and as $1 / f$ if we pick high-degree nodes in $V_{a_i}$.
        Since at least one of $v_{a_1}, \ldots, v_{a_{\ell}}, v_{a_{\alpha - \ell}}, \ldots, v_{a_{\alpha - 1}}$ has high degree in this case and $\ell + 1 \le \frac{\alpha + 1}{2}$, this instance of Biased Cycle is $(\ell + 1)$-reachable.
        Thus, we can apply \cref{lm:biased-cycle-algorithm} to create a full encoding of this Biased Cycle in time $\Oh(m^{2 - f})$.
        As there is an edge $\{a_0, a_{\alpha}\}$ in this instance of Biased Cycle, we can attach the resulting tree decomposition with a partial encoding to $B_{\textup{ends}}$ by a bag containing both $a_0$ and $a_{\alpha}$.
        Hence, we create a partial $H$-encoding of $G$.
        Note that this partial $H$-encoding encodes exactly such $H$-subgraphs $\bv$, for which $(v_{c_1}, \ldots, v_{c_{\gamma}})$ is a low-degree tuple, because for $B_{\textup{ends}}$ only low-degree tuples $(v_{c_1}, \ldots, v_{c_{\gamma}})$ are listed.
    \end{itemize}
    \end{itemize}

    \medskip

    We note that the cases are mutually exclusive and cover all $H$-subgraphs $\bv$ in $G$.
    Hence, we create a full $H$-encoding of $G$.
\end{proof}

\subsection{Algorithm for $P(\alpha, \beta, \gamma \times 2)$}
\label{sec:upper-bound-P-graphs-alpha-beta-gamma}

\begin{lemma} \label{lm:pabc-ub}
    For any $\alpha \ge \beta \ge 3$, $\gamma \ge 1$, and $H = P(\alpha, \beta, \gamma \times 2)$, \Henciso can be solved in time $\Oh(m^{2 - \frac{1}{\funcc(\alpha, \beta, \gamma)}})$.
    \lipicsEnd
\end{lemma}

\cref{sec:upper-bound-P-graphs-alpha-beta-gamma} is devoted to the proof of this lemma.

    Let $f \coloneqq \frac{1}{\funcc(\alpha, \beta, \gamma)}$. We want to solve the problem in time $\Oh(m^{2-f}) = \Oh(m^{1 + (\funcc(\alpha, \beta, \gamma) - 1) f})$.

    Consider some $H$-subgraph $\bv$ of $G$.
    We create several partial $H$-encodings of $G$ and claim that exactly one of them encodes $\bv$.

    We split the nodes in $V_{a_0}, V_{a_1}, V_{a_2}, \ldots, V_{a_{\alpha}}, V_{b_1}, V_{b_2}, \ldots, V_{b_{\beta - 1}}$ of $G$ into groups called subparts by their degrees: degrees in $[1, m^f)$ (sets $V_{a_i}^1$ and $V_{b_i}^1$), in $[m^f, m^{2f})$ (sets $V_{a_i}^2$ and $V_{b_i}^2$), $\ldots$, in $[m^{(\funcc(\alpha, \beta, \gamma) - 1) \cdot f}, m]$ (sets $V_{a_i}^{\funcc(\alpha, \beta, \gamma)}$ and $V_{b_i}^{\funcc(\alpha, \beta, \gamma)}$).
    For each of the $\fabc^{\alpha + \beta} = \Oh(1)$ choices of subparts in parts $V_{a_i}$ and $V_{b_i}$, we construct several partial $H$-encodings that encode exactly the $H$-subgraphs that satisfy these degree constraints.
    As $\bv$ satisfies exactly one of these $\fabc^{\alpha + \beta}$ cases, it will be encoded exactly once.
    Let $\jj_d$ be such that $v_d$ belongs to $V_d^{\jj_d}$ for any $d \in \{a_i \mid i \in [0, \alpha] \} \cup \{b_i \mid i \in [\beta - 1] \}$.
    Note that all nodes in $V_d^{\jj_d}$ have degrees at most $m^{\jj_d f}$, and $|V_d^{\jj_d}| \le m^{1 - (\jj_d - 1) f}$ due to \cref{high-degree-small-cnt-nodes} as all nodes in $V_d^{\jj_d}$ have degrees at least $m^{(\jj_d - 1) f}$ in $G$.
    Denote $\jmax \coloneqq \max \{ \jj_d \mid d \in \{a_i \mid i \in [0, \alpha] \} \cup \{b_i \mid i \in [\beta - 1]\} \}$ and $\vmax \coloneqq \argmax \{ \jj_d \mid d \in \{a_i \mid i \in [0, \alpha] \} \cup \{b_i \mid i \in [\beta - 1]\} \}$.
    Furthermore, we denote $\jj_{a_{x \dd y}} \coloneqq \sum_{x \le i \le y} \jj_{a_i}$ for any $x, y \in [0, \alpha]$ and $\jj_{b_{x \dd y}} \coloneqq \sum_{x \le i \le y} \jj_{b_i}$ for any $x, y \in [0, \beta]$.
    Once we fix for each of the parts $V_{a_i}$ and $V_{b_i}$ to which subpart belongs the corresponding node of $\bv$, we filter out all other nodes.
    By a slight abuse of notation we still call the new filtered graph $G$.
    
    After the degree splitting, in some cases simple algorithms are already applicable. We first deal with some such cases.

\subsubsection{High-degree-low-degree Algorithms}

We consider seven special cases. The precise condition of each case includes the negations of all previous cases, but we omit them for conciseness.

\paragraph*{Special Case 1: \boldmath$\jj_{a_0} \ge 2$.}

In this case $|V_{a_0}^{\jj_{a_0}}| \le m^{1 - f}$.
Since $H - \{a_0\}$ is a tree, it has a tree decomposition where every bag consists of exactly two adjacent nodes.
Adding $a_0$ to all these bags, we obtain a tree decomposition of $H$.
Each one of these bags can be materialized in time $\Oh(m^{2-f})$ because there are $\Oh(m)$ choices for the two adjacent nodes and $|V_{a_0}^{\jj_{a_0}}| \le m^{1 - f}$ choices for a node in~$V_{a_0}^{\jj_{a_0}}$.

\paragraph*{Special Case 1': \boldmath$\jj_{a_{\alpha}} \ge 2$.}

This case is symmetric to the previous one.

\smallskip

Hence, from now on we may assume $\jj_{a_0} = \jj_{a_{\alpha}} = 1$.

\paragraph*{Special Case 2: \boldmath$\jmax \le 2 \gamma + 1$ and $\alpha = \beta$.}

We create a tree decomposition of $H$ consisting of two adjacent bags:
\[B_1 \coloneqq \{a_0, a_1, \ldots, a_{\alpha - 1}, b_1, c_1, c_2, \ldots, c_{\gamma}\}\]
and
\[B_2 \coloneqq \{a_{\alpha - 1}, b_1, b_2, \ldots, b_{\beta - 1}, b_{\beta}, c_1, c_2, \ldots, c_{\gamma}\}.\]
It is easy to see that it is indeed a valid tree decomposition of $H$ (see \cref{Paac-upper-bound-all-small-degrees}).

To materialize $B_1$, note that there are $\Oh(m)$ choices for an adjacent pair of nodes from $V_{a_0}^{\jj_{a_0}}$ and $V_{b_1}^{\jj_{b_1}}$.
Then, for each $i \in [\gamma]$, there are $\Oh(m^f)$ choices to pick a node in $V_{c_i}$ as a neighbor of the chosen node in $V_{a_0}^{\jj_{a_0}}$.
Furthermore, for each $i \in [\alpha - 1]$, there are $\Oh(m^{\jj_{a_{i - 1}} f})$ choices to pick a node in $V_{a_i}^{\jj_{a_i}}$ as a neighbor of the chosen node in $V_{a_{i-1}}^{\jj_{a_{i-1}}}$.
In total, this takes time
\begin{align*}
    \Oh(m^{1 + (\gamma + \jj_{a_{0 \dd \alpha - 2}}) f}) &\le \Oh(m^{1 + (\gamma + 1 + (\alpha - 2) \cdot \jmax) f})\\
                                            &\le \Oh(m^{1 + (\gamma + 1 + (\alpha - 2) \cdot (2 \gamma + 1)) f})\\
                                               &= \Oh(m^{1 + ((2\alpha - 3) \gamma + \alpha - 1) f})\\
                                               &\le \Oh(m^{2 - f}),
\end{align*}
where the last inequality holds due to \cref{a-a-c-linear-bound}.

We materialize $B_2$ analogously as $\alpha = \beta$.
Hence, we get a full $H$-encoding of $G$.

\begin{figure}
    \begin{center}
        \includegraphics[scale=0.8]{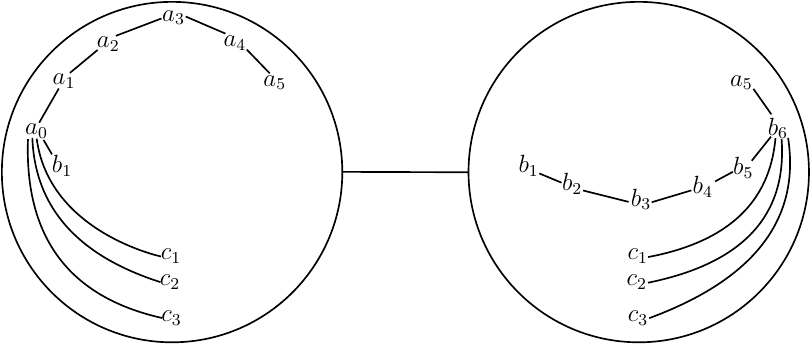}
    \end{center}

    \caption{An example of a tree decomposition of $P(\alpha, \beta, \gamma \times 2)$ for Special Case 2 of \cref{lm:pabc-ub} for $\alpha = \beta = 6$ and $\gamma = 3$.}
    \label{Paac-upper-bound-all-small-degrees}
\end{figure}

\paragraph*{Special Case 3: \boldmath$\jmax \le 2 \gamma + 1$ and $\alpha > \beta$.}

We create a tree decomposition of $H$ as follows (see \cref{Pabc-upper-bound-all-small-degrees}).
Let $B_b \coloneqq \{b_0, b_1, \ldots, b_{\beta}\}$.
To create a submaterialization $S_b$ of $B_b$, 
note that there are $\Oh(m)$ choices for an adjacent pair of nodes from $V_{b_0}^{\jj_{b_0}}$ and $V_{b_1}^{\jj_{b_1}}$. Then, for each $i \in [2, \beta]$, there are $\Oh(m^{\jj_{b_{i - 1}} f}) \le \Oh(m^{\jmax f})$ choices to pick a node in $V_{b_i}^{\jj_{b_i}}$ as a neighbor of the chosen node in $V_{b_{i-1}}^{\jj_{b_{i-1}}}$.
Hence, we materialize $B_b$ in time $\Oh(m^{1 + (\beta - 1) \jmax f}) \le \Oh(m^{1 + (\beta - 1) (2 \gamma + 1) f}) = \Oh(m^{2 - \ell f})$, where $\ell \coloneqq \funcc(\alpha, \beta, \gamma) - (\beta - 1) \cdot (2\gamma + 1)$.
Note that $\ell \ge 2$ due to \cref{a-b-c-linear-bound}.

Furthermore, we create bags $B_{c_i} \coloneqq B_b \cup \{c_i\}$ for all $i \in [\gamma]$. We materialize $B_{c_i}$ by starting from the submaterialization $S_b$ of $B_b$ and extending each tuple $u \in S_b$ by picking a node in $V_{c_i}$ as one of the $\Oh(m^{\jj_{b_0} f}) = \Oh(m^f)$ neighbors of the node in $V_{b_0}^{\jj_{b_0}}$ that is picked by $u$.
Hence, each $B_{c_i}$ is materialized in time $\Oh(m^{(2 - \ell f) + f}) \le \Oh(m^{2 - f})$.
Furthermore, we create a bag $B_{\textup{ends}} \coloneqq \{a_0, a_{\alpha}\}$ and create its materialization $S_{\textup{ends}}$ as a set of projections of the elements of $S_b$ onto $B_{\textup{ends}}$.
This takes time $\tOh(|S_b|) = \tOh(m^{2 - \ell f}) \le \Oh(m^{2 - f})$, and we have $|S_{\textup{ends}}| \le |S_b| = \Oh(m^{2 - \ell f})$.

We connect the created bags in a line: $B_b \edg B_{c_1} \edg B_{c_2} \edg \cdots \edg B_{c_{\gamma}} \edg B_{\textup{ends}}$.
To create a tree decomposition of $H$, it remains to attach a decomposition of the path $a_0 \edg a_1 \edg \cdots \edg a_{\alpha}$ to $B_{\textup{ends}}$.

We consider $\subp{a_0}, \ldots, \subp{a_{\alpha}}$ as an instance of Biased Cycle, where we think of edges of $G$ between $\subp{a_{i - 1}}$ and $\subp{a_i}$ for $i \in [\alpha]$ as the blue edges, and we think of $S_{\textup{ends}}$ as the red edges between $\subp{a_0}$ and $\subp{a_{\alpha}}$.
We apply \cref{lm:biased-cycle-algorithm} to create a full encoding of this Biased Cycle.
If $\ell > \frac{\alpha + 1}{2}$, it works in time $\Oh(m^{2 - f})$, and if $\ell \le \frac{\alpha + 1}{2}$, it works in time \begin{align*}
    \Oh(m^{2 - f} + m^{1 + (\jj_{a_{\ell \dd \alpha - \ell}} + \ell - 1) f}) &\le \Oh(m^{2 - f} + m^{1 + ((\alpha - 2 \ell + 1) \cdot \jmax + \ell - 1) f})\\
                                                                &\le \Oh(m^{2 - f} + m^{1 + ((\alpha - 2 \ell + 1) \cdot (2 \gamma + 1) + \ell - 1) f})\\
                                                                &\le \Oh(m^{2 - f}),
\end{align*}
where the last inequality holds due to \cref{a-b-c-linear-bound-2}.

As there is an edge $\{a_0, a_{\alpha}\}$ in this instance of Biased Cycle, we can attach the resulting tree decomposition with a partial encoding to $B_{\textup{ends}}$ by a bag containing both $a_0$ and $a_{\alpha}$.
Therefore, we create a full $H$-encoding of $G$.

\begin{figure}
    \begin{center}
        \includegraphics[scale=0.85]{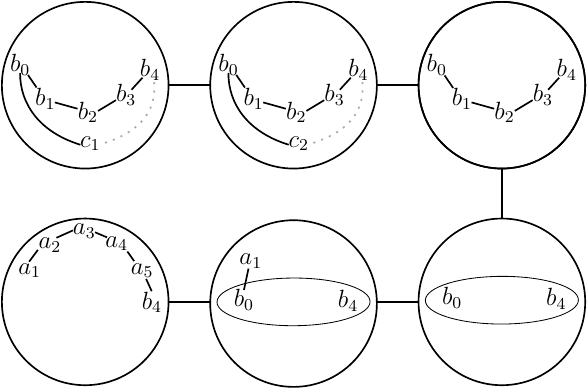}
    \end{center}

    \caption{An example of a tree decomposition of $P(\alpha, \beta, \gamma \times 2)$ for Special Case 3 of \cref{lm:pabc-ub} for $\alpha = 6$, $\beta = 4$, and $\gamma = 2$.}
    \label{Pabc-upper-bound-all-small-degrees}
\end{figure}

\medskip

Hence, from now on we may assume $\jmax \ge 2 \gamma + 2$.

\paragraph*{Special Case 4: for some positive integer \boldmath$\ell \le \alpha - 1$, we have $\jj_{a_1} = \jj_{a_2} = \cdots = \jj_{a_{\ell - 1}} = 1$, $\jj_{a_{\ell}} \ge \gamma + 2$, and $\jj_{b_{\beta - 1}} \ge \gamma + \ell + 1$.} \label{prg:case4}

We create the following tree decomposition of $H$ (see \cref{Pabc-upper-bound-split-by-middle-node}).
Define a bag
\[B_{1} \coloneqq \{ a_0, a_1, a_2, \ldots, a_{\ell}, b_{\beta-1}, c_1, c_2, \ldots, c_{\gamma} \},\]
a bag
\[B_{2} \coloneqq \{ a_{\ell}, b_{\beta-1}, b_{\beta}, c_1, c_2, \ldots, c_{\gamma} \},\]
bags $B_{b_i} \coloneqq \{b_{i-1}, b_i, b_{\beta-1}\}$ for $i \in [\beta-2]$, and bags $B_{a_i} \coloneqq \{a_{\ell}, a_i, a_{i + 1}\}$ for $i \in [\ell + 1, \alpha - 1]$.
We connect these bags in a line: $B_{b_{\beta-2}} \edg B_{b_{\beta-3}} \edg \cdots \edg B_{b_1} \edg B_{1} \edg B_{2} \edg B_{a_{\alpha-1}} \edg B_{a_{\alpha-2}} \edg \cdots \edg B_{a_{\ell+1}}$.
It is easy to see that it is indeed a valid tree decomposition of $H$.

To materialize $B_1$, note that the number of choices to pick a node from $V_{b_{\beta - 1}}^{\jj_{b_{\beta - 1}}}$ is $\Oh(|V_{b_{\beta - 1}}^{\jj_{b_{\beta - 1}}}|) = \Oh(m^{1 - (\jj_{b_{\beta - 1}} - 1) f}) \le \Oh(m^{1 - (\gamma + \ell) f})$. 
There are $\Oh(m)$ choices for an adjacent pair of nodes from $V_{a_0}^{\jj_{a_0}}$ and $V_{a_1}^{\jj_{a_1}}$.
Then, for each $i \in [\gamma]$, there are $\Oh(m^f)$ choices to pick a node in $V_{c_i}$ of the chosen node in $V_{a_0}^{\jj_{a_0}}$.
Finally, for each $i \in [2, \ell]$, there are $\Oh(m^{\jj_{a_{i - 1}} f}) = \Oh(m^f)$ choices to pick a node in $V_{a_i}^{\jj_{a_i}}$ as a neighbor of the chosen node in $V_{a_{i-1}}^{\jj_{a_{i-1}}}$.
In total, this takes time $\Oh(m^{(1 - (\gamma + \ell) f) + 1 + (\gamma + (\ell - 1)) f}) = \Oh(m^{2 - f})$.

To materialize $B_2$, note that the number of choices to pick a node from $V_{a_{\ell}}^{\jj_{a_{\ell}}}$ is $\Oh(|V_{a_{\ell}}^{\jj_{a_{\ell}}}|) = \Oh(m^{1 - (\jj_{a_{\ell}} - 1) f}) \le \Oh(m^{1 - (\gamma + 1) f})$.
There are $\Oh(m)$ choices to pick an adjacent pair of nodes from $V_{b_{\beta - 1}}^{\jj_{b_{\beta - 1}}}$ and $V_{b_{\beta}}^{\jj_{b_{\beta}}}$.
Then, for each $i \in [\gamma]$, there are $\Oh(m^{f})$ choices to pick a node in $V_{c_i}$ as a neighbor of the chosen node in $V_{b_{\beta}}^{\jj_{b_{\beta}}}$.
In total, this takes time $\Oh(m^{(1 - (\gamma + 1) f) + 1 + \gamma f}) = \Oh(m^{2 - f})$.

Each bag $B_{b_i}$, for $i \in [\beta - 2]$, can be materialized in time $\Oh(m^{2 - f})$, because there are $\Oh(m)$ choices for a pair of adjacent nodes from $\subp{b_{i - 1}}$ and $\subp{b_i}$, and $\Oh(|\subp{b_{\beta - 1}}|) \le \Oh(m^{1 - (\gamma + \ell) f}) \le \Oh(m^{1 - f})$ choices for a node from $\subp{b_{\beta - 1}}$.

Analogously, bags $B_{a_i}$ for $i \in [\ell + 1, \alpha - 1]$ can also be materialized in time $\Oh(m^{2 - f})$.
Hence, we get a full $H$-encoding of $G$ in time $\Oh(m^{2 - f})$.

\begin{figure}
\begin{center}
    \includegraphics[scale=0.85]{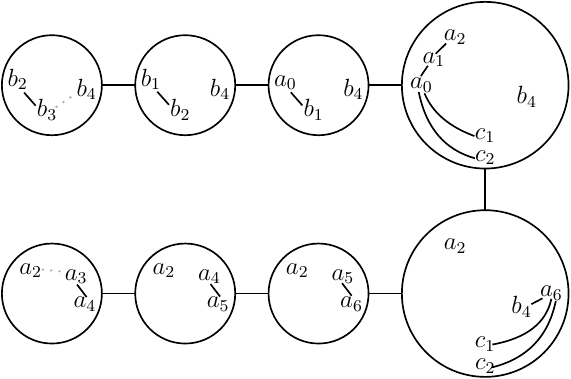}
\end{center}

    \caption{An example of a tree decomposition of $P(\alpha, \beta, \gamma \times 2)$ for Special Case 4 of \cref{lm:pabc-ub} for $\alpha = 6$, $\beta = 5$, $\gamma = 2$, and $\ell = 2$.}
\label{Pabc-upper-bound-split-by-middle-node}
\end{figure}

\paragraph*{Special Case 4': for some positive integer \boldmath$\ell \le \alpha - 1$, we have $\jj_{a_{\alpha - \ell + 1}} = \cdots = \jj_{a_{\alpha - 2}} = \jj_{a_{\alpha - 1}} = 1$, $\jj_{a_{\alpha - \ell}} \ge \gamma + 2$, and $\jj_{b_1} \ge \gamma + \ell + 1$.} \label{prg:case4p}

This case is symmetric to the previous one.

\paragraph*{Special Case 5: for some positive integer \boldmath$\ell \le \alpha / 2$, we have $\jj_{a_r} = 1$ for all $r \in \{0, \ldots, \ell - 1, \alpha - \ell + 1, \ldots, \alpha\}$, $\jj_{a_{\ell}} \ge \gamma + 2$, $\jj_{b_1} \le \gamma + \ell$, and $\beta = 3$.} \label{prg:case5}

We create the following tree decomposition of $H$ (see \cref{Pa3c-upper-bound-split-by-middle-node}).
Define a bag
\[B_{1} \coloneqq \{ a_0, a_1, \ldots, a_{\ell}, b_1, b_2, c_1, c_2, \ldots, c_{\gamma} \},\]
a bag
\[B_{2} \coloneqq \{ a_{\ell}, b_2, b_3, c_1, c_2, \ldots, c_{\gamma} \},\]
and bags $B_{a_i} \coloneqq \{a_{\ell}, a_i, a_{i + 1}\}$ for $i \in [\ell + 1, \alpha - 1]$.
We connect these bags in a line: $B_{1} \edg B_{2} \edg B_{a_{\alpha-1}} \edg B_{a_{\alpha-2}} \edg \cdots \edg B_{a_{\ell+1}}$.
It is easy to see that it is indeed a valid tree decomposition of $H$.

To materialize $B_1$, note that there are $\Oh(m)$ choices for an adjacent pair of nodes from $V_{a_0}^{\jj_{a_0}}$ and $V_{b_1}^{\jj_{b_1}}$.
Then there are $\Oh(m^{\jj_{b_1} f}) \le \Oh(m^{(\gamma + \ell) f})$ choices to pick a node in $\subp{b_2}$ as a neighbor of the chosen node in $V_{b_1}^{\jj_{b_1}}$. 
Then, for each $i \in [\gamma]$, there are $\Oh(m^f)$ choices to pick a node in $V_{c_i}$ as a neighbor of the chosen node in $V_{a_0}^{\jj_{a_0}}$. 
Finally, for each $i \in [\ell]$, there are $\Oh(m^{\jj_{a_{i - 1}} f}) = \Oh(m^f)$ choices to pick a node in $V_{a_i}^{\jj_{a_i}}$ as a neighbor of the chosen node in $V_{a_{i-1}}^{\jj_{a_{i-1}}}$.
In total, this takes time $\Oh(m^{1 + ((\gamma + \ell) + \gamma + \ell) f}) \le \Oh(m^{1 + (2 \gamma + \alpha) f}) \le \Oh(m^{2 - f})$ due to the fact that $\ell \le \alpha / 2$ and \cref{lm:a3c-linear-bound}.

To materialize $B_2$, note that the number of choices to pick a node in $V_{a_{\ell}}^{\jj_{a_{\ell}}}$ is $\Oh(|V_{a_{\ell}}^{\jj_{a_{\ell}}}|) = \Oh(m^{1 - (\jj_{a_{\ell}} - 1) f}) \le \Oh(m^{1 - (\gamma + 1) f})$.
There are $\Oh(m)$ choices to pick an adjacent pair of nodes $V_{b_2}^{\jj_{b_2}}$ and $\subp{b_3}$.
Then, for each $i \in [\gamma]$, there are $\Oh(m^f)$ choices to pick a node in $V_{c_i}$ as a neighbor of the chosen node in $\subp{b_3}$.
In total, this takes time $\Oh(m^{(1 - (\gamma + 1) f) + 1 + \gamma f}) = \Oh(m^{2 - f})$.

Each bag $B_{a_i}$, for $i \in [\ell + 1, \alpha - 1]$, can be materialized in time $\Oh(m^{2 - f})$, because there are $\Oh(m)$ choices for a pair of adjacent nodes from $\subp{a_i}$ and $\subp{a_{i + 1}}$, and the number of choices for a node from $\subp{a_{\ell}}$ is $\Oh(|\subp{a_{\ell}}|) \le \Oh(m^{1 - (\gamma + 1) f}) \le \Oh(m^{1 - f})$.
Hence, we get a full $H$-encoding of $G$ in time $\Oh(m^{2 - f})$.

\begin{figure}
    \begin{center}
        \includegraphics[scale=0.85]{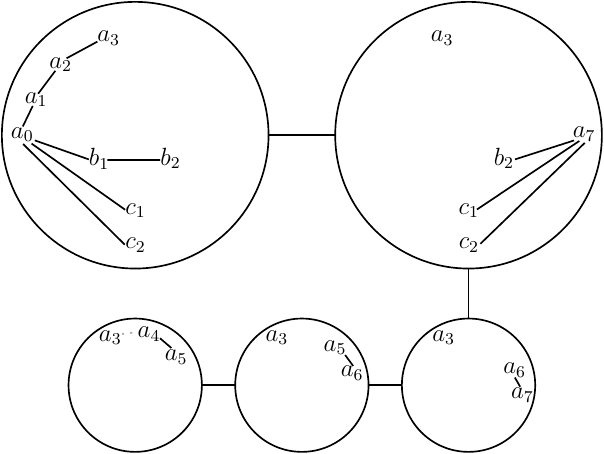}
    \end{center}

    \caption{An example of a tree decomposition of $P(\alpha, \beta, \gamma \times 2)$ for Special Case 5 of \cref{lm:pabc-ub} for $\alpha = 7$, $\beta = 3$, $\gamma = 2$, and $\ell = 3$.}
    \label{Pa3c-upper-bound-split-by-middle-node}
\end{figure}

\bigskip

This concludes the list of high-degree-low-degree algorithms we use. We now turn to the general case.

\subsubsection{General Case: Algorithms Involving ``Hyper-Degree'' Splitting}

Let us recall which assumptions we already have. We have $\jj_{a_0} = \jj_{a_{\alpha}} = 1$, $\jmax \ge 2 \gamma + 2$, and none of the cases above are applicable.

We first build a partial $H$-encodings $\partenc_1$ of $G$ and a partial $H'$-encoding $\partenc'_2$ of $G[H']$, where $H' \coloneqq H - \{a_1, \ldots, a_{\alpha - 1}\}$, such that $\partenc'_2$ has $\bends \coloneqq \{b_0, b_{\beta}\}$ as one of the bags of its underlying tree decomposition.
We will later extend $\partenc'_2$ to a partial $H$-encoding $\partenc_2$ of $G$, such that $\partenc_1$ and $\partenc_2$ together constitute a full $H$-encoding of $G$.
That is, each $H$-subgraph $\bv$ of $G$ is encoded by exactly one of them.

We have two options of how we will accomplish that.

\paragraph*{Option 1}

In the first option in time $\Oh(m^{2 - f})$ we create $\partenc_1$ and $\partenc'_2$, where the submaterialization $S_{\textup{ends}}$ of $\bends$ in $\partenc'_2$ has size $\Oh(m^{2 - kf})$ for $k \coloneqq \jmax - 2 \gamma \ge 2$.

As the degrees of nodes in $\subp{b_0}$ are smaller than $m^f$, in time $\Oh(m^{1 + (\gamma - 1) f})$ we can materialize a bag $B_{b_0 \to c} \coloneqq \{b_0, c_1, c_2, \ldots, c_{\gamma}\}$: there are $\Oh(m)$ choices for an adjacent pair of nodes from $\subp{b_0}$ and $V_{c_1}$, and there are $\Oh(m^f)$ choices to pick a node in $V_{c_i}$ as a neighbor of the chosen node in $\subp{b_0}$, for each $i \in [2, \gamma]$.
By using a binary search tree as a dictionary we can store for every tuple of nodes from parts $V_{c_1},V_{c_2}, \ldots, V_{c_{\gamma}}$ all their common neighbors in $\subp{b_0}$ in time $\tOh(m^{1 + (\gamma - 1) f}) \le \Oh(m^{2 - f})$ (due to \cref{F-gets-big}).
Note that $(v_{c_1}, \ldots, v_{c_{\gamma}})$ is one of the listed tuples as they have a common neighbor $v_{b_0} \in \subp{b_0}$.
Now we split the listed tuples of nodes into the ones that have less than $m^{1 - (\jmax - \gamma - 1) f}$ common neighbors in $\subp{b_0}$ (low-degree tuples $C^{\textup{lo}}$) and the ones that have at least $m^{1 - (\jmax - \gamma - 1) f}$ common neighbors (high-degree tuples $C^{\textup{hi}}$).
Similarly to nodes from parts $V_{a_i}$ and $V_{b_i}$, we deal with low- and high-degree tuples separately.
We distinguish two cases, in one of which we create $\partenc_1$, and in the other one we create $\partenc'_2$.

\subparagraph{Option 1: Case 1: \boldmath$(v_{c_1}, \ldots, v_{c_{\gamma}})$ is a high-degree tuple.}

The total number of copies of $B_{b_0 \to c}$ that we generated is at most $m^{1 + (\gamma - 1) f}$.
Hence, there are $\Oh(m^{1 + (\gamma - 1) f} / m^{1 - (\jmax - \gamma - 1) f}) = \Oh(m^{(\jmax - 2) f})$ high-degree tuples of nodes from parts $V_{c_1}, V_{c_2}, \ldots, V_{c_{\gamma}}$.
Furthermore, due to \cref{high-degree-small-cnt-nodes}, there are at most $m^{1 - (\jmax - 1) f}$ nodes in $\subp{\vmax}$.
Consider the graph $H - \{\vmax, c_1, \ldots, c_{\gamma}\}$.
Since it is a tree, it has a tree decomposition where each bag consists of two adjacent nodes.
Adding $\vmax, c_1, c_2, \ldots, c_{\gamma}$ to every bag, we obtain a tree decomposition of $H$ (see \cref{Pabc-upper-bound-take-out-C}).

We materialize each bag from this tree decomposition as follows.
There are $\Oh(m)$ choices for the two adjacent nodes.
There are $\Oh(m^{(\jmax - 2) f})$ choices for a high-degree tuple from $V_{c_1} \times V_{c_2} \times \cdots \times V_{c_{\gamma}}$.
Finally, there are at most $m^{1 - (\jmax - 1) f}$ choices for a node from $\subp{\vmax}$.
Hence, we can materialize these bags in time $\Oh(m^{1 + (\jmax - 2) f + (1 - (\jmax - 1) f)}) = \Oh(m^{2 - f})$.

Note that this partial $H$-encoding $\partenc_1$ encodes exactly such $H$-subgraphs $\bv$ of $G$, for which $(v_{c_1}, \ldots, v_{c_{\gamma}})$ is a high-degree tuple because for all bags, only high-degree tuples $(v_{c_1}, \ldots, v_{c_{\gamma}})$ are listed.

\begin{figure}
    \begin{center}
        \includegraphics[scale=0.85]{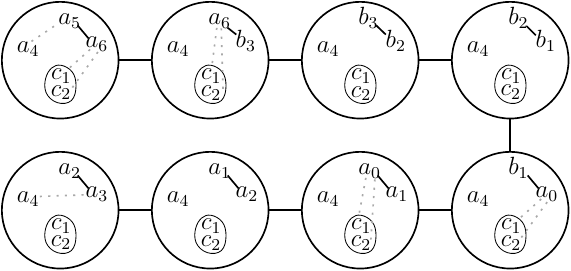}
    \end{center}

    \caption{An example of a tree decomposition of $P(\alpha, \beta, \gamma \times 2)$ for Case 1 of Option 1 of \cref{lm:pabc-ub} for $\alpha = 6$, $\beta = 4$, $\gamma = 2$, and $\vmax = a_4$.}
    \label{Pabc-upper-bound-take-out-C}
\end{figure}

\subparagraph{Option 1: Case 2: \boldmath$(v_{c_1}, \ldots, v_{c_{\gamma}})$ is a low-degree tuple.}

In time $\Oh(m^{1 + (\gamma - 1)f})$ we materialize a bag $B_{b_{\beta} \to c} \coloneqq \{b_{\beta}, c_1, c_2, \ldots, c_{\gamma}\}$ analogously to the materialization of $B_{b_0 \to c}$.
Then for each low-degree tuple of nodes from $V_{c_1}, V_{c_2}, \ldots, V_{c_{\gamma}}$, we spend time $\Oh(m^{1 - (\jmax - \gamma - 1) f})$ to materialize all their common neighbors in $\subp{b_0}$ by accessing the dictionary.
This way we create a submaterialization $S_{\textup{down}}$ of $B_{\textup{down}} \coloneqq B_{b_{\beta} \to c} \cup \{b_0\}$ in time $\Oh(m^{1 + (\gamma - 1) f} \cdot (\log m + m^{1 - (\jmax - \gamma - 1) f})) = \Oh(m^{2 - (\jmax - 2\gamma) f}) = \Oh(m^{2-kf}) \le \Oh(m^{2-2f})$.
Note that for an $H$-subgraph $\bv$, we have that $(v_{c_1}, \ldots, v_{c_{\gamma}})$ is a low-degree tuple if and only if $\prj{\bv}{B_{\textup{down}}} \in S_{\textup{down}}$.
We attach a bag $B_{\textup{ends}} \coloneqq \{b_0, b_{\beta}\}$ to $B_{\textup{down}}$ and create its submaterialization $S_{\textup{ends}}$ as a set of projections of the elements of $S_{\textup{down}}$ onto $B_{\textup{ends}}$.
    This takes time $\tOh(|S_{\textup{down}}|) = \tOh(m^{2 - k f}) \le \tOh(m^{2 - 2f}) \le \Oh(m^{2 - f})$, and we have $|S_{\textup{ends}}| \le |S_{\textup{down}}| = \Oh(m^{2 - k f})$ as promised.
    It remains to attach the decomposition of the path $b_0 \edg b_1 \edg \cdots \edg b_{\beta}$ to $\bends$ to complete the partial $H'$-encoding $\partenc'_2$ of $G[H']$.

We consider $V_{b_0}, \ldots, V_{b_{\beta}}$ as an instance of Biased Cycle for $\ell = k$, where we think of edges of $G$ between $V_{b_{i - 1}}$ and $V_{b_i}$ for $i \in [\beta]$ as the blue edges, and we think of $S_{\textup{ends}}$ as the red edges between $V_{b_0}$ and $V_{b_{\beta}}$.
We apply \cref{lm:biased-cycle-algorithm} to create a full encoding of this Biased Cycle.
If $k > \frac{\beta + 1}{2}$, it works in time $\Oh(m^{2 - f})$, and if $k \le \frac{\beta + 1}{2}$, it works in time
\begin{align*}
    \Oh(m^{2 - f} + m^{1 + (\jj_{b_{k \dd \beta- k}} + k - 1) f}) &\le \Oh(m^{2 - f} + m^{1 + ((\beta - 2 k + 1) \cdot \jmax + k - 1) f})\\
                                                            &= \Oh(m^{2 - f} + m^{1 + ((\beta - 2 k + 1) \cdot (k + 2 \gamma) + k - 1) f})\\
                                                            &\le \Oh(m^{2 - f}),
\end{align*}
where the last inequality holds due to \cref{a-b-c-quadratic-bound-b}.

As there is an edge $\{b_0, b_{\beta}\}$ in this instance of Biased Cycle, we can attach the resulting tree decomposition with a partial encoding to $B_{\textup{ends}}$ by a bag containing both $b_0$ and $b_{\beta}$.
Hence, we create a partial $H'$-encoding $\partenc'_2$ of $G[H']$.
Note that $\partenc'_2$ encodes exactly such $H'$-subgraphs $\bv$ of $G[H']$, for which $(v_{c_1}, \ldots, v_{c_{\gamma}})$ is a low-degree tuple because for $B_{\textup{ends}}$, only low-degree tuples $(v_{c_1}, \ldots, v_{c_{\gamma}})$ are listed.

Combining the two cases together, we get that $\partenc_1$ encodes only such $H$-subgraphs $\bv$ of $G$, for which $\prj{\bv}{\{c_1, \ldots, c_{\gamma}\}}$ is a high-degree tuple, and $\partenc'_2$ encodes only such $H'$-subgraphs $\bv'$ of $G[H']$, for which $\prj{\bv'}{\{c_1, \ldots, c_{\gamma}\}}$ is a low-degree tuple.
Hence, when we extend $\partenc'_2$ to $\partenc_2$, no $H$-subgraph of $G$ can be encoded by both $\partenc_1$ and $\partenc_2$.
And on the other hand, for each $H$-subgraph $\bv$ of $G$, we have that $\prj{\bv}{\{c_1, \ldots, c_{\gamma}\}}$ is either a high-degree tuple or a low-degree tuple.
Hence, either $\bv$ is encoded by $\partenc_1$ or $\prj{\bv}{V(H')}$ is encoded by $\partenc'_2$, and thus we later will encode $\bv$ by $\partenc_2$.

\paragraph*{Option 2}

In the second option in time $\Oh(m^{2 - f}) + \tOh(m^{2 - k' f})$ we create $\partenc_1$ and $\partenc'_2$, where the submaterialization $S_{\textup{ends}}$ of $\bends$ in $\partenc'_2$ has size $\Oh(m^{2 - k'f})$ for $k' \coloneqq \fabc - (2 \gamma + 1 + \jj_{b_1} + \jj_{b_{3 \dd \beta - 1}})$.

We materialize the bag
\[B_1 \coloneqq \{ b_0, b_1, b_2, c_1, c_2, \ldots, c_{\gamma} \}.\]
To this end, note that there are $\Oh(m)$ choices for a pair of adjacent nodes from $\subp{b_0}$ and $\subp{b_1}$.
Then there are $\Oh(m^{\jj_{b_1} f})$ choices to pick a node in $\subp{b_2}$ as a neighbor of the chosen node in $\subp{b_1}$. 
Finally, for each $i \in [\gamma]$, there are $\Oh(m^f)$ choices to pick a node in $V_{c_i}$ as a neighbor of the chosen node in $\subp{b_0}$.
In total, this takes time $\Oh(m^{1 + (\jj_{b_1} + \gamma) f}) \le \Oh(m^{2 - k' f})$.
By using a binary search tree as a dictionary we can store for every tuple of nodes from parts $\subp{b_2}, V_{c_1}, V_{c_2}, \ldots, V_{c_{\gamma}}$ all their common pairs of neighbors in $\subp{b_0}$ and $\subp{b_1}$ in time $\tOh(m^{1 + (\jj_{b_1} + \gamma) f}) \le \Oh(m^{2 - k' f})$.
Note that $(v_{b_2}, v_{c_1}, \ldots, v_{c_{\gamma}})$ is one of the listed tuples as they have a common pair of neighbors $(v_{b_0}, v_{b_1}) \in \subp{b_0} \times \subp{b_1}$.
Now we split the listed tuples of nodes into the ones that have less than $m^{(\gamma + \jj_{b_1} + 1) f}$ common pairs of neighbors in $\subp{b_0}$ and $\subp{b_1}$ (low-degree tuples $D^{\textup{lo}}$) and the ones that have at least $m^{(\gamma + \jj_{b_1} + 1) f}$ common pairs of neighbors (high-degree tuples $D^{\textup{hi}}$).
Similarly to nodes from parts $V_{a_i}$ and $V_{b_i}$, we deal with low- and high-degree tuples separately.
We distinguish two cases, in one of which we create $\partenc_1$, and in the other one we create $\partenc'_2$.

\subparagraph{Option 2: Case 1: \boldmath$(v_{b_2}, v_{c_1}, \ldots, v_{c_{\gamma}})$ is a high-degree tuple.}

The total number of copies of $B_1$ that we generated is $\Oh(m^{1 + (\jj_{b_1} + \gamma) f})$.
Hence, there are $\Oh(m^{1 + (\jj_{b_1} + \gamma) f} / m^{(\gamma + \jj_{b_1} + 1) f}) = \Oh(m^{1 - f})$ high-degree tuples of nodes from parts $\subp{b_2}, V_{c_1}, \ldots, V_{c_{\gamma}}$.
Consider the graph $H - \{b_2, c_1, \ldots, c_{\gamma}\}$.
Since it is a tree, it has a tree decomposition, where each bag consists of two adjacent nodes.
Adding $b_2, c_1, c_2, \ldots, c_{\gamma}$ to every bag, we obtain a tree decomposition of $H$.
For each bag of this tree decomposition, we create its submaterialization in time $\Oh(m^{2 - f})$, as there are $\Oh(m)$ choices for the two adjacent nodes and $\Oh(m^{1 - f})$ choices for a high-degree tuple.

Note that this partial $H$-encoding $\partenc_1$ encodes exactly such $H$-subgraphs $\bv$, for which $(v_{b_2}, v_{c_1}, \ldots, v_{c_{\gamma}})$ is a high-degree tuple as for all bags, only high-degree tuples $(v_{b_2}, v_{c_1}, \ldots, v_{c_{\gamma}})$ are listed.

\subparagraph{Option 2: Case 2: \boldmath$(v_{b_2}, v_{c_1}, \ldots, v_{c_{\gamma}})$ is a low-degree tuple.}

We materialize the bag
\[B_2 \coloneqq \{ b_2, b_3, \ldots, b_{\beta}, c_1, c_2, \ldots, c_{\gamma} \}.\]
To this end, note that there are $\Oh(m)$ choices for a pair of adjacent nodes from $\subp{b_{\beta}}$ and $\subp{b_{\beta - 1}}$.
Then, for each $i \in [\gamma]$, there are $\Oh(m^f)$ choices to pick a node in $V_{c_i}$ as a neighbor of the chosen node in $\subp{b_{\beta}}$.
Finally, for each $i \in \{\beta - 2, \beta - 3, \ldots, 2\}$, there are $\Oh(m^{\jj_{b_{i + 1}}f})$ choices to pick a node in $\subp{b_i}$ as a neighbor of the chosen node in $\subp{b_{i+1}}$.
In total, this takes time $\Oh(m^{1 + (\gamma + \jj_{b_{3 \dd \beta - 1}}) f}) \le \Oh(m^{2 - k' f})$.

If the chosen nodes from parts $V_{b_2}, V_{c_1}, V_{c_2}, \ldots, V_{c_{\gamma}}$ form a low-degree tuple, we spend time $\Oh(m^{(\gamma + \jj_{b_1} + 1) f})$ to materialize all their common pairs of neighbors in $\subp{b_0}$ and $\subp{b_1}$ by accessing the dictionary.
This way we create a submaterialization $S_{\textup{down}}$ of $B_{\textup{down}} \coloneqq B_2 \cup \{b_0, b_1\}$ in time \[\Oh(m^{1 + (\gamma + \jj_{b_{3 \dd \beta - 1}}) f} \cdot (\log m + m^{(\gamma + \jj_{b_1} + 1) f})) = \Oh(m^{1 + (2 \gamma + \jj_{b_1} + \jj_{b_{3 \dd \beta - 1}} + 1) f}) = \Oh(m^{2-k' f}).\]
Note that for an $H$-subgraph $\bv$, we have that $(v_{b_2}, v_{c_1}, \ldots, v_{c_{\gamma}})$ is a low-degree tuple if and only if $\prj{\bv}{B_{\textup{down}}} \in S_{\textup{down}}$.
We attach a bag $B_{\textup{ends}} \coloneqq \{b_0, b_{\beta}\}$ to $B_{\textup{down}}$ and create its submaterialization $S_{\textup{ends}}$ as a set of projections of the elements of $S_{\textup{down}}$ onto $B_{\textup{ends}}$.
    It takes time $\tOh(|S_{\textup{down}}|) = \tOh(m^{2 - k' f})$, and we have $|S_{\textup{ends}}| \le |S_{\textup{down}}| = \Oh(m^{2 - k' f})$ as promised.
    Note that $B_{\textup{down}} \edg \bends$ is a tree decomposition of $H'$, hence we get the partial $H'$-encoding $\partenc'_2$ of $G[H']$.
    Note that $\partenc'_2$ encodes exactly such $H'$-subgraphs $\bv$ of $G[H']$, for which $(v_{b_2}, v_{c_1}, \ldots, v_{c_{\gamma}})$ is a low-degree tuple because for $B_{\textup{ends}}$, only low-degree tuples $(v_{b_2}, v_{c_1}, \ldots, v_{c_{\gamma}})$ are listed.

    Combining the two cases together, we get that $\partenc_1$ encodes only such $H$-subgraphs $\bv$ of $G$, for which $\prj{\bv}{\{b_2, c_1, \ldots, c_{\gamma}\}}$ is a high-degree tuple, and $\partenc'_2$ encodes only such $H'$-subgraphs $\bv'$ of $G[H']$, for which $\prj{\bv'}{\{b_2, c_1, \ldots, c_{\gamma}\}}$ is a low-degree tuple.
Hence, when we extend $\partenc'_2$ to $\partenc_2$, no $H$-subgraph of $G$ can be encoded by both $\partenc_1$ and $\partenc_2$.
And on the other hand, for each $H$-subgraph $\bv$ of $G$, we have that $\prj{\bv}{\{b_2, c_1, \ldots, c_{\gamma}\}}$ is either a high-degree tuple or a low-degree tuple.
Hence, either $\bv$ is encoded by $\partenc_1$ or $\prj{\bv}{V(H')}$ is encoded by $\partenc'_2$, and thus we later will encode $\bv$ by $\partenc_2$.

\bigskip

We now decide on which of the two options we use.
Let $\ell \coloneqq \max \{k, k'\}$, where we recall that $k = \jmax - 2 \gamma \ge 2$ and $k' = \fabc - (2 \gamma + 1 + \jj_{b_1} + \jj_{b_{3 \dd \beta - 1}})$.
We claim that in time $\Oh(m^{2 - f})$ we can create a partial $H$-encoding $\partenc_1$ of $G$ and a partial $H'$-encoding $\partenc'_2$ of $G[H']$, such that $\partenc'_2$ has $\bends \coloneqq \{b_0, b_{\beta}\}$ as one of its bags with submaterialization $S_{\textup{ends}}$ of size $\Oh(m^{2 - \ell f})$. To do that, we use Option 1 if $k \ge k'$ and Option 2 if $k' \ge k$.
Hence, indeed $|S_{\textup{ends}}| = \Oh(m^{2 - \ell f})$.
Furthermore, in Option 1 we work in time $\Oh(m^{2 - f})$, and in the second option we work in time $\Oh(m^{2 - f}) + \tOh(m^{2 - k' f})$.
We use the second option only if $k' \ge k \ge 2$, and hence we have $\Oh(m^{2 - f}) + \tOh(m^{2 - k' f}) = \Oh(m^{2 - f})$.

We have a partial $H'$-encoding $\partenc'_2$ of $G[H']$, and now we need to extend it to a partial $H$-encoding $\partenc_2$ of $G$, such that $\partenc_2$ encodes some $H$-subgraph $\bv$ if and only if $\partenc'_2$ encodes $\prj{\bv}{V(H')}$.

Consider $\subp{a_0}, \ldots, \subp{a_{\alpha}}$ as an instance of Biased Cycle, where we think of edges of $G$ between $\subp{a_{i - 1}}$ and $\subp{a_i}$ for $i \in [\alpha]$ as the blue edges, and we think of $S_{\textup{ends}}$ as the red edges between $\subp{a_0}$ and $\subp{a_{\alpha}}$.
We apply the algorithm from \cref{lm:biased-cycle-algorithm} to create a full encoding of this Biased Cycle. As we prove later, it works in time $\Oh(m^{2 - f})$.
As there is an edge $\{a_0, a_{\alpha}\}$ in this instance of Biased Cycle, we can attach the resulting tree decomposition with a partial encoding to $B_{\textup{ends}}$ by a bag containing both $a_0$ and $a_{\alpha}$.
Hence, we create the desired partial $H$-encoding $\partenc_2$ of $G$.

It remains to show that when we use the algorithm from \cref{lm:biased-cycle-algorithm}, it runs in $\Oh(m^{2 - f})$ time.
To prove it, we consider nine cases.
One can verify that it is a complete set of cases, covering all possible situations.

\subparagraph{Case 1: \boldmath$\alpha \le 2 \ell - 2$.}

In this case $\ell > \frac{\alpha + 1}{2}$, and the algorithm from \cref{lm:biased-cycle-algorithm} works in time $\Oh(m^{2 - f})$.

\subparagraph{Case 2: \boldmath$\alpha = 2 \ell - 1$.}

In this case the algorithm from \cref{lm:biased-cycle-algorithm} works in time
\begin{align*}
    \Oh(m^{2 - f} + m^{1 + (\jj_{a_{\ell \dd \alpha - \ell}} + \ell - 1) f}) &= \Oh(m^{2 - f} + m^{1 + (\ell - 1) f})\\
                                                            &= \Oh(m^{2 - f} + m^{1 + (\alpha - 1) / 2 \cdot f})\\
                                                            &\le \Oh(m^{2 - f} + m^{1 + (\alpha / 2 + \gamma - 1) \cdot f})\\
                                                            &= \Oh(m^{2 - f}),
\end{align*}
where the last equality holds due to \cref{F-gets-big}.

\subparagraph{Case 3: \boldmath$\alpha \ge 2 \ell$ and $\jj_{a_r} \ge 2$ for some $r \in \{0, \ldots, \ell - 1, \alpha - \ell + 1, \ldots, \alpha\}$.}

In this case the instance of Biased Cycle is $\ell$-reachable, and the algorithm from \cref{lm:biased-cycle-algorithm} works in time $\Oh(m^{2 - f})$.

\subparagraph{Case 4: \boldmath$\alpha \ge 2 \ell$, $\jj_{a_r} = 1$ for all $r \in \{0, \ldots, \ell - 1, \alpha - \ell + 1, \ldots, \alpha\}$, $\jj_{a_{\ell}} \ge \gamma + 2$, and $\jj_{b_{\beta - 1}} \ge \gamma + \ell + 1$.}

This cannot happen as Special Case 4 (see \cref{prg:case4}) would be applicable.

\subparagraph{Case 4': \boldmath$\alpha \ge 2 \ell$, $\jj_{a_r} = 1$ for all $r \in \{0, \ldots, \ell - 1, \alpha - \ell + 1, \ldots, \alpha\}$, $\jj_{a_{\alpha - \ell}} \ge \gamma + 2$, and $\jj_{b_1} \ge \gamma + \ell + 1$.}

This cannot happen as Special Case 4' (see \cref{prg:case4p}) would be applicable.

\subparagraph{Case 5: \boldmath$\alpha \ge 2 \ell$ and $\jj_{a_{\ell}} \le \gamma + 1$.}

In this case the algorithm from \cref{lm:biased-cycle-algorithm} runs time
\begin{align*}
    \Oh(m^{2 - f} + m^{1 + (\jj_{a_{\ell \dd \alpha - \ell}} + \ell - 1) f}) &\le \Oh(m^{2 - f} + m^{1 + (\jj_{a_{\ell}} + (\alpha - 2 \ell) \cdot \jmax + \ell - 1) f})\\
                                                            &\le \Oh(m^{2 - f} + m^{1 + (\gamma + 1 + (\alpha - 2 \ell) \cdot (2 \gamma + k) + \ell - 1) f})\\
                                                            &\le \Oh(m^{2 - f} + m^{1 + (\gamma + (\alpha - 2 \ell) \cdot (2 \gamma + \ell) + \ell) f})\\
                                                            &= \Oh(m^{2 - f}).
\end{align*}

The last equality holds because of \cref{a-b-c-quadratic-bound-a-t-2}. To see that the preconditions of this lemma hold, note that $\ell \ge k \ge 2$ and
\begin{align*}
    \ell &\ge k'\\
         &= \fabc - (2 \gamma + 1 + \jj_{b_1} + \jj_{b_{3 \dd \beta - 1}})\\
         &\ge \fabc - (2 \gamma + 1 + (\beta - 2) \cdot \jmax)\\
         &= \fabc - (2 \gamma + 1 + (\beta - 2) \cdot (2 \gamma + k))\\
         &\ge \fabc - (2 \gamma + 1 + (\beta - 2) \cdot (2 \gamma + \ell)).
\end{align*}

\subparagraph{Case 5': \boldmath$\alpha \ge 2 \ell$ and $\jj_{a_{\alpha - \ell}} \le \gamma + 1$.}

This case is symmetric to the previous one.

\subparagraph{Case 6: \boldmath$\alpha \ge 2 \ell$, $\jj_{b_1} \le \gamma + \ell$, $\jj_{b_{\beta - 1}} \le \gamma + \ell$, and $\beta \ge 4$.}

In this case the algorithm from \cref{lm:biased-cycle-algorithm} runs in time
\begin{align*}
    \Oh(m^{2 - f} + m^{1 + (\jj_{a_{\ell \dd \alpha - \ell}} + \ell - 1) f}) &\le \Oh(m^{2 - f} + m^{1 + ((\alpha - 2 \ell + 1) \cdot \jmax + \ell - 1) f})\\
                                                            &= \Oh(m^{2 - f} + m^{1 + ((\alpha - 2 \ell + 1) \cdot (2 \gamma + k) + \ell - 1) f})\\
                                                            &\le \Oh(m^{2 - f} + m^{1 + ((\alpha - 2 \ell + 1) \cdot (2 \gamma + \ell) + \ell - 1) f})\\
                                                            &= \Oh(m^{2 - f}).
\end{align*}

The last equality holds because of \cref{a-b-c-quadratic-bound-a-t-0}. To see that the preconditions of this lemma hold, note that $\ell \ge k \ge 2$ and
\begin{align*}
    \ell &\ge k'\\
         &= \fabc - (2 \gamma + 1 + \jj_{b_1} + \jj_{b_{3 \dd \beta - 1}})\\
         &\ge \fabc - (2 \gamma + 1 + (\gamma + \ell) + (\beta - 4) \cdot \jmax + (\gamma + \ell))\\
         &= \fabc - (2 \gamma + 1 + (\gamma + \ell) + (\beta - 4) \cdot (2 \gamma + k) + (\gamma + \ell))\\
         &\ge \fabc - (2 \gamma + 1 + (\beta - 4) \cdot (2 \gamma + \ell) + 2 \cdot (\gamma + \ell)).
\end{align*}

\subparagraph{Case 7: \boldmath$\alpha \ge 2 \ell$, $\jj_{a_r} = 1$ for all $r \in \{0, \ldots, \ell - 1, \alpha - \ell + 1, \ldots, \alpha\}$, $\jj_{a_{\ell}} \ge \gamma + 2$, $\jj_{b_1} \le \gamma + \ell$, and $\beta = 3$.}

This cannot happen as Special Case 5 (see \cref{prg:case5}) would be applicable.

\bigskip

This concludes the proof of \cref{lm:pabc-ub}.

\subsection{Proof of \cref{lm:p-graph-upper-bound}}
\label{sec:upper-bound-P-graphs-combination}

Now we can put the pieces together to prove \cref{lm:p-graph-upper-bound}.

\pgraphupperbound*

\begin{proof}[Proof of \cref{lm:p-graph-upper-bound}]
    Let $(\alpha, \beta, \gamma) \in \Tfamily$.

    For $(\alpha, \beta, \gamma) = (1, 0, 0)$, \Henciso can be solved in time $\Oh(m) = \Oh(m^{2 - 1 / 1})$ by \cref{edge-ub} and $\fabc = 1$ by \cref{lm:p1-time-complexity}.
    For $(\alpha, \beta, \gamma) = (2, 1, 0)$, \Henciso can be solved in time $\Oh(m^{2 - 1 / 2})$ by \cref{cycle-ub} and $\fabc = 2$ by \cref{lm:c3-time-complexity}.
    For $(\alpha, \beta, \gamma) = (k - 2, 2, 0)$ with $k \ge 4$, \Henciso can be solved in time $\Oh(m^{2 - 1 / \left\lceil k / 2 \right\rceil})$ by \cref{cycle-ub}, and $\fabc = \left\lceil k / 2 \right\rceil$ by \cref{lm:ck-time-complexity}.

    Otherwise, we have $\alpha \ge \beta \ge 2$ and $\gamma \ge 1$.
    If $\alpha = \beta = 2$, then \Henciso can be solved in time $\Oh(m^{2 - 1 / (\gamma + 2)})$ by \cref{biclique-ub}, and $\fabc = \gamma + 2$ by \cref{lm:kk2-time-complexity}.
    If $\alpha > \beta = 2$, then $\fabc = 2 \gamma + 2 + \left\lfloor \frac{\alpha - 1}{2} \right\rfloor$ by \cref{a-c-time-complexity}, and \cref{lm:pa2c-ub} shows an algorithms with the desired time complexity.
    Finally, if $\alpha \ge \beta > 2$, \cref{lm:pabc-ub} shows an algorithm with the desired time complexity.
\end{proof}

\section{Algebraic Lemmas} \label{sec:algebraic-lemmas}

Here we prove inequalities on $\fabc$ that we used in the previous sections.

In this section numbers $\alpha$, $\beta$, $\gamma$, $k$, and $\ell$ are always assumed to be integers.

Let us recall that $\fabc=$

\[ 
            \begin{cases}
                2 \beta \gamma + \frac{\alpha \beta}{2} - \frac{\beta^2}{2} + \frac{\beta}{2} - \frac{\alpha}{2} - 2 \gamma + 2, & \text{if $\alpha + \beta$ even, $\alpha > \beta$, and $\beta < \gamma + 2$;}\\
                2 \beta \gamma + \frac{\alpha \beta}{2} - \frac{\beta^2}{2} + \frac{3\beta}{2} - \frac{\alpha}{2} - 3 \gamma, & \text{if $\alpha + \beta$ even, $3 \beta < \alpha + 6 \gamma + 8$, and ($\alpha = \beta$ or $\beta \ge \gamma + 2$);}\\
                2 \beta \gamma + \frac{\alpha \beta}{2} - \frac{\beta^2}{2} + 3\beta - \alpha - 6 \gamma - 4, & \text{if $\alpha + \beta$ even, $2 \beta \le \alpha + 4 \gamma + 6$, and $3 \beta \ge \alpha + 6 \gamma + 8$;}\\
                2 \beta \gamma + \frac{\alpha \beta}{2} - \frac{\beta^2}{2} + \beta - \frac{\alpha}{2} - 2 \gamma + \frac{3}{2}, & \text{if $\alpha + \beta$ odd and $\beta < 2 \gamma + 3$;}\\
                2 \beta \gamma + \frac{\alpha \beta}{2} - \frac{\beta^2}{2} + 2 \beta - \frac{\alpha}{2} - 4 \gamma - \frac{3}{2}, & \text{if $\alpha + \beta$ odd, $2 \beta \le \alpha + 4 \gamma + 6$, and $\beta \ge 2 \gamma + 3$;}\\
                2 \gamma^2 + \alpha \gamma + \frac{\alpha^2}{8} + \frac{\alpha}{2}, & \text{if $\alpha = 0 \bmod 4$ and $2 \beta > \alpha + 4 \gamma + 6$;}\\
                2 \gamma^2 + \alpha \gamma + \frac{\alpha^2}{8} + \frac{\alpha}{2} + \frac{3}{8}, & \text{if $\alpha$ is odd and $2 \beta > \alpha + 4 \gamma + 6$;}\\
                2 \gamma^2 + \alpha \gamma + \frac{\alpha^2}{8} + \frac{\alpha}{2} + \frac{1}{2}, & \text{if $\alpha = 2 \bmod 4$ and $2 \beta > \alpha + 4 \gamma + 6$.}\\
            \end{cases}
    \]

\begin{lemma} \label{lm:funcc-correctness}
    The function $\funcc$ is well-defined. That is, the eight cases from the definition of $\funcc$ are disjoint and cover all $(\alpha, \beta, \gamma) \in \ZZ^3_{\ge 0}$ with $\alpha \ge \beta$.
\end{lemma}

\begin{proof}
    Fix some $(\alpha, \beta, \gamma) \in \ZZ^3_{\ge 0}$. We prove that there exists exactly one case from the definition of $\funcc$ that $(\alpha, \beta, \gamma)$ satisfies.

    If $2 \beta > \alpha + 4 \gamma + 6$, then the last three cases from the definition of $\funcc$ disjointly cover all triples $(\alpha, \beta, \gamma) \in \ZZ^3_{\ge 0}$ depending on the reminder of $\alpha$ modulo four.
    On the other hand, if $2 \beta > \alpha + 4 \gamma + 6$, none of the first five cases are applicable.
    Case one fails as $2 \beta > \alpha + 4 \gamma + 6 > 2 \gamma + 4$ contradicts $\beta < \gamma + 2$.
    Case two fails as $2 \beta > \alpha + 4 \gamma + 6$ implies $3 \beta = 3 / 2 \cdot (2 \beta) \ge 3 / 2 \alpha + 3 / 2 \cdot (4 \gamma) + 3 / 2 \cdot 6 \ge \alpha + 6 \gamma + 9$, contradicting $3 \beta < \alpha + 6 \gamma + 8$.
    Conditions of case three directly contradict $2 \beta > \alpha + 4 \gamma + 6$.
    Case four fails as $2 \beta > \alpha + 4 \gamma + 6 \ge 4 \gamma + 6$ contradicts $\beta < 2 \gamma + 3$.
    And finally conditions of case five directly contradict $2 \beta > \alpha + 4 \gamma + 6$.

    Therefore, from now on we may assume $2 \beta \le \alpha + 4 \gamma + 6$.
    Depending on the parity of $\alpha + \beta$, either one of the first three cases can be applicable, or one of the fourth and the fifth. We consider these cases separately.

    First, assume that $\alpha + \beta$ is even. If $\alpha > \beta$ and $\beta < \gamma + 2$, then case one is applicable. Otherwise, cases two and three cover all triples $(\alpha, \beta, \gamma) \in \ZZ^3_{\ge 0}$ depending on whether $3 \beta < \alpha + 6 \gamma + 8$ or not.
    Furthermore, cases one, two, and three are disjoint. Indeed, case two requires one of the conditions from case one to fail. Conditions of case three directly contradicts conditions of case two. And if case three holds, we have $3 \beta \ge \alpha + 6 \gamma + 8 \ge 3 \gamma + 6$ contradicting one of the conditions of case one.

    Now assume that $\alpha + \beta$ is odd. In this case assuming $2 \beta \le \alpha + 4 \gamma + 6$, cases four and five disjointly cover all triples $(\alpha, \beta, \gamma) \in \ZZ^3_{\ge 0}$ depending on whether $\beta < 2 \gamma + 3$ holds or not.

    Therefore, indeed for every $(\alpha, \beta, \gamma) \in \ZZ^3_{\ge 0}$, there exists exactly one case from the definition of $\funcc$ that is applicable.
\end{proof}

\begin{lemma} \label{lm:p1-time-complexity}
    For $\alpha = 1, \beta = 0$, and $\gamma = 0$, we have $\funcc(\alpha, \beta, \gamma) = 1$.
\end{lemma}

\begin{proof}
    In this case $\alpha + \beta$ is odd and $\beta < 2\gamma + 3$.
    Hence, the fourth case from the definition of $\funcc$ is applicable.
    We substitute and get $\fabc = 1$.
\end{proof}

\begin{lemma} \label{lm:c3-time-complexity}
    For $\alpha = 2, \beta = 1$, and $\gamma = 0$, we have $\fabc = 2$.
\end{lemma}

\begin{proof}
    In this case $\alpha + \beta$ is odd and $\beta < 2\gamma + 3$.
    Hence, the fourth case from the definition of $\funcc$ is applicable.
    We substitute and get $\fabc = 2$.
\end{proof}

\begin{lemma} \label{lm:funcc-is-positive-integer}
    For any $(\alpha, \beta, \gamma) \in \mathcal{P}'$, we have $\fabc \in \mathbb{N}_{>0}$.
\end{lemma}

\begin{proof}
    If $(\alpha, \beta, \gamma) = (1, 0, 0)$, \cref{lm:p1-time-complexity} implies that $\fabc = 1$.
    If $(\alpha, \beta, \gamma) = (2, 1, 0)$, \cref{lm:c3-time-complexity} implies that $\fabc = 2$.
    Otherwise, for all $(\alpha, \beta, \gamma) \in \mathcal{P}'$, it holds that $\alpha \ge \beta \ge 2$ and $\gamma \ge 0$.
    We consider all the cases from the definition of $\funcc$.
    \begin{enumerate}
        \item \casea In this case $\fabc = 2 \beta \gamma + \frac{\alpha \beta}{2} - \frac{\beta^2}{2} + \frac{\beta}{2} - \frac{\alpha}{2} - 2 \gamma + 2 = 2 (\beta - 1) \gamma + 2 + \frac{\alpha - \beta}{2} \cdot (\beta - 1)$. As $\alpha + \beta$ is required to be even in this case, all the summands are integers. Furthermore, all of them are non-negative, and the summand $2$ is positive. Hence, $\fabc \in \NN_{>0}$.
        \item \caseb In this case $\fabc = 2 \beta \gamma + \frac{\alpha \beta}{2} - \frac{\beta^2}{2} + \frac{3\beta}{2} - \frac{\alpha}{2} - 3 \gamma = (2 \beta - 3) \cdot \gamma + \frac{\alpha - \beta}{2} \cdot (\beta - 1) + \beta$. As $\alpha + \beta$ is required to be even in this case, all the summands are integers. Furthermore, all of them are non-negative, and the summand $\beta$ is positive. Hence, $\fabc \in \NN_{>0}$.
        \item \casec In this case $\fabc = 2 \beta \gamma + \frac{\alpha \beta}{2} - \frac{\beta^2}{2} + 3\beta - \alpha - 6 \gamma - 4 = 2 \beta \gamma + \frac{\alpha - \beta}{2} \cdot \beta + (3 \beta - \alpha - 6 \gamma - 8) + 4$. As $\alpha + \beta$ is required to be even in this case, all the summands are integers. Furthermore, as $3 \beta \ge \alpha + 6 \gamma + 8$ is required in this case, all the summands are non-negative, and the summand $4$ is positive. Hence, $\fabc \in \NN_{>0}$.
        \item \cased In this case $\fabc = 2 \beta \gamma + \frac{\alpha \beta}{2} - \frac{\beta^2}{2} + \beta - \frac{\alpha}{2} - 2 \gamma + \frac{3}{2} = 2 (\beta - 1) \gamma + \frac{\alpha - \beta + 1}{2} \cdot (\beta - 1) + 2$. As $\alpha + \beta$ is required to be odd in this case, all the summands are integers. Furthermore, all of them are non-negative, and the summand $2$ is positive. Hence, $\fabc \in \NN_{>0}$.
        \item \casee In this case the value of $\fabc$ is larger by $\beta - 2 \gamma - 3$ than the one in the fourth case. As $\alpha + \beta$ is odd in both cases, $\fabc$ is also an integer in this case. Furthermore, as $\beta \ge 2 \gamma + 3$ is required in this case, $\beta - 2 \gamma - 3 \ge 0$, and the value of $\fabc$ in this case is also positive.
        \item \casef All the summands are non-negative and $\frac{\alpha}{2}$ is positive. Hence, $\fabc > 0$. Furthermore, as $\alpha$ is divisible by four, all the summands are integers, and thus $\fabc \in \NN_{>0}$.
        \item \caseg All the summands are non-negative and $\frac{\alpha}{2}$ is positive. Hence, $\fabc > 0$. Furthermore, $\fabc = 2 \gamma^2 + \alpha \gamma + \frac{(\alpha + 1) \cdot (\alpha + 3)}{8}$.
            The first two summands are clearly integers. As $\alpha$ is odd in this case, $\alpha + 1$ and $\alpha + 3$ are both divisible by two. Furthermore, as $\alpha + 1$ and $\alpha + 3$ differ by exactly two, one of them is also divisible by four. Therefore, $(\alpha + 3) \cdot (\alpha + 1)$ is divisible by eight, and the last summand is also an integer.
        \item \caseh All the summands are non-negative and $\frac{\alpha}{2}$ is positive. Hence, $\fabc > 0$.
            Furthermore, $\fabc = 2 \gamma^2 + \alpha \gamma + \frac{(\alpha + 2)^2}{8}$. As $\alpha + 2$ is divisible by four in this case, $(\alpha + 2)^2$ is divisible by eight, and thus all the summands are integers.
    \end{enumerate}
\end{proof}

\begin{lemma} \label{lm:funcc-last-cases}
    For $(\alpha, \beta, \gamma) \in \mathcal{P}'$, such that $2 \beta > \alpha + 4 \gamma + 6$, we have $\fabc = \left\lceil 2 \gamma^2 + \alpha \gamma + \frac{\alpha^2}{8} + \frac{\alpha}{2} \right\rceil$.
\end{lemma}

\begin{proof}
    If $2 \beta > \alpha + 4 \gamma + 6$, one of the last three cases from the definition of $\funcc$ is applicable.
    Note that in these cases $0 \le \fabc - (2 \gamma^2 + \alpha \gamma + \frac{\alpha^2}{8} + \frac{\alpha}{2}) < 1$.
    As \cref{lm:funcc-is-positive-integer} implies that $\fabc$ is an integer, we obtain $\fabc = \left\lceil 2 \gamma^2 + \alpha \gamma + \frac{\alpha^2}{8} + \frac{\alpha}{2} \right\rceil$.
\end{proof}

\begin{lemma} 
\label{F-gets-big} 
    For any $(\alpha,\beta,\gamma) \in \Tfamily$ we have $\funcc(\alpha,\beta,\gamma) \ge \tfrac \alpha 2 + \gamma$.
\end{lemma}
\begin{proof}
  Since $\funcc(1,0,0) = 1$ due to \cref{lm:p1-time-complexity} and $\funcc(2,1,0) = 2$ due to \cref{lm:c3-time-complexity}, the claim is obvious for $(1,0,0), (2,1,0) \in \Tfamily$. All remaining $(\alpha,\beta,\gamma) \in \Tfamily$ have $\beta \ge 2$. We use this inequality to simplify $\funcc(\alpha,\beta,\gamma)$:
  
  Case 1: $\funcc(\alpha,\beta,\gamma) \ge 4 \gamma + \tfrac{\alpha \beta}2 - \tfrac{\beta^2}2 + \tfrac \beta 2 - \tfrac \alpha 2 - 2 \gamma + 2 \ge \gamma + \tfrac{\alpha \beta}2 - \tfrac{\beta^2}2 + \tfrac \beta 2 - \tfrac \alpha 2 + 2$. By computing the derivative, we see that this is minimized for $\beta = (\alpha + 1)/2$. Plugging this in, we obtain $\funcc(\alpha,\beta,\gamma) \ge 2 \gamma + (\alpha-1)^2/8 + 2$. Finally, $(\alpha-1)^2/8 + 2 \ge \alpha / 2$ yields the claimed bound.
  
  Case 2: $\funcc(\alpha,\beta,\gamma) \ge 4 \gamma + \tfrac{\alpha \beta}2 - \tfrac{\beta^2}2 + (\tfrac \beta 2 + 2) - \tfrac \alpha 2 - 3 \gamma = \gamma + \tfrac{\alpha \beta}2 - \tfrac{\beta^2}2 + \tfrac \beta 2 - \tfrac \alpha 2 + 2$. From here we can argue as in the previous case.
  
  Case 3: Since in this case $3 \beta \ge \alpha + 6 \gamma + 8 \ge \beta + 8$, we must have $\beta \ge 4$. We use this to bound $\funcc(\alpha,\beta,\gamma) \ge 8 \gamma + \tfrac{\alpha \beta}2 - \tfrac{\beta^2}2 + (\beta + 8) - \alpha - 6 \gamma - 4 \ge \gamma + \tfrac{\alpha \beta}2 - \tfrac{\beta^2}2 + \beta - \alpha + 4$. By computing the derivative, we see that this is minimized for $\beta = \alpha/2 + 1$. Plugging this in, we obtain $\funcc(\alpha,\beta,\gamma) \ge \gamma + (\alpha-2)^2/8 + 4$. Finally, $(\alpha-2)^2/8 + 4 \ge \alpha / 2$ yields the claimed bound.
  
  Case 4: Using $\tfrac \beta 2 \ge \tfrac 1 2$ we see that the function value in this case it at least the function value in case 1, so we can argue as in case 1.
  
  Case 5: Since $\beta \ge 2 \gamma + 3$ we have $\beta \ge 3$. Using this, we get $\funcc(\alpha,\beta,\gamma) \ge 6 \gamma + \tfrac{\alpha \beta}2 - \tfrac{\beta^2}2 + 2 \beta - \tfrac \alpha 2 - 4 \gamma - \tfrac 3 2$. Using $\tfrac 3 2 \beta \ge \tfrac 9 2 \ge \tfrac 3 2 + 2$ we obtain $\funcc(\alpha,\beta,\gamma) \ge \gamma + \tfrac{\alpha \beta}2 - \tfrac{\beta^2}2 + \tfrac \beta 2 - \tfrac \alpha 2 + 2$. From here we can argue as in case 1.
  
  Cases 6-8: Clear since $2 \gamma^2 \ge \gamma$.
\end{proof}

\begin{lemma} \label{lm:alternative-definition-of-funcc}
    Let
    \[f_2(\alpha, \beta, \gamma) \coloneqq 2 \beta \gamma + \frac{\alpha \beta}{2} - \frac{\beta^2}{2} + \frac{3 \beta}{2} - \frac{\alpha}{2} - 3 \gamma,\]
    \[f_4(\alpha, \beta, \gamma) \coloneqq 2 \beta \gamma + \frac{\alpha \beta}{2} - \frac{\beta^2}{2} + \beta - \frac{\alpha}{2} - 2 \gamma + \frac{3}{2},\]
    and
    \[f_5(\alpha, \beta, \gamma) \coloneqq 2 \beta \gamma + \frac{\alpha \beta}{2} - \frac{\beta^2}{2} + 2 \beta - \frac{\alpha}{2} - 4 \gamma - \frac{3}{2}.\]

    For every $(\alpha, \beta, \gamma) \in \family'$, we have $\fabc=$
    \[ 
            \begin{cases}
                f_4(\alpha - 1, \beta, \gamma), & \text{if $\alpha + \beta$ is even, $\alpha > \beta$, and $\beta < \gamma + 2$;}\\
                f_2(\alpha, \beta, \gamma), & \text{if $\alpha + \beta$ is even, $3 \beta < \alpha + 6 \gamma + 8$, and ($\alpha = \beta$ or $\beta \ge \gamma + 2$);}\\
                f_5(\alpha, \beta - 1, \gamma), & \text{if $\alpha + \beta$ is even, $2 \beta \le \alpha + 4 \gamma + 6$, and $3 \beta \ge \alpha + 6 \gamma + 8$;}\\
                f_4(\alpha, \beta, \gamma), & \text{if $\alpha + \beta$ is odd and $\beta < 2 \gamma + 3$;}\\
                f_5(\alpha, \beta, \gamma), & \text{if $\alpha + \beta$ is odd, $2 \beta \le \alpha + 4 \gamma + 6$, and $\beta \ge 2 \gamma + 3$;}\\
                f_5(\alpha, 2 \gamma + \frac{\alpha + ((\alpha - 1) \bmod 4) + 3}{2}, \gamma), & \text{if $2 \beta > \alpha + 4 \gamma + 6$.}\\
            \end{cases}
    \]
    Furthermore, if $2 \beta > \alpha + 4 \gamma + 6$, then $2 \gamma + \frac{\alpha + ((\alpha - 1) \bmod 4) + 3}{2}$ is a positive integer and $2 \gamma + \frac{\alpha + ((\alpha - 1) \bmod 4) + 3}{2} \le \beta - 1$.

    In other words, Case 1 from the definition of $\funcc$ can be defined via Case 4 for a smaller triple $(\alpha - 1, \beta, \gamma)$,
    Case 3 from the definition of $\funcc$ can be defined via Case 5 for a smaller triple $(\alpha, \beta - 1, \gamma)$,
    and Cases 6, 7, and 8 can be defined via Case 5 for a smaller triple $(\alpha, 2 \gamma + \frac{\alpha + ((\alpha - 1) \bmod 4) + 3}{2}, \gamma)$.
\end{lemma}

\begin{proof}
    We consider all the cases from the definition of $\funcc$.
    \begin{enumerate}
        \item \casea As $\alpha > \beta$ in this case, we have $\alpha - 1 \ge \beta$.
            Furthermore, as $\alpha + \beta$ is even, $(\alpha - 1) + \beta$ is odd.
            As $\beta < \gamma + 2$ in this case, we have $\beta < 2 \gamma + 3$.
            Therefore, triple $(\alpha - 1, \beta, \gamma)$ satisfies the conditions of Case 4 from the definition of $\funcc$.
            Finally, we have
            \begin{align*}
                f_4(\alpha - 1, \beta, \gamma) &= 2\beta\gamma + \frac{(\alpha - 1)\beta}{2} - \frac{\beta^2}{2} + \beta - \frac{\alpha - 1}{2} - 2\gamma + \frac{3}{2}\\
                                &= 2\beta\gamma + \frac{\alpha\beta}{2} - \frac{\beta^2}{2} + \frac{\beta}{2} - \frac{\alpha}{2} - 2\gamma + 2\\
                                &= \funcc(\alpha, \beta, \gamma).
            \end{align*}
        \item \caseb This case is trivial.
        \item \casec In this case $3 \beta \ge \alpha + 6 \gamma + 8 \ge \beta + 6 \gamma + 8$, so $\beta - 1 \ge (3 \gamma + 4) - 1 \ge 2 \gamma + 3$.
            Furthermore, as $\alpha + \beta$ is even, $\alpha + (\beta - 1)$ is odd.
            Moreover, $2 (\beta - 1) < 2 \beta \le \alpha + 4 \gamma + 6$ is required in this case.
            Therefore, triple $(\alpha, \beta - 1, \gamma)$ satisfies the conditions of Case 5 from the definition of $\funcc$.
            Finally, we have
            \begin{align*}
                f_5(\alpha, \beta - 1, \gamma) &= 2 (\beta - 1) \gamma + \frac{\alpha (\beta - 1)}{2} - \frac{(\beta - 1)^2}{2} + 2 (\beta - 1) - \frac{\alpha}{2} - 4 \gamma - \frac{3}{2}\\
                                    &= 2 \beta \gamma + \frac{\alpha \beta}{2} - \frac{\beta^2}{2} + 3 \beta - \alpha - 6 \gamma - 4\\
                                    &= \funcc(\alpha, \beta, \gamma).
            \end{align*}
        \item \cased This case is trivial.
        \item \casee This case is trivial.
        \item \casef Here we consider all cases $6$ through $8$ together.
            Define $x \coloneqq (((\alpha - 1) \bmod 4) + 3) / 2$.
            Denote $\beta' \coloneqq \frac{\alpha}{2} + 2\gamma + x$.
            It is easy to see that $\beta'$ is a positive integer.
            Furthermore, $2 \alpha + 2 \beta' = 2 \alpha + \alpha + 4 \gamma + (\alpha - 1) + 3 = 2 \bmod 4$, hence $\alpha + \beta'$ is odd.
            By the case condition we have $\alpha + \beta \ge 2 \beta > \alpha + 4 \gamma + 6$, therefore $\alpha \ge \beta \ge 4 \gamma + 6 \ge 6$.
            Thus, we have $\beta' = \frac{\alpha}{2} + 2 \gamma + x \ge 2 \gamma + 3$.
            Furthermore, $\beta' < \beta \le \alpha$ because $x \le 3$ and $2 \beta > \alpha + 4 \gamma + 6$ by the case assumption.
            Finally, $2 \beta' \le \alpha + 4 \gamma + 6$ as $x \le 3$.
            Therefore, triple $(\alpha, \beta', \gamma)$ satisfies the conditions of Case 5 from the definition of $\funcc$.
            We obtain
            \begin{align*}
                f_5(\alpha, \beta', \gamma) &= 2 \beta' \gamma + \frac{\alpha \beta'}{2} - \frac{\beta'^2}{2} + 2 \beta' - \frac{\alpha}{2} - 4 \gamma - \frac{3}{2}\\
                                        &= 2 (\frac{\alpha}{2} + 2 \gamma + x) \gamma + \frac{\alpha (\frac{\alpha}{2} + 2 \gamma + x)}{2} - \frac{(\frac{\alpha}{2} + 2 \gamma + x)^2}{2} \\&+ 2 (\frac{\alpha}{2} + 2 \gamma + x) - \frac{\alpha}{2} - 4 \gamma - \frac{3}{2}\\
                                        &= 2 \gamma^2 + \alpha \gamma + \frac{\alpha^2}{8} + \frac{\alpha}{2} + (2x - \frac{x^2}{2} - \frac{3}{2})\\
                                        &=\fabc,
            \end{align*}
            where the last equality can be checked for each value of $\alpha \bmod 4$.
            \begin{enumerate}
                \item $\alpha = 0 \bmod 4$. In this case $x = (3 + 3) / 2 = 3$. Hence, $2x - \frac{x^2}{2} - \frac{3}{2} = 0$.
                \item $\alpha = 1 \bmod 4$. In this case $x = (0 + 3) / 2 = 3 / 2$. Hence, $2x - \frac{x^2}{2} - \frac{3}{2} = 3 / 8$.
                \item $\alpha = 2 \bmod 4$. In this case $x = (1 + 3) / 2 = 2$. Hence, $2x - \frac{x^2}{2} - \frac{3}{2} = 1 / 2$.
                \item $\alpha = 3 \bmod 4$. In this case $x = (2 + 3) / 2 = 5 / 2$. Hence, $2x - \frac{x^2}{2} - \frac{3}{2} = 3 / 8$.
            \end{enumerate}
        \item \caseg See case $6$.
        \item \caseh See case $6$.
    \end{enumerate}
\end{proof}

\begin{lemma} \label{lm:ck-time-complexity}
    For $\alpha \ge 2, \beta = 2$, and $\gamma = 0$, we have $\fabc = \left\lceil (\alpha + \beta) / 2 \right\rceil $.
\end{lemma}

\begin{proof}
    If $\alpha$ is even, we have that $\alpha + \beta$ is even, $3 \beta < \alpha + 6 \gamma + 8$, and $\beta \ge \gamma + 2$.
    Hence, the second case from the definition of $\funcc$ is applicable.
    We substitute and get $\fabc = \alpha / 2 + 1 = \left\lceil (\alpha + \beta) / 2 \right\rceil$.

    If, on the other hand, $\alpha$ is odd, we have that $\alpha + \beta$ is odd and $\beta < 2\gamma + 3$.
    Hence, the fourth case from the definition of $\funcc$ is applicable.
    We substitute and get $\fabc = \alpha / 2 + 3 / 2 = \left\lceil (\alpha + \beta) / 2 \right\rceil$.
\end{proof}

\begin{lemma} \label{lm:kk2-time-complexity}
    For $\alpha = 2, \beta = 2$, and $\gamma \ge 0$, we have $\fabc = \gamma + 2$.
\end{lemma}

\begin{proof}
    In this case $\alpha + \beta$ is even, $3 \beta < \alpha + 6 \gamma + 8$, and $\alpha = \beta$.
    Hence, the second case from the definition of $\funcc$ is applicable.
    We substitute and get $\fabc = \gamma + 2$.
\end{proof}

\begin{lemma} \label{a-c-time-complexity}
For $\alpha \ge 3$, $\beta = 2$, and $\gamma \ge 1$, we have $\funcc(\alpha, \beta, \gamma) = 2 \gamma + 2 + \left\lfloor \frac{\alpha - 1}{2} \right\rfloor$.
\end{lemma}

\begin{proof}
    Consider two cases.
    \begin{enumerate}
        \item \emph{$\alpha$ is even.} Then $\alpha + \beta$ is even, $\alpha > \beta$, and $\beta < \gamma + 2$. So the first case from the definition of $\funcc$ applies.
            We substitute and get $\fabc = 4\gamma + \alpha - 2 + 1 - \frac{\alpha}{2} - 2\gamma + 2 = 2\gamma + \frac{\alpha}{2} + 1 = 2 \gamma + 2 + \left\lfloor \frac{\alpha - 1}{2} \right\rfloor$.
        \item \emph{$\alpha$ is odd.} Then $\alpha + \beta$ is odd and $\beta < 2\gamma + 3$. So the fourth case from the definition of $\funcc$ applies.
            We substitute and get $\fabc = 4\gamma + \alpha - 2 + 2 - \frac{\alpha}{2} - 2\gamma + \frac{3}{2} = 2 \gamma + 2 + \left\lfloor \frac{\alpha - 1}{2} \right\rfloor$.
    \end{enumerate}
\end{proof}

\begin{lemma} \label{a-a-c-linear-bound}
    For $\alpha = \beta \ge 3$ and $\gamma \ge 1$ we have $\funcc(\alpha, \beta, \gamma) \ge (2\alpha - 3) \cdot \gamma + \alpha$.
\end{lemma}

\begin{proof}
    We need to prove $\funcc(\alpha, \alpha, \gamma) \ge 2\alpha\gamma + \alpha - 3\gamma$. We consider all the cases from the definition of $\funcc$.
    \begin{enumerate}
        \item \casea This case requires $\alpha > \beta$ which contradicts the fact that $\alpha = \beta$.
        \item \caseb We have $\funcc(\alpha, \beta, \gamma) = 2\beta\gamma + \frac{\alpha\beta}{2} - \frac{\beta^2}{2} + \frac{3\beta}{2} - \frac{\alpha}{2} - 3\gamma = 2\alpha\gamma + \alpha - 3\gamma$.
        \item \casec In this case $3\beta \ge \alpha + 6\gamma + 8$ is required which implies $\alpha \ge 3\gamma + 4$ because $\alpha=\beta$. Thus, $\funcc(\alpha, \beta, \gamma) = 2\beta\gamma + \frac{\alpha\beta}{2} - \frac{\beta^2}{2} + 3\beta - \alpha - 6\gamma - 4 = 2\alpha\gamma + 2\alpha - 6\gamma - 4 \ge 2\alpha\gamma + \alpha + (3\gamma + 4) - 6\gamma - 4 = 2\alpha\gamma + \alpha - 3\gamma$.
        \item \cased This case requires $\alpha + \beta$ being odd which contradicts the fact that $\alpha=\beta$.
        \item \casee This case requires $\alpha + \beta$ being odd which contradicts the fact that $\alpha=\beta$.
        \item \casef In this case $2\beta > \alpha + 4\gamma + 6$ is required which implies $\alpha > 4\gamma + 6$ because $\alpha=\beta$.
            Thus, $\frac{\alpha}{2} \ge 2 \gamma + 2$.
            Therefore, we get $\fabc = 2\gamma^2 + \alpha\gamma + \frac{\alpha^2}{8} + \frac{\alpha}{2} = \frac{1}{2} \cdot (\frac{\alpha}{2} - 2\gamma) \cdot (\frac{\alpha}{2} - 2\gamma - 2) + \gamma + 2\alpha\gamma + \alpha - 3\gamma \ge 2\alpha\gamma + \alpha - 3\gamma$.
        \item \caseg In this case the value of $\funcc$ is larger by $\frac{3}{8}$ than the one in the sixth case while $2\beta > \alpha + 4\gamma + 6$ still holds.
            Thus, the inequality follows from the sixth case.
        \item \caseh In this case the value of $\funcc$ is larger by $\frac{1}{2}$ than the one in the sixth case while $2\beta > \alpha + 4\gamma + 6$ still holds.
            Thus, the inequality follows from the sixth case.
    \end{enumerate}
\end{proof}

\begin{lemma} \label{a-b-c-linear-bound}
    For $\alpha > \beta \ge 3$ and $\gamma \ge 1$ we have $\funcc(\alpha, \beta, \gamma) \ge (\beta - 1) \cdot (2\gamma + 1) + 2$.
\end{lemma}

\begin{proof}
    We want to prove
    \begin{align} 
        \funcc(\alpha, \beta, \gamma) \stackrel{?}{\ge} (\beta - 1) \cdot (2\gamma+1) + 2 = 2\beta\gamma+\beta - 2\gamma + 1. \tag{\textasteriskcentered}
    \end{align}
    Consider all the cases from the definition of $\funcc$.
    \begin{enumerate}
        \item \casea We have $\funcc(\alpha, \beta, \gamma) = 2\beta\gamma+\frac{\alpha\beta}{2} - \frac{\beta^2}{2} + \frac{\beta}{2} - \frac{\alpha}{2} - 2\gamma + 2$.
            In terms of $\alpha$ this is a linear function with slope $\frac{\beta-1}{2} > 0$.
            As $\alpha > \beta$ and $\alpha + \beta$ is even in this case, we obtain $\alpha \ge \beta + 2$.
            Therefore, $\fabc \ge \funcc(\beta + 2, \beta, \gamma) = 2\beta\gamma + \frac{(\beta + 2) \cdot \beta}{2} - \frac{\beta^2}{2} + \frac{\beta}{2} - \frac{\beta + 2}{2} - 2\gamma + 2 = 2\beta\gamma + \beta - 2\gamma + 1$.
        \item \caseb This case requires $\alpha = \beta$ or $\beta \ge \gamma + 2$. Since $\alpha > \beta$, we have $\beta \ge \gamma + 2$. The value of $\funcc$ in this case is different from the one in the previous case by $\beta - \gamma - 2 \ge 0$. As the only condition that we used in the previous case to prove \inq was the fact that $\alpha+\beta$ is even, the same arguments apply here.
        \item \casec The value of $\funcc$ in this case is different from the one in the first case by $\frac{5}{2}\beta - \frac{\alpha}{2} - 4\gamma - 6 = (3\beta - \alpha - 6\gamma - 8) + \frac{1}{2} \cdot (\alpha - \beta) + (2\gamma + 2)$, which is nonnegative due to the fact that $3\beta \ge \alpha + 6\gamma + 8$ in this case.
            As the only condition that we used in the first case to prove \inq was the fact that $\alpha+\beta$ is even, the same arguments apply here.
        \item \cased
            In terms of $\alpha$, the value of $\funcc(\alpha, \beta, \gamma)$ in this case is a linear function with slope $\frac{\beta-1}{2} > 0$.
            Since $\alpha \ge \beta + 1$, we obtain $\fabc \ge \funcc(\beta + 1, \beta, \gamma) =  2\beta\gamma + \frac{(\beta + 1) \cdot \beta}{2} - \frac{\beta^2}{2} + \beta - \frac{\beta + 1}{2} - 2\gamma + \frac{3}{2} = 2\beta\gamma + \beta - 2\gamma + 1$.
        \item \casee The value of $\funcc$ in this case is different from the one in the previous case by $\beta - 2\gamma - 3$, which is nonnegative due to the fact that $\beta \ge 2\gamma+3$ in this case. Therefore, the same arguments as in the previous case apply.
        \item \casef We have $\funcc(\alpha, \beta, \gamma) = 2\gamma^2 + \alpha\gamma + \frac{\alpha^2}{8} + \frac{\alpha}{2}$. This is monotonically increasing in $\alpha$ (for positive $\alpha$).
            Furthermore, in this case $\alpha + 4\gamma + 6 < 2\beta < \alpha + \beta$ is required, which yields $\beta - 4\gamma > 6$.
            Since $\alpha \ge \beta + 1$, we obtain $\fabc \ge \funcc(\beta + 1, \beta, \gamma) = 2\gamma^2 + (\beta+1)\gamma + \frac{(\beta+1)^2}{8} + \frac{\beta+1}{2} = \frac{(\beta-4\gamma) (\beta-4\gamma-2)}{8} + (2\gamma - \frac{3}{8}) + 2\beta\gamma+\beta-2\gamma+1 \ge 2\beta\gamma+\beta-2\gamma+1$.
        \item \caseg In this case the value of $\funcc$ is larger by $\frac{3}{8}$ than the one in the sixth case while $2\beta > \alpha + 4\gamma + 6$ still holds.
            Thus, the inequality follows from the sixth case.
        \item \caseh In this case the value of $\funcc$ is larger by $\frac{1}{2}$ than the one in the sixth case while $2\beta > \alpha + 4\gamma + 6$ still holds.
            Thus, the inequality follows from the sixth case.
    \end{enumerate}
\end{proof}

\begin{lemma} \label{a-b-c-linear-bound-2}
    For $\alpha > \beta \ge 3$ and $\gamma \ge 1$, we have $\funcc(\alpha, \beta, \gamma) \ge (\alpha - 2\ell + 1) \cdot (2\gamma + 1) + \ell$ where $\ell \coloneqq \funcc(\alpha, \beta, \gamma) - (\beta - 1) \cdot (2\gamma + 1)$.
\end{lemma}

\begin{proof}
    We need to prove
    \begin{align*}
        \fabc &\stackrel{?}{\ge} (\alpha - 2\ell + 1) \cdot (2 \gamma + 1) + \ell\\
              &= (\alpha + 1) \cdot (2 \gamma + 1) - (4 \gamma + 1) \ell\\
              &= (\alpha + 1) \cdot (2 \gamma + 1) - (4 \gamma + 1) \cdot (\funcc(\alpha, \beta, \gamma) - (\beta - 1) \cdot (2\gamma + 1))\\
              &= (\alpha + 1) \cdot (2 \gamma + 1) + (4 \gamma + 1) \cdot (\beta - 1) \cdot (2 \gamma + 1) - (4 \gamma + 1) \cdot \fabc.
    \end{align*}

    By rearranging, we simplify our proof goal to
    \begin{align*}
        \fabc &\stackrel{?}{\ge} \frac{1}{4 \gamma + 2} \cdot ((\alpha + 1) \cdot (2 \gamma + 1) + (4 \gamma + 1) \cdot (\beta - 1) \cdot (2 \gamma + 1))\\
              &= \frac{1}{2} (\alpha + 1 + (4 \gamma + 1) \cdot (\beta - 1))\\
              &= \frac{\alpha + \beta}{2} + 2 \beta \gamma - 2 \gamma. \tag{\textasteriskcentered}
    \end{align*}

    We consider all the cases from the definition of $\funcc$.
    \begin{enumerate}
        \item \casea As $\alpha > \beta$ by the lemma statement, and $\alpha + \beta$ is required to be even in this case, we get $\alpha \ge \beta + 2$.
            Therefore, we obtain
            \begin{align*}
                \fabc &= 2 \beta \gamma + \frac{\alpha \beta}{2} - \frac{\beta^2}{2} + \frac{\beta}{2} - \frac{\alpha}{2}-2\gamma+2\\
                      &= 2 \beta \gamma + \frac{\alpha}{2} + \frac{\alpha (\beta - 2)}{2} - \frac{\beta^2}{2} + \frac{\beta}{2} - 2\gamma + 2\\
                      &\ge 2 \beta \gamma + \frac{\alpha}{2} + \frac{(\beta + 2) \cdot (\beta - 2)}{2} - \frac{\beta^2}{2} + \frac{\beta}{2} - 2\gamma + 2\\
                      &= \frac{\alpha + \beta}{2} + 2 \beta \gamma - 2 \gamma,
            \end{align*}
            thus proving \inq.
        
        \item \caseb The value of $\funcc$ in this case is different from the one in the previous case by $\beta - \gamma - 2$, which is nonnegative as in this case either $\alpha = \beta$ or $\beta \ge \gamma + 2$ is required, and by lemma statement $\alpha > \beta$.
            As the only condition that we used in the previous case to prove \inq was the fact that $\alpha+\beta$ is even, the same arguments apply here.
        \item \casec The value of $\funcc$ in this case is different from the one in the first case by $\frac{5}{2}\beta - \frac{\alpha}{2} - 4\gamma - 6 = (3\beta - \alpha - 6\gamma - 8) + \frac{\alpha - \beta}{2} + 2\gamma + 2$, which is nonnegative as $3\beta \ge \alpha + 6\gamma + 8$ is required in this case.
            As the only condition that we used in the first case to prove \inq was the fact that $\alpha+\beta$ is even, the same arguments apply here.
        \item \cased As $\alpha > \beta$ by the lemma statement, we get $\alpha \ge \beta + 1$.
            Therefore, we obtain
            \begin{align*}
                \fabc &= 2 \beta \gamma + \frac{\alpha \beta}{2} - \frac{\beta^2}{2} + \beta - \frac{\alpha}{2}-2\gamma+\frac{3}{2}\\
                      &= 2 \beta \gamma + \frac{\alpha}{2} + \frac{\alpha (\beta - 2)}{2} - \frac{\beta^2}{2} + \beta - 2\gamma + \frac{3}{2}\\
                      &\ge 2 \beta \gamma + \frac{\alpha}{2} + \frac{(\beta + 1) \cdot (\beta - 2)}{2} - \frac{\beta^2}{2} + \beta - 2\gamma + \frac{3}{2}\\
                      &= \frac{\alpha + \beta}{2} + 2 \beta \gamma - 2 \gamma + \frac{1}{2}\\
                      &\ge \frac{\alpha + \beta}{2} + 2 \beta \gamma - 2 \gamma,
            \end{align*}
            thus proving \inq.
        \item \casee The value of $\funcc$ in this case is different from the one in the previous case by $\beta - 2\gamma - 3$, which is nonnegative due to the fact that $\beta \ge 2\gamma+3$ is required in this case.
            Therefore, the same arguments as in the previous case apply.
        \item \casef As $\alpha + 4 \gamma + 6 < 2 \beta \le \alpha + \beta$ is required in this case, we get $\beta - 4 \gamma - 2 \ge 0$.
            Furthermore, as $\alpha > \beta$ by the lemma statement, we get $\alpha \ge \beta + 1$.
            Combining these facts together, we obtain
            \begin{align*}
                \fabc &= 2 \gamma^2 + \alpha \gamma + \frac{\alpha^2}{8} + \frac{\alpha}{2}\\
                      &\ge 2 \gamma^2 + (\beta + 1) \cdot \gamma + \frac{(\beta + 1)^2}{8} + \frac{\alpha}{2}\\
                      &= \frac{\alpha + \beta}{2} + 2 \beta \gamma - 2 \gamma + (2 \gamma^2 - \beta \gamma + 3 \gamma - \frac{\beta}{2} + \frac{(\beta + 1)^2}{8})\\
                      &= \frac{\alpha + \beta}{2} + 2 \beta \gamma - 2 \gamma + (\frac{(\beta - 4 \gamma) \cdot (\beta - 4 \gamma - 2)}{8} + 2 \gamma + \frac{1}{8})\\
                      &\ge \frac{\alpha + \beta}{2} + 2 \beta \gamma - 2 \gamma,
            \end{align*}
            thus proving \inq.
        \item \caseg In this case the value of $\funcc$ is different by $\frac{3}{8}$ than the one in the sixth case, and the same arguments apply.
        \item \caseh In this case the value of $\funcc$ is different by $\frac{1}{2}$ than the one in the sixth case, and the same arguments apply.
    \end{enumerate}
\end{proof}

\begin{lemma} \label{lm:a3c-linear-bound}
    For $\alpha \ge \beta = 3$ and $\gamma \ge 1$, we have $\fabc \ge 2 \gamma + \alpha + 1$.
\end{lemma}

\begin{proof}
    We consider four cases.
    \begin{enumerate}
        \item \emph{$\alpha = \beta$.}
            In this case $\alpha + \beta$ is even, $3 \beta < \alpha + 6 \gamma + 8$, and $\alpha = \beta$, thus the second case from the definition of $\funcc$ is applicable.
            We obtain $\funcc(\alpha, \beta, \gamma) = 2\beta\gamma + \frac{\alpha\beta}{2} - \frac{\beta^2}{2} + \frac{3\beta}{2} - \frac{\alpha}{2} - 3\gamma = 3 \gamma + 3 \ge 2 \gamma + 4 = 2 \gamma + \alpha + 1$.
        \item \emph{$\alpha \neq \beta$ and $\alpha$ is even.}
            In this case $\alpha + \beta$ is odd and $\beta < 2\gamma + 3$.
            Thus, the fourth case from the definition of $\funcc$ is applicable.
            We obtain $\funcc(\alpha, \beta, \gamma) = 2\beta\gamma + \frac{\alpha\beta}{2} - \frac{\beta^2}{2} + \beta - \frac{\alpha}{2} - 2\gamma + \frac{3}{2} = 4\gamma + \alpha \ge 2 \gamma + \alpha + 1$.
        \item \emph{$\alpha \neq \beta$, $\alpha$ is odd, and $\gamma = 1$.}
            In this case $\alpha + \beta$ is even, $3\beta < \alpha + 6\gamma + 8$, and $\beta \ge \gamma + 2$.
            Thus, the second case from the definition of $\funcc$ is applicable.
            We obtain $\funcc(\alpha, \beta, \gamma) = 2\beta\gamma + \frac{\alpha\beta}{2} - \frac{\beta^2}{2} + \frac{3\beta}{2} - \frac{\alpha}{2} - 3\gamma = \alpha + 3 = 2 \gamma + \alpha + 1$.
        \item \emph{$\alpha \neq \beta$, $\alpha$ is odd, and $\gamma \ge 2$.}
            In this case $\alpha + \beta$ is even, $\alpha > \beta$, and $\beta < \gamma + 2$.
            Thus, the first case from the definition of $\funcc$ is applicable.
            We obtain $\funcc(\alpha, \beta, \gamma) = 2\beta\gamma + \frac{\alpha\beta}{2} - \frac{\beta^2}{2} + \frac{\beta}{2} - \frac{\alpha}{2} - 2\gamma + 2 = 4\gamma + \alpha - 1 \ge 2 \gamma + \alpha + 1$.
    \end{enumerate}
\end{proof}

\begin{lemma} \label{a-b-c-quadratic-bound-b}
    For $\alpha \ge \beta \ge 3$, $\gamma \ge 1$, and $k \ge 2$ we have $\funcc(\alpha, \beta, \gamma) \ge (\beta - 2k + 1) \cdot (k + 2\gamma) + k$.
\end{lemma}

\begin{proof}
    We want to prove
    \begin{align} 
        \funcc(\alpha, \beta, \gamma) \stackrel{?}{\ge} (\beta - 2k + 1) \cdot (k + 2\gamma) + k = -2k^2 + (\beta - 4\gamma + 2) k + (2\beta\gamma + 2\gamma). \tag{\textasteriskcentered}
    \end{align}
    The right-hand side of \inq is a quadratic function $k$ with a negative coefficient for $k^2$.
    It achieves the maximum value for $k = k_{\max} \coloneqq \frac{\beta - 4 \gamma + 2}{4}$.
    As $k_{\max}$ may be fractional, the maximum value over integer values of $k$ is achieved for the closest integer to $k_{\max}$.
    That is, $k'_{\max} \coloneqq \left\lfloor \frac{\beta - 4 \gamma + 3}{4} \right\rfloor = \left\lceil \frac{\beta}{4} \right\rceil - \gamma$.
    Furthermore, as we have $k \ge 2$, the maximum over possible values of $k$ is achieved for $k = k_0 \coloneqq \max\{k'_{\max}, 2\}$.
    Hence, $-2k_0^2 + (\beta - 4\gamma + 2) k_0 + (2\beta\gamma + 2\gamma) \ge -2k^2 + (\beta - 4\gamma + 2) k + (2\beta\gamma + 2\gamma)$.
    Therefore, it suffices to prove \inq for $k = k_0 = \max\{\left\lceil \frac{\beta}{4} \right\rceil - \gamma, 2\}$.
    We consider three cases.
    \begin{itemize}
        \item \emph{Case 1: $\beta \le 4\gamma+5$ and $\alpha = \beta$.} In this case $k = k_0 = 2$.
            Hence, the right-hand side of \inq is $2\beta\gamma + 2\beta - 6\gamma - 4$. Consider all the cases from the definition of $\funcc$.
        \begin{enumerate}
            \item \casea This case requires $\alpha > \beta$ which contradicts the fact that $\alpha=\beta$.
            \item \caseb This case requires $3\beta < \alpha + 6\gamma+8 = \beta + 6\gamma + 8$.
                Thus, $\beta < 3\gamma+4$.
                Therefore, we obtain $\funcc(\alpha, \beta, \gamma) = 2\beta\gamma + \frac{\alpha\beta}{2} - \frac{\beta^2}{2} + \frac{3\beta}{2} - \frac{\alpha}{2} - 3\gamma = 2\beta\gamma + \beta - 3\gamma \ge 2\beta\gamma+2\beta-6\gamma-4$.
            \item \casec We have $\funcc(\alpha, \beta, \gamma) = 2\beta\gamma + \frac{\alpha\beta}{2} - \frac{\beta^2}{2} + 3\beta - \alpha - 6\gamma - 4 = 2\beta\gamma+2\beta -6\gamma-4$.
            \item \cased This case requires $\alpha + \beta$ being odd which contradicts $\alpha=\beta$.
            \item \casee This case requires $\alpha + \beta$ being odd which contradicts $\alpha=\beta$.
            \item \casef This case requires $2\beta  > \alpha + 4\gamma + 6 = \beta + 4\gamma + 6$ which implies $\beta > 4\gamma+6$. We arrive at a contradiction with the case assumption $\beta \le 4\gamma+5$.
            \item \caseg This case requires $2\beta  > \alpha + 4\gamma + 6 = \beta + 4\gamma + 6$ which implies $\beta > 4\gamma+6$. We arrive at a contradiction with the case assumption $\beta \le 4\gamma+5$.
            \item \caseh This case requires $2\beta  > \alpha + 4\gamma + 6 = \beta + 4\gamma + 6$ which implies $\beta > 4\gamma+6$. We arrive at a contradiction with the case assumption $\beta \le 4\gamma+5$.
        \end{enumerate}
        \item \emph{Case 2: $\beta \le 4\gamma+5$ and $\alpha > \beta$.}
            In this case $k = k_0 = 2$.
            Hence, the right-hand side of \inq is $2\beta\gamma + 2\beta - 6\gamma - 4$.
            Now we use the case assumption $\beta \le 4 \gamma + 5$ and \cref{a-b-c-linear-bound} to obtain $\fabc \ge (\beta-1) \cdot (2\gamma+1) + 2 = 2\beta\gamma+\beta-2\gamma+1 \ge 2\beta\gamma+2\beta -6\gamma-4$.
        \item \emph{Case 3: $\beta \ge 4\gamma+6$.}
            In this case $k = k_0 = \left\lceil \frac{\beta}{4} \right\rceil - \gamma$.
            Let $x \in \{0, 1, 2, 3\}$ be such that $\left\lceil \frac{\beta}{4} \right\rceil = \frac{\beta+x}{4}$.
            We substitute $k$ in the right-hand side of \inq:
            \begin{align*}
                \fabc & \stackrel{?}{\ge} -2(\left\lceil \frac{\beta}{4} \right\rceil - \gamma)^2 + (\beta-4\gamma+2)(\left\lceil \frac{\beta}{4} \right\rceil - \gamma)+(2\beta\gamma+2\gamma) \\
                      & = -2\left\lceil \frac{\beta}{4} \right\rceil^2 + 2\gamma^2 + \beta \cdot \left\lceil \frac{\beta}{4} \right\rceil + \beta\gamma + 2 \left\lceil \frac{\beta}{4} \right\rceil \\
                      & = -2(\frac{\beta+x}{4})^2 + 2\gamma^2 + \beta \cdot \frac{\beta+x}{4} + \beta\gamma + 2 \frac{\beta+x}{4} \\
                      & = \frac{\beta^2}{8} + \beta\gamma + 2\gamma^2 + \frac{\beta}{2} + \frac{x}{2} - \frac{x^2}{8}. \tag{\textasteriskcentered \textasteriskcentered}
            \end{align*}
            We consider all the cases from the definition of $\funcc$.
        \begin{enumerate}
            \item \casea In this case $\beta < \gamma+2$ is required which contradicts the case assumption $\beta \ge 4\gamma+6$.
            \item \caseb $\funcc(\alpha, \beta, \gamma) = 2\beta\gamma+\frac{\alpha\beta}{2} - \frac{\beta^2}{2} + \frac{3\beta}{2} - \frac{\alpha}{2} - 3\gamma$. In terms of $\alpha$, $\fabc$ is a linear function with slope $\frac{\beta-1}{2} > 0$.
                Hence, $\fabc$ is a strictly increasing function in $\alpha$.
                In this case $3 \beta < \alpha + 6 \gamma + 8$ is required, which yields $\alpha \ge 3 \beta - 6 \gamma - 7$.
                As $\alpha + \beta$ is even in this case, we get $\alpha \ge 3 \beta - 6 \gamma - 6$.
                Furthermore, the case assumption requires $\beta \ge 4 \gamma + 6$.
                Hence, $0 < \gamma \le \frac{\beta - 6}{4}$.
                Finally, $\frac{\beta^2}{4} + \beta + 3 \ge \frac{9}{2} + \frac{x}{2} - \frac{x^2}{8}$ as $\beta \ge 3$ and $x \in \{0, 1, 2, 3\}$.
                Combining these facts together, we obtain

                \begin{align*}
                    \fabc & \ge \funcc(3 \beta - 6 \gamma - 6, \beta, \gamma) \\
                          & = 2\beta\gamma+\frac{(3 \beta - 6 \gamma - 6) \cdot \beta}{2} - \frac{\beta^2}{2} + \frac{3\beta}{2} - \frac{3 \beta - 6 \gamma - 6}{2} - 3\gamma \\
                          & = \beta^2 + 3 - \beta \gamma - 3 \beta \\
                          & \ge \beta^2 + 3 - \beta \cdot \frac{\beta - 6}{4} - 3 \beta \\
                          & = \frac{\beta^2}{2} - \frac{5 \beta}{2} + (\beta + 3 + \frac{\beta^2}{4}) \\
                          & \ge \frac{\beta^2}{2} - \frac{5 \beta}{2} + (\frac{9}{2} + \frac{x}{2} - \frac{x^2}{8}) \\
                          & = \frac{\beta^2}{8} + \beta \cdot \frac{\beta - 6}{4} + 2 (\frac{\beta - 6}{4})^2 + \frac{\beta}{2} + \frac{x}{2} - \frac{x^2}{8} \\
                          & \ge \frac{\beta^2}{8} + \beta \gamma + 2 \gamma^2 + \frac{\beta}{2} + \frac{x}{2} - \frac{x^2}{8},
                \end{align*}

                thus proving \inqq.

            \item \casec We have $\funcc(\alpha, \beta, \gamma) = 2\beta\gamma+\frac{\alpha\beta}{2} - \frac{\beta^2}{2} + 3\beta - \alpha - 6\gamma - 4$.
                In terms of $\alpha$, $\fabc$ is a linear function with slope $\frac{\beta-2}{2} > 0$.
                Hence, $\fabc$ is a strictly increasing function in $\alpha$.
                In this case $2\beta  \le \alpha + 4\gamma + 6$ is required which yields $\alpha \ge 2\beta -4\gamma-6$.
                Furthermore, the case assumption requires $\beta \ge 4 \gamma + 6$.
                Hence, $0 < \gamma \le \frac{\beta - 6}{4}$.
                Finally, one can verify that $5 \ge \frac{x}{2} + \frac{9}{2} - \frac{x^2}{8}$ for any $x \in \{0, 1, 2, 3\}$.
                Combining these facts together, we obtain
                \begin{align*}
                    \fabc &\ge \funcc(2 \beta - 4 \gamma - 6, \beta, \gamma)\\
                          &= 2\beta\gamma+\frac{(2 \beta - 4 \gamma - 6) \cdot \beta}{2} - \frac{\beta^2}{2} + 3\beta - (2 \beta - 4 \gamma - 6) - 6\gamma - 4\\
                          &= \frac{\beta^2}{2} - 2 \beta - 2 \gamma + 2\\
                          &\ge \frac{\beta^2}{2} - 2 \beta - 2 \frac{\beta - 6}{4} + 2\\
                          &= \frac{\beta^2}{2} - \frac{5 \beta}{2} + 5\\
                          &\ge \frac{\beta^2}{2} - \frac{5 \beta}{2} + (\frac{x}{2} + \frac{9}{2} - \frac{x^2}{8})\\
                          &=\frac{\beta^2}{8} + \beta \cdot \frac{\beta - 6}{4} + 2 (\frac{\beta - 6}{4})^2 + \frac{\beta}{2} + \frac{x}{2} - \frac{x^2}{8}\\
                          &\ge \frac{\beta^2}{8} + \beta \gamma + 2 \gamma^2 + \frac{\beta}{2} + \frac{x}{2} - \frac{x^2}{8},
                \end{align*}
                thus proving \inqq.

            \item \cased In this case $\beta < 2\gamma+3$ is required which contradicts the case assumption $\beta \ge 4\gamma+6$.

            \item \casee The value of $\funcc$ in this case is different from the one in the third case by $-\beta + \frac{\alpha}{2} + 2\gamma + \frac{5}{2} = \frac{\alpha + 4\gamma + 6 - 2\beta}{2} - \frac{1}{2} \ge -\frac{1}{2}$, where the last inequality holds due to the fact that $2\beta  \le \alpha + 4\gamma+6$ is required in these two cases.
                This is the only case assumption that we used to prove the inequality for the third case.
                Hence, the same arguments yield \inqq but with $-\frac{1}{2}$ on the right-hand side.
                As the values on both sides of \inqq are integers due to \cref{lm:funcc-is-positive-integer}, it follows that \inqq holds (since $n \ge m - \frac{1}{2}$ for integers $n$ and $m$ implies $n \ge m$).
            
            \item \casef Here we consider all cases $6$ through $8$ together.
                We have $\funcc(\alpha, \beta, \gamma) = 2\gamma^2 + \alpha\gamma + \frac{\alpha^2}{8} + \frac{\alpha}{2} + y$, where
                \[y = \begin{cases} 0, & \text{if $\alpha = 0 \bmod 4$,}\\ \frac{3}{8}, & \text{if $\alpha$ is odd,}\\ \frac{1}{2}, & \text{if $\alpha = 2 \bmod 4$.} \end{cases}\]
                
                We claim that $(\alpha - \beta) \gamma + \frac{x^2}{8} - \frac{x}{2} + y \ge 0$ holds.
                If $\alpha > \beta$, we get $(\alpha - \beta) \gamma + \frac{x^2}{8} - \frac{x}{2} + y \ge 1 + \frac{x^2}{8} - \frac{x}{2} \ge 0$, where the last inequality can be checked for all values $x \in \{0, 1, 2, 3\}$.
                On the other hand, if $\alpha = \beta$, we have $\alpha = \beta \bmod 4$, and thus for each fixed value of $x \in \{0, 1, 2, 3\}$, we have a fixed value of $y$.
                We get that $(\alpha - \beta) \gamma + \frac{x^2}{8} - \frac{x}{2} + y = \frac{x^2}{8} - \frac{x}{2} + y \ge 0$, where the last inequality can be checked for all values $x \in \{0, 1, 2, 3\}$ and corresponding values of $y$.
                Hence, indeed $(\alpha - \beta) \gamma + \frac{x^2}{8} - \frac{x}{2} + y \ge 0$ always holds.

                Therefore, we obtain
                \begin{align*}
                    \fabc &= 2\gamma^2 + \alpha\gamma + \frac{\alpha^2}{8} + \frac{\alpha}{2} + y\\
                          &\ge 2\gamma^2 + \alpha\gamma + \frac{\beta^2}{8} + \frac{\beta}{2} + y\\
                          &= \frac{\beta^2}{8} + \beta \gamma + 2 \gamma^2 + \frac{\beta}{2} + \frac{x}{2} - \frac{x^2}{8} + ((\alpha - \beta) \gamma + \frac{x^2}{8} - \frac{x}{2} + y)\\
                          &\ge \frac{\beta^2}{8} + \beta \gamma + 2 \gamma^2 + \frac{\beta}{2} + \frac{x}{2} - \frac{x^2}{8},
                \end{align*}
                thus proving \inqq.
        \item \caseg See case $6$.
        \item \caseh See case $6$.
        \end{enumerate}
    \end{itemize}
\end{proof}

\begin{lemma} \label{small-lemma-about-mod-4}
    If $\alpha \ge \beta \ge 3$, $\gamma \ge 1$, $\alpha + \beta$ is even, and $2\beta  \le \alpha + 4\gamma + 5$, then $\frac{\alpha - \beta + 2}{2} \ge \left\lfloor \frac{\alpha + 2 - 4 \gamma}{4} \right\rfloor$.
\end{lemma}

\begin{proof}
    We consider four cases of the remainder of $\alpha$ modulo four.

    \begin{enumerate}
        \item \emph{$\alpha = 0 \bmod 4$.}
            In this case the right-hand side of the inequality $2 \beta \le \alpha + 4 \gamma + 5$ is odd, while the left-hand side is even.
            Hence, we get $2 \beta \le \alpha + 4 \gamma + 4$, and consequently $\beta \le \alpha / 2 + 2 \gamma + 2$.
            We thus obtain $\frac{\alpha - \beta +  2}{2} \ge \frac{\alpha - (\alpha / 2 + 2 \gamma + 2) + 2}{2} = \frac{\alpha - 4 \gamma}{4} = \left\lfloor \frac{\alpha + 2 - 4 \gamma}{4} \right\rfloor$.
        \item \emph{$\alpha = 1 \bmod 4$.}
            As $2 \beta \le \alpha + 4 \gamma + 5$, we get $\beta \le \alpha / 2 + 2 \gamma + 5 / 2$.
            Hence, $\frac{\alpha - \beta + 2}{2} \ge \frac{\alpha - (\alpha / 2 + 2 \gamma + 5 / 2) + 2}{2} = \frac{\alpha - 1 - 4 \gamma}{4} = \left\lfloor \frac{\alpha + 2 - 4 \gamma}{4} \right\rfloor$.
        \item \emph{$\alpha = 2 \bmod 4$.}
            As $\alpha$ is even in this case, we get that $\beta$ is also even because $\alpha + \beta$ is even.
            Hence, $2 \beta - \alpha - 4 \gamma = 2 \bmod 4$.
            Therefore, as we have $2 \beta - \alpha - 4 \gamma \le 5$, we obtain $2 \beta - \alpha - 4 \gamma \le 2$, and thus $\beta \le \alpha / 2 + 2 \gamma + 1$.
            We obtain $\frac{\alpha - \beta + 2}{2} \ge \frac{\alpha - (\alpha / 2 + 2 \gamma + 1) + 2}{2} = \frac{\alpha + 2 - 4 \gamma}{4} = \left\lfloor \frac{\alpha + 2 - 4 \gamma}{4} \right\rfloor$.
        \item \emph{$\alpha = 3 \bmod 4$.}
            As $\alpha$ is odd in this case, we get that $\beta$ is also odd because $\alpha + \beta$ is even.
            Hence, $2 \beta - \alpha - 4 \gamma = 3 \bmod 4$.
            Therefore, as we have $2 \beta - \alpha - 4 \gamma \le 5$, we obtain $2 \beta - \alpha - 4 \gamma \le 3$, and thus $\beta \le \alpha / 2 + 2 \gamma + 3 / 2$.
            We obtain $\frac{\alpha - \beta + 2}{2} \ge \frac{\alpha - (\alpha / 2 + 2 \gamma + 3 / 2) + 2}{2} = \frac{\alpha + 1 - 4 \gamma}{4} = \left\lfloor \frac{\alpha + 2 - 4 \gamma}{4} \right\rfloor$.
    \end{enumerate}
\end{proof}

\begin{lemma} \label{second-small-lemma-about-mod-4}
    If $\alpha \ge \beta \ge 3$, $\gamma \ge 1$, $\alpha + \beta$ is odd, and $2\beta  \le \alpha + 4\gamma + 6$, then $\frac{\alpha - \beta + 3}{2} \ge \left\lfloor \frac{\alpha + 3 - 4 \gamma}{4} \right\rfloor$.
\end{lemma}

\begin{proof}
    We consider four cases of the remainder of $\alpha$ modulo four.

    \begin{enumerate}
        \item \emph{$\alpha = 0 \bmod 4$.}
            As $2 \beta \le \alpha + 4 \gamma + 6$, we get $\beta \le \alpha / 2 + 2 \gamma + 3$.
            We thus obtain $\frac{\alpha - \beta + 3}{2} \ge \frac{\alpha - (\alpha / 2 + 2 \gamma + 3) + 3}{2} = \frac{\alpha - 4 \gamma}{4} = \left\lfloor \frac{\alpha + 3 - 4 \gamma}{4} \right\rfloor$.
        \item \emph{$\alpha = 1 \bmod 4$.}
            As $\alpha$ is odd in this case, we get that $\beta$ is even because $\alpha + \beta$ is odd.
            Hence, $2 \beta - \alpha - 4 \gamma = 3 \bmod 4$.
            Therefore, as we have $2 \beta - \alpha - 4 \gamma \le 6$, we derive $2 \beta - \alpha - 4 \gamma \le 3$.
            Thus, $\beta \le \alpha / 2 + 2 \gamma + 3 / 2$.
            We obtain $\frac{\alpha - \beta + 3}{2} \ge \frac{\alpha - (\alpha / 2 + 2 \gamma + 3 / 2) + 3}{2} = \frac{\alpha + 3 - 4 \gamma}{4} = \left\lfloor \frac{\alpha + 3 - 4 \gamma}{4} \right\rfloor$.
        \item \emph{$\alpha = 2 \bmod 4$.}
            As $\alpha$ is even in this case, we get that $\beta$ is odd because $\alpha + \beta$ is odd.
            Hence, $2 \beta - \alpha - 4 \gamma = 0 \bmod 4$.
            Therefore, as we have $2 \beta - \alpha - 4 \gamma \le 6$, we derive $2 \beta - \alpha - 4 \gamma \le 4$, and thus $\beta \le \alpha / 2 + 2 \gamma + 2$.
            We obtain $\frac{\alpha - \beta + 3}{2} \ge \frac{\alpha - (\alpha / 2 + 2 \gamma + 2) + 3}{2} = \frac{\alpha + 2 - 4 \gamma}{4} = \left\lfloor \frac{\alpha + 3 - 4 \gamma}{4} \right\rfloor$.
        \item \emph{$\alpha = 3 \bmod 4$.}
            In this case the right-hand side of the inequality $2 \beta \le \alpha + 4 \gamma + 6$ is odd, while the left-hand side is even.
            Hence, we get $2 \beta \le \alpha + 4 \gamma + 5$, and consequently $\beta \le \alpha / 2 + 2 \gamma + 5 / 2$.
            We thus obtain $\frac{\alpha - \beta + 3}{2} \ge \frac{\alpha - (\alpha / 2 + 2 \gamma + 5 / 2) + 3}{2} = \frac{\alpha + 1 - 4 \gamma}{4} = \left\lfloor \frac{\alpha + 3 - 4 \gamma}{4} \right\rfloor$.
    \end{enumerate}
\end{proof}

\begin{lemma} \label{a-b-c-quadratic-bound-a-t-2}
    For $\alpha \ge \beta \ge 3$, $\gamma \ge 1$, and $\ell \ge 2$ such that $\ell \ge \fabc - (2\gamma + 1 + (\beta - 2) \cdot (2\gamma + \ell))$, we have $\funcc(\alpha, \beta, \gamma) \ge \gamma + (\alpha - 2 \ell) \cdot (2\gamma + \ell) + \ell + 1$.
\end{lemma}

\begin{proof}
    We have $\ell \ge \fabc - (2\gamma + 1 + (\beta - 2) \cdot (2\gamma + \ell)) = (\fabc + 2 \gamma - 2 \beta \gamma - 1) - (\beta - 2) \cdot \ell$. 
    By rearranging, we obtain $\ell \ge \frac{\fabc + 2 \gamma - 2 \beta \gamma - 1}{\beta - 1}$.
    As $\ell$ is an integer, $\ell \ge \ell_{\lo} \coloneqq \left\lceil \frac{\fabc + 2 \gamma - 2 \beta \gamma - 1}{\beta - 1} \right\rceil$ follows.

    We need to prove
    \begin{align*}
        \fabc &\stackrel{?}{\ge} \gamma + (\alpha - 2 \ell) \cdot (2\gamma + \ell) + \ell + 1\\
              &= -2 \ell^2 + (\alpha - 4 \gamma + 1) \ell + (2 \alpha \gamma + \gamma + 1). \tag{\textasteriskcentered}
    \end{align*}

    The right-hand side of \inq is a quadratic function in $\ell$ with a negative coefficient for $\ell^2$.
    It achieves the maximum value for $\ell = \ell_{\max} \coloneqq \frac{\alpha - 4 \gamma + 1}{4}$.
    As $\ell_{\max}$ may be fractional, the maximum value over integer values of $\ell$ is achieved for the closest integer to $\ell_{\max}$.
    That is, $\ell'_{\max} \coloneqq \left\lfloor \frac{\alpha - 4 \gamma + 2}{4} \right\rfloor$.
    Furthermore, as we have $\ell \ge \ell_{\lo}$ and $\ell \ge 2$, the maximum over possible values of $\ell$ is achieved for $\ell = \ell_0 \coloneqq \max\{\ell'_{\max}, \ell_{\lo}, 2\}$.
    Hence, $-2 \ell_0^2 + (\alpha - 4 \gamma + 1) \ell_0 + (2 \alpha \gamma + \gamma + 1) \ge -2 \ell^2 + (\alpha - 4 \gamma + 1) \ell + (2 \alpha \gamma + \gamma + 1)$.
    Therefore, it suffices to prove \inq for $\ell = \ell_0 = \max\{\left\lfloor \frac{\alpha - 4 \gamma + 2}{4} \right\rfloor, \left\lceil \frac{\fabc + 2 \gamma - 2 \beta \gamma - 1}{\beta - 1} \right\rceil, 2\}$.
    We consider all the cases from the definition of $\funcc$.

\begin{enumerate}
    \item \casea In this case we have 
        \begin{align*}
            \ell_{\lo} &= \left\lceil \frac{\funcc(\alpha, \beta, \gamma) + 2\gamma - 2\beta\gamma - 1}{\beta - 1} \right\rceil\\
                       &= \left\lceil \frac{2\beta\gamma + \frac{\alpha\beta}{2} - \frac{\beta^2}{2} + \frac{\beta}{2} - \frac{\alpha}{2} - 2\gamma + 2 + 2\gamma - 2\beta\gamma - 1}{\beta - 1} \right\rceil\\
                       &= \left\lceil \frac{\alpha - \beta}{2} + \frac{1}{\beta - 1} \right\rceil\\
                       &= \frac{\alpha - \beta}{2} + \left\lceil \frac{1}{\beta - 1} \right\rceil\\
                       &= \frac{\alpha - \beta + 2}{2},
        \end{align*}
        where $\frac{\alpha-\beta}{2}$ is an integer as $\alpha+\beta$ is even in this case.
        As $\beta < \gamma + 2$ is required in this case, we get $\beta \le \gamma + 1$.
        Hence, $\ell_{\lo} = \frac{\alpha - \beta + 2}{2} \ge \frac{\alpha - (\gamma + 1) + 2}{2} \ge \frac{\alpha + 2}{4} - \frac{2\gamma + 2}{4} \ge \frac{\alpha + 2 - 4 \gamma}{4} \ge \ell'_{\max}$.
        Furthermore, $\frac{\alpha - \beta + 2}{2} \ge 2$ as $\alpha > \beta$ in this case and $\alpha + \beta$ is even.
        Therefore, $\ell = \ell_0 = \ell_{\lo} = \frac{\alpha - \beta + 2}{2}$ in this case.
        In order to reuse the arguments in later cases, we substitute $\ell = \frac{\alpha - \beta + x}{2}$ in the right-hand side of \inq, where $x = 2$ in this case:
        \begin{align*}
            \fabc &\stackrel{?}{\ge} -2 \ell^2 + (\alpha - 4 \gamma + 1) \ell + (2 \alpha \gamma + \gamma + 1)\\
                  &= -2 (\frac{\alpha - \beta + x}{2})^2 + (\alpha - 4 \gamma + 1) \cdot \frac{\alpha - \beta + x}{2} + (2 \alpha \gamma + \gamma + 1)\\
                  &= 2 \beta \gamma + \frac{\alpha \beta}{2} - \frac{\beta^2}{2} + (x - \frac{1}{2}) \beta - \frac{x - 1}{2} \alpha - (2x - 1) \gamma + (\frac{x}{2} - \frac{x^2}{2} + 1). \tag{\textasteriskcentered \textasteriskcentered}
        \end{align*}
        This case requires $\beta < \gamma + 2$. Therefore, we obtain
        \begin{align*}
            \fabc &= 2 \beta \gamma + \frac{\alpha \beta}{2} - \frac{\beta^2}{2} + \frac{\beta}{2} - \frac{\alpha}{2} - 2 \gamma + 2\\
                  &\ge 2 \beta \gamma + \frac{\alpha \beta}{2} - \frac{\beta^2}{2} + \frac{\beta}{2} - \frac{\alpha}{2} - 2 \gamma + 2 + (\beta - \gamma - 2)\\
                  &= 2 \beta \gamma + \frac{\alpha \beta}{2} - \frac{\beta^2}{2} + \frac{3\beta}{2} - \frac{\alpha}{2} - 3 \gamma\\
                  &= 2 \beta \gamma + \frac{\alpha \beta}{2} - \frac{\beta^2}{2} + (x - \frac{1}{2}) \beta - \frac{x - 1}{2} \alpha - (2x - 1) \gamma + (\frac{x}{2} - \frac{x^2}{2} + 1),
        \end{align*}
        thus proving \inqq.

    \item \caseb In this case $\fabc$ is different from the value in the first case by $\beta-\gamma-2$.
        Thus, we analogously get $\ell_{\lo} = \left\lceil \frac{\alpha-\beta}{2} + \frac{1 + (\beta - \gamma - 2)}{\beta-1} \right\rceil = \frac{\alpha-\beta}{2} + \left\lceil \frac{\beta-\gamma-1}{\beta - 1} \right\rceil$, where $\frac{\alpha - \beta}{2}$ is an integer as $\alpha + \beta$ is even in this case. In the conditions for this case we have that either $\alpha = \beta$ or $\beta \ge \gamma + 2$. We consider two cases
    \begin{enumerate}
        \item \emph{$\alpha = \beta$.} We get $\ell_{\lo} = \frac{\alpha - \beta}{2} + \left\lceil \frac{\beta-\gamma-1}{\beta - 1} \right\rceil = \left\lceil \frac{\beta-\gamma-1}{\beta - 1} \right\rceil \le 1$.
        Furthermore, this case requires $3 \alpha = 3 \beta < \alpha + 6 \gamma + 8$.
        Therefore, $\alpha < 3 \gamma + 4$.
        Hence, $\alpha \le 3 \gamma + 3$.
        We obtain $\ell'_{\max} = \left\lfloor \frac{\alpha - 4 \gamma + 2}{4} \right\rfloor \le \left\lfloor \frac{(3 \gamma + 3) - 4 \gamma + 2}{4} \right\rfloor = \left\lfloor \frac{5 - \gamma}{4} \right\rfloor \le 1$.
        Thus, $\ell = \ell_0 = 2$ in this case.
        From $\alpha \le 3 \gamma + 3$ we derive $\alpha \le 4 \gamma + 5$.
        Therefore, we obtain
        \begin{align*}
            \fabc &= 2 \beta \gamma + \frac{\alpha \beta}{2} - \frac{\beta^2}{2} + \frac{3 \beta}{2} - \frac{\alpha}{2} - 3 \gamma\\
                  &= 2 \alpha \gamma + \alpha - 3 \gamma\\
                  &\ge 2 \alpha \gamma + \alpha - 3 \gamma + (\alpha - 4 \gamma - 5)\\
                  &= -2 \ell^2 + (\alpha - 4 \gamma + 1) \ell + (2 \alpha \gamma + \gamma + 1),
        \end{align*}
        thus proving \inq.
        \item \emph{$\alpha > \beta$ and $\beta \ge \gamma + 2$.}
            In this case we have $\beta - \gamma - 1 \ge 1$.
            Thus, $\ell_{\lo} = \frac{\alpha - \beta}{2} + \left\lceil \frac{\beta-\gamma-1}{\beta - 1} \right\rceil = \frac{\alpha - \beta + 2}{2}$.
            Note that $\ell_{\lo} = \frac{\alpha - \beta + 2}{2} \ge 2$ as $\alpha > \beta$ and $\alpha + \beta$ is even in this case.

            We now prove $\ell_{\lo} \ge \ell'_{\max}$.
            If $\alpha \le 9$, we get $\ell_{\lo} \ge 2 \ge \frac{\alpha - 4 + 2}{4} \ge \frac{\alpha - 4 \gamma + 2}{4} \ge \ell'_{\max}$.
            Otherwise, $\alpha \ge 10$. As $3 \beta < \alpha + 6 \gamma + 8$ is required in this case, we get $\beta < (\alpha + 6 \gamma + 8) / 3$.
            Combining these two facts, we obtain
            \begin{align*}
                \ell_{\lo} &= \frac{\alpha - \beta + 2}{2}\\
                           &\ge \frac{\alpha - (\alpha + 6 \gamma + 8) / 3 + 2}{2}\\
                           &= \frac{\alpha - 4 \gamma + 2 + (\alpha - 10) / 3}{4}\\
                           &\ge \frac{\alpha - 4 \gamma + 2}{4}\\
                           &\ge \ell'_{\max}.
            \end{align*}
            Hence, indeed $\ell_{\lo} \ge \ell'_{\max}$ always holds, and we have $\ell = \ell_0 = \ell_{\lo} = \frac{\alpha - \beta + 2}{2} = \frac{\alpha - \beta + x}{2}$ for $x = 2$ in this case.
            We obtain
            \begin{align*}
                \fabc &= 2 \beta \gamma + \frac{\alpha \beta}{2} - \frac{\beta^2}{2} + \frac{3\beta}{2} - \frac{\alpha}{2} - 3 \gamma\\
                      &= 2 \beta \gamma + \frac{\alpha \beta}{2} - \frac{\beta^2}{2} + (x - \frac{1}{2}) \beta - \frac{x - 1}{2} \alpha - (2x - 1) \gamma + (\frac{x}{2} - \frac{x^2}{2} + 1),
            \end{align*}
            thus proving \inqq.
    \end{enumerate}
    \item \casec In this case $\funcc(\alpha, \beta, \gamma)$ is different from the value in the first case by $\frac{5}{2}\beta-\frac{\alpha}{2}-4\gamma-6$.
        Thus, we analogously get
        $\ell_{\lo} = \left\lceil \frac{\alpha-\beta}{2} + \frac{1 + (\frac{5}{2}\beta-\frac{\alpha}{2}-4\gamma-6)}{\beta-1} \right\rceil = \frac{\alpha-\beta }{2} + \left\lceil \frac{\frac{3\beta-\alpha-6\gamma-8}{2} + (\beta - \gamma - 1)}{\beta - 1} \right\rceil \ge \frac{\alpha-\beta+2}{2}$, where $\frac{\alpha - \beta}{2}$ is an integer as $\alpha + \beta$ is even in this case, and by the case requirements we have $3\beta \ge \alpha + 6\gamma + 8 \ge 3 \cdot (\gamma + 2)$.
        We consider two cases.
    \begin{enumerate}
        \item \emph{$\alpha + 4 \gamma + 6 = 2 \beta$.}
            Hence, $\beta = \frac{\alpha}{2} + 2 \gamma + 3$.
            As $\beta$ is an integer, we get that $\alpha$ is even.
            As $\alpha + \beta$ is even in this case, we get that $\beta$ is also even.
            Hence, for the equality $\alpha + 4 \gamma + 6 = 2 \beta$ to hold modulo four, we get $\alpha = 2 \bmod 4$.
            Thus, in this case $\ell'_{\max} = \left\lceil \frac{\alpha + 2 - 4 \gamma}{4} \right\rceil = \frac{\alpha + 2 - 4 \gamma}{4}$.
            Combining these facts, we obtain
            \begin{align*}
                \fabc &= 2 \beta \gamma + \frac{\alpha \beta}{2} - \frac{\beta^2}{2} + 3 \beta - \alpha - 6 \gamma - 4\\
                      &= (\alpha + 4 \gamma + 6) \cdot \gamma + \frac{\alpha \cdot (\alpha / 2 + 2 \gamma + 3)}{2} - \frac{(\alpha / 2 + 2 \gamma + 3)^2}{2} \\
                      &+ 3 \cdot (\frac{\alpha}{2} + 2 \gamma + 3) - \alpha - 6 \gamma - 4\\
                      &= \frac{\alpha^2}{8} + \alpha \gamma + 2 \gamma^2 + \frac{2 \alpha + 2}{4}\\
                      &\ge \frac{\alpha^2}{8} + \alpha \gamma + 2 \gamma^2 + \frac{\alpha + 4}{4}\\
                      &= - 2 (\frac{\alpha + 2 - 4 \gamma}{4})^2 + (\alpha - 4 \gamma + 1) \cdot \frac{\alpha + 2 - 4 \gamma}{4} + (2 \alpha \gamma + \gamma + 1)\\
                      &= -2 \ell'^2_{\max} + (\alpha - 4 \gamma + 1) \cdot \ell'_{\max} + (2 \alpha \gamma + \gamma + 1)\\
                      &\ge -2 \ell_0^2 + (\alpha - 4 \gamma + 1) \cdot \ell_0 + (2 \alpha \gamma + \gamma + 1),
            \end{align*}
            thus proving \inq, where the last inequality holds as $\ell'_{\max}$ is the point in which such a quadratic function takes its maximum value over integers.

        \item \emph{$\alpha + 4 \gamma + 6 \neq 2 \beta$.} As $\alpha + 4 \gamma + 6 \ge 2 \beta$ is required in this case, we obtain $\alpha + 4 \gamma + 5 \ge 2 \beta$.
            Thus, \cref{small-lemma-about-mod-4} implies $\ell_0 \ge \ell_{\lo} \ge \frac{\alpha - \beta + 2}{2} \ge \ell_{\max}$.
            Hence, it is sufficient to prove \inq for $\ell = \frac{\alpha - \beta + 2}{2}$ as then $-2 \ell^2 + (\alpha - 4 \gamma + 1) \ell + (2 \alpha \gamma + \gamma + 1) \ge -2 \ell_0^2 + (\alpha - 4 \gamma + 1) \ell_0 + (2 \alpha \gamma + \gamma + 1)$.
            As $3 \beta \ge \alpha + 6 \gamma + 8$ is required in this case, we obtain
            \begin{align*}
                \fabc &= 2 \beta \gamma + \frac{\alpha \beta}{2} - \frac{\beta^2}{2} + 3 \beta - \alpha - 6 \gamma - 4\\
                      &\ge 2 \beta \gamma + \frac{\alpha \beta}{2} - \frac{\beta^2}{2} + 3 \beta - \alpha - 6 \gamma - 4 + (\alpha + 6 \gamma + 8 - 3 \beta) / 2\\
                &= 2 \beta \gamma + \frac{\alpha \beta}{2} - \frac{\beta^2}{2} + \frac{3\beta}{2} - \frac{\alpha}{2} - 3 \gamma\\
                      &= 2 \beta \gamma + \frac{\alpha \beta}{2} - \frac{\beta^2}{2} + (x - \frac{1}{2}) \beta - \frac{x - 1}{2} \alpha - (2x - 1) \gamma + (\frac{x}{2} - \frac{x^2}{2} + 1),
            \end{align*}
            thus proving \inqq because $x = 2$ in this case.
    \end{enumerate}
        
    \item \cased In this case $\funcc(\alpha, \beta, \gamma)$ is different from the value in the first case by $\frac{\beta}{2}-\frac{1}{2}$.
        Thus, we analogously get $\ell_{\lo} = \left\lceil \frac{\alpha-\beta}{2} + \frac{1 + (\frac{\beta}{2} - \frac{1}{2})}{\beta-1} \right\rceil = \frac{\alpha-\beta+1}{2} + \left\lceil \frac{1}{\beta - 1} \right\rceil = \frac{\alpha-\beta+3}{2}$, where $\frac{\alpha - \beta + 1}{2}$ is an integer as $\alpha + \beta$ is odd in this case.
        In this case $\beta < 2\gamma + 3$ is required, thus $\ell_{\lo} = \frac{\alpha-\beta+3}{2} \ge \frac{\alpha-(2\gamma+3)+3}{2} = \frac{2\alpha - 4 \gamma}{4} \ge \frac{\alpha + 2 - 4 \gamma}{4} \ge \left\lfloor \frac{\alpha + 2 - 4 \gamma}{4} \right\rfloor = \ell'_{\max}$.
        Furthermore, $\ell_{\lo} = \frac{\alpha - \beta + 3}{2} \ge 2$ as $\alpha + \beta$ is odd.
        Consequently, we get that $\ell = \ell_0 = \ell_{\lo} = \frac{\alpha - \beta + 3}{2}$ in this case.
        As $\beta < 2 \gamma + 3$ is required in this case, $3 \beta < 6 \gamma + 9 \le 6 \gamma + 7 + \alpha$ follows.
        We thus obtain
        \begin{align*}
            \fabc &= 2 \beta \gamma + \frac{\alpha \beta}{2} - \frac{\beta^2}{2} + \beta - \frac{\alpha}{2} - 2 \gamma + \frac{3}{2}\\
                  &\ge 2 \beta \gamma + \frac{\alpha \beta}{2} - \frac{\beta^2}{2} + \beta - \frac{\alpha}{2} - 2 \gamma + \frac{3}{2} + (3 \beta - 6 \gamma - 7 - \alpha) / 2\\
            &= 2 \beta \gamma + \frac{\alpha \beta}{2} - \frac{\beta^2}{2} + \frac{5\beta}{2} - \alpha - 5 \gamma - 2\\
                  &= 2 \beta \gamma + \frac{\alpha \beta}{2} - \frac{\beta^2}{2} + (x - \frac{1}{2}) \beta - \frac{x - 1}{2} \alpha - (2x - 1) \gamma + (\frac{x}{2} - \frac{x^2}{2} + 1),
        \end{align*}
        thus proving \inqq because $x = 3$ in this case.
    \item \casee In this case $\funcc(\alpha, \beta, \gamma)$ is different from the value in the first case by $\frac{3}{2}\beta-2\gamma-\frac{7}{2}$.
        Thus, we analogously get $\ell_{\lo} = \left\lceil \frac{\alpha-\beta }{2} + \frac{1 + (\frac{3}{2}\beta-2\gamma-\frac{7}{2})}{\beta-1} \right\rceil = \frac{\alpha-\beta+1}{2} + \left\lceil \frac{\beta - 2\gamma - 2}{\beta - 1} \right\rceil = \frac{\alpha-\beta+3}{2}$, where $\frac{\alpha - \beta + 1}{2}$ is an integer as $\alpha + \beta$ is odd in this case, and by the case requirements we have $\beta - 2 \gamma - 2 \ge 1$.
        Furthermore, $\ell_{\lo} = \frac{\alpha - \beta + 3}{2} \ge \left\lfloor \frac{\alpha + 3 - 4 \gamma}{4} \right\rfloor \ge \left\lfloor \frac{\alpha + 2 - 4 \gamma}{4} \right\rfloor = \ell'_{\max}$ due to \cref{second-small-lemma-about-mod-4}.
        Finally, $\ell_{\lo} = \frac{\alpha - \beta + 3}{2} \ge 2$ as $\alpha + \beta$ is odd.
        Consequently, we get that $\ell = \ell_0 = \ell_{\lo} = \frac{\alpha - \beta + 3}{2}$ in this case.
        As $\alpha \ge \beta$, we derive $\alpha / 2 + \gamma + 1 / 2 \ge \beta / 2$.
        We thus obtain
        \begin{align*}
            \fabc &= 2 \beta \gamma + \frac{\alpha \beta}{2} - \frac{\beta^2}{2} + 2 \beta - \frac{\alpha}{2} - 4 \gamma - \frac{3}{2}\\
                  &\ge 2 \beta \gamma + \frac{\alpha \beta}{2} - \frac{\beta^2}{2} + 2 \beta - \frac{\alpha}{2} - 4 \gamma - \frac{3}{2} + (\beta / 2 - \alpha / 2 - \gamma - 1 / 2)\\
            &= 2 \beta \gamma + \frac{\alpha \beta}{2} - \frac{\beta^2}{2} + \frac{5\beta}{2} - \alpha - 5 \gamma - 2\\
                  &= 2 \beta \gamma + \frac{\alpha \beta}{2} - \frac{\beta^2}{2} + (x - \frac{1}{2}) \beta - \frac{x - 1}{2} \alpha - (2x - 1) \gamma + (\frac{x}{2} - \frac{x^2}{2} + 1),
        \end{align*}
        thus proving \inqq because $x = 3$ in this case.

    \item \casef Here we consider all cases $6$ through $8$ together.
                We have $\funcc(\alpha, \beta, \gamma) = 2\gamma^2 + \alpha\gamma + \frac{\alpha^2}{8} + \frac{\alpha}{2} + y$, where
                \[y = \begin{cases} 0, & \text{if $\alpha = 0 \bmod 4$,}\\ \frac{3}{8}, & \text{if $\alpha$ is odd,}\\ \frac{1}{2}, & \text{if $\alpha = 2 \bmod 4$.} \end{cases}\]
                Furthermore, let $x \in \{-1, 0, 1, 2\}$ by such that $\ell'_{\max} = \left\lfloor \frac{\alpha + 2 - 4 \gamma}{4} \right\rfloor = \left\lfloor \frac{\alpha + 2}{4} \right\rfloor - \gamma = \frac{\alpha + x}{4} - \gamma$.
                We claim that $y + \frac{\alpha}{4} \ge 1 + \frac{x}{4} - \frac{x^2}{8}$ holds.
                To prove it, we consider four cases of the remainder of $\alpha$ modulo four.
                \begin{enumerate}
                    \item \emph{$\alpha = 0 \bmod 4$.} In this case $x = y = 0$ and $\alpha \ge 4$. By substitution the inequality holds.
                    \item \emph{$\alpha = 1 \bmod 4$.} In this case $x = -1$, $y = 3 / 8$, and $\alpha \ge 3$. By substitution the inequality holds.
                    \item \emph{$\alpha = 2 \bmod 4$.} In this case $x = 2$, $y = 1 / 2$, and $\alpha \ge 3$. By substitution the inequality holds.
                    \item \emph{$\alpha = 3 \bmod 4$.} In this case $x = 1$, $y = 3 / 8$, and $\alpha \ge 3$. By substitution the inequality holds.
                \end{enumerate}
                Thus, indeed we have $y + \frac{\alpha}{4} \ge 1 + \frac{x}{4} - \frac{x^2}{8}$.
                Therefore, we obtain
            \begin{align*}
                \fabc &= 2\gamma^2 + \alpha\gamma + \frac{\alpha^2}{8} + \frac{\alpha}{2} + y\\
                      &\ge 2\gamma^2 + \alpha\gamma + \frac{\alpha^2}{8} + \frac{\alpha}{2} + y + (1 + \frac{x}{4} - \frac{x^2}{8} - y - \frac{\alpha}{4})\\
                      &=2 \gamma^2 + \alpha \gamma + \frac{\alpha^2}{8} + \frac{\alpha}{4} + (1 + \frac{x}{4} - \frac{x^2}{8})\\
                      &= -2(\frac{\alpha + x}{4} - \gamma)^2 + (\alpha - 4 \gamma + 1) \cdot (\frac{\alpha + x}{4} - \gamma) + (2 \alpha \gamma + \gamma + 1)\\
                      &= -2 \ell'^2_{\max} + (\alpha - 4 \gamma + 1) \cdot \ell'_{\max} + (2 \alpha \gamma + \gamma + 1)\\
                      &\ge - 2 \ell_0^2 + (\alpha - 4 \gamma + 1) \cdot \ell_0 + (2 \alpha \gamma + \gamma + 1),
            \end{align*}
            thus proving \inq, where the last inequality holds as $\ell'_{\max}$ is the point in which such a quadratic function takes its maximum value over integers.
    \item \caseg See case $6$.
    \item \caseh See case $6$.
\end{enumerate}
\end{proof}

\begin{lemma} \label{a-b-c-quadratic-bound-a-t-0}
    For $\alpha \ge \beta \ge 4$, $\gamma \ge 1$, and $\ell \ge 2$, such that $\ell \ge \fabc - (2\gamma + 1 + (\beta - 4) \cdot (2\gamma + \ell) + 2 \cdot (\gamma + \ell))$, we have $\fabc \ge (\alpha - 2\ell + 1) \cdot (2\gamma + \ell) + \ell$.
\end{lemma}

\begin{proof}
    We have $\ell \ge \fabc - (2\gamma + 1 + (\beta - 4) \cdot (2\gamma + \ell) + 2 \cdot (\gamma + \ell)) = (\fabc + 4 \gamma - 2 \beta \gamma - 1) - (\beta - 2) \cdot \ell$.
    By rearranging, we obtain
    $\ell \ge \frac{\fabc + 4 \gamma - 2 \beta \gamma - 1}{\beta - 1}$.
    As $\ell$ is an integer, $\ell \ge \ell_{\lo} \coloneqq \left\lceil \frac{\fabc + 4 \gamma - 2 \beta \gamma - 1}{\beta - 1} \right\rceil$ follows.

    We need to prove
    \begin{align*}
        \fabc &\stackrel{?}{\ge} (\alpha - 2\ell + 1) \cdot (2\gamma + \ell) + \ell\\
              &=- 2 \ell^2 + (\alpha - 4 \gamma + 2) \ell + (2 \alpha \gamma + 2 \gamma). \tag{\textasteriskcentered}
    \end{align*}

    The right-hand side of \inq is a quadratic function in $\ell$ with a negative coefficient for $\ell^2$.
    It achieves the maximum value for $\ell = \ell_{\max} \coloneqq \frac{\alpha - 4 \gamma + 2}{4}$.
    As $\ell_{\max}$ may be fractional, the maximum value over integer values of $\ell$ is achieved for the closest integer to $\ell_{\max}$.
    That is, $\ell'_{\max} \coloneqq \left\lfloor \frac{\alpha - 4 \gamma + 3}{4} \right\rfloor$.
    Furthermore, as we have $\ell \ge \ell_{\lo}$ and $\ell \ge 2$, the maximum over possible values of $\ell$ is achieved for $\ell = \ell_0 \coloneqq \max\{\ell'_{\max}, \ell_{\lo}, 2\}$.
    Hence, $-2 \ell_0^2 + (\alpha - 4 \gamma + 2) \ell_0 + (2 \alpha \gamma + 2\gamma) \ge -2 \ell^2 + (\alpha - 4 \gamma + 2) \ell + (2 \alpha \gamma + 2\gamma)$.
    Therefore, it suffices to prove \inq for $\ell = \ell_0 = \max\{\left\lfloor \frac{\alpha - 4 \gamma + 3}{4} \right\rfloor, \left\lceil \frac{\fabc + 4 \gamma - 2 \beta \gamma - 1}{\beta - 1} \right\rceil, 2\}$.
    We consider all the cases from the definition of $\funcc$.

\begin{enumerate}
    \item \casea In this case we have
        \begin{align*}
            \ell_{\lo} &= \left\lceil \frac{\funcc(\alpha, \beta, \gamma) + 4\gamma - 2\beta\gamma - 1}{\beta - 1} \right\rceil\\
                       &= \left\lceil \frac{2\beta\gamma + \frac{\alpha\beta}{2} - \frac{\beta^2}{2} + \frac{\beta}{2} - \frac{\alpha}{2} - 2\gamma + 2 + 4\gamma - 2\beta\gamma - 1}{\beta - 1} \right\rceil = \left\lceil \frac{\frac{\beta \cdot (\alpha - \beta)}{2} - \frac{\alpha - \beta}{2} + 2\gamma + 1}{\beta - 1} \right\rceil\\
                       &= \frac{\alpha - \beta}{2} + \left\lceil \frac{2\gamma + 1}{\beta - 1} \right\rceil \ge \frac{\alpha - \beta + 6}{2},
        \end{align*}
        where $\frac{\alpha-\beta}{2}$ is an integer as $\alpha+\beta$ is even in this case, and $\frac{2\gamma + 1}{\beta - 1} > 2$ as $\beta \le \gamma + 1$ in this case.
        Furthermore, $\ell_0 \ge \ell_{\lo} \ge \frac{\alpha - \beta + 6}{2} \ge \frac{\alpha - (\gamma + 1) + 6}{2} \ge \frac{\alpha + 3 - 4 \gamma}{4} \ge \left\lfloor \frac{\alpha + 3 - 4 \gamma}{4} \right\rfloor = \ell'_{\max}$.
        Hence, it is sufficient to prove \inq for $\ell = \frac{\alpha - \beta + 6}{2}$ as then $-2 \ell^2 + (\alpha - 4 \gamma + 2) \ell + (2 \alpha \gamma + 2\gamma) \ge -2 \ell_0^2 + (\alpha - 4 \gamma + 2) \ell_0 + (2 \alpha \gamma + 2\gamma)$.
        In order to reuse the arguments in later cases, we substitute $\ell = \frac{\alpha - \beta + x}{2}$ in the right-hand side of \inq, where $x = 6$ in this case:
        \begin{align*}
            \fabc &\stackrel{?}{\ge} -2 \ell^2 + (\alpha - 4 \gamma + 2) \ell + (2 \alpha \gamma + 2 \gamma)\\
                  &= -2(\frac{\alpha-\beta+x}{2})^2 + (\alpha-4\gamma+2) \cdot \frac{\alpha-\beta+x}{2} + (2\alpha\gamma + 2\gamma)\\
                  &= 2\beta\gamma + \frac{\alpha\beta}{2} - \frac{\beta^2}{2} + (x - 1) \beta + (1 - \frac{x}{2}) \alpha + (2-2x)\gamma + (x - \frac{x^2}{2}). \tag{\textasteriskcentered \textasteriskcentered}
        \end{align*}

        As $\beta < \gamma + 2$ in this case, we derive $3 \alpha / 2 + 8 \gamma + 14 \ge 3 \alpha / 2 + 3 \cdot (\gamma + 2) \ge 3 \beta / 2 + 3 \beta = 9 \beta / 2$.
        Therefore, we obtain
        \begin{align*}
            \fabc &= 2 \beta \gamma + \frac{\alpha \beta}{2} - \frac{\beta^2}{2} + \frac{\beta}{2} - \frac{\alpha}{2} - 2 \gamma + 2\\
                  &\ge 2 \beta \gamma + \frac{\alpha \beta}{2} - \frac{\beta^2}{2} + \frac{\beta}{2} - \frac{\alpha}{2} - 2 \gamma + 2 + (9 \beta / 2 - 3 \alpha / 2 - 8 \gamma - 14)\\
                  &=2 \beta \gamma + \frac{\alpha \beta}{2} - \frac{\beta^2}{2} + 5 \beta - 2 \alpha - 10 \gamma - 12\\
                  &=2\beta\gamma + \frac{\alpha\beta}{2} - \frac{\beta^2}{2} + (x - 1) \beta + (1 - \frac{x}{2}) \alpha + (2-2x)\gamma + (x - \frac{x^2}{2}),
        \end{align*}
        thus proving \inqq.
    \item \caseb In this case $\funcc(\alpha, \beta, \gamma)$ is different from the value in the first case by $\beta-\gamma-2$.
        Thus, we analogously get $\ell_{\lo} = \left\lceil \frac{\alpha-\beta }{2} + \frac{2\gamma + 1 + (\beta - \gamma - 2)}{\beta-1} \right\rceil = \frac{\alpha-\beta }{2} + \left\lceil \frac{\beta-1+\gamma}{\beta - 1} \right\rceil = \frac{\alpha - \beta + 2}{2} + \left\lceil \frac{\gamma}{\beta - 1} \right\rceil \ge \frac{\alpha - \beta + 4}{2}$, where $\frac{\alpha - \beta}{2}$ is an integer as $\alpha + \beta$ is even in this case.
        As $3 \beta < \alpha + 6 \gamma + 8$ is required in this case, we derive $\ell_0 \ge \ell_{\lo} \ge \frac{\alpha - \beta + 4}{2} = \frac{3 \alpha - 3 \beta + 12}{6} \ge \frac{3 \alpha - (\alpha + 6 \gamma + 8) + 12}{6} = \frac{4 \alpha + 8 - 12 \gamma}{12} \ge \frac{3 \alpha + 9 - 12 \gamma}{12} = \frac{\alpha + 3 - 4 \gamma}{4} \ge \left\lfloor \frac{\alpha + 3 - 4 \gamma}{4} \right\rfloor = \ell'_{\max}$.
        Hence, it is sufficient to prove \inq for $\ell = \frac{\alpha - \beta + 4}{2} = \frac{\alpha - \beta + x}{2}$ where $x=4$ in this case, as then $-2 \ell^2 + (\alpha - 4 \gamma + 2) \ell + (2 \alpha \gamma + 2\gamma) \ge -2 \ell_0^2 + (\alpha - 4 \gamma + 2) \ell_0 + (2 \alpha \gamma + 2\gamma)$.
        As $\alpha + 6 \gamma + 8 > 3 \beta$ is required in this case, we obtain
        \begin{align*}
            \fabc &= 2 \beta \gamma + \frac{\alpha \beta}{2} - \frac{\beta^2}{2} + \frac{3\beta}{2} - \frac{\alpha}{2} - 3 \gamma\\
                  &\ge 2 \beta \gamma + \frac{\alpha \beta}{2} - \frac{\beta^2}{2} + \frac{3\beta}{2} - \frac{\alpha}{2} - 3 \gamma + (3 \beta - \alpha - 6 \gamma - 8) / 2\\
                  &= 2 \beta \gamma + \frac{\alpha \beta}{2} - \frac{\beta^2}{2} + 3 \beta - \alpha - 6 \gamma - 4\\
                  &=2\beta\gamma + \frac{\alpha\beta}{2} - \frac{\beta^2}{2} + (x - 1) \beta + (1 - \frac{x}{2}) \alpha + (2-2x)\gamma + (x - \frac{x^2}{2}),
        \end{align*}
        thus proving \inqq.
    \item \casec In this case $\funcc(\alpha, \beta, \gamma)$ is different from the value in the first case by $\frac{5}{2}\beta-\frac{\alpha}{2}-4\gamma-6$.
        Thus, we analogously get $\ell_{\lo} = \left\lceil \frac{\alpha-\beta }{2} + \frac{2\gamma + 1 + (\frac{5}{2}\beta-\frac{\alpha}{2}-4\gamma-6)}{\beta-1} \right\rceil = \frac{\alpha-\beta }{2} + \left\lceil \frac{\frac{3\beta-\alpha-6\gamma-8}{2} + (\beta - 1) + \gamma}{\beta - 1} \right\rceil = \frac{\alpha - \beta + 2}{2} + \left\lceil \frac{\frac{3\beta - \alpha - 6\gamma - 8}{2} + \gamma}{\beta - 1} \right\rceil \ge \frac{\alpha - \beta + 4}{2}$, where $\frac{\alpha - \beta}{2}$ is an integer as $\alpha+\beta$ is even in this case, and $3\beta \ge \alpha + 6\gamma + 8$ in this case.

        We claim that $\frac{\alpha - \beta + 4}{2} = \frac{\alpha + 2}{4} + \frac{\alpha + 4 \gamma + 6 - 2 \beta}{4} - \gamma \ge \left\lfloor \frac{\alpha + 3}{4} \right\rfloor - \gamma = \left\lfloor \frac{\alpha + 3 - 4 \gamma}{4} \right\rfloor = \ell_{\max}$.
        Note that $\alpha + 4 \gamma + 6 \ge 2 \beta$ is required in this case.
        If $\alpha + 4 \gamma + 6 = 2 \beta$, then $\alpha$ is even, and we have $\frac{\alpha + 2}{4} + \frac{\alpha + 4 \gamma + 6 - 2 \beta}{4} - \gamma \ge \frac{\alpha + 2}{4} - \gamma \ge \left\lfloor \frac{\alpha + 3}{4} \right\rfloor - \gamma$.
        Otherwise, $\alpha + 4 \gamma + 5 \ge 2 \beta$, and we have $\frac{\alpha + 2}{4} + \frac{\alpha + 4 \gamma + 6 - 2 \beta}{4} - \gamma \ge \frac{\alpha + 2}{4} + \frac{1}{4} - \gamma \ge \left\lfloor \frac{\alpha + 3}{4} \right\rfloor - \gamma$.
        Therefore, indeed $\ell_0 \ge \ell_{\lo} \ge \frac{\alpha - \beta + 4}{2} \ge \ell'_{\max}$.
        Hence, it is sufficient to prove \inq for $\ell = \frac{\alpha - \beta + 4}{2} = \frac{\alpha - \beta + x}{2}$ where $x=4$ in this case, as then $-2 \ell^2 + (\alpha - 4 \gamma + 2) \ell + (2 \alpha \gamma + 2\gamma) \ge -2 \ell_0^2 + (\alpha - 4 \gamma + 2) \ell_0 + (2 \alpha \gamma + 2\gamma)$.
        We obtain
        \begin{align*}
            \fabc &= 2 \beta \gamma + \frac{\alpha \beta}{2} - \frac{\beta^2}{2} + 3 \beta - \alpha - 6 \gamma - 4\\
            &=2\beta\gamma + \frac{\alpha\beta}{2} - \frac{\beta^2}{2} + (x - 1) \beta + (1 - \frac{x}{2}) \alpha + (2-2x)\gamma + (x - \frac{x^2}{2}),
        \end{align*}
        thus proving \inqq.
        
    \item \cased In this case $\funcc(\alpha, \beta, \gamma)$ is different from the value in the first case by $\frac{\beta}{2}-\frac{1}{2}$.
        Thus, we analogously get $\ell_{\lo} = \left\lceil \frac{\alpha-\beta }{2} + \frac{2\gamma + 1 + (\frac{\beta}{2} - \frac{1}{2})}{\beta-1} \right\rceil = \frac{\alpha-\beta+1}{2} + \left\lceil \frac{2\gamma + 1}{\beta - 1} \right\rceil \ge \frac{\alpha-\beta+3}{2}$, where $\frac{\alpha - \beta + 1}{2}$ is an integer as $\alpha + \beta$ is odd in this case.
        Furthermore, as $2 \gamma + 3 > \beta$ is required in this case, we derive $\ell_0 \ge \ell_{\lo} \ge \frac{\alpha - \beta + 3}{2} \ge \frac{\alpha-(2\gamma+3)+3}{2} = \frac{2\alpha - 4 \gamma}{4} \ge \frac{\alpha + 3 - 4 \gamma}{4} \ge \left\lfloor \frac{\alpha + 3 - 4 \gamma}{4} \right\rfloor = \ell'_{\max}$.
        Hence, it is sufficient to prove \inq for $\ell = \frac{\alpha - \beta + 3}{2} = \frac{\alpha - \beta + x}{2}$ where $x=3$ in this case, as then $-2 \ell^2 + (\alpha - 4 \gamma + 2) \ell + (2 \alpha \gamma + 2\gamma) \ge -2 \ell_0^2 + (\alpha - 4 \gamma + 2) \ell_0 + (2 \alpha \gamma + 2\gamma)$.

        This case requires $2 \gamma + 3 > \beta$ to hold. Therefore, we obtain
        \begin{align*}
            \fabc &= 2 \beta \gamma + \frac{\alpha \beta}{2} - \frac{\beta^2}{2} + \beta - \frac{\alpha}{2} - 2 \gamma + \frac{3}{2}\\
                  &\ge 2 \beta \gamma + \frac{\alpha \beta}{2} - \frac{\beta^2}{2} + \beta - \frac{\alpha}{2} - 2 \gamma + \frac{3}{2} + (\beta - 2 \gamma - 3)\\
                  &= 2 \beta \gamma + \frac{\alpha \beta}{2} - \frac{\beta^2}{2} + 2 \beta - \frac{\alpha}{2} - 4 \gamma - \frac{3}{2}\\
                  &=2\beta\gamma + \frac{\alpha\beta}{2} - \frac{\beta^2}{2} + (x - 1) \beta + (1 - \frac{x}{2}) \alpha + (2-2x)\gamma + (x - \frac{x^2}{2}),
        \end{align*}
        thus proving \inqq.
    \item \casee In this case $\funcc(\alpha, \beta, \gamma)$ is different from the value in the first case by $\frac{3}{2}\beta-2\gamma-\frac{7}{2}$.
        Thus, we analogously get $\ell_{\lo} = \left\lceil \frac{\alpha-\beta }{2} + \frac{2\gamma + 1 + (\frac{3}{2}\beta-2\gamma-\frac{7}{2})}{\beta-1} \right\rceil = \frac{\alpha-\beta+1}{2} + \left\lceil \frac{\beta - 2}{\beta - 1} \right\rceil = \frac{\alpha-\beta+3}{2}$, where $\frac{\alpha - \beta + 1}{2}$ is an integer as $\alpha + \beta$ is odd in this case.

        \cref{second-small-lemma-about-mod-4} implies that $\ell_{\lo} = \frac{\alpha - \beta + 3}{2} \ge \left\lfloor \frac{\alpha + 3 - 4 \gamma}{4} \right\rfloor = \ell_{\max}$.
        Finally, $\ell_{\lo} = \frac{\alpha - \beta + 3}{2} \ge 2$ as $\alpha + \beta$ is odd.
        Hence, $\ell_0 = \ell_{\lo} = \frac{\alpha - \beta + 3}{2}$ in this case.
        We obtain
        \begin{align*}
            \fabc &= 2 \beta \gamma + \frac{\alpha \beta}{2} - \frac{\beta^2}{2} + 2 \beta - \frac{\alpha}{2} - 4 \gamma - \frac{3}{2}\\
                  &=2\beta\gamma + \frac{\alpha\beta}{2} - \frac{\beta^2}{2} + (x - 1) \beta + (1 - \frac{x}{2}) \alpha + (2-2x)\gamma + (x - \frac{x^2}{2}),
        \end{align*}
        thus proving \inqq, where $x = 3$ in this case.
    \item \casef Here we consider all cases $6$ through $8$ together.
            We have $\funcc(\alpha, \beta, \gamma) = 2\gamma^2 + \alpha\gamma + \frac{\alpha^2}{8} + \frac{\alpha}{2} + y$, where
            \[y = \begin{cases} 0, & \text{if $\alpha = 0 \bmod 4$,}\\ \frac{3}{8}, & \text{if $\alpha$ is odd,}\\ \frac{1}{2}, & \text{if $\alpha = 2 \bmod 4$.} \end{cases}\]
            Furthermore, let $x \in \{0, 1, 2, 3\}$ by such that $\ell'_{\max} = \left\lfloor \frac{\alpha + 3 - 4 \gamma}{4} \right\rfloor = \left\lfloor \frac{\alpha + 3}{4} \right\rfloor - \gamma = \frac{\alpha + x}{4} - \gamma$.
            We claim that $y = \frac{x}{2} - \frac{x^2}{8}$ holds.
            To prove it, we consider four cases of the remainder of $\alpha$ modulo four.
            \begin{enumerate}
                \item \emph{$\alpha = 0 \bmod 4$.} In this case $x = y = 0$. The equality obviously holds.
                \item \emph{$\alpha = 1 \bmod 4$.} In this case $x = 3$ and $y = 3 / 8$. The equality obviously holds.
                \item \emph{$\alpha = 2 \bmod 4$.} In this case $x = 2$ and $y = 1 / 2$. The equality obviously holds.
                \item \emph{$\alpha = 3 \bmod 4$.} In this case $x = 1$ and $y = 3 / 8$. The equality obviously holds.
            \end{enumerate}
            Thus, indeed we have $y = \frac{x}{2} - \frac{x^2}{8}$.
            Therefore, we obtain
            \begin{align*}
                \fabc &= 2\gamma^2 + \alpha\gamma + \frac{\alpha^2}{8} + \frac{\alpha}{2} + y\\
                      &= 2\gamma^2 + \alpha\gamma + \frac{\alpha^2}{8} + \frac{\alpha}{2} + (\frac{x}{2} - \frac{x^2}{8})\\
                      &= -2(\frac{\alpha + x}{4} - \gamma)^2 + (\alpha - 4 \gamma + 2) \cdot (\frac{\alpha + x}{4} - \gamma) + (2 \alpha \gamma + 2\gamma)\\
                      &= -2 \ell'^2_{\max} + (\alpha - 4 \gamma + 2) \cdot \ell'_{\max} + (2 \alpha \gamma + 2\gamma)\\
                      &\ge -2 \ell_0^2 + (\alpha - 4 \gamma + 2) \ell_0 + (2 \alpha \gamma + 2\gamma),
            \end{align*}
            thus proving \inq, where the last inequality holds as $\ell'_{\max}$ is the point in which such a quadratic function takes its maximum value over integers.
    \item \caseg See case $6$.
    \item \caseh See case $6$.
\end{enumerate}
\end{proof}

\bibliography{refs}

\appendix

\section{Further Discussion and Related Work}
\label{app:further-related-work}

Beyond the variants of subgraph isomorphism discussed in the introduction, further variations arise from the following changes:

\begin{itemize}
\item \emph{Parameter $m$ vs $n$:} Measuring running time in terms of $m$ is natural, as it describes the input size. Nevertheless, a significant fraction of work on subgraph finding problems studies the parameter $n$. We leave it as an open problem whether similar results as in this paper can be shown for parameter $n$.

\item \emph{Induced vs non-induced:} In this paper we search for subgraphs that are not necessarily induced. Induced subgraph isomorphism behaves quite differently (see, e.g.,~\cite{WilliamsWWY15,DalirrooyfardVW21,
DalirrooyfardW22,EppsteinLW23}), so we do not expect that similar results as in this paper hold for the induced version.
\end{itemize}

Further related work considers subgraph isomorphism with restrictions on the host graph (see, e.g.,~\cite{ChibaN85,Kowalik03,MarxP14}) or with respect to other parameters (see, e.g.,~\cite{MarxP14,DalirrooyfardMVWX23}). 
Closely related are problems in which the pattern is not a fixed subgraph, but the goal is to detect the smallest subgraph from a family of subgraphs, e.g., finding the shortest cycle (see, e.g.,~\cite{LincolnWW18}).

\paragraph*{Fixed-Parameter Tractability}
In parameterized algorithms, a class $\mathcal C$ of pattern graphs is called fixed-parameter tractable if there is a constant $c>0$ such that $H$-subgraph detection can be solved in time $\Oh(m^c)$ for each $H \in \mathcal C$. As the hidden constant may depend on $H$, this running time is typically written as $f(H) \cdot m^c$. Significant effort goes into optimizing the factor $f(H)$, e.g., the undirected $k$-path pattern can be detected in time $\Oh(1.66^k \textup{poly}(n))$~\cite{BjorklundHKK17}. Similar improvements of the factor $f(H)$ have been found for the $\tOh(n^{\tw(H)+1})$-time algorithm for $H$-subgraph detection~\cite{AlonYZ95,AminiFS12,FominLRSR12,Pratt19}.

In contrast, in this paper we assume $H$ to have constant size, and therefore we ignore the constant factor $f(H)$.

\paragraph*{Database Theory}

One of the most fundamental operations on relational databases are join queries. 
Without going into the details of their definition, it suffices to say that join queries are a generalization of subgraph listing to hypergraphs (specifically, subgraph isomorphism is the special case of join queries where all tables have two columns and where the query is self-join-free).

The last 15 years have brought two major advances for join queries in database theory. The first is the AGM bound, which bounds the maximal output size of a join for a fixed query and is tight up to constant factors~\cite{AtseriasGM13}. More precisely, for a join query $Q$ the AGM bound is a number $\textup{AGM}(Q)$ such that on any input database of size $m$ the output size is at most $\Oh(m^{\textup{AGM}(Q)})$. 
The second major advancement was the development of worst-case optimal join algorithms, which solve a join query in running time $\tOh(m^{\textup{AGM}(Q)})$, and thus are optimal up to logarithmic factors for queries attaining the worst-case output size~\cite{NgoPRR18,Veldhuizen12}. 

The natural next step is to study \emph{output-sensitive} algorithms, for which the running time depends on the input size $m$ and the output size $t$. Specifically, the next research challenge is to determine for each join query $Q$ the optimal constant $c(Q)$ such that it can be solved in time $\tOh(m^{c(Q)} + t)$. (Alternatively, one can ask for the optimal preprocessing time $\tOh(m^{c(Q)})$ for enumeration with delay $\tOh(1)$.) Note that worst-case optimal join algorithms attain $c(Q) \le \textup{AGM}(Q)$, and that any algorithm with $c(Q) \le \textup{AGM}(Q)$ is also a worst-case optimal join algorithm, so this challenge asks for a strengthening of worst-case optimal join algorithms.
Much of the recent interest in the fine-grained complexity of listing and enumeration in database theory can be viewed as working on this challenge, see, e.g.~\cite{DurandG07,DeepHK20,BringmannCM22,
CarmeliZBCKS22,CarmeliS23}. Our results in this paper also contribute to this challenge, as we determine the optimal constant $c(Q)$ for listing and enumeration for graph queries (more precisely, self-join-free join queries of arity 2) of subquadratic complexity.
In full generality, this research challenge remains as a difficult guiding question.

\paragraph*{Submodular Width}\label{app:further-related-work:subw}
Submodular width is an algorithmic approach to subgraph detection (and, more generally, to join queries) with heavy influences from both parameterized algorithms and database theory.
Marx~\cite{Marx13} defined the submodular width $\subw(H)$ and showed that $H$-subgraph detection has time complexity $m^{\Theta(\subw(H)^{\Theta(1)})}$, where the lower bound assumes the Exponential Time Hypothesis.
The upper bound was improved to time $\tOh(m^{\subw(H)})$ by Abo Khamis, Ngo, and Suciu~\cite{Khamis0S17}; this is known as the PANDA algorithm, see also \cite{KhamisNS24} for a simplified exposition.

Let us compare our results in this paper with the submodular width approach. It turns out that the exponents that we obtained are equal to the submodular width.
More precisely, we show in \cref{sec:p-graph-subw} that $\min\{\textup{subw}(H),2\} = \min\{c(H),2\}$, where $c(H)$ is the exponent of our algorithms for listing, enumeration and min-weight subgraph isomorphism (this follows by comparing \cref{subw-full-characterization} with our main result \cref{full-characterization}).

That is, we could have obtained (some of) our results by invoking PANDA's $\tOh(m^{\subw(H)})$ running time as a black box. However, to prove that time $\tOh(m^{\subw(H)})$ matches our lower bound we would still need to determine the submodular width of all patterns in $\family$. While computing the submodular width of a specific pattern can be done with the help of a computer (at least for small patterns), here we need bounds on infinite families of graphs, for which this kind of computer search does not suffice. 
Since the definition of submodular width is quite complicated to work with, analyzing the submodular width of families of graphs is very difficult. In fact, the only tight upper bound analysis of $\min\{\textup{subw}(H),2\}$ that we know \emph{mimics our algorithms} for patterns in $\family$ from \cref{sec:p-graph-upper-bounds}, and is thus equally long, with the same large number of cases. See \cref{sec:p-graph-subw}, where we present this analysis of $\min\{\textup{subw}(H),2\}$, for the reader to convince themselves that the route via submodular width and the PANDA algorithm would not lead to a shorter paper.

In addition to not leading to a shorter paper, using PANDA for the algorithmic part of this paper would come with some disadvantages:
(1) Currently we present a self-contained algorithm, which would not be the case if we invoked the PANDA algorithm.
(2) PANDA has a large number of logarithmic factors hidden by the $\tOh$-notation, while our algorithms do not need any logarithmic factors.
(3) Our approach solves the intermediate problem that we call \Henciso, which can be used to solve a larger family of problems compared to PANDA; we elaborate on this in the next paragraph.

\paragraph*{Potential for Improved Counting Algorithms}

One can view PANDA \cite{KhamisNS24} as an algorithm solving a relaxed version of \Henciso, in which we ask for a collection of partial $H$-encodings that cover every solution \emph{at least} once instead of \emph{exactly} once.
This relaxation influences the class of problems that can be solved: PANDA can only solve a problem if we can handle the fact that some solutions appear multiple times.
This is the case for the three main problems we study in this paper: \Hminiso, \Hlistiso, and \Henumiso.
However, it is not the case for \emph{counting} the number of $H$-subgraphs (and more generally for any non-idempotent aggregative query).
To circumvent this issue, Abo Khamis et al.~\cite{KhamisCMNNOS20} developed the \#PANDA algorithm, which in contrast to PANDA can handle non-idempotent aggregative queries, at the cost of running in time $\tOh(m^{\sharpw(H)})$, where $\sharpw(H)$ is a width notion satisfying $\sharpw(H) \ge \subw(H)$.

Since our algorithms solve \Henciso and thus encode every $H$-subgraph of $G$ exactly once, they can be easily generalized to any non-idempotent aggregative query (by generalizing the proof of \cref{all-reductions}). In particular, we can \emph{count} the number of $H$-subgraphs in the same time as we can solve \Hminiso; see \cref{sec:counting}. (We did not discuss this application earlier, because for counting our algorithm does not always achieve near-optimal running time, as, e.g., counting cliques can be sped up using fast matrix multiplication and the lower bounds for $P(3, 3, 3)$ based on the \threesumconj and the \minplusconvconj do not transfer to \Hcountiso.) 
Due to the equality with submodular width shown in \cref{sec:p-graph-subw}, our algorithm can count the number of $H$-subgraphs in time $\Oh(m^{\subw(H)})$ for any pattern $H$ with $\subw(H) < 2$. This is better than \#PANDA's time $\tOh(m^{\sharpw(H)})$. Specifically, similarly to the discussion about high-degree-low-degree-based algorithms in \cref{sec:p922-example}, via a computer-aided search one can show that for the pattern $H = P(9, 2, 2)$ we have $\sharpw(H) \ge 2 - 1 / 9$ while $\subw(H) = 2 - 1 / 8$.
(Proving this fact goes beyond the scope of this paper and will be discussed in detail in our upcoming work~\cite{BG26}.)
Therefore, our algorithm gives a polynomial improvement upon the state of the art for counting $H$-subgraphs, for at least one pattern.

We hope that the techniques introduced in Lemmas \ref{lm:p922-ub}, \ref{lm:pa2c-ub}, and \ref{lm:pabc-ub} can lead to a framework for counting that is faster than \#PANDA.
It is a fascinating question whether this approach can solve counting in PANDA's time complexity of $\tOh(m^{\subw(H)})$ in general.

\section{Counting Algorithms} \label{sec:counting}

In the \emph{\Hcountiso} problem, given a host graph $G$, the task is to compute the total number of $H$-subgraphs of $G$. 
In this section we show that essentially the same algorithms that we designed for \Hminiso also work for \Hcountiso. Specifically, we show the following analogue of the algorithmic part of our main result \cref{full-characterization} for \Hcountiso. 

\begin{theorem} \label{thm:main-result-counting}
    Let $H$ be a pattern. Assume that every graph in $D(H)$ is isomorphic to a graph in $\family$, 
    so up to isomorphism we have $D(H) = \{P(\alpha_1, \beta_1, \gamma_1 \times 2), \ldots, P(\alpha_k, \beta_k, \gamma_k \times 2) \}$ for some $(\alpha_1, \beta_1, \gamma_1), \ldots, (\alpha_k, \beta_k, \gamma_k) \in \Tfamily$.
    Let $c(H) := 2 - 1 / \max_{i \in [k]} \funcc(\alpha_i, \beta_i, \gamma_i)$.
    Then \Hcountiso can be solved in time $\Oh(m^{c(H)})$.
    \lipicsEnd
\end{theorem}

Note that the algorithm of \cref{thm:main-result-counting} is not necessarily near-optimal, as the analogue of the lower bound part of our main result does not hold for \Hcountiso. Indeed, for some patterns (e.g.\ triangles) fast matrix multiplication can be used to design faster algorithms.

\medskip
Our proof of \cref{thm:main-result-counting} relies on the following variant of \cref{all-reductions} for \Hcountiso.

\begin{lemma}\label{lm:counting-reduction}
    For any pattern $H$, if \Henciso can be solved in time $T(m)$ then \Hcountiso can be solved in time $\Oh(T(m))$.
\end{lemma}

\begin{proof}
    We closely follow the proof of the analogous fact for \Hminiso from \cref{all-reductions}.
    We first solve \Henciso on $G$ in time $T(m)$, which returns a full $H$-encoding consisting of partial $H$-encodings $\partenc_1, \partenc_2, \ldots, \partenc_k$.
    Recall that $k$ must be constant, i.e., it may depend on $H$ but not on $G$.
    Since the running time is $T(m)$, each partial $H$-encoding has size $\Oh(T(m))$. 
    For each partial encoding $\partenc_i$, we apply \cref{lm:reduce-encoding-to-tree} to build a tree $H'_i$ and a host graph $G'_i$ for pattern $H'_i$.
    We solve \countiso{H'_i} on $G'_i$ in linear time as $H'_i$ is a tree.
    This yields for each $i$ the number of $H'_i$-subgraphs in $G'_i$.
    We return the sum of these values.
    
    Since $\partenc_1, \partenc_2, \ldots, \partenc_k$ is a full $H$-encoding, every $H$-subgraph of $G$ is encoded by exactly one partial encoding $\partenc_i$. By the bijection property of \cref{lm:reduce-encoding-to-tree}, \countiso{H'_i} on $G'_i$ counts exactly the number of $H$-subgraphs of $G$ encoded by $\partenc_i$.
    This proves correctness. 
    The running time of this algorithm is $\Oh(T(m))$ as $k$ is constant.
\end{proof}

\begin{proof}[Proof of \cref{thm:main-result-counting}]
	By \cref{lm:p-graph-upper-bound}, \enciso{P(\alpha_i, \beta_i, \gamma_i \times 2)} can be solved in time $\Oh(m^{2 - 1 / \funcc(\alpha_i, \beta_i, \gamma_i)}) \le \Oh(m^{c(H)})$ for each $i$.
    By \cref{tc-from-D}, \Henciso can be solved in time $\Oh(m^{c(H)})$. By \cref{lm:counting-reduction}, the same holds for \Hcountiso.
\end{proof}

\section{Computing the Decomposition \boldmath$D(H)$} \label{sec:appendix-decomposition-algorithm}

In this section we show that the decomposition $D(H)$ defined in \cref{def:clique-separator-decomp} can be computed in polynomial time (in terms of $|V(H)|$). This is not needed for our main result as in this paper we always assume that $H$ has constant size. In this section we make an exception and do not assume that $H$ has constant size to show this interesting side result.

We build on the following tool due to Whitesides \cite{Whitesides81}.

\begin{lemma}[\cite{Whitesides81}] \label{lm:find-clique-separator}
    There is an algorithm that given a graph $H$, in polynomial time either decides that $H$ does not have a clique separator or computes some clique separator $C$ of $H$.
    \lipicsEnd
\end{lemma}

\begin{lemma} \label{min-clique-separator}
    There is an algorithm that given a graph $H$, in polynomial time either decides that $H$ does not have a clique separator or computes a minimal clique separator $C$ of~$H$.
\end{lemma}

\begin{proof}
    We apply the algorithm from \cref{lm:find-clique-separator}.
    In polynomial time it either decides that $H$ does not have a clique separator or computes some clique separator $\bar{C}$ of $H$. In the first case we are done, and in the second case we still have to transform $\bar{C}$ into a minimal clique separator. To this end, we iteratively try to remove nodes from $\bar{C}$ one by one while $\bar{C}$ is still a clique separator. At the end of this process we get some clique separator $C$, for which $C \setminus \{b\}$ is not a clique separator for every $b \in C$. Such a procedure can clearly be performed in polynomial time. We claim that such $C$ is a minimal clique separator.

    For the sake of contradiction, assume that $C$ is not a minimal clique separator.
    That is, there exists a clique separator $C' \subsetneq C$. All nodes of $C$ are adjacent, so all nodes of $C \setminus C'$ lie in one connected component of $H[V \setminus C']$. Let $a$ be a node from some other connected component of $H[V \setminus C']$ (there are at least two because $C'$ is a clique separator). Let $b$ be some node of $C \setminus C'$. As $b$ and $a$ lie in different connected components of $H[V(H) \setminus (C \setminus \{b\})]$, we have that $C \setminus \{b\}$ is a clique separator in $H$, which leads to a contradiction because we assumed that $C \setminus \{b\}$ is not a clique separator for every $b \in C$.
\end{proof}

The above lemma gives an efficient way of computing $D(H)$: search for a minimal clique separator of $H$, and if it exists, recursively split into smaller graphs, and if it does not exist, $H$ does not have any clique separator and thus $D(H) = \{H\}$. We formalize this procedure in the following lemma.

\begin{lemma} \label{d-set-algo}
    There is a $\poly(|V(H)|)$-time algorithm that computes $D(H)$ for a given graph $H$.
\end{lemma}

\begin{proof}
    We develop a recursive procedure that finds $D(H)$. First, apply Lemma \ref{min-clique-separator}. If there is no clique separator in $H$, return $D(H) \coloneqq \{H\}$ because $V(H)$ is a subset of nodes with no clique separator, and by maximality it is the only element of $D(H)$. Otherwise, Lemma \ref{min-clique-separator} returns some minimal clique separator $C$, and we recurse into $H[S_i \cup C]$ for all $i \in [k]$, where $S_i$ are connected components of $H[V(H) \setminus C]$. Due to Lemma \ref{d-set-computation}, it is sufficient to return the union of all $D(H[S_i \cup C])$ that we get from these recursive calls.

    Set $n \coloneqq |V(H)|$. We now claim that the number of recursive calls in this procedure is polynomial in $n$, and as each terminal call of this recursion returns a single element in its $D$-set, and in the upper levels of recursion the $D$-sets are only united, $D$-sets in all recursive calls have sizes polynomial in $n$, and thus the whole procedure takes $\poly(n)$ time.

    It remains to show that the number of recursive calls is polynomial in $n$. We prove by complete induction on $n$ that the number of non-terminal recursive calls $R(H)$ for a graph $H$ is at most $\overline{E}(H) \coloneqq {n \choose 2} - |E(H)| = \poly(n)$. That is, the number of non-edges in $H$. Consequently, the total number of recursive calls is $\poly(n)$ because each non-terminal recursive call creates at most $n$ terminal recursive calls.
    We prove the step of induction. We assume that the claim is proven for all graphs $H'$ with $|V(H')| \le n - 1$ and prove it for $H$. If $H$ does not have a clique separator, it has $R(H) = 0 \le \overline{E}(H)$ non-terminal calls. Otherwise, $R(H) = 1 + \sum_{i \in [k]} R(H[S_i \cup C])$, where $C$ is some minimal clique separator in $H$, and $S_i$ are connected components of $H[V(H) \setminus C]$. Note that if some pair of nodes $\{a, b\}$ forms a non-edge in some $H[S_i \cup C]$, it also forms a non-edge in $H$.
    Furthermore, any non-edge $\{a, b\}$ from $H$ can be a non-edge in at most one of $H[S_i \cup C]$ because if $\{a, b\} \subseteq S_i \cup C$ and $\{a, b\} \subseteq S_j \cup C$ for some $i \neq j$, then $\{a, b\} \subseteq (S_i \cup C) \cap (S_j \cup C) = C$, which is impossible as $C$ is a clique. Thus, $\overline{E}(H) \ge \sum_{i \in [k]} \overline{E}(H[S_i \cup C])$ follows. Furthermore, for any $a \in S_1$ and $b \in S_2$, $\{a, b\}$ is a non-edge in $H$, but this non-edge is not present in any $H[S_i \cup C]$, thus $\overline{E}(H) \ge 1 + \sum_{i \in [k]} \overline{E}(H[S_i \cup C])$.
    Each of the graphs $H[S_i \cup C]$ has at most $n-1$ nodes, and thus by the induction hypothesis $\overline{E}(H[S_i \cup C]) \ge R(H[S_i \cup C])$ holds for all $i \in [k]$. Consequently, $\overline{E}(H) \ge 1 + \sum_{i \in [k]} \overline{E}(H[S_i \cup C]) \ge 1 + \sum_{i \in [k]} R(H[S_i \cup C]) = R(H)$, thus proving the claim.
\end{proof}

\section{Results for Submodular Width} \label{sec:p-graph-subw}

The problem of subgraph isomorphism is closely connected to the notion of submodular width (or $\subw$ for short), see \cref{app:further-related-work:subw} for a discussion and see \cref{def:subw} for the definition of submodular width.
In this section we prove that submodular width is equal to the exponent of our algorithms for min-weight, listing and enumeration from our main result. This is implied by the following theorem (by comparing it with \cref{full-characterization}).

\renewcommand{\maybe}{\lipicsEnd}
\begin{restatable}{theorem}{subwfullcharacterization} \label{subw-full-characterization}
    Let $H$ be a pattern graph. If there exists $H' \in D(H)$ such that $H'$ is not isomorphic to any graph in $\family$, then $\subw(H) \ge 2$.
    
    Otherwise, every graph in $D(H)$ is isomorphic to a graph in $\family$, so up to isomorphism we have $D(H) = \{P(\alpha_1, \beta_1, \gamma_1 \times 2), \ldots, P(\alpha_k, \beta_k, \gamma_k \times 2) \}$ for some $(\alpha_1, \beta_1, \gamma_1), \ldots, (\alpha_k, \beta_k, \gamma_k) \in \Tfamily$.
    In this case, we have $\subw(H) = \clemb(H) = 2 - 1 / \max_{i \in [k]} \funcc(\alpha_i, \beta_i, \gamma_i)$.\maybe
\end{restatable}
\renewcommand{\maybe}{}

In other words, we classify all graphs $H$ with $\subw(H) < 2$, and for each such graph we determine the exact value of $\subw(H)$.
Our result is the first classification of all graphs of submodular width $w$ for any $w > 1$. 

We remark that submodular width is very difficult to work with, due to its complicated definition. 
While for small, concrete graphs it can be feasible to compute $\subw(H)$ with computer aid, in this work we require bounds on the submodular width of infinite families of graphs, so computer search is not applicable. 
In this section we overcome these difficulties for graphs of sumodular width less than 2, but this requires a long proof with many cases. In fact, even though in this section we prove bounds on the numerical value $\subw(H)$ instead of designing algorithms, our proof follows virtually the same steps as our algorithm from \cref{sec:p-graph-upper-bounds} (and thus is equally long and has equally many cases).

\medskip
Throughout this section, $\alpha, \beta$, and $\gamma$ denote integers with $\alpha \ge \beta \ge 3$ and $\gamma \ge 1$.
We start by stating some general facts about submodular width in \cref{sec:subw-general-case-to-p-graphs}.
After that, the proof structure in this section closely follows the proof structure of \cref{sec:p-graph-upper-bounds}:
In \cref{sec:subw-P-graphs-alpha-gamma} we cover all graphs of the form $P(\alpha, \gamma \times 2)$, in \cref{sec:subw-P-graphs-alpha-beta-gamma} we cover all graphs of the form $P(\alpha, \beta, \gamma \times 2)$, and in \cref{sec:subw-P-graphs-combination} we put the pieces together to prove \cref{subw-full-characterization}.
Finally, as an application of the above theorem in \cref{sec:computing_min_subw_2} we show that the value $\min\{\subw(H), 2\}$ can be computed in polynomial time in terms of $|V(H)|$ for a given pattern $H$.

\subsection{Submodular Width Toolbox} \label{sec:subw-general-case-to-p-graphs}

\begin{definition}[\cite{Marx13}]\label{def:submodularity}
    For a pattern $H$, a function $\subf \colon 2^{V(H)} \to \RR_{\ge 0}$ is \emph{submodular} if $\subf(X) + \subf(Y) \ge \subf(X \cap Y) + \subf(X \cup Y)$ holds for all $X, Y \subseteq V(H)$.
    Furthermore, $\subf$ is \emph{monotone} if $\subf(X) \le \subf(X')$ holds for all $X, X' \subseteq V(H)$ with $X \subseteq X'$.
    Finally, $\subf$ is \emph{edge-dominated} if $\subf(e) \le 1$ holds for every $e \in E(H)$.
    \lipicsEnd
 \end{definition}

 \begin{definition}[\cite{Marx13}] \label{def:subw}
     For a pattern $H$, let $\submodset(H)$ be the set of all submodular monotone edge-dominated functions $\subf \colon 2^{V(H)} \to \RR_{\ge 0}$ such that $\subf(\emptyset) = 0$.

     Furthermore, for a function $\subf \colon 2^{V(H)} \to \RR_{\ge 0}$, let $\fwidth(H)$ be the minimum over all tree decompositions $(T, (B_t)_{t \in V(T)})$ of $H$ of the maximum value of $\subf$ on any bag, \ie, $\max \{\subf(B_t) \mid t \in V(T) \}$.

     The \emph{submodular width} of $H$ is defined as $\subw(H) \coloneqq \sup \{ \fwidth(H) \mid \subf \in \submodset(H) \}$.
     \lipicsEnd
 \end{definition}

 \begin{lemma} \label{lm:induced-minor-subw} 
     Let $H$ and $H'$ be patterns such that $H'$ is an induced minor of $H$.
     Then $\subw(H') \le \subw(H)$.
\end{lemma}

\begin{proof}
    Consider some $\subf' \in \submodset(H')$.
    We show that there is some $\subf \in \submodset(H)$ such that $\funwidth{\subf'}(H') \le \fwidth(H)$.
    By taking the supremum over all $\subf' \in \submodset(H')$, we obtain $\subw(H') \le \subw(H)$.

    As $H'$ is an induced minor of $H$, there exists a function $f \colon V(H') \to 2^{V(H)}$, such that nodes of $H'$ are mapped to disjoint non-empty connected subsets of nodes of $H$, where for any $x \neq y$ we have  $\{x, y\} \in E(H')$ if and only if $\{a, b\} \in E(H)$ for some $a \in f(x)$ and $b \in f(y)$. Pick an arbitrary such function $f$.
    Denote by $f^{-1} \colon V(H) \to V(H') \cup \{\bot\}$ a function that maps each node $a$ of $H$ to the only node $x$ of $H'$, such that $a \in f(x)$, or $\bot$ if no such node $x$ exists.
    As the sets $f(x)$ are disjoint, $f^{-1}$ is uniquely defined.
    Furthermore, for $B \subseteq V(H)$, we write $f^{-1}(B) = \{f^{-1}(a) \mid a \in B, f^{-1}(a) \neq \bot\}$.
    
    We define $\subf \colon 2^{V(H)} \to \RR_{\ge 0}$ as $\subf(B) \coloneqq \subf'(f^{-1}(B))$ for all $B \subseteq V(H)$.
    We first show that $\subf \in \submodset(H)$. For that, we need to show that $\subf(\emptyset) = 0$, $\subf$ is edge-dominated, monotone, and submodular. We prove all these properties:
    \begin{itemize}
        \item \textbf{\boldmath$\subf(\emptyset) = 0$:} We have $\subf(\emptyset) = \subf'(f^{-1}(\emptyset)) = \subf'(\emptyset) = 0$.
        \item \textbf{\boldmath$\subf$ is edge-dominated:} Fix some $\{a, b\} \in E(H)$.
            Consider the set $f^{-1}(\{a, b\})$. This set is either empty, consists of one node of $H'$, or of two nodes of $H'$ connected by an edge (here we use that $H'$ is an \emph{induced} minor of $H$).
            We have $\subf(\{a, b\}) = \subf'(f^{-1}(\{a, b\}))$.
            If $f^{-1}(\{a, b\})$ is empty, we get $\subf(\{a, b\}) = \subf'(\emptyset) = 0 \le 1$.
            If $f^{-1}(\{a, b\})$ consists of two nodes connected by an edge, we get $\subf(\{a, b\}) \le 1$ as $\subf'$ is edge-dominated.
            If $f^{-1}(\{a, b\}) = \{x\}$ for some node $x \in V(H')$, we have $\subf(\{a, b\}) = \subf'(\{x\}) \le \subf'(\{x, y\}) \le 1$ for any node $y$ adjacent to $x$ in $H'$, as $\subf'$ is monotone and edge-dominated.
            Such a node $y$ always exists as $H'$ is a pattern.
        \item \textbf{\boldmath$\subf$ is monotone:} Consider some $X, Y \in V(H)$, such that $X \subseteq Y$.
            We have $f^{-1}(X) \subseteq f^{-1}(Y)$, and thus $\subf(X) = \subf'(f^{-1}(X)) \le \subf'(f^{-1}(Y)) = \subf(Y)$ as $\subf'$ is monotone.
        \item \textbf{\boldmath$\subf$ is submodular:} Consider some $X, Y \in V(H)$.
            Note that $f^{-1}(X \cap Y) \subseteq f^{-1}(X) \cap f^{-1}(Y)$ and $f^{-1}(X \cup Y) = f^{-1}(X) \cup f^{-1}(Y)$, and thus
            \begin{align*}
                \varphi(X) + \varphi(Y) &= \varphi'(f^{-1}(X)) + \varphi'(f^{-1}(Y))\\
                                        &\ge \varphi'(f^{-1}(X) \cap f^{-1}(Y)) + \varphi'(f^{-1}(X) \cup f^{-1}(Y))\\
                                        &\ge \varphi'(f^{-1}(X \cap Y)) + \varphi'(f^{-1}(X \cup Y))\\
                                        &= \varphi(X \cap Y) + \varphi(X \cup Y)
        \end{align*}
            due to monotonicity and submodularity of $\varphi'$.
    \end{itemize}

    This shows $\varphi \in S(H)$.
    It remains to show $\funwidth{\subf'}(H') \le \fwidth(H)$.
    Due to the definition of $\fwidth(H)$, there is some tree decomposition $\mathcal{T}$ of $H$ such that $\subf(B) \le \fwidth(H)$ for all bags $B$ of $\mathcal{T}$.
    We consider a tree decomposition $\mathcal{T}'$ of $H'$ that has the same tree structure as $\mathcal{T}$, but each bag $B$ is replaced with $f^{-1}(B)$.
    It is clear that $\mathcal{T}'$ is a tree decomposition of $H'$ (see proof of \cref{induced-minor-encoding}).
    Furthermore, by the definition of $\subf$, we have $\subf'(f^{-1}(B)) = \subf(B) \le \fwidth(H)$.
    Therefore, $\subf'(B') \le \fwidth(H)$ for all bags $B'$ of $\mathcal{T}'$, and thus $\funwidth{\subf'}(H') \le \fwidth(H)$.
\end{proof}

 \begin{lemma} \label{clique-separator-subw-combination}
    Let $H$ be a pattern and $C$ be a clique separator in $H$. Let $S_1, \ldots, S_k \subset V(H)$ be the connected components of $H - C$.
    We have $\subw(H) = \max \{\subw(H[S_i \cup C]) \mid i \in [k] \}$.
\end{lemma}

\begin{proof}
    As $H[S_i \cup C]$ for each $i \in [k]$ is an induced subgraph of $H$, we have $\subw(H) \ge \subw(H[S_i \cup C])$ due to \cref{lm:induced-minor-subw}.
    Hence, $\subw(H) \ge \max \{\subw(H[S_i \cup C]) \mid i \in [k] \}$.

    We now prove $\subw(H) \le \max \{\subw(H[S_i \cup C]) \mid i \in [k] \}$.
    Fix some $\subf \in \submodset(H)$.
    We show $\fwidth(H) \le \max \{\subw(H[S_i \cup C]) \mid i \in [k] \}$.
    Consider functions $\subf_i \colon 2^{S_i \cup C} \to \RR_{\ge 0}$ for each $i \in [k]$ defined as the restrictions of $\subf$ to $2^{S_i \cup C}$.
    Note that $\subf_i \in \submodset(H[S_i \cup C])$.
    Hence, $\funwidth{\subf_i}(H[S_i \cup C]) \le \subw(H[S_i \cup C])$ by the definition of $\subw$.
    Therefore, there exists a tree decomposition $\mathcal{T}_i$ of $H[S_i \cup C]$, for which $\subf(B) = \subf_i(B) \le \subw(H[S_i \cup C])$ for each bag $B$ of $\mathcal{T}_i$.
    As $C$ is a clique, due to \cref{induced-clique-td}, we have that each $\mathcal{T}_i$ has a bag containing $C$.
    If we pick some such bag in $\mathcal{T}_1$ and connect it by edges with any such bags in each $\mathcal{T}_i$ for $i \in [2, k]$, we obtain a tree decomposition $\mathcal{T}$ of $H$.
    Then, for each bag $B$ of $\mathcal{T}$, we have $\subf(B) \le \max \{\subw(H[S_i \cup C]) \mid i \in [k] \}$.
    Hence, $\fwidth(H) \le \max \{\subw(H[S_i \cup C]) \mid i \in [k] \}$.
    By taking the supremum over $\varphi \in S(H)$, we obtain $\subw(H) \le \max \{\subw(H[S_i \cup C]) \mid i \in [k] \}$.
\end{proof}

\begin{lemma} \label{subw-from-D}
    For any pattern $H$, we have $\subw(H) = \max\{\subw(H') \mid H' \in D(H)\}$.
\end{lemma}

\begin{proof}
    As each $H' \in D(H)$ is an induced minor of $H$, we have $\subw(H) \ge \subw(H')$ due to \cref{lm:induced-minor-subw}.
    Hence, $\subw(H) \ge \max\{\subw(H') \mid H' \in D(H)\}$.

    We now prove $\subw(H) \le \max\{\subw(H') \mid H' \in D(H)\}$.
    We prove this inequality by induction on $|V(H)|$.
    If $H$ does not have a clique separator, then $D(H) = \{H\}$, and the claim is trivial.
    Otherwise, there is some minimal clique separator $C$ in $H$.
    Let $S_1, S_2, \ldots, S_k$ be the connected components of $H - C$.
    Note that by \cref{d-set-computation}, we have $D(H[S_i \cup C]) \subseteq D(H)$ for each $i$.
    Therefore, by the induction hypothesis, we have $\subw(H[S_i \cup C]) \le \max\{\subw(H') \mid H' \in D(H[S_i \cup C])\} \le \max\{\subw(H') \mid H' \in D(H)\}$.
    Applying \cref{clique-separator-subw-combination}, we get $\subw(H) = \max\{ \subw(H[S_i \cup C]) \mid i \in [k]\} \le \max\{\subw(H') \mid H' \in D(H)\}$.
\end{proof}

We use the following results from the literature on submodular width.

\begin{lemma}[{\cite[Corollary 9]{FanKZ23}}] \label{lm:clemb-bounded-by-subw}
    For a pattern $H$, we have $\clemb(H) \le \subw(H)$.
    \lipicsEnd
\end{lemma}

\begin{lemma}[{\cite[Section E.6]{JoglekarR18}}] \label{lm:p333-subw}
    We have $\subw(P(3, 3, 3)) \ge 2$.
    \lipicsEnd
\end{lemma}

\begin{lemma}[{\cite[Proposition 18]{FanKZ23}}] \label{cycle-subw}
For $k \ge 3$, we have $\subw(C_k) \le 2 - 1 / \left\lceil \frac{k}{2} \right\rceil$.
\lipicsEnd
\end{lemma}

\begin{lemma}[{\cite[Proposition 19]{FanKZ23}}] \label{biclique-subw}
    For $k \ge 2$, we have $\subw(K_{2, k}) \le 2 - 1 / k$.
    \lipicsEnd
\end{lemma}

\begin{lemma} \label{edge-subw}
    We have $\subw(K_2) \le 1$.
\end{lemma}

\begin{proof}
    Consider a tree decomposition of $K_2$ consisting of a single bag $B \coloneqq V(K_2)$.
    For any $\subf \in \submodset(K_2)$, we have $\subf(B) \le 1$ as $\subf$ is edge-dominated.
    Hence, $\fwidth(K_2) \le 1$, and thus $\subw(K_2) \le 1$.
\end{proof}

\subsection{Submodular Width of $P(\alpha, \gamma \times 2)$}
\label{sec:subw-P-graphs-alpha-gamma}

\begin{definition}
    Let $H$ be a pattern and $\subf \colon 2^{V(H)} \to \RR_{\ge 0}$ be a monotone edge-dominated function.
    For each $a \in V(H)$, we define $\subdeg^{\subf}_a \coloneqq 1 - \subf(\{a\})$.
    When the function $\subf$ is clear from the context, we simply write $\subdeg_a$.
    \lipicsEnd
\end{definition}

Note that $\subdeg^{\subf}_a \ge 0$ as $\subf(\{a\}) \le \subf(\{a, b\}) \le 1$ for any $b$ adjacent to $a$ in $H$, as $\subf$ is monotone and edge-dominated.
Such a node $b$ always exists as $H$ is a pattern.

\begin{lemma} \label{lm:subw-simple-materializationsi}
    Let $H$ be a pattern and  $\subf \in \submodset(H)$.
    For any $V_1, V_2 \subseteq V(H)$, we have $\subf(V_1 \cup V_2) \le \subf(V_1) + \subf(V_2)$.
\end{lemma}

\begin{proof}
   As $\subf$ is submodular, we have
   $\subf(V_1) + \subf(V_2) \ge \subf(V_1 \cup V_2) + \subf(V_1 \cap V_2) \ge \subf(V_1 \cup V_2).$
\end{proof}

\begin{lemma} \label{lm:subw-simple-materializationsii}
    Let $H$ be a pattern and  $\subf \in \submodset(H)$.
    For any $V \subseteq V(H)$, $a \in V$, and $b \in V(H) \setminus V$ with $\{a, b\} \in E(H)$, we have $\subf(V \cup \{b\}) \le \subf(V) + \subdeg_a$.
\end{lemma}

\begin{proof}
   As $\subf$ is submodular and edge-dominated, we have
   $\subf(V) + 1 \ge \subf(V) + \subf(\{a, b\}) \ge \subf(V \cup \{a, b\}) + \subf(V \cap \{a, b\}) = \subf(V \cup \{b\}) + \subf(\{a\}).$
   Therefore,
   $\subf(V \cup \{b\}) \le \subf(V) + (1 - \subf(\{a\})) = \subf(V) + \subdeg_a.$
\end{proof}

We thus can interpret $\subdeg^{\subf}_a$ as an analogue of the degree in \cref{sec:p-graph-upper-bounds}.

\begin{lemma}[Biased Cycle] \label{lm:biased-cycle-subw}
    Let $H$ be a $(k+1)$-cycle on nodes $d_0, d_1, \ldots, d_{k}$.
    Let $\subf \in \submodset((V(H), E(H) \setminus \{\{d_0, d_k\}\}))$ be a function such that $\subf(\{d_0, d_k\}) \le 2 - \ell f$ for some $\ell \in \NN$ and $f \in (0, 1)$. We call such a pair $(\varphi, H)$ a \emph{Biased Cycle instance}. Then the following bounds on $\fwidth(H)$ hold.

    \begin{itemize}
        \item We have $\fwidth(H) \le \max\{2 - f, 1 + \subdeg_{\ell \dd k - \ell} + (\ell - 1) f\}$ if $\ell \le \frac{k + 1}{2}$, where $\subdeg_{x \dd y} \coloneqq \sum_{x \le i \le y} \subdeg_{d_i}$ for any $x, y \in [0, k]$.\footnote{Note that if $\ell = \frac{k + 1}{2}$, we have $\subdeg_{\ell \dd k - \ell} = \subdeg_{\ell \dd \ell - 1} = 0$.}
        \item We have $\fwidth(H) \le 2 - f$ if $\ell > \frac{k + 1}{2}$.
    \end{itemize}
    Furthermore, if $\ell \le \frac{k + 1}{2}$ and $\subdeg_r \ge f$ for some $r \in \{0, \ldots, \ell - 1, k - \ell + 1, \ldots, k\}$, we have $\fwidth(H) \le 2 - f$.
    We call such an instance of Biased Cycle \emph{$\ell$-reachable}.

\end{lemma}

\begin{proof}
    We consider four cases depending on the given values $\subdeg$.

    \begin{itemize}
        \item \emph{Case 1: $\subdeg_r \ge f$ for some $r \in [0, \min\{\ell - 1, k - 1\}]$.}
            Choose minimal such $r$. Consider the following tree decomposition of $H$.
            Bag $B_0$ consists of nodes $d_k, d_0, d_1, \ldots, d_r$.
            Furthermore, we create bags $B_i \coloneqq \{d_r, d_i, d_{i + 1}\}$ for $i \in [r + 1, k - 1]$.
            We connect them in a line: $B_0 \edg B_{k - 1} \edg B_{k - 2} \edg \cdots \edg B_{r + 1}$.
            It is clear that this is a tree decomposition of $H$.
            For each $i \in [r + 1, k - 1]$, we have $\subf(B_i) \le \subf(\{d_r\}) + \subf(\{d_i, d_{i + 1}\}) \le (1 - \subdeg_r) + 1 \le 2 - f$, due to \cref{lm:subw-simple-materializationsi} and as $\subf$ is edge-dominated.
            As $r$ is the minimal index such that $\subdeg_r \ge f$, we have for each $i \in [r]$
            \begin{align*}
                \subf(\{d_k, d_0, d_1, d_2, \ldots, d_{i}\}) &\le \subf(\{d_k, d_0, d_1, d_2, \ldots, d_{i - 1}\}) + \subdeg_{i - 1}\\
                                                             &\le \subf(\{d_k, d_0, d_1, d_2, \ldots, d_{i - 1}\}) + f
            \end{align*}
            due to \cref{lm:subw-simple-materializationsii}.
            As $\subf(\{d_0, d_k\}) \le 2 - \ell f$ by the lemma statement, by induction we arrive at $\subf(B_0) \le 2 - \ell f + r f \le 2 - f$ as $r \le \ell - 1$.
            Therefore, $\fwidth(H) \le 2 - f$.
        \item \emph{Case 1': $\subdeg_r < f$ for all $r \in [0, \min\{\ell - 1, k - 1\}]$ and $\subdeg_r \ge f$ for some $r \in [\max\{k - \ell + 1, 1\}, k]$.} This case is symmetric to the previous one.
        \item \emph{Case 2: $\subdeg_r < f$ for all $r \in [0, \min\{\ell - 1, k - 1\}] \cup [\max\{k - \ell + 1, 1\}, k]$ and $\ell > \frac{k + 1}{2}$.}
            In this case $k \le 2 \ell - 2$. Hence, $\subdeg_r < f$ for all $r \in [0, k]$.
            We consider a tree decomposition of $H$ consisting of two bags: $B_1 \coloneqq \{d_k, d_0, d_1, \ldots, d_{\min\{\ell - 1, k - 1\}}\}$ and $B_2 \coloneqq \{d_0, d_k, d_{k - 1}, \ldots, d_{\max\{k - \ell + 1, 1\}}\}$.
            It is easy to see that it is a tree decomposition of $H$.
            Similarly to Case 1, we arrive at $\subf(B_1) \le 2 - \ell f + \min\{\ell - 1, k - 1\} f \le 2 - f$.
            Analogously, $\subf(B_2) \le 2 - f$.
            Therefore, $\fwidth(H) \le 2 - f$.
        \item \emph{Case 3: $\subdeg_r < f$ for all $r \in [0, \min\{\ell - 1, k - 1\}] \cup [\max\{k - \ell + 1, 1\}, k]$ and $\ell \le \frac{k + 1}{2}$.}
            In this case $\min\{\ell - 1, k - 1\} = \ell - 1$ and $\max\{k - \ell + 1, 1\} = k - \ell + 1$.
            Furthermore, $\ell - 1 < k - \ell + 1$.
            We consider a tree decomposition of $H$ consisting of two bags: $B_1 \coloneqq \{d_k, d_0, d_1, \ldots, d_{\ell - 1}\}$ and $B_2 \coloneqq \{d_{\ell - 1}, d_{\ell}, \ldots, d_k\}$.
            It is easy to see that it is a tree decomposition of $H$.
            The same way as in Case 2, we have $\subf(B_1) \le 2 - f$.

            Furthermore, for each $i \in [\ell - 1, k - 2]$, we have
            \[\subf(\{d_i, d_{i + 1}, \ldots, d_k\}) \le \subf(\{d_{i + 1}, d_{i + 2}, \ldots, d_k\}) + \subdeg_{i + 1} \]
            due to \cref{lm:subw-simple-materializationsii}.
            As $\subf$ is edge-dominated, we have $\subf(\{d_{k - 1}, d_k\}) \le 1$.
            Therefore, by induction we arrive at
            \[\subf(B_2) \le 1 + \subdeg_{\ell \dd k - 1} \le 1 + \subdeg_{\ell \dd k - \ell} + (\ell - 1) f.\]
            Hence, we have $\fwidth(H) \le \max\{2 - f, 1 + \subdeg_{\ell \dd k - \ell} + (\ell - 1) f\}$.
    \end{itemize}

    Note that if $\ell > \frac{k + 1}{2}$, only Cases 1, 1', and 2 appear, and if $(H, \subf)$ is an $\ell$-reachable instance of Biased Cycle, only Cases 1 and 1' appears. In both situation we obtain $\fwidth(H) \le 2 - f$.
\end{proof}

\begin{lemma} \label{lm:pa2c-subw}
    For any $\alpha \ge 3$, $\gamma \ge 2$, and $H = P(\alpha, \gamma \times 2)$, we have $\subw(H) \le 2 - 1 / (2 \gamma + \left\lfloor \frac{\alpha - 1}{2} \right\rfloor)$.
\end{lemma}

\begin{proof}
    We prove the lemma for even values of $\alpha$, because for odd values of $\alpha$ the right-hand side of the desired inequality for $P(\alpha, \gamma \times 2)$ is the same as for $P(\alpha + 1, \gamma \times 2)$, and thus by \cref{lm:induced-minor-subw}, the claim follows from the even case.

    Define $\ell \coloneqq \frac{\alpha - 2}{2}$ (so $\alpha = 2\ell + 2$) and $f \coloneqq \frac{1}{2 \gamma + \ell}$.
    We need to show $\subw(H) \le 2 - f = 1 + (2\gamma + \ell - 1) f$.
    Fix any $\subf \in \submodset(H)$; we will show $\fwidth(H) \le 2 - f$.
    By taking the supremum over all such $\subf$, we thus prove $\subw(H) \le 2 - f$.

    \medskip

    We consider four cases.
    \begin{itemize}
        \item \emph{Case 1: $\subdeg_{a_0} \ge f$.}
            Since $H - \{a_0\}$ is a tree, it has a tree decomposition where every bag consists of exactly two adjacent nodes.
            Adding $a_0$ to all bags, we obtain a tree decomposition of $H$.
            For each of its bags $\{a_0\} \cup e$, we have $\subf(\{a_0\} \cup e) \le \subf(\{a_0\}) + \subf(e) \le 1 - \subdeg_{a_0} + 1 \le 2 - f$ due to \cref{lm:subw-simple-materializationsi} and as $\subf$ is edge-dominated.
            Therefore, $\fwidth(H) \le 2 - f$.
        \item \emph{Case 1': $\subdeg_{a_0} < f$ and $\subdeg_{a_{\alpha}} \ge f$.}
            This case is symmetric to the previous one.
        \item \emph{Case 2: $\subdeg_{a_0}, \subdeg_{a_1}, \ldots, \subdeg_{a_{\ell}}, \subdeg_{a_{\alpha - \ell}}, \ldots, \subdeg_{a_{\alpha - 1}}, \subdeg_{a_{\alpha}} < f$.}
            As $\alpha - \ell = \ell + 2$, we have $\subdeg_{a_i} < f$ for all $i \in [0, \alpha] \setminus \{\ell + 1\}$.
            We consider a tree decomposition consisting of two adjacent bags: $B_1 \coloneqq \{a_0, a_1, \ldots, a_{\ell + 1}, c_1, c_2, \ldots, c_{\gamma}\}$ and $B_2 \coloneqq \{a_{\ell + 1}, a_{\ell + 2}, \ldots, a_{\alpha}, c_1, c_2,\allowbreak \ldots, c_{\gamma}\}$.
        It is easy to see that it is a tree decomposition of $H$ (see \cref{Pac-upper-bound-2}).
        To bound $\subf(B_1)$, we start from $\subf(\{a_0, a_1\}) \le 1$ as $\subf$ is edge-dominated.
        For each $c_i$ for $i \in [\gamma]$, we have $\subf(\{a_0, a_1, c_1, c_2, \ldots, c_i\}) \le \subf(\{a_0, a_1, c_1, c_2, \ldots, c_{i - 1}\}) + \subdeg_{a_{0}} \le \subf(\{a_0, a_1, c_1, c_2, \ldots, c_{i - 1}\}) + f$ due to \cref{lm:subw-simple-materializationsii}.
        Analogously, for each $i \in [2, \ell+1]$, we have
        \begin{align*}
            \subf(\{a_0, a_1, \ldots, a_i, c_1, c_2, \ldots, c_{\gamma}\}) &\le \subf(\{a_0, a_1, \ldots, a_{i - 1}, c_1, c_2, \ldots, c_{\gamma}\}) + \subdeg_{a_{i - 1}}\\
                                                        &\le \subf(\{a_0, a_1, \ldots, a_{i - 1}, c_1, c_2, \ldots, c_{\gamma}\}) + f.
    \end{align*}
        We thus obtain by induction $\subf(B_1) \le 1 + \gamma f + \ell f \le 1 + (2 \gamma + \ell - 1) f = 2 - f$.
        Analogously, we have $\subf(B_2) \le 2 - f$.
            Therefore, $\fwidth(H) \le 2 - f$.
        \item \emph{Case 3: $\subdeg_{a_0}, \subdeg_{a_{\alpha}} < f$, and at least one of $\subdeg_{a_1}, \ldots, \subdeg_{a_{\ell}}, \subdeg_{a_{\alpha - \ell}}, \ldots, \subdeg_{a_{\alpha - 1}}$ is greater or equal to $f$.}
            We first bound $\subf(B_{a_0 \to c})$ for $B_{a_0 \to c} \coloneqq \{a_0, c_1, c_2, \ldots, c_{\gamma}\}$.
            We start from $\subf(\{a_0, c_1\}) \le 1$ as $\subf$ is edge-dominated.
            For each $c_i$ for $i \in [2, \gamma]$, we have $\subf(\{a_0, c_1, c_2, \ldots, c_i\}) \le \subf(\{a_0, c_1, c_2, \ldots, c_{i - 1}\}) + \subdeg_{a_{0}} \le \subf(\{a_0, c_1, c_2, \ldots, c_{i - 1}\}) + f$ due to \cref{lm:subw-simple-materializationsii}.
            We thus obtain by induction $\subf(B_{a_0 \to c}) \le 1 + (\gamma - 1) f$.
            Analogously, we have $\subf(B_{a_{\alpha} \to c}) \le 1 + (\gamma - 1) f$ for $B_{a_{\alpha} \to c} \coloneqq \{a_{\alpha}, c_1, c_2, \ldots, c_{\gamma}\}$.
            As $\subf$ is submodular, we have
            \begin{align*}
                2 + (2 \gamma - 2) f &\ge \subf(B_{a_0 \to c}) + \subf(B_{a_{\alpha} \to c})\\
                                     &\ge \subf(B_{a_0 \to c} \cup B_{a_{\alpha} \to c}) + \subf(B_{a_0 \to c} \cap B_{a_{\alpha} \to c})\\
                                     &= \subf(B_{\textup{down}}) + \subf(\{c_1, c_2, \ldots, c_{\gamma}\}),
        \end{align*}
            where $B_{\textup{down}} \coloneqq \{a_0, a_{\alpha}, c_1, c_2, \ldots, c_{\gamma}\}$.
            We distinguish two cases.
    \begin{itemize}
        \item \emph{Case 3.1: $\subf(\{c_1, c_2, \ldots, c_{\gamma}\}) \le 1 - f$.}
            Consider the graph $H - \{c_1, \ldots, c_{\gamma}\}$.
            Since it is a tree, it has a tree decomposition where each bag consists of two adjacent nodes.
            Adding $c_1, c_2, \ldots, c_{\gamma}$ to every bag, we obtain a tree decomposition of $H$ (see \cref{Pac-upper-bound-1}).
            For each of its bags $e \cup \{c_1, c_2, \ldots, c_{\gamma}\}$ we have $\subf(e \cup \{c_1, c_2, \ldots, c_{\gamma}\}) \le \subf(e) + \subf(\{c_1, c_2, \ldots, c_{\gamma}\}) \le 1 + 1 - f = 2 - f$, due to \cref{lm:subw-simple-materializationsi} and as $\subf$ is edge-dominated.
            Therefore, $\fwidth(H) \le 2 - f$.
        \item \emph{Case 3.2: $\subf(\{c_1, c_2, \ldots, c_{\gamma}\}) > 1 - f$.}
            As $\subf(B_{\textup{down}}) + \subf(\{c_1, c_2, \ldots, c_{\gamma}\}) \le 2 + (2 \gamma - 2) f$, we obtain $\subf(B_{\textup{down}}) \le 1 + (2 \gamma - 1) f = 2 - (\ell + 1) f$.
            Furthermore, as $\subf$ is monotone, we have $\subf(\{a_0, a_{\alpha}\}) \le \subf(B_{\textup{down}}) \le 2 - (\ell + 1) f$.
            Define a graph $H' \coloneqq (\{a_0, a_1, \ldots, a_{\alpha}\}, E(H[\{a_0, a_1, \ldots, a_{\alpha}\}]) \cup \{\{a_0, a_{\alpha}\}\})$.
We consider $H'$ as an instance of Biased Cycle for $\subf'$ defined as the restriction of $\subf$ to $2^{\{a_0, a_1, \ldots, a_{\alpha}\}}$.
            Since at least one of $\subdeg_{a_1}, \ldots, \subdeg_{a_{\ell}},\allowbreak \subdeg_{a_{\alpha - \ell}}, \ldots, \subdeg_{a_{\alpha - 1}}$ is greater or equal to $f$, this instance of Biased Cycle is $(\ell + 1)$-reachable.
            Thus, we can apply \cref{lm:biased-cycle-subw} for the parameter $\ell + 1$ and derive that there is a tree decomposition of $H[\{a_0, a_1, \ldots, a_{\alpha}\}]$, where one of the bags contains both $a_0$ and $a_{\alpha}$ such that $\subf(B) \le 2 - f$ for all bags $B$ of this tree decomposition.
            By attaching $B_{\textup{down}}$ to a bag of this tree decomposition that contains $a_0$ and $a_{\alpha}$, we get a tree decomposition of $H$, such that $\subf(B) \le 2 - f$ for all bags $B$ of this tree decomposition.
            Therefore, $\fwidth(H) \le 2 - f$.
    \end{itemize}
    \end{itemize}
\end{proof}

\subsection{Submodular Width of $P(\alpha, \beta, \gamma \times 2)$}
\label{sec:subw-P-graphs-alpha-beta-gamma}

\begin{lemma} \label{lm:pabc-subw}
    For any $\alpha \ge \beta \ge 3$, $\gamma \ge 1$, and $H = P(\alpha, \beta, \gamma \times 2)$, we have $\subw(H) \le 2 - 1 /\fabc$.
    \lipicsEnd
\end{lemma}

\cref{sec:subw-P-graphs-alpha-beta-gamma} is devoted to the proof of this lemma.
    Let $f \coloneqq \frac{1}{\funcc(\alpha, \beta, \gamma)}$.
    Consider some $\subf \in \submodset(H)$.
    We show $\fwidth(H) \le 2 - f = 1 + (\fabc - 1) f$.
    By taking the supremum over all such $\subf$, we thus prove $\subw(H) \le 2 - f$.

    \medskip

    Denote $\submax \coloneqq \max \{ \subdeg_d \mid d \in \{a_i \mid i \in [0, \alpha] \} \cup \{b_i \mid i \in [\beta - 1]\} \}$ and $\vmax \coloneqq \argmax \{ \subdeg_d \mid d \in \{a_i \mid i \in [0, \alpha] \} \cup \{b_i \mid i \in [\beta - 1]\} \}$.
    Furthermore, we denote $\subdeg_{a_{x \dd y}} \coloneqq \sum_{x \le i \le y} \subdeg_{a_i}$ for any $x, y \in [0, \alpha]$ and $\subdeg_{b_{x \dd y}} \coloneqq \sum_{x \le i \le y} \subdeg_{b_i}$ for any $x, y \in [0, \beta]$.

    We first consider some simple special cases.
    
\subsubsection{Special Cases}

We consider nine special cases. The precise condition of each case includes the negations of all previous cases, but we omit them for conciseness.

\paragraph*{Special Case 1: \boldmath$\subdeg_{a_0} > f$.}

Since $H - \{a_0\}$ is a tree, it has a tree decomposition where every bag consists of exactly two adjacent nodes.
Adding $a_0$ to all these bags, we obtain a tree decomposition of $H$.
For each bag $e \cup \{a_0\}$ of this tree decomposition, we have $\subf(e \cup \{a_0\}) \le \subf(e) + \subf(\{a_0\}) \le 1 + (1-\subdeg_{a_0}) \le 1 + (1 - f) = 2 - f$, due to \cref{lm:subw-simple-materializationsi} and as $\subf$ is edge-dominated.
Therefore, $\fwidth(H) \le 2 - f$.

\paragraph*{Special Case 1': \boldmath$\subdeg_{a_{\alpha}} > f$.}

This case is symmetric to the previous one.

\smallskip

Hence, from now on we may assume $\subdeg_{a_0}, \subdeg_{a_{\alpha}} \le f$.

\paragraph*{Special Case 2: \boldmath$\submax \le (2 \gamma + 1) f$ and $\alpha = \beta$.}

We consider a tree decomposition of $H$ consisting of two adjacent bags:
\[B_1 \coloneqq \{a_0, a_1, \ldots, a_{\alpha - 1}, b_1, c_1, c_2, \ldots, c_{\gamma}\}\]
and
\[B_2 \coloneqq \{a_{\alpha - 1}, b_1, b_2, \ldots, b_{\beta - 1}, b_{\beta}, c_1, c_2, \ldots, c_{\gamma}\}.\]
It is easy to see that it is indeed a valid tree decomposition of $H$ (see \cref{Paac-upper-bound-all-small-degrees}).

To bound $\subf(B_1)$, we first note that $\subf(\{a_0, b_1\}) \le 1$ as $\subf$ is edge-dominated.
Furthermore, for each $i \in [\gamma]$, we have
\[\subf(\{a_0, c_1, \ldots, c_i, b_1\}) \le \subf(\{a_0, c_1, \ldots, c_{i - 1}, b_1\}) + \subdeg_{a_{0}} \le \subf(\{a_0, c_1, \ldots, c_{i - 1}, b_1\}) + f\]
due to \cref{lm:subw-simple-materializationsii}.
Finally, for each $i \in [\alpha - 1]$, we have
\[\subf(\{a_0, a_1, \ldots, a_i, b_1, c_1, \ldots, c_{\gamma}\}) \le \subf(\{a_0, a_1, \ldots, a_{i - 1}, b_1, c_1, \ldots, c_{\gamma}\}) + \subdeg_{a_{i - 1}}\]
due to \cref{lm:subw-simple-materializationsii}.
Therefore, by induction, we obtain $\subf(B_1) \le 1 + \gamma f + \subdeg_{a_{0 \dd \alpha - 2}} \le 1 + (\gamma + 1) f + (\alpha - 2) \submax \le 1 + (\gamma + 1 + (\alpha - 2) \cdot (2 \gamma + 1)) f = 1 + ((2\alpha - 3) \gamma + \alpha - 1) f \le 2 - f$, where the last inequality holds due to \cref{a-a-c-linear-bound}.
Analogously, we get $\subf(B_2) \le 2 - f$ since $\alpha = \beta$.
Therefore, $\fwidth(H) \le 2 - f$.

\paragraph*{Special Case 3: \boldmath$\submax \le (2 \gamma + 1) f$ and $\alpha > \beta$.}

We create a tree decomposition of $H$ as follows (see \cref{Pabc-upper-bound-all-small-degrees}).
Let $B_b \coloneqq \{b_0, b_1, \ldots, b_{\beta}\}$.
To bound $\subf(B_b)$, we first note that $\subf(\{b_0, b_1\}) \le 1$ as $\subf$ is edge-dominated.
Furthermore, for each $i \in [2, \beta]$, we have
\[ \subf(\{b_0, b_1, \ldots, b_i\}) \le \subf(\{b_0, b_1, \ldots, b_{i - 1}\}) + \subdeg_{b_{i - 1}} \]
due to \cref{lm:subw-simple-materializationsii}.
Therefore, by induction, we obtain $\subf(B_b) \le 1 + \subdeg_{b_{1 \dd \beta - 1}} \le 1 + (\beta - 1) \submax \le 1 + (\beta - 1) \cdot (2 \gamma + 1) f = 2 - \ell f$, where $\ell \coloneqq \funcc(\alpha, \beta, \gamma) - (\beta - 1) \cdot (2\gamma + 1)$.
Note that $\ell \ge 2$ due to \cref{a-b-c-linear-bound}.

    Furthermore, we create bags $B_{c_i} \coloneqq B_b \cup \{c_i\}$ for all $i \in [\gamma]$.
    Due to \cref{lm:subw-simple-materializationsii}, we have $\subf(B_{c_i}) \le \subf(B_b) + \subdeg_{b_0} \le 2 - \ell f + f \le 2 - f$.
    Moreover, as $\subf$ is monotone, we have $\subf(\{a_0, a_{\alpha}\}) = \subf(\{b_0, b_{\beta}\}) \le \subf(B_b) \le 2 - \ell f$.
    
Define a graph $H' \coloneqq (\{a_0, a_1, \ldots, a_{\alpha}\}, E(H[\{a_0, a_1, \ldots, a_{\alpha}\}]) \cup \{\{a_0, a_{\alpha}\}\})$.
We consider $H'$ as an instance of Biased Cycle for $\subf'$ defined as the restriction of $\subf$ to $2^{\{a_0, a_1, \ldots, a_{\alpha}\}}$.
    We claim that $\funwidth{\subf'}(H') \le 2 - f$.
    If $\ell > \frac{\alpha + 1}{2}$, the inequality $\funwidth{\subf'}(H') \le 2 - f$ holds due to \cref{lm:biased-cycle-subw}.
    Otherwise, if $\ell \le \frac{\alpha + 1}{2}$, due to \cref{lm:biased-cycle-subw}, we have
    \begin{align*}
        \funwidth{\subf'}(H') &\le \max\{2 - f, 1 + \subdeg_{a_{\ell \dd \alpha - \ell}} + (\ell - 1) f\}\\
                              &\le \max\{2 - f, 1 + (\alpha - 2 \ell + 1) \cdot \submax + (\ell - 1) f\}\\
                              &\le \max\{2 - f, 1 + (\alpha - 2 \ell + 1) \cdot (2 \gamma + 1) f + (\ell - 1) f\}\\
                              &= 2 - f,
\end{align*}
where the last equality holds due to \cref{a-b-c-linear-bound-2}.
As $\funwidth{\subf'}(H') \le 2 - f$, there is a tree decomposition of $H[\{a_0, a_1, \ldots, a_{\alpha}\}]$, where one of the bags contains both $a_0$ and $a_{\alpha}$ such that $\subf(B) = \subf'(B) \le 2 - f$ for all bags $B$ of this tree decomposition.
By attaching $B_b, B_{c_1}, \ldots, B_{c_{\gamma}}$ in a line to a bag of this tree decomposition that contains $a_0$ and $a_{\alpha}$, we obtain a tree decomposition of $H$ such that $\subf(B) \le 2 - f$ for all bags $B$ of this tree decomposition.
Therefore, $\fwidth(H) \le 2 - f$.

\medskip

Hence, from now on we may assume $\submax > (2 \gamma + 1) f$.

\paragraph*{Special Case 4: for some positive integer \boldmath$\ell \le \alpha - 1$, we have $\subdeg_{a_1}, \subdeg_{a_2}, \ldots, \subdeg_{a_{\ell - 1}} \le f$, $\subdeg_{a_{\ell}} > (\gamma + 1) f$, and $\subdeg_{b_{\beta - 1}} > (\gamma + \ell) f$.} \label{prg:case4subw}

We create the following tree decomposition of $H$ (see \cref{Pabc-upper-bound-split-by-middle-node}).
Define a bag
\[B_{1} \coloneqq \{ a_0, a_1, \ldots, a_{\ell}, b_{\beta-1}, c_1, c_2, \ldots, c_{\gamma} \},\]
a bag
\[B_{2} \coloneqq \{ a_{\ell}, b_{\beta-1}, b_{\beta}, c_1, c_2, \ldots, c_{\gamma} \},\]
bags $B_{b_i} \coloneqq \{b_{i-1}, b_i, b_{\beta-1}\}$ for $i \in [\beta-2]$, and bags $B_{a_i} \coloneqq \{a_{\ell}, a_i, a_{i + 1}\}$ for $i \in [\ell + 1, \alpha - 1]$.
We connect these bags in a line: $B_{b_{\beta-2}} \edg B_{b_{\beta-3}} \edg \cdots \edg B_{b_1} \edg B_{1} \edg B_{2} \edg B_{a_{\alpha-1}} \edg B_{a_{\alpha-2}} \edg \cdots \edg B_{a_{\ell+1}}$.
It is easy to see that it is indeed a valid tree decomposition of $H$.

To bound $\subf(B_1)$, we first note that $\subf(\{a_0, a_1\}) \le 1$ as $\subf$ is edge-dominated.
Furthermore, for each $i \in [\gamma]$, we have
\[\subf(\{a_0, a_1, c_1, \ldots, c_i\}) \le \subf(\{a_0, a_1, c_1, \ldots, c_{i - 1}\}) + \subdeg_{a_{0}} \le \subf(\{a_0, a_1, c_1, \ldots, c_{i - 1}\}) + f\]
due to \cref{lm:subw-simple-materializationsii}.
Moreover, for each $i \in [2, \ell]$, we have
\[\subf(\{a_0, a_1, \ldots, a_i, c_1, \ldots, c_{\gamma}\}) \le \subf(\{a_0, a_1, \ldots, a_{i - 1}, c_1, \ldots, c_{\gamma}\}) + \subdeg_{a_{i - 1}}\]
due to \cref{lm:subw-simple-materializationsii}.
Therefore, by induction and \cref{lm:subw-simple-materializationsi}, we obtain $\subf(B_1) \le \subf(\{a_0, a_1, \ldots, a_{\ell}, c_1, \ldots, c_{\gamma}\}) + \subf(\{b_{\beta - 1}\}) \le (1 + \gamma f + \subdeg_{a_{1 \dd \ell - 1}}) + (1 - \subdeg_{b_{\beta - 1}}) \le (1 + (\gamma + \ell - 1) f) + (1 - (\gamma + \ell) f) = 2 - f$.

To bound $\subf(B_2)$, we first note that $\subf(\{b_{\beta - 1}, b_{\beta}\}) \le 1$ as $\subf$ is edge-dominated.
Furthermore, for each $i \in [\gamma]$, we have
\[\subf(\{b_{\beta - 1}, b_{\beta}, c_1, \ldots, c_i\}) \le \subf(\{b_{\beta - 1}, b_{\beta}, c_1, \ldots, c_{i - 1}\}) + \subdeg_{b_{\beta}} \le \subf(\{b_{\beta - 1}, b_{\beta}, c_1, \ldots, c_{i - 1}\}) + f\]
due to \cref{lm:subw-simple-materializationsii}.
Therefore, by induction and \cref{lm:subw-simple-materializationsi}, we obtain $\subf(B_2) \le \subf(\{ b_{\beta-1}, b_{\beta}, c_1, c_2, \ldots, c_{\gamma} \}) + \subf(\{a_{\ell}\}) \le (1 + \gamma f) + (1 - \subdeg_{a_{\ell}}) \le (1 + \gamma f) + (1 - (\gamma + 1) f) = 2 - f$.

For each bag $B_{b_i}$, for $i \in [\beta - 2]$, we have $\subf(B_{b_i}) \le \subf(\{b_{i - 1}, b_i\}) + \subf(\{b_{\beta - 1}\}) \le 1 + (1 - \subdeg_{b_{\beta - 1}}) \le 2 - f$ due to \cref{lm:subw-simple-materializationsi} and as $\subf$ is edge-dominated.
Analogously, we have $\subf(B_{a_i}) \le 2 - f$ for $i \in [\ell + 1, \alpha - 1]$.
Therefore, we arrive at $\fwidth(H) \le 2 - f$.

\paragraph*{Special Case 4': for some positive integer \boldmath$\ell \le \alpha - 1$, we have $\subdeg_{a_{\alpha - \ell + 1}}, \ldots, \subdeg_{a_{\alpha - 2}}, \subdeg_{a_{\alpha - 1}} \le f$, $\subdeg_{a_{\alpha - \ell}} > (\gamma + 1) f$, and $\subdeg_{b_1} > (\gamma + \ell) f$.} \label{prg:case4psubw}

This case is symmetric to the previous one.

\paragraph*{Special Case 5: for some positive integer \boldmath$\ell \le \alpha / 2$, we have $\subdeg_{a_r} \le f$ for all $r \in \{0, \ldots, \ell - 1, \alpha - \ell + 1, \ldots, \alpha\}$, $\subdeg_{a_{\ell}} > (\gamma + 1) f$, $\subdeg_{b_1} \le (\gamma + \ell) f$, and $\beta = 3$.} \label{prg:case5subw}

Consider the following tree decomposition of $H$ (see \cref{Pa3c-upper-bound-split-by-middle-node}).
We define a bag
\[B_{1} \coloneqq \{ a_0, a_1, \ldots, a_{\ell}, b_1, b_2, c_1, c_2, \ldots, c_{\gamma} \},\]
a bag
\[B_{2} \coloneqq \{ a_{\ell}, b_2, b_3, c_1, c_2, \ldots, c_{\gamma} \},\]
and bags $B_{a_i} \coloneqq \{a_{\ell}, a_i, a_{i + 1}\}$ for $i \in [\ell + 1, \alpha - 1]$.
We connect these bags in a line: $B_{1} \edg B_{2} \edg B_{a_{\alpha-1}} \edg B_{a_{\alpha-2}} \edg \cdots \edg B_{a_{\ell+1}}$.
It is easy to see that it is indeed a valid tree decomposition of $H$.

To bound $\subf(B_1)$, we first note that $\subf(\{a_0, b_1\}) \le 1$ as $\subf$ is edge-dominated.
Due to \cref{lm:subw-simple-materializationsii}, we have $\subf(\{a_0, b_1, b_2\}) \le \subf(\{a_0, b_1\}) + \subdeg_{b_1} \le 1 + (\gamma + \ell) f$.
Furthermore, for each $i \in [\gamma]$, we have
\[\subf(\{a_0, b_1, b_2, c_1, \ldots, c_i\}) \le \subf(\{a_0, b_1, b_2, c_1, \ldots, c_{i - 1}\}) + \subdeg_{a_{0}} \le \subf(\{a_0, b_1, b_2, c_1, \ldots, c_{i - 1}\}) + f\]
due to \cref{lm:subw-simple-materializationsii}.
Moreover, for each $i \in [\ell]$, we have
\begin{align*}
    \subf(\{a_0, a_1, \ldots, a_i, b_1, b_2, c_1, \ldots, c_{\gamma}\}) &\le \subf(\{a_0, a_1, \ldots, a_{i - 1}, b_1, b_2, c_1, \ldots, c_{\gamma}\}) + \subdeg_{a_{i - 1}}\\
                                                        &\le \subf(\{a_0, a_1, \ldots, a_{i - 1}, b_1, b_2, c_1, \ldots, c_{\gamma}\}) + f
\end{align*}
due to \cref{lm:subw-simple-materializationsii}.
Therefore, by induction, we obtain $\subf(B_1) \le 1 + (\gamma + \ell) f + \gamma f + \ell f \le 1 + (2 \gamma + \alpha) f \le 2 - f$ due to the fact that $\ell \le \alpha / 2$ and \cref{lm:a3c-linear-bound}.

To bound $\subf(B_2)$, we first note that $\subf(\{b_2, b_3\}) \le 1$ as $\subf$ is edge-dominated.
Furthermore, for each $i \in [\gamma]$, we have
\[\subf(\{b_2, b_3, c_1, \ldots, c_i\}) \le \subf(\{b_2, b_3, c_1, \ldots, c_{i - 1}\}) + \subdeg_{b_3} \le \subf(\{b_2, b_3, c_1, \ldots, c_{i - 1}\}) + f\]
due to \cref{lm:subw-simple-materializationsii}.
Therefore, by induction and \cref{lm:subw-simple-materializationsi}, we obtain
\begin{align*}
    \subf(B_2) &\le \subf(\{ b_2, b_3, c_1, c_2, \ldots, c_{\gamma} \}) + \subf(\{a_{\ell}\})\\
               &\le (1 + \gamma f) + (1 - \subdeg_{a_{\ell}}) \le (1 + \gamma f) + (1 - (\gamma + 1) f)\\
               &= 2 - f.
\end{align*}

For each bag $B_{a_i}$, for $i \in [\ell + 1, \alpha - 1]$, we have $\subf(B_{a_i}) \le \subf(\{a_{i}, a_{i + 1}\}) + \subf(\{a_{\ell}\}) \le 1 + (1 - \subdeg_{a_{\ell}}) \le 2 - f$, due to \cref{lm:subw-simple-materializationsi} and as $\subf$ is edge-dominated.
Therefore, we arrive at $\fwidth(H) \le 2 - f$.

\paragraph*{Special Case 6: \boldmath$\subf(\{c_1, c_2, \ldots, c_{\gamma}\}) \le \submax - f$.} \label{prg:case6subw}

Consider the graph $H - \{\vmax, c_1, c_2, \ldots, c_{\gamma}\}$.
It is a tree, so it has a tree decomposition where each bag consists of two adjacent nodes.
If we add $\vmax, c_1, \ldots, c_{\gamma}$ to each bag of this tree decomposition, we obtain a tree decomposition of $H$.
For each bag $e \cup \{\vmax, c_1, c_2, \ldots, c_{\gamma}\}$ of this tree decomposition, we have
$\subf(e \cup \{\vmax, c_1, c_2, \ldots, c_{\gamma}\}) \le \subf(e) + \subf(\{\vmax\}) + \subf(\{c_1, c_2, \ldots, c_{\gamma}\}) \le 1 + (1 - \submax) + (\submax - f) = 2 - f$, due to \cref{lm:subw-simple-materializationsi} and as $\subf$ is edge-dominated.
Therefore, we arrive at $\fwidth(H) \le 2 - f$.

\medskip

Hence, from now on we may assume $\subf(\{c_1, c_2, \ldots, c_{\gamma}\}) > \submax - f$.

\paragraph*{Special Case 7: \boldmath$\subf(\{b_2, c_1, c_2, \ldots, c_{\gamma}\}) \le 1 - f$.} \label{prg:case7subw}

Consider the graph $H - \{b_2, c_1, c_2, \ldots, c_{\gamma}\}$.
It is a tree, so it has a tree decomposition where each bag consists of two adjacent nodes.
If we add $b_2, c_1, \ldots, c_{\gamma}$ to each bag of this tree decomposition, we obtain a tree decomposition of $H$.
For each bag $e \cup \{b_2, c_1, c_2, \ldots, c_{\gamma}\}$ of this tree decomposition, we have
$\subf(e \cup \{b_2, c_1, c_2, \ldots, c_{\gamma}\}) \le \subf(e) + \subf(\{b_2, c_1, c_2, \ldots, c_{\gamma}\}) \le 1 + (1 - f) = 2 - f$ due to \cref{lm:subw-simple-materializationsi}, and as $\subf$ is edge-dominated.
Therefore, we arrive at $\fwidth(H) \le 2 - f$.

\medskip

Hence, from now on we may assume $\subf(\{b_2, c_1, c_2, \ldots, c_{\gamma}\}) > 1 - f$.

\bigskip

This concludes the list of special cases. We now turn to the general case.

\subsubsection{General Case}

Let us recall which assumptions we already have. We have
\begin{align}
    \subdeg_{a_0}, \subdeg_{a_{\alpha}} &\le f, \label{eq:1} \\
    \submax &> (2 \gamma + 1) f, \label{eq:2} \\
    \subf(\{c_1, c_2, \ldots, c_{\gamma}\}) &> \submax - f, \label{eq:3} \\
    \subf(\{b_2, c_1, c_2, \ldots, c_{\gamma}\}) &> 1 - f, \label{eq:4}
\end{align}
and none of the special cases 4, 4', and 5 are applicable.

    We first bound $\subf(B_{a_0 \to c})$ for $B_{a_0 \to c} \coloneqq \{a_0, c_1, c_2, \ldots, c_{\gamma}\}$.
    We start from $\subf(\{a_0, c_1\}) \le 1$ as $\subf$ is edge-dominated.
    For each $c_i$ for $i \in [2, \gamma]$, we have $\subf(\{a_0, c_1, c_2, \ldots, c_i\}) \le \subf(\{a_0, c_1, c_2, \ldots, c_{i - 1}\}) + \subdeg_{a_{0}} \le \subf(\{a_0, c_1, c_2, \ldots, c_{i - 1}\}) + f$ due to \cref{lm:subw-simple-materializationsii}.
    We thus obtain by induction $\subf(B_{a_0 \to c}) \le 1 + (\gamma - 1) f$.
    Analogously, we have $\subf(B_{a_{\alpha} \to c}) \le 1 + (\gamma - 1) f$  for $B_{a_{\alpha} \to c} \coloneqq \{a_{\alpha}, c_1, c_2, \ldots, c_{\gamma}\}$.
    As $\subf$ is submodular, we have $2 + (2 \gamma - 2) f \ge \subf(B_{a_0 \to c}) + \subf(B_{a_{\alpha} \to c}) \ge \subf(B_{a_0 \to c} \cup B_{a_{\alpha} \to c}) + \subf(B_{a_0 \to c} \cap B_{a_{\alpha} \to c}) = \subf(B_{a_0 \to c} \cup B_{a_{\alpha} \to c}) + \subf(\{c_1, c_2, \ldots, c_{\gamma}\})$.
    As we have $\subf(\{c_1, c_2, \ldots, c_{\gamma}\}) > \submax - f$, we obtain $\subf(B_{a_0 \to c} \cup B_{a_{\alpha} \to c}) < 2 + (2 \gamma - 2) f - (\submax - f) \le 2 - kf$, where $k \coloneqq \left\lfloor \frac{\submax - (2 \gamma - 1) f}{f} \right\rfloor$ (and so $\submax < (2 \gamma + k) f$).
    Note that $k \ge 2$ as $\submax > (2 \gamma + 1) f$.
    As $\subf$ is monotone, we have
    \begin{equation}
        \subf(\{a_0, a_{\alpha}\}) \le \subf(B_{a_0 \to c} \cup B_{a_{\alpha} \to c}) < 2 - kf.\label{eq:5}
    \end{equation}

    We now bound $\subf(B_1)$ for $B_1 \coloneqq \{ b_0, b_1, b_2, c_1, c_2, \ldots, c_{\gamma} \}$.
    We start from $\subf(\{b_0, b_1\}) \le 1$ as $\subf$ is edge-dominated.
    Due to \cref{lm:subw-simple-materializationsii}, we have $\subf(\{b_0, b_1, b_2\}) \le \subf(\{b_0, b_1\}) + \subdeg_{b_1} \le 1 + \subdeg_{b_1}$.
    Finally, for each $i \in [\gamma]$, we have $\subf(\{b_0, b_1, b_2, c_1, \ldots, c_i\}) \le \subf(\{b_0, b_1, b_2, c_1, \ldots, c_{i - 1}\}) + \subdeg_{b_0} \le \subf(\{b_0, b_1, b_2, c_1, \ldots, c_{i - 1}\}) + f$ due to \cref{lm:subw-simple-materializationsii}.
    Therefore, by induction, we have
    \begin{equation}
        \subf(B_1) \le 1 + \gamma f + \subdeg_{b_1}. \label{eq:6}
    \end{equation}

    We now bound $\subf(B_2)$ for $B_2 \coloneqq \{ b_2, b_3, \ldots, b_{\beta}, c_1, c_2, \ldots, c_{\gamma} \}$.
    We start from $\subf(\{b_{\beta - 1}, b_{\beta}\}) \le 1$ as $\subf$ is edge-dominated.
    Furthermore, for each $i \in [\gamma]$, we have $\subf(\{b_{\beta - 1}, b_{\beta}, c_1, \ldots, c_i\}) \le \subf(\{b_{\beta - 1}, b_{\beta}, c_1, \ldots, c_{i - 1}\}) + \subdeg_{b_{\beta}} \le \subf(\{b_{\beta - 1}, b_{\beta}, c_1, \ldots, c_{i - 1}\}) + f$ due to \cref{lm:subw-simple-materializationsii}.
    Finally, for each $i \in [2, \beta - 2]$, we have $\subf(\{b_i, b_{i + 1}, \ldots, b_{\beta}, c_1, \ldots, c_{\gamma}\}) \le \subf(\{b_{i + 1}, \ldots, b_{\beta}, c_1, \ldots, c_{\gamma}\}) + \subdeg_{b_{i + 1}}$.
    Therefore, by induction, we have
    \begin{equation}    
        \subf(B_2) \le 1 + \gamma f + \subdeg_{b_{3 \dd \beta - 1}}. \label{eq:7}
    \end{equation}
   
    As $\subf$ is submodular, we have $2 + 2 \gamma f + \subdeg_{b_1} + \subdeg_{b_{3 \dd \beta - 1}} \ge \subf(B_1) + \subf(B_2) \ge \subf(B_1 \cup B_2) + \subf(B_1 \cap B_2) = \subf(B_1 \cup B_2) + \subf(\{b_2, c_1, c_2, \ldots, c_{\gamma}\})$.
    Due to \cref{eq:4,eq:6,eq:7}, we obtain
    $\subf(B_1 \cup B_2) < 1 + (2 \gamma + 1) f + \subdeg_{b_1} + \subdeg_{b_{3 \dd \beta - 1}} \le 2 - k'f$, where $k' \coloneqq \fabc - \left\lceil \frac{(2 \gamma + 1) f + \subdeg_{b_1} + \subdeg_{b_{3 \dd \beta - 1}}}{f} \right\rceil$.
    As $\subf$ is monotone, we have $\subf(\{a_0, a_{\alpha}\}) \le \subf(B_1 \cup B_2) < 2 - k'f$.

    Combining this fact with \cref{eq:5}, we get
    \begin{equation}
        \subf(\{a_0, a_{\alpha}\}) \le 2 - \ell f \label{eq:8}
    \end{equation}
    for $\ell \coloneqq \max\{k, k'\} \ge k \ge 2$.
    We build a tree decomposition of $H - \{a_1, a_2, \ldots, a_{\alpha - 1}\}$ such that $\subf(B) \le 2 - f$ for all bags $B$ of this tree decomposition, and there is a bag $\{a_0, a_{\alpha}\}$ in this tree decomposition.
    To this end, define a graph $H' \coloneqq (\{b_0, b_1, \ldots, b_{\beta}\}, E(H[\{b_0, b_1, \ldots, b_{\beta}\}]) \cup \{\{b_0, b_{\beta}\}\})$.
We consider $H'$ as an instance of Biased Cycle for $\subf'$ defined as the restriction of $\subf$ to $2^{\{b_0, b_1, \ldots, b_{\beta}\}}$.
    We claim that $\funwidth{\subf'}(H') \le 2 - f$.
    If $\ell > \frac{\beta + 1}{2}$, the inequality $\funwidth{\subf'}(H') \le 2 - f$ holds due to \cref{lm:biased-cycle-subw}.
    Otherwise, if $\ell \le \frac{\beta + 1}{2}$, due to \cref{lm:biased-cycle-subw}, we have
\begin{align*}
    \funwidth{\subf'}(H') &\le \max\{2 - f, 1 + \subdeg_{b_{\ell \dd \beta - \ell}} + (\ell - 1) f\}\\
                                                                &\le\max\{2 - f, 1 + (\beta - 2 \ell + 1) \cdot \submax + (\ell - 1) f\}\\
                                                                &\le\max\{2 - f, 1 + (\beta - 2 k + 1) \cdot \submax + (k - 1) f\}\\
                                                                &\le\max\{2 - f, 1 + (\beta - 2 k + 1) \cdot (k + 2 \gamma) f + (k - 1) f\}\\
                                                                &= 2 - f,
\end{align*}
where the third inequality holds as $k \le \ell$ and $\submax \ge f$, the fourth inequality is due to the definition of $k$, and the last equality holds due to \cref{a-b-c-quadratic-bound-b}.
As $\funwidth{\subf'}(H') \le 2 - f$, there is a tree decomposition of $H[\{b_0, b_1, \ldots, b_{\beta}\}]$, where one of the bags contains both $a_0$ and $a_{\alpha}$ such that $\subf(B) = \subf'(B) \le 2 - f$ for all bags $B$ of this tree decomposition.
By attaching $\{a_0, a_{\alpha}\}, \{a_0, a_{\alpha}, c_1\}, \ldots, \{a_0, a_{\alpha}, c_{\gamma}\}$ in a line to a bag of this tree decomposition that contains $a_0$ and $a_{\alpha}$, we get a tree decomposition $\mathcal{T}$ of $H - \{a_1, a_2, \ldots, a_{\alpha - 1}\}$, such that $\subf(B) \le 2 - f$ for all bags $B$ of this tree decomposition, and there is a bag $\{a_0, a_{\alpha}\}$ in this tree decomposition.

Define a graph $H'' \coloneqq (\{a_0, a_1, \ldots, a_{\alpha}\}, E(H[\{a_0, a_1, \ldots, a_{\alpha}\}]) \cup \{\{a_0, a_{\alpha}\}\})$.
We consider $H''$ as an instance of Biased Cycle for $\subf''$ defined as the restriction of $\subf$ to $2^{\{a_0, a_1, \ldots, a_{\alpha}\}}$.
As we will prove, we have $\funwidth{\subf''}(H'') \le 2 - f$.
Therefore, there is a tree decomposition of $H[\{a_0, a_1, \ldots, a_{\alpha}\}]$, such that $\subf(B) = \subf''(B) \le 2 - f$ for all bags $B$ of this tree decomposition as due to \cref{lm:subw-simple-materializationsii}, for any $i \in [\gamma]$ we have $\subf(\{a_0, a_{\alpha}, c_i\}) \le \subf(\{a_0, a_{\alpha}\}) + \subdeg_{a_0} \le (2 - \ell f) + f \le (2 - 2f) + f = 2 - f$, and one of the bags of this tree decomposition contains both $a_0$ and $a_{\alpha}$.
By attaching this bag to the bag containing $a_0$ and $a_{\alpha}$ of the already built tree decomposition $\mathcal{T}$ of $H - \{a_1, a_2, \ldots, a_{\alpha - 1}\}$, we obtain a tree decomposition of $H$, such that $\subf(B) \le 2 - f$ for all bags $B$ of this tree decomposition.
Therefore, $\fwidth(H) \le 2 - f$.

Hence, it remains to prove $\funwidth{\subf''}(H'') \le 2 - f$.
To this end, we use \cref{lm:biased-cycle-subw} and consider nine cases.
One can verify that it is a complete set of cases, covering all possible situations.

\subparagraph{Case 1: \boldmath$\alpha \le 2 \ell - 2$.}

In this case $\ell > \frac{\alpha + 1}{2}$, and \cref{lm:biased-cycle-subw} yields $\funwidth{\subf''}(H'') \le 2 - f$.

\subparagraph{Case 2: \boldmath$\alpha = 2 \ell - 1$.}

In this case \cref{lm:biased-cycle-subw} yields
\begin{align*}
    \funwidth{\subf''}(H'') &\le \max\{2 - f, 1 + \subdeg_{a_{\ell \dd \alpha - \ell}} + (\ell - 1) f\}\\
                            &= \max\{2 - f, 1 + (\ell - 1) f\}\\
                            &= \max\{2 - f, 1 + (\alpha - 1) / 2 \cdot f\}\\
                            &\le \max\{2 - f, 1 + (\alpha / 2 + \gamma - 1) \cdot f\}\\
                            &= 2 - f,
\end{align*}
where the last equality holds due to \cref{F-gets-big}.

\subparagraph{Case 3: \boldmath$\alpha \ge 2 \ell$ and $\subdeg_{a_r} > f$ for some $r \in \{0, \ldots, \ell - 1, \alpha - \ell + 1, \ldots, \alpha\}$.}

In this case the instance of Biased Cycle is $\ell$-reachable, and thus \cref{lm:biased-cycle-subw} yields $\funwidth{\subf''}(H'') \le 2 - f$.

\subparagraph{Case 4: \boldmath$\alpha \ge 2 \ell$, $\subdeg_{a_r} \le f$ for all $r \in \{0, \ldots, \ell - 1, \alpha - \ell + 1, \ldots, \alpha\}$, $\subdeg_{a_{\ell}} > (\gamma + 1) f$, and $\subdeg_{b_{\beta - 1}} > (\gamma + \ell) f$.}

This cannot happen as Special Case 4 (see \cref{prg:case4subw}) would be applicable.

\subparagraph{Case 4': \boldmath$\alpha \ge 2 \ell$, $\subdeg_{a_r} \le f$ for all $r \in \{0, \ldots, \ell - 1, \alpha - \ell + 1, \ldots, \alpha\}$, $\subdeg_{a_{\alpha - \ell}} > (\gamma + 1) f$, and $\subdeg_{b_1} > (\gamma + \ell) f$.}

This cannot happen as Special Case 4' (see \cref{prg:case4psubw}) would be applicable.

\subparagraph{Case 5: \boldmath$\alpha \ge 2 \ell$ and $\subdeg_{a_{\ell}} \le (\gamma + 1) f$.}

In this case \cref{lm:biased-cycle-subw} yields
\begin{align*}
    \funwidth{\subf''}(H'') &\le \max\{2 - f, 1 + \subdeg_{a_{\ell \dd \alpha - \ell}} + (\ell - 1) f\}\\
                            &\le \max\{2 - f, 1 + \subdeg_{a_{\ell}} + (\alpha - 2 \ell) \cdot \submax + (\ell - 1) f\}\\
                            &\le \max\{2 - f, 1 + (\gamma + (\alpha - 2 \ell) \cdot (2 \gamma + k) + \ell) f\}\\
                            &\le \max\{2 - f, 1 + (\gamma + (\alpha - 2 \ell) \cdot (2 \gamma + \ell) + \ell) f\}\\
                            &= 2 - f.
\end{align*}

The last equality holds because of \cref{a-b-c-quadratic-bound-a-t-2}. To see that the preconditions of this lemma hold, note that $\ell \ge k \ge 2$ and
\begin{align*}
    \ell \ge k' &= \fabc - \left\lceil \frac{(2 \gamma + 1) f + \subdeg_{b_1} + \subdeg_{b_{3 \dd \beta - 1}}}{f} \right\rceil\\
         &\ge \fabc - \left\lceil \frac{(2 \gamma + 1) f + (\beta - 2) \submax}{f} \right\rceil\\
         &\ge \fabc - \left\lceil \frac{(2 \gamma + 1) f + (\beta - 2) (2 \gamma + k) f}{f} \right\rceil\\
         &= \fabc - (2 \gamma + 1 + (\beta - 2) \cdot (2 \gamma + k))\\
         &\ge \fabc - (2 \gamma + 1 + (\beta - 2) \cdot (2 \gamma + \ell)).
\end{align*}

\subparagraph{Case 5': \boldmath$\alpha \ge 2 \ell$ and $\subdeg_{a_{\alpha - \ell}} \le (\gamma + 1) f$.}

This case is symmetric to the previous one.

\subparagraph{Case 6: \boldmath$\alpha \ge 2 \ell$, $\subdeg_{b_1} \le (\gamma + \ell) f$, $\subdeg_{b_{\beta - 1}} \le (\gamma + \ell) f$, and $\beta \ge 4$.}

In this case \cref{lm:biased-cycle-subw} yields
\begin{align*}
    \funwidth{\subf''}(H'') &\le \max\{2 - f, 1 + \subdeg_{a_{\ell \dd \alpha - \ell}} + (\ell - 1) f\}\\
                            &\le \max\{2 - f, 1 + (\alpha - 2 \ell + 1) \cdot \submax + (\ell - 1) f\}\\
                            &\le \max\{2 - f, 1 + ((\alpha - 2 \ell + 1) \cdot (2 \gamma + k) + \ell - 1) f\}\\
                            &\le \max\{2 - f, 1 + ((\alpha - 2 \ell + 1) \cdot (2 \gamma + \ell) + \ell - 1) f\}\\
                            &= 2 - f.
\end{align*}

The last equality holds because of \cref{a-b-c-quadratic-bound-a-t-0}. To see that the preconditions of this lemma hold, note that $\ell \ge k \ge 2$ and
\begin{align*}
    \ell &\ge k'\\
         &= \fabc - \left\lceil \frac{(2 \gamma + 1) f + \subdeg_{b_1} + \subdeg_{b_{3 \dd \beta - 1}}}{f} \right\rceil\\
         &\ge \fabc - \left\lceil \frac{(2 \gamma + 1) f + (\gamma + \ell) f + (\beta - 4) \cdot \submax + (\gamma + \ell) f}{f} \right\rceil\\
         &\ge \fabc - \left\lceil \frac{(2 \gamma + 1) f + (\gamma + \ell) f + (\beta - 4) \cdot (2 \gamma + k) f + (\gamma + \ell) f}{f} \right\rceil\\
         &= \fabc - (2 \gamma + 1 + (\gamma + \ell) + (\beta - 4) \cdot (2 \gamma + k) + (\gamma + \ell))\\
         &\ge \fabc - (2 \gamma + 1 + (\beta - 4) \cdot (2 \gamma + \ell) + 2 \cdot (\gamma + \ell)).
\end{align*}

\subparagraph{Case 7: \boldmath$\alpha \ge 2 \ell$, $\subdeg_{a_r} \le f$ for all $r \in \{0, \ldots, \ell - 1, \alpha - \ell + 1, \ldots, \alpha\}$, $\subdeg_{a_{\ell}} > (\gamma + 1) f$, $\subdeg_{b_1} \le (\gamma + \ell) f$, and $\beta = 3$.}

This cannot happen as Special Case 5 (see \cref{prg:case5subw}) would be applicable.

\bigskip

This concludes the case distinction that proves $\funwidth{\subf''}(H'') \le 2 - f$, which as we argued above shows $\fwidth(H) \le 2 - f$, thus finishing the proof of \cref{lm:pabc-subw}.

\subsection{Proof of \cref{subw-full-characterization}}
\label{sec:subw-P-graphs-combination}

Now we can put the pieces together to prove \cref{subw-full-characterization}.
We first calculate submodular width of all patterns in $\family$.

\begin{lemma}\label{lm:p-graph-subw}
    For any $(\alpha, \beta, \gamma) \in \Tfamily$, we have $\subw(P(\alpha, \beta, \gamma \times 2)) \le 2 - 1 / \fabc$.
\end{lemma}

\begin{proof}
    Let $(\alpha, \beta, \gamma) \in \Tfamily$.

    For $(\alpha, \beta, \gamma) = (1, 0, 0)$, we have $\subw(P(\alpha, \beta, \gamma \times 2)) \le 1 = 2 - 1 / 1$ due to \cref{edge-subw}, and $\fabc = 1$ due to \cref{lm:p1-time-complexity}.
    For $(\alpha, \beta, \gamma) = (2, 1, 0)$, we have $\subw(P(\alpha, \beta, \gamma \times 2)) \le 2 - 1 / 2$ due to \cref{cycle-subw}, and $\fabc = 2$ due to \cref{lm:c3-time-complexity}.
    For $(\alpha, \beta, \gamma) = (k - 2, 2, 0)$ with $k \ge 4$, we have $\subw(P(\alpha, \beta, \gamma \times 2)) \le 2 - 1 / \left\lceil k / 2 \right\rceil$ due to \cref{cycle-subw}, and $\fabc = \left\lceil k / 2 \right\rceil$ due to \cref{lm:ck-time-complexity}.

    Otherwise, we have $\alpha \ge \beta \ge 2$ and $\gamma \ge 1$.
    If $\alpha = \beta = 2$, then we have $\subw(P(\alpha, \beta, \gamma \times 2)) \le 2 - 1 / (\gamma + 2)$ due to \cref{biclique-subw}, and $\fabc = \gamma + 2$ due to \cref{lm:kk2-time-complexity}.
    If $\alpha > \beta = 2$, then $\fabc = 2 \gamma + 2 + \left\lfloor \frac{\alpha - 1}{2} \right\rfloor$ by \cref{a-c-time-complexity}, and \cref{lm:pa2c-subw} shows $\subw(P(\alpha, \beta, \gamma \times 2)) \le 2 - 1 / (2 \gamma + 2 + \left\lfloor \frac{\alpha - 1}{2} \right\rfloor)$.
    Finally, if $\alpha \ge \beta > 2$, \cref{lm:pabc-subw} shows $\subw(P(\alpha, \beta, \gamma \times 2)) \le 2 - 1 / \fabc$.
\end{proof}

Finally, we prove \cref{subw-full-characterization}.

\subwfullcharacterization*

\begin{proof}[Proof of \cref{subw-full-characterization}]
   By the same arguments as in the proof of \cref{full-characterization}, we have that each $H' \in D(H)$ is a pattern, i.e., $H'$ is connected and $|V(H')| \ge 2$.

    Assume that some $H' \in D(H)$ is not isomorphic to any graph in $\family$.
    Due to \cref{small-clemb-implies-p}, we have that either $\clemb(H') \ge 2$ or $H'$ is a $P$-graph.
    If $\clemb(H') \ge 2$, we have $\subw(H) \ge \subw(H') \ge \clemb(H') \ge 2$ due to \cref{lm:induced-minor-subw} and \cref{lm:clemb-bounded-by-subw}.
    Otherwise, assume that $H'$ is a $P$-graph, i.e, $H'$ is isomorphic to $P(\ell_1,\ldots,\ell_k)$ for some integers $\ell_1,\ldots,\ell_k \ge 1$.
	Suppose that $H'$ has at least three paths of length at least three. Then $H'$ (and thus $H$) has $P(3, 3, 3)$ as an induced minor, so  
    we obtain $\subw(H) \ge \subw(P(3, 3, 3)) \ge 2$ due to \cref{lm:induced-minor-subw} and \cref{lm:p333-subw}.
	So we can assume that $H'$ is a $P$-graph with at most two paths of length at least three and without clique separator. In this situation we show that $H'$ is isomorphic to a graph in $\family$, contradicting the choice of $H'$. 
	Note that $H'$ is isomorphic to $P(\ell_1,\ldots,\ell_k)$ with $\ell_3,\ldots,\ell_k \le 2$.
	If $k = 1$, then $H'$ is isomorphic to $P(\ell_1)$. This graph has a clique separator of size one, unless $\ell_1 = 1$. So $H'$ is isomorphic to $P(1) = P(1, 0, 0 \times 2) \in \family$.
	If $k = 2$, then $H'$ is isomorphic to $P(\ell_1,\ell_2)$, and thus to the cycle of length $\ell := \ell_1 + \ell_2$. If $\ell = 3$ then $H'$ is isomorphic to $P(2,1) = P(2,1,0 \times 2) \in \family$, and if $\ell \ge 4$ then $H'$ is isomorphic to $P(\ell-2,2) = P(\ell-2,2,0 \times 2) \in \family$.
	If $k \ge 3$, then none of $\ell_1,\ldots,\ell_k$ can be 1, as the path of length 1 would be a size-2 clique separator of $H'$. Therefore, $H'$ is isomorphic to $P(\alpha,\beta,\gamma \times 2)$ for some $\alpha \ge \beta \ge 2$ and $\gamma \ge 1$, which is in $\family$. 
    This contradicts the choice of $H'$, thus proving the first part of the theorem.

    \smallskip
	For the second part of the theorem, assume that every $H' \in D(H)$ is isomorphic to some pattern in $\family$, i.e., up to isomorphism we have $D(H) = \{P(\alpha_1, \beta_1, \gamma_1 \times 2), \ldots, P(\alpha_k, \beta_k, \gamma_k \times 2) \}$ for some $(\alpha_1, \beta_1, \gamma_1), \ldots, (\alpha_k, \beta_k, \gamma_k) \in \Tfamily$.
    Let $c(H) := 2 - 1 / \max_{i \in [k]} \funcc(\alpha_i, \beta_i, \gamma_i)$. By \cref{lm:p-graph-subw}, we have $\subw(P(\alpha_i, \beta_i, \gamma_i \times 2)) \le 2 - 1 / \funcc(\alpha_i, \beta_i, \gamma_i) \le c(H)$ for each $i$.
    Hence, by \cref{subw-from-D}, we have $\subw(H) \le c(H)$.
    On the other hand, by \cref{lm:p-graph-lower-bound} we have $\clemb(P(\alpha_i, \beta_i, \gamma_i \times 2)) \ge 2 - 1 / \funcc(\alpha_i, \beta_i, \gamma_i)$ for each $i$.
    Since each $P(\alpha_i, \beta_i, \gamma_i \times 2)$ is an induced subgraph of $H$, by \cref{clemb-of-induced-minors} we obtain $\clemb(H) \ge \max_i \clemb(P(\alpha_i, \beta_i, \gamma_i \times 2)) \ge 2 - 1/\max_i \funcc(\alpha_i, \beta_i, \gamma_i) = c(H)$.
    Therefore, we have $\subw(H) \le c(H) \le \clemb(H)$.
    Due to \cref{lm:clemb-bounded-by-subw}, we have $\clemb(H) \le \subw(H)$.
    Thus, we obtain $\subw(H) = c(H) = \clemb(H)$.
\end{proof}

\subsection{Computing $\min\{\subw(H), 2\}$}
\label{sec:computing_min_subw_2}

We now combine \cref{subw-full-characterization} with the algorithm for computing $D(H)$ from \cref{sec:appendix-decomposition-algorithm} to show that the number $\min\{\subw(H), 2\}$ can be computed in time $\poly(|V(H)|)$. Note that in this section we make an exception and do not assume that $H$ has constant size to show this interesting side result. To the best of our knowledge, this is the first polynomial-time algorithm computing submodular width up to a constant larger than $1$.

\begin{lemma}
    There is a $\poly(|V(H)|)$-time algorithm that for a given pattern $H$ computes the number $\min\{\subw(H), 2\}$.
\end{lemma}

\begin{proof}
    We first use \cref{d-set-algo} to compute $D(H)$.
    For each pattern from $D(H)$, in polynomial time we can decide whether it is isomorphic to some pattern from $\family$ or not: it should either be a single edge, a cycle, or there are exactly two nodes of degrees larger than two, and the remaining graph is a collection of paths between these nodes.
    If some element of $D(H)$ is not isomorphic to any pattern from $\family$, \cref{subw-full-characterization} implies that $\subw(H) \ge 2$, and thus $\min\{\subw(H), 2\} = 2$.
    Otherwise, $D(H) = \{P(\alpha_1, \beta_1, \gamma_1 \times 2), \ldots, P(\alpha_k, \beta_k, \gamma_k \times 2) \}$ for some $(\alpha_1, \beta_1, \gamma_1), \ldots, (\alpha_k, \beta_k, \gamma_k) \in \Tfamily$.
    All values $\alpha_i, \beta_i, \gamma_i$ can be computed in polynomial time.
    Given $\alpha_i, \beta_i, \gamma_i$, the value $\funcc(\alpha_i, \beta_i, \gamma_i)$ can be computed in constant time.
    Therefore, we can compute $2 - 1 / \max_{i \in [k]} \funcc(\alpha_i, \beta_i, \gamma_i)$ in polynomial time.
    Due to \cref{subw-full-characterization}, we have $\min\{\subw(H), 2\} = \subw(H) = 2 - 1 / \max_{i \in [k]} \funcc(\alpha_i, \beta_i, \gamma_i)$.
\end{proof}

\end{document}